%% file: thesis.tex
\pdfoutput=1 
\documentclass[twoside,
openright,
titlepage,numbers=noenddot,headinclude,
                footinclude=true,
cleardoublepage=plain,
abstractoff, 
                BCOR=5mm,
paper=letter,
pagesize,
fontsize=11pt,
                american%
                ]{scrreprt}

\input{classicthesis-config}

\usepackage[left=1in,right=1in,top=.9in,bottom=.9in]{geometry}

\titleformat{\chapter}[display]  {\relax}{\vspace*{-3\baselineskip}\makebox[\linewidth][r]{\color{halfgray}\chapterNumber\thechapter}}{0pt}  {\raggedright\spacedallcaps}[\normalsize\vspace*{.8\baselineskip}\titlerule]

\begin{document}
\frenchspacing
\raggedbottom
\selectlanguage{canadian} 
\pagenumbering{roman}
\pagestyle{plain}
\include{frontmatter/Titlepage}
\cleardoublepage\include{frontmatter/Abstract}
\cleardoublepage\include{frontmatter/Acknowledgments}

\pagestyle{scrheadings}
\cleardoublepage\include{frontmatter/Contents}
\cleardoublepage\include{frontmatter/Preface}
\pagenumbering{arabic}
\cleardoublepage
\include{chapters/ch1-introduction}
\include{chapters/ch2-ggmatching}
\include{chapters/ch3-tdmatching}
\include{chapters/ch7-strongmatching}
\include{chapters/ch4-bottleneckmatching}
\include{chapters/ch4-bottleneckmatching-bipartite}
\include{chapters/ch5-planematching}
\include{chapters/ch6-matchingpacking}

\cleardoublepage\include{frontmatter/Colophon}
\end{document}

%% file: classicthesis-config.tex
\PassOptionsToPackage{eulerchapternumbers,
				 pdfspacing,floatperchapter,
				 subfig,
                 eulermath,
                 parts}{classicthesis}										

\usepackage{ifthen}
\newboolean{enable-backrefs} 
\setboolean{enable-backrefs}{false} 

\newcommand{\myTitle}{Matchings in Geometric Graphs\xspace}

\newcommand{\myName}{Ahmad Biniaz\xspace}

\newcommand{\myUni}{Carleton University\xspace}

\newcommand{\myVersion}{version 4.1\xspace} 

    {\end{itemize}}
    {\end{enumerate}}

\newcounter{dummy} 
\providecommand{\mLyX}{L\kern-.1667em\lower.25em\hbox{Y}\kern-.125emX\@}

\PassOptionsToPackage{latin9}{inputenc}	
 \usepackage{inputenc}				

\usepackage{chapterbib}
 
\PassOptionsToPackage{canadian}{babel}   
 \usepackage{babel}					

\PassOptionsToPackage{square,numbers}{natbib}
 \usepackage{natbib}				

\PassOptionsToPackage{fleqn}{amsmath}		
 \usepackage{amsmath}

\usepackage{amssymb}
\usepackage{amsthm}

\newtheoremstyle{caps}
  {\topsep}   
  {\topsep}   
  {\itshape}  
  {0pt}       
  {}          
  {.}         
  {5pt plus 1pt minus 1pt} 
  {\spacedlowsmallcaps{\thmname{#1}\thmnumber{ #2}}\thmnote{ (#3)}} 

\theoremstyle{caps}
\newtheorem{definition}{Definition}
\newtheorem{theorem}[definition]{Theorem}
\newtheorem{lemma}[definition]{Lemma}
\newtheorem{observation}[definition]{Observation}
\newtheorem{corollary}[definition]{Corollary}
\newtheorem{conjecture}{Conjecture}

\newtheorem{proposition}{Proposition}

\newtheoremstyle{algo}
  {\topsep}   
  {\topsep}   
  {\upshape}  
  {0pt}       
  {}          
  {}          
  {\newline}  
  {\spacedlowsmallcaps{\thmname{#1}\thmnumber{ #2}}\thmnote{ (#3)}} 

\theoremstyle{algo}
\newtheorem{algorithm}{Algorithm}

\numberwithin{definition}{chapter} 
\numberwithin{figure}{chapter} 

\PassOptionsToPackage{T1}{fontenc} 
	\usepackage{fontenc}     
\usepackage[full]{textcomp} 
\usepackage{scrhack} 
\usepackage{xspace} 
\usepackage{mparhack} 
\usepackage{fixltx2e} 
\PassOptionsToPackage{printonlyused,smaller}{acronym}
	\usepackage{acronym} 
\usepackage{algpseudocode} 
\algtext*{EndWhile}
\algtext*{EndIf}
\algtext*{EndFor}
\usepackage{algorithm}
\usepackage{wasysym}
\usepackage{multirow}

\usepackage{tabularx} 
\usepackage{longtable} 
\usepackage{caption}
\usepackage{subcaption}  

\usepackage{wrapfig}

\usepackage{listings} 
\PassOptionsToPackage{pdftex,hyperfootnotes=false,pdfpagelabels}{hyperref}
\usepackage{hyperref}  
\PassOptionsToPackage{pdftex}{graphicx}
	\usepackage{graphicx} 

\graphicspath{{figures/}}

\newcommand{\backrefnotcitedstring}{\relax}
\newcommand{\backrefcitedsinglestring}[1]{(Cited on page~#1.)}
\newcommand{\backrefcitedmultistring}[1]{(Cited on pages~#1.)}
\ifthenelse{\boolean{enable-backrefs}}%
{%
		\PassOptionsToPackage{hyperpageref}{backref}
		\usepackage{backref} 
		   \renewcommand*{\backref}[1]{}  
		   \renewcommand*{\backrefalt}[4]{
		      \ifcase #1 %
		         \backrefnotcitedstring%
		      \or%
		         \backrefcitedsinglestring{#2}%
		      \else%
		         \backrefcitedmultistring{#2}%
		      \fi}%
}{\relax}    

\hypersetup{%
    colorlinks=true, linktocpage=true, pdfstartpage=3, pdfstartview=FitV,%
    breaklinks=true, pdfpagemode=UseNone, pageanchor=true, pdfpagemode=UseOutlines,%
    plainpages=false, bookmarksnumbered, bookmarksopen=true, bookmarksopenlevel=1,%
    hypertexnames=true, pdfhighlight=/O,
    urlcolor=webbrown, linkcolor=RoyalBlue, citecolor=webgreen, 
    pdftitle={\myTitle},%
    pdfauthor={\textcopyright\ 2015 \myName},%
    pdfsubject={},%
    pdfkeywords={},%
    pdfcreator={pdfLaTeX},%
    pdfproducer={LaTeX with hyperref and classicthesis}%
}   

\makeatletter
\@ifpackageloaded{babel}%
    {%
       \addto\extrasamerican{%
				}%
       \addto\extrasngerman{%
				}%
    }{\relax}
\makeatother

\usepackage{classicthesis} 

\makeatletter
\def\ttl@tocpart{%
  \def\ttl@a{\protect\numberline{\thepart}\@gobble{}}}
\makeatother

\KOMAoption{bibliography}{leveldown}

\newboolean{carleton-version}
\setboolean{carleton-version}{false}
\ifthenelse{\boolean{carleton-version}}{

\renewcommand{\href}[2]{#2}
\renewcommand{\url}[1]{#1}

}{
}

\newcommand{\MC}{M_\times}

\newcommand{\MOPT}{M^*}
\newcommand{\bt}{\lambda}
\newcommand{\btopt}{\lambda^*}
\newcommand{\cw}{cw}

\newcommand{\CD}[2]{D[#1,#2]}

\newcommand{\require}{\textbf{Input: }}
\newcommand{\ensure}{\textbf{Output: }}
\newcommand{\LS}[1]{\text{LS$(#1)$}}
\newcommand{\ConvexShape}{\bigtriangledown}
\newcommand{\kTD}[2]{$#1$\text{-}TD#2}
\newcommand{\kDT}[2]{$#1$\text{-}DG#2}
\newcommand{\kDG}[2]{$#1$\text{-}DG#2}
\newcommand{\kGG}[2]{$#1$\text{-}GG#2}

\newcommand{\kRNG}[2]{$#1$\text{-}RNG#2}
\newcommand{\Tm}{t_{min}}
\newcommand{\WS}[1]{\text{WS$(#1)$}}

\newcommand{\BC}{\text{\em BalancedCut}}

\newcommand{\RGB}{\text{\em RGB-matching}}
\newcommand{\RB}{\text{\em RB-matching}}
\newcommand{\CH}[1]{\text{$CH(#1)$}}
\newcommand{\HC}[1]{\text{$H(#1)$}}
\newcommand{\M}[2]{\text{$M_{#1}(#2)$}}
\newcommand{\Min}[2]{\text{\sf Min$(#1,#2)$}}
\newcommand{\Cut}[2]{\text{\sf Cut$(#1,#2)$}}
\newcommand{\Tangent}[2]{\text{\sf Tangent$(#1,#2)$}}

\newcommand{\LC}[1]{{\em left$(#1)$}}
\newcommand{\RC}[1]{{\em right$(#1)$}}
\newcommand{\Kn}[1]{K#1}
\newcommand{\e}[2]{e_{#1}{(#2)}}
\newcommand{\pmp}[1]{pmp(#1)}
\newcommand{\MST}{\text{\em MST}}
\newcommand{\SMGG}{\text{\em Strong-matching}}

\newcommand{\dg}[1]{deg(#1)}

\newcommand{\tr}[1]{t(#1)}

\newcommand{\trmin}{t}
\newcommand{\emin}{e}
\newcommand{\cmin}{D(u,v)}
\newcommand{\cone}[2]{C^{#1}_{#2}}
\newcommand{\hex}[2]{X(#1,#2)}
\newcommand{\tra}[2]{{#1}({#2})}
\newcommand{\G}[2]{G_{#1}({#2})}
\newcommand{\Inf}[1]{\text{Inf}(#1)}
\newcommand{\disc}{\Circle}
\newcommand{\discs}{\scriptsize \Circle}
\newcommand{\ddisc}{\ominus}
\newcommand{\ddiscs}{\scriptsize \ominus}
\newcommand{\sqr}{\Box}
\newcommand{\sqrs}{\scriptsize\Box}
\newcommand{\trid}{\bigtriangledown}
\newcommand{\trids}{\scriptsize\bigtriangledown}
\newcommand{\triu}{\bigtriangleup}
\newcommand{\trius}{\scriptsize\bigtriangleup}

\newcommand{\GUD}{G_{\bigtriangledown \hspace*{-8.2pt} \bigtriangleup}}

\newcommand{\SP}[2]{S^{\text{\tiny +#1}}_\text{\tiny #2}}
\newcommand{\SM}[2]{S^{\text{-\tiny #1}}_\text{\tiny #2}}



%% file: frontmatter/Titlepage.tex
\begin{titlepage}
  \begin{addmargin}[-1cm]{-3cm}
    \begin{center}
        \large  

        \hfill

        \vfill

        \begingroup
            \color{Maroon}\spacedallcaps{\myTitle} \\ \bigskip
        \endgroup

        \spacedlowsmallcaps{\myName}

        \vfill


        {\footnotesize A thesis submitted to the Faculty of Graduate and Post Doctoral Affairs\\
         in partial fulfillment of the requirements for the degree of} \\ \medskip
        Doctor of Philosophy {\footnotesize in} Computer Science \\ \bigskip
        \myUni \\ 
        {\footnotesize Ottawa, Ontario, Canada, 2016}

        \vfill                      

        {\footnotesize \textcopyright\ 2016 \myName}
    \end{center}  
  \end{addmargin}       
\end{titlepage}   

%% file: frontmatter/Abstract.tex
\phantomsection
\begingroup
\let\clearpage\relax
\let\cleardoublepage\relax
\let\cleardoublepage\relax

\chapter*{Abstract}
A {\em geometric graph} is a graph whose vertex set is a set of points in the plane and whose edge set contains straight-line segments. 
A {\em matching} in a graph is a subset of edges of the graph with no shared vertices. A matching is called {\em perfect} if it matches all the vertices of the underling graph. A {\em geometric matching} is a matching in a geometric graph. In this thesis, we study matching problems in various geometric graphs. Among the family of geometric graphs we look at complete graphs, complete bipartite graphs, complete multipartite graphs, Delaunay graphs, Gabriel graphs, and $\Theta$-graphs. The classical {\em matching problem} is to find a matching of maximum size in a given graph. We study this problem as well as some of its variants on geometric graphs. The {\em bottleneck matching problem} is to find a maximum matching that minimizes the length of the longest edge. The {\em plane matching problem} is to find a maximum matching so that the edges in the matching are pairwise
non-crossing. A geometric matching is {\em strong} with respect to a given shape $S$ if we can assign to each edge in the matching a scaled version of $S$ such that the shapes representing the edges are pairwise disjoint. The {\em strong matching problem} is to find the maximum strong matching with respect to a given shape. The {\em matching packing problem} is to pack as many edge-disjoint perfect matchings as possible into a geometric graph. We study these problems and establish lower and upper bounds on the size of different kinds of matchings in various geometric graphs. We also present algorithms for computing such matchings. Some of the presented bounds are tight, while the others need to be sharpened.
\endgroup			

\vfill

%% file: frontmatter/Acknowledgments.tex
\phantomsection


\bigskip

\begingroup
\let\clearpage\relax
\let\cleardoublepage\relax
\let\cleardoublepage\relax
\chapter*{Acknowledgments}

I would not have been able to write this thesis without the help and support of many people.

I wish to express my profound and sincere gratitude to Anil Maheshwari and Michiel Smid. They are all I could have wished for and more in my supervisors. Their continuous encouragement and belief in me were vital in the advancement of my graduate career, and made these past four years into a journey of exploration and excitement.

I also want to thank the other members and students of the Computational Geometry lab at Carleton University for making it such a nice place to work (and occasionally not work). 

Finally, I am deeply thankful to my family for all the things they have done for my sake. I am grateful to my parents for their unconditional support, constant encouragement, and unflinching faith in me throughout my entire life. I am also deeply thankful to my wife, Mohadeseh, for her unfailing love and support, and to my lovely son, Mehrdad.

Thank you!

\endgroup

%% file: frontmatter/Contents.tex
\refstepcounter{dummy}
\setcounter{tocdepth}{2} 
\setcounter{secnumdepth}{3} 
\manualmark
\markboth{\spacedlowsmallcaps{\contentsname}}{\spacedlowsmallcaps{\contentsname}}
\tableofcontents 
\automark[section]{chapter}
\renewcommand{\chaptermark}[1]{\markboth{\spacedlowsmallcaps{#1}}{\spacedlowsmallcaps{#1}}}
\renewcommand{\sectionmark}[1]{\markright{\thesection\enspace\spacedlowsmallcaps{#1}}}

\cleardoublepage

%% file: frontmatter/Preface.tex
\phantomsection
\addcontentsline{toc}{chapter}{\tocEntry{Preface}}
\begingroup
\let\clearpage\relax
\let\cleardoublepage\relax
\let\cleardoublepage\relax

\chapter*{Preface}
This thesis is in ``integrated article format'' in which each chapter is based on published papers, conference proceedings, or papers awaiting publication.



\begin{itemize}
 \item Chapter~\ref{ch:gg} considers the matching problems in Gabriel graphs. This chapter is a combination of results that have been published in the journal of Theoretical Computer Science~\cite{Biniaz2015-ggmatching-TCS} and results that have been presented in the 32nd European Workshop on Computational Geometry (EuroCG'16)~\cite{Biniaz2016-ewcg}. 

 \item Chapter~\ref{ch:td} considers matching problems in triangular-distance Delaunay graphs. This chapter presents the results that have been published in the journal of Computational Geometry: Theory and Applications~\cite{Biniaz2015-hotd-CGTA}. A preliminary version of these results have been published in the proceedings of the First International Conference on Algorithms and Discrete Applied Mathematics (CALDAM 2015)~\cite{Biniaz2015-hotd-CALDAM}.

 \item Chapter~\ref{ch:sm} considers the strong matching problem. This chapter is based on the results that have been accepted for publication in the journal of Computational Geometry: Theory and Applications, special issue in memoriam: Ferran Hurtado~\cite{Biniaz2015-strong}. 

 \item Chapter~\ref{ch:bm} considers the non-crossing bottleneck matching problem in a point set. The results of this chapter have been published in the journal of Computational Geometry: Theory and Applications~\cite{Abu-Affash2015-bottleneck}. 

 \item Chapter~\ref{ch:bmb} considers the non-crossing bottleneck matching problem in bipartite geometric graphs. The results of this chapter have appeared in the proceedings of the 26th Canadian Conference on Computational Geometry (CCCG 2014) \cite{Biniaz2014-bichromatic}. 

 \item Chapter~\ref{ch:pm} considers the non-crossing maximum matching problem in complete multipartite geometric graphs. The results of this chapter have been published in the proceedings of the 14th International Symposium on Algorithms and Data Structures (WADS 2015) \cite{Biniaz2015-RGB}.

 \item Chapter~\ref{ch:mp} considers the problem of packing edge-disjoint perfect matchings in to a complete geometric graph. The results of this chapter have been published in the journal of Discrete Mathematics {\&} Theoretical Computer Science~\cite{Biniaz2015-packing}. 
\end{itemize}
\endgroup
\bibliographystyle{abbrv}
\bibliography{../thesis}

%% file: chapters/ch1-introduction.tex
\chapter{Introduction}
\label{ch:summary}

A matching in a graph is a subset of edges of the graph with no shared vertices. A geometric graph is a graph whose vertex set is a set of points in the plane and whose edge set contains straight-line segments between the points.
In this thesis we consider matching problems in various geometric graphs. 

In this section we provide a summary of the thesis.
First we give a description of various geometric graphs that are considered in this thesis. We also provide the definition of different kinds of matching problems that we are looking at. Then we provide a brief review of the previous work on each problem, followed by a summary of the obtained results. Moreover, we present a brief description of the main ideas that are used to obtain these results.
\section{Geometric Graphs}
Let $P$ be a set of $n$ points in the plane. A {\em geometric graph}, $G=(P,E)$, is a graph whose vertex set is $P$ and whose edge set $E$ contains straight-line segments with endpoints in $P$. 
A {\em complete geometric
graph}, $K(P)$, is a geometric graph on $P$ that contains a straight-line edge between every pair
of points in $P$. Let $\{P_1,\dots,P_k\}$, where $k\ge 2$, be a partition of $P$. A {\em complete multipartite geometric
graph}, $K(P_1,\dots, P_k)$, is a geometric graph on $P$ that contains a straight-line edge between every point in $P_i$ and every point in $P_j$ where $i\neq j$. 

For two points $p$ and $q$ in $P$ let $D(p,q)$ be the closed disk having $\overline{pq}$ as diameter, and let $\bigtriangledown(p,q)$ (resp. $\bigtriangleup(p,q)$) be the smallest downward (resp. upward) equilateral triangle having $p$ and $q$ on its boundary. Moreover, let the {\em lune} $L(p,q)$ be the intersection of two disks of radius $|pq|$ that are centered at $p$ and $q$, where $|pq|$ denote the Euclidean distance between $p$ and $q$. See Figure~\ref{graphs-intro} for illustration of the following geometric graphs. Assume $P$ does not contain any four co-circular points. Then the {\em Delaunay triangulation} on $P$, denoted by $DT(P)$, has an edge between two points $p$ and $q$ if and only if there exists a circle that contains $p$ and $q$ on its boundary and whose interior does not contain any point of $P$. The {\em Gabriel graph} on $P$, denoted by $GG(P)$, is defined to have an edge between two
points $p$ and $q$ if and only if $D(p,q)$ is empty. The {\em relative neighborhood graph} on $P$, denoted by $RNG(P)$, has an edge $(p,q)$ if and only if $L(p,q)$ is empty. Similarly, the {\em triangular-distance Delaunay graph} on $P$, denoted by $TD(P)$, is defined to have an edge between two
points $p$ and $q$ if and only if $\bigtriangledown(p,q)$ is empty. The  theta-six graph on $P$, denoted by $\Theta_6(P)$, has an edge $(p,q)$ if and only if $\bigtriangledown(p,q)$ or $\bigtriangleup(p,q)$ is empty. The {\em  $L_\infty$-Delaunay graph} on $P$ is defined to have an edge between two
points $p$ and $q$ if and only if there exists an empty axis-parallel square that has $p$ and $q$ on its boundary. 

The {\em order-$k$ Gabriel graph} on $P$, denoted by $\kGG{k}$, has an edge $(p,q)$ if and only if $D(p,q)$ contains at most $k$ points of $P\setminus\{p,q\}$. Note that $\kGG{0}$ is equal to $GG(P)$. The $\kRNG{k}$ graph is defined similarly. The {\em order-$k$ triangular-distance Delaunay graph} on $P$, denoted by $\kTD{k}$, has an edge $(p,q)$ if and only if $\bigtriangledown(p,q)$ contains at most $k$ points of $P\setminus\{p,q\}$. Note that $\kTD{0}$ is equal to $TD(P)$.

Given a radius $r\ge 0$, the {\em disk graph} on $P$, denoted by $DG(P,r)$, has an edge $(p,q)$ if and only if the Euclidean distance between $p$ and $q$ is at most $r$. The {\em unit disk graph} on $P$, which is denoted by $UDG(P)$, is $DG(P,1)$. 

\begin{figure}[H]
 
$\begin{tabular}{ccc}
\multicolumn{1}{m{.31\textwidth}}{\centering\includegraphics[width=.25\textwidth]{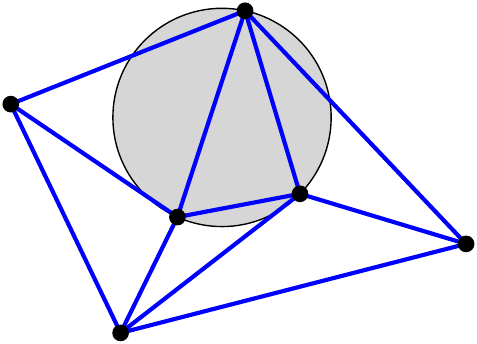}}
&\multicolumn{1}{m{.31\textwidth}}{\centering\includegraphics[width=.25\textwidth]{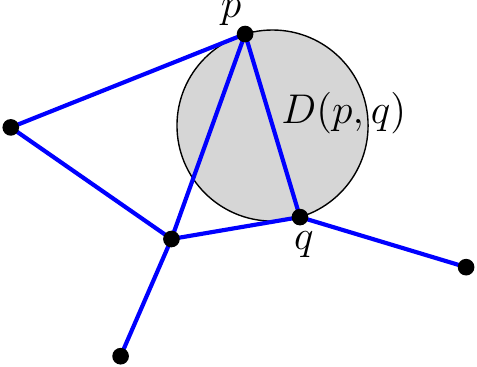}} &\multicolumn{1}{m{.31\textwidth}}{\centering\includegraphics[width=.25\textwidth]{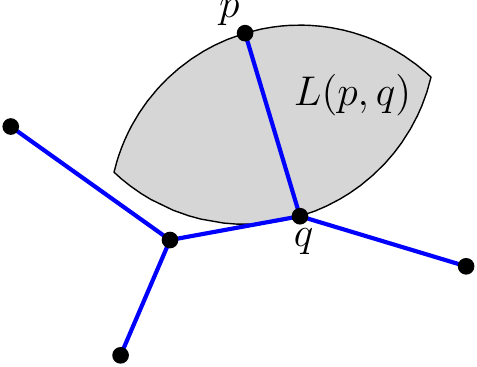}}
\\
 DT & GG (\kGG{0})& RNG
\\
\\
\multicolumn{1}{m{.31\textwidth}}{\centering\includegraphics[width=.25\textwidth]{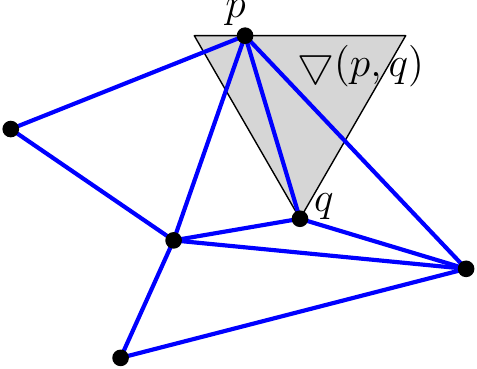}}
&\multicolumn{1}{m{.31\textwidth}}{\centering\includegraphics[width=.25\textwidth]{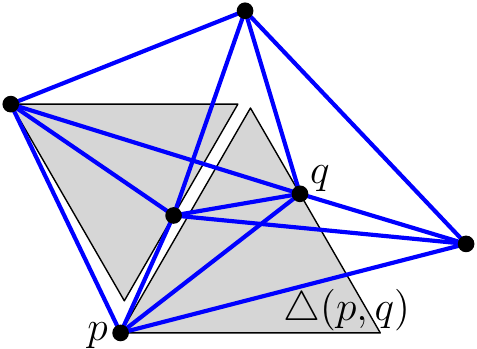}} &\multicolumn{1}{m{.31\textwidth}}{\centering\includegraphics[width=.25\textwidth]{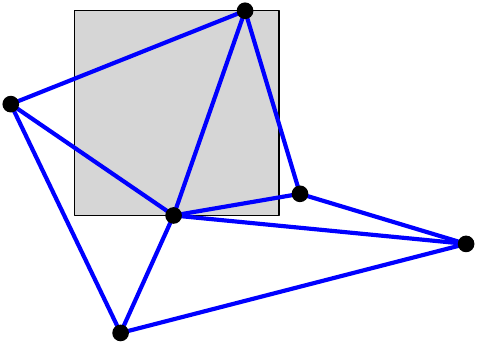}}
\\
 TD (\kTD{0})&$\Theta_6$  & $L_\infty$\text{-Delaunay}
\\
\\
\multicolumn{1}{m{.31\textwidth}}{\centering\includegraphics[width=.25\textwidth]{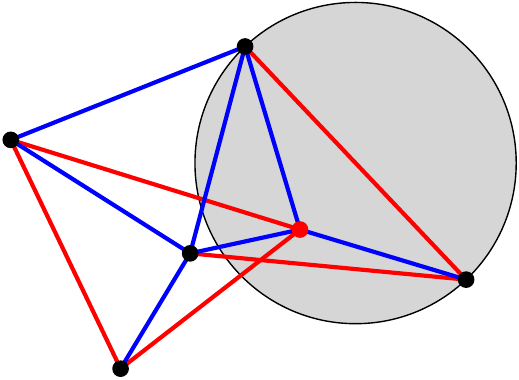}}
&\multicolumn{1}{m{.31\textwidth}}{\centering\includegraphics[width=.25\textwidth]{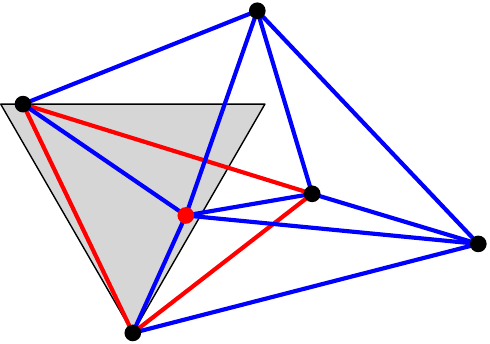}} &\multicolumn{1}{m{.31\textwidth}}{\centering\includegraphics[width=.25\textwidth]{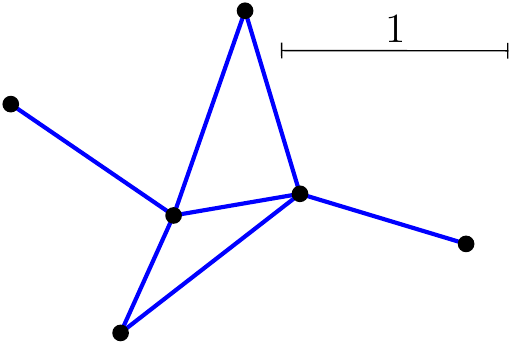}}
\\
\kGG{1} & \kTD{1} & UDG
\end{tabular}$
  \caption{Some geometric graphs.}

\label{graphs-intro}
\end{figure}

\section{Matching Problems}
Consider a geometric graph $G$. We say that two edges of $G$ {\em cross} if
they have a point in common that is interior to both edges. Two edges are {\em disjoint} if they have
no point in common. A {\em matching} in $G$ is a set of edges that do not share vertices. A {\em maximum
matching} is a matching with maximum cardinality. A {\em perfect matching} is a matching that matches all the vertices of $G$. When we talk about perfect matchings we assume that $G$ has an even number of points. A {\em plane matching} is a matching whose edges do not cross. A {\em bottleneck matching} is a maximum matching in $G$ in which the length of the longest edge is minimized. Given a matching $M$ and geometric shape $S$, we say that $M$ is a {\em strong matching} if we can assign to each edge $(p,q)$ in $M$ a scaled copy of the shape $S$ such that $p$ and $q$ are on the boundary of $S$ and all the shapes are pairwise disjoint (do not overlap). A {\em bottleneck biconnected subgraph} of $G$ is a 2-connected spanning subgraph of $G$ in which the length of the longest edge is minimized. A {\em bottleneck Hamiltonian cycle} is defined similarly.

Let $\{P_1,\dots,P_k\}$, where $k\ge 2$, be a partition of $P$. Assume the points in $P_i$ are colored $C_i$. $P$ is called {\em color-balanced} if no color is in strict majority, i.e., $|P_i|\le |P|/2$ for all $i\in\{1,\dots,k\}$. A matching in $K(P_1,\dots,P_k)$ is called a {\em colored matching}. A {\em balanced cut} is a line $\ell$ that partitions a color-balanced point set $P$ into two color-balanced point sets $Q_1$ and $Q_2$ such that $\max\{|Q_1|,|Q_2|\}\le \frac{2}{3}|P|$. See Figure~\ref{matchings-intro}(a). The ham-sandwich cut is a balanced cut; given a set of red and blue points, there exists a line that simultaneously bisects the red points and the blue points. A ham-sandwich cut can be computed in $O(n)$ time~\cite{Lo1994}.

\begin{figure}[H]
  \centering
$\begin{tabular}{cc}
\multicolumn{1}{m{.55\textwidth}}{\centering\includegraphics[width=.38\textwidth]{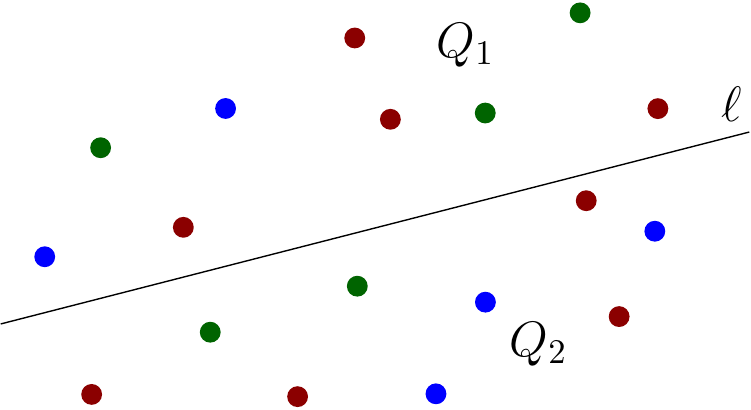}}
&\multicolumn{1}{m{.45\textwidth}}{\centering\includegraphics[width=.3\textwidth]{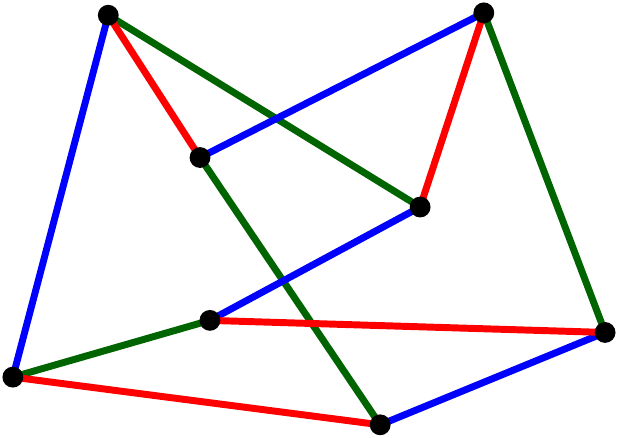}} 
\\
(a) & (b)
\end{tabular}$
  \caption{(a) A balanced cut. (b) A packing example of three edge-disjoint plane perfect matchings.}
\label{matchings-intro}
\end{figure}

We say that a set of subgraphs is {\em packed into} $K(P)$, if the subgraphs in the set are pairwise edge-disjoint (do not share any edge). In a packing problem, we ask for the largest number of subgraphs of a given type that can be packed into $K(P)$. Among all subgraphs of $K(P)$, plane perfect matchings, plane spanning trees, and plane spanning paths are of interest. See Figure~\ref{matchings-intro}(b). We say that a graph $G$ is {\em matching persistent} if by removing any perfect matching $M$ from $G$, the resulting graph, $G-M$, has a perfect
matching. We define the {\em matching persistency} of $G$ as the size of the smallest set $\mathcal{M}$ of edge-disjoint perfect matchings that can be removed from $G$ such that $G - \mathcal{M}$ does not have any perfect matching. The {\em plane matching persistency} of $G$ is defined similarly while each of the matchings in $\mathcal{M}$ is plane.
\addtocontents{toc}{\protect\setcounter{tocdepth}{0}}
\section{Previous Work}
\addtocontents{toc}{\protect\setcounter{tocdepth}{2}}
In a (strong) matching problem, usually we are looking for a (strong) matching of maximum cardinality in a specific geometric graph. Dillencourt~\cite{Dillencourt1990} proved that the Delaunay triangulation of $P$, denoted by $DT(P)$, contains a perfect matching. 
\'{A}brego et al. \cite{Abrego2004} proved that $DT(P)$ has a strong circle matching of size at least $\lceil(n-1)/8\rceil$.
\'{A}brego et al. \cite{Abrego2004, Abrego2009} proved that $L_\infty{\textrm -}DT(P)$ has a perfect matching and a strong matching of size at least $\lceil n/5\rceil$.

A problem related to the bottleneck matching problem is the {\em containment problem} in which we want to determine the smallest value $k^*$ for which \kGG{k^*}{} contains a bottleneck perfect matching of every point set $P$ (or of $K(P)$). Chang et al.~\cite{Chang1992} proved that a bottleneck perfect matching of $P$ is contained in $\kRNG{16}$. Since for a given point set $\kRNG{k}$ is a subgraph of $\kGG{k}$, this implies that $\kGG{16}$ contains a bottleneck matching of $P$. Chang et al.~\cite{Chang1992b} showed that a bottleneck biconnected subgraph of $P$ is contained in $\kRNG{1}$, and hence in $\kGG{1}$.
In~\cite{Abellanas2009} it is shown that a bottleneck Hamiltonian cycle for $P$ is contained in $\kGG{15}$. Recently, Kaiser et
al.~\cite{Kaiser2015} improved the bound by showing that $\kGG{10}$ contains a bottleneck Hamiltonian cycle of $P$. This implies that $\kGG{10}$ has a perfect matching. They also showed that $\kGG{5}$ may not contain any bottleneck Hamiltonian cycle of $P$.

The bottleneck plane matching problem is to compute a plane perfect matching in a geometric graph in which the length of the longest edge is minimized. 
Let $\{R,B\}$ be a partition of $P$ such that $|R|=|B|$. Carlsson et al.~\cite{Carlsson2010} showed that the bottleneck plane perfect matching problem in $K(R,B)$ is NP-hard. When $R\cup B$ is in general position, they presented an $O(n^4\log n)$-time algorithm.
Abu-Affash et al.~\cite{Abu-Affash2014} showed that the bottleneck plane perfect matching problem in $K(P)$ is NP-hard and does not admit a PTAS (Polynomial Time Approximation Scheme), unless P=NP. Let $\lambda^*$ be the length of the longest edge in an optimal bottleneck plane perfect matching in $K(P)$. Abu-Affash et al.~\cite{Abu-Affash2014} presented a polynomial-time algorithm that computes a plane perfect matching whose edges have length at most $2\sqrt{10}\lambda^*$.

A plane perfect matching in $K(R,B)$, with $|R|=|B|$, can be computed optimally in $\Theta(n\log n)$ time by recursively applying the ham-sandwich theorem~\cite{Lo1994}. Let $\{P_1,\dots,P_k\}$ with $k\ge 2$, be a partition of $P$. Aichholzer et al.~\cite{Aichholzer2010} and Kano et al.~\cite{Kano2013} showed that $K(P_1,\dots,P_k)$ has a plane colored matching if and only if $P$ is color-balanced. Kano et al.~\cite{Kano2013} presented an algorithm that computes such a matching in $O(n^2\log n)$ time.

In a packing problem we usually ask for the maximum number of edge-disjoint structures that can be packed into a complete graph. Aichholzer et al.~\cite{Aichholzer2014-packing} considered the problem of packing plane spanning trees and plane Hamiltonian paths into $K(P)$. They showed that at least $\sqrt{n}$ plane spanning trees can be packed into $K(P)$. They also showed how to pack 2 plane Hamiltonian paths into $K(P)$.
\addtocontents{toc}{\protect\setcounter{tocdepth}{0}}
\section{Obtained Results}
\addtocontents{toc}{\protect\setcounter{tocdepth}{2}}
We looked at the maximum matching problem in Gabriel graphs and TD-Delau\-nay graphs. We showed that $\kGG{0}$, $\kGG{1}$, and $\kGG{2}$ has a matching of size at least $\frac{n-1}{4}$, $\frac{2(n-1)}{5}$, and $\frac{n-1}{2}$, respectively~\cite{Biniaz2015-ggmatching-TCS}. We  also showed that $\kTD{0}$, $\kTD{1}$, $\kTD{2}$ contains a matching of size at least $\frac{n-1}{3}$, $\frac{2(n-1)}{5}$, and $\frac{n-1}{2}$, respectively~\cite{Biniaz2015-hotd-CALDAM, Biniaz2015-hotd-CGTA}.
As for Euclidean bottleneck matching, we showed that a bottleneck perfect matching of $P$ is contained in $\kGG{9}$, but $\kGG{8}$ may not have any bottleneck perfect matching~\cite{Biniaz2015-ggmatching-TCS, Biniaz2016-ewcg}. As for triangular-distance bottleneck matching, we showed that a bottleneck perfect matching of $P$ is contained in $\kTD{6}$, but $\kTD{5}$ may not have any bottleneck perfect matching~\cite{Biniaz2015-hotd-CALDAM, Biniaz2015-hotd-CGTA}. In~\cite{Babu2013-WALCOM, Babu2014} we considered the maximum matching problem in $\kTD{0}$.

In~\cite{Biniaz2015-strong} we considered the strong matching problem. We showed that if shape $S$ is a diametral-disk (a disk whose diameter is a line segment between two input points), a downward equilateral triangle, and an axis-aligned square then $P$ admits a strong matching of size at least $\frac{n -1}{17}$, $\frac{n - 1}{9}$, and $\frac{n - 1}{4}$, respectively. If both downward and upward equilateral triangles are allowed, we compute a strong matching of size at least $\frac{n - 1}{4}$.

For the plane matching problem in complete multipartite geometric graphs, we show how to compute a plane maximum matching in $K(P_1,\allowbreak\dots,\allowbreak P_k)$ in $\Theta(n\log n)$ time~\cite{Biniaz2015-RGB}. We also extended this results to the case where the points are in the interior of a simple polygon~\cite{Biniaz2015-geodesic, Biniaz2016-CGTA-geodesic}. 

In~\cite{Abu-Affash2015-bottleneck} we considered the plane bottleneck matching problem in $K(P)$, which is NP-hard. We show how to compute a plane bottleneck matching of size at least $\frac{n}{5}$ in $K(P)$ with edges of length at most $\lambda^*$ in $O(n \log^2 n)$ time, where $\lambda^*$ is the length of the longest edge in an optimal bottleneck plane perfect matching. We also presented an $O(n \log n)$-time approximation algorithm that computes a plane matching of size at least $\frac{2n}{5}$ in $K(P)$ whose edges have length at most $(\sqrt{2}+\sqrt{3})\lambda^*$. In~\cite{Biniaz2014-bichromatic} we considered the bottleneck matching problem in complete bipartite geometric graphs.

In~\cite{Biniaz2015-packing} we considered the problem of packing plane edge-disjoint matchings into $K(P)$. We proved that at least $\lfloor\log_2n\rfloor-1$ edge-disjoint plane perfect matchings can be packed into $K(P)$. 

\section{Summary of the thesis}
In this section we describe the obtained results in more detail. Moreover we give a brief description of the approaches used to obtain these results.   
\subsection{Maximum Matchings}
In~\cite[Chapter~\ref{ch:gg}]{Biniaz2015-ggmatching-TCS,Biniaz2016-ewcg} and \cite[Chapter~\ref{ch:td}]{Biniaz2015-hotd-CGTA, Biniaz2015-hotd-CALDAM} we consider the maximum matching and bottleneck matching containment problems in $\kGG{k}$ and $\kTD{k}$ graphs. In~\cite{Biniaz2015-ggmatching-TCS} we showed that $\kGG{0}$ has a matching of size at least $\frac{n-1}{4}$ and this bound is tight. We also proved that $\kGG{1}$ has a matching of size at least
$\frac{2(n-1)}{5}$ and $\kGG{2}$ has a perfect matching. As for Euclidean bottleneck matching, we showed that a bottleneck matching of $P$ is contained in $\kGG{9}$, but $\kGG{8}$ may not have any bottleneck matching; see~\cite{Biniaz2015-ggmatching-TCS, Biniaz2016-ewcg}. In~\cite{Biniaz2016-ewcg} we also showed that $\kGG{7}$ may not contain any bottleneck Hamiltonian cycle of $P$.
In~\cite{Babu2014, Babu2013-WALCOM} we have shown that $\kTD{0}$ contains a matching of size at least $\frac{n-1}{3}$ and this bound is tight (this result is not included in the thesis). In~\cite{Biniaz2015-hotd-CGTA, Biniaz2015-hotd-CALDAM} we proved that $\kTD{1}$ has a matching of size at least $\frac{2(n-1)}{5}$ and $\kTD{2}$ has a perfect matching. As for triangular-distance bottleneck matching, we showed that a bottleneck matching of $P$ is contained in $\kTD{6}$, but $\kTD{5}$ may not have any bottleneck matching. We also showed that a bottleneck biconnected subgraph of $P$ is contained in $\kTD{1}$. In addition, we showed that $\kTD{7}$ contains a bottleneck Hamiltonian cycle of $P$ while $\kTD{5}$ may not contain any. Tables~\ref{table1-intro} and~\ref{table2-intro} summarize the results.
\begin{table}[htb]
\centering
\caption{Lower bounds on the size of maximum (strong) matchings.}
\vspace{3pt}
\label{table1-intro}
    \begin{tabular}{|c|c|c||c|c|}
         \hline
             Graph 	& Max matching & Reference& Strong matching &Reference \\  \hline \hline  
	      {$DT$}& 	${\lfloor \frac{n}{2}\rfloor}$&\cite{Dillencourt1990}& $\lceil\frac{n-1}{8}\rceil$&\cite{Abrego2004}\\\hline
	      
	    \multirow{2}{*}{$L_\infty{\textrm -}DT$}& \multirow{2}{*}{$\lfloor \frac{n}{2}\rfloor$} &\multirow{2}{*}{\cite{Abrego2004, Abrego2009}}&
	    $\lceil\frac{n}{5}\rceil$ & \cite{Abrego2004, Abrego2009}       \\ 
		& &	& $\lceil\frac{n-1}{4}\rceil$& \cite[Section~\ref{sm:infty-Delaunay-section}]{Biniaz2015-strong}\\ 
	    \hline
	      {$\kGG{0}$} &${\lceil \frac{n-1}{4}\rceil}$&\cite[Section~\ref{gg:max-matching-section}]{Biniaz2015-ggmatching-TCS} &$\lceil\frac{n-1}{17}\rceil$&\cite[Section~\ref{sm:Gabriel-section}]{Biniaz2015-strong}\\  
	      {$\kGG{1}$} &${\lceil \frac{2n-2}{5}\rceil}$&\cite[Section~\ref{gg:max-matching-section}]{Biniaz2015-ggmatching-TCS} &-&-\\  
	      {$\kGG{2}$} &${\lfloor \frac{n}{2}\rfloor}$&\cite[Section~\ref{gg:max-matching-section}]{Biniaz2015-ggmatching-TCS} &-&-\\ \hline              
	    {$\kTD{0}$} &${\lceil \frac{n-1}{3}\rceil}$&\cite{Babu2014} &$\lceil\frac{n-1}{9}\rceil$&\cite[Section~\ref{sm:half-theta-six-section}]{Biniaz2015-strong}\\ 
	    {$\kTD{1}$} &${\lceil \frac{2n-2}{5}\rceil}$&\cite[Section~\ref{td:matching2}]{Biniaz2015-hotd-CGTA} &-&-\\
	    {$\kTD{2}$} &${\lfloor \frac{n}{2}\rfloor}$&\cite[Section~\ref{td:matching2}]{Biniaz2015-hotd-CGTA} &-&-\\ 
	    {$\Theta_6$} &$\lceil \frac{n-1}{3}\rceil$& \cite{Babu2014} &$\lceil\frac{n-1}{4}\rceil$&\cite[Section~\ref{sm:theta-six-section}]{Biniaz2015-strong}\\ 
	    \hline
    \end{tabular}
\end{table}

\begin{paragraph}{Our approach:}In order to provide a lower bound on the size of a maximum matching in $\kGG{k}$ and $\kTD{k}$, we first give a lower bound on the number of components that result after removing a set $S$ of vertices from $\kGG{k}$ and $\kTD{k}$. Then we use the following theorem of Tutte and Berge.
For a graph $G=(V,E)$ and $S\subseteq V$, let $G-S$ be the subgraph obtained from $G$ by removing all vertices in $S$, and let $o(G-S)$ be the number of odd components in $G-S$, i.e., connected components with an odd number of vertices. In a graph $G$, the {\em deficiency}, $\text{def}_G(S)$, is $o(G-S)-|S|$. Let $\text{def}(G)=\max_{S\subseteq V}{\text{def}_G(S)}$.
\end{paragraph}
\begin{theorem}[Tutte-Berge formula~\cite{Berge1958}] 
\label{gg:Berge-intro} 
The size of a maximum matching in $G$ is $$\frac{1}{2}(n-\mathrm{def}(G)).$$
\end{theorem}

\begin{table}
\centering
\caption{Bottleneck structures containment in $\kGG{k}$ and $\kTD{k}$.}
\vspace{3pt}
\label{table2-intro}
    \begin{tabular}{|c|c|c||c|c|c|}
         \hline
             Distance 	& Bott. struct.&$\notin$ & Reference& $\in$ &Reference \\  \hline \hline 
	      	      
	    \multirow{3}{*}{Euclidean}
	  &biconnected subgraph& {$\kGG{0}$} &\cite{Chang1992b}& $\kGG{1}$& \cite{Chang1992b} \\ 
	  &matching& {$\kGG{8}$} &\cite[Section~\ref{8GGsec}]{Biniaz2015-ggmatching-TCS}& $\kGG{9}$& \cite[Section~\ref{9GGsec}]{Biniaz2016-ewcg} \\
	  &cycle& {$\kGG{7}$} &\cite[Section~\ref{7GGsec}]{Biniaz2016-ewcg}& $\kGG{10}$& \cite{Kaiser2015} \\
	    \hline

	  \multirow{3}{*}{Triangular}&biconnected subgraph& 	$\kTD{0}$&\cite[Section~\ref{td:biconnected-section}]{Biniaz2015-hotd-CGTA}& $\kTD{1}$&\cite[Section~\ref{td:biconnected-section}]{Biniaz2015-hotd-CGTA}\\ 
	&matching& 	$\kTD{5}$&\cite[Section~\ref{td:bottleneck-matching-section}]{Biniaz2015-hotd-CGTA}& $\kTD{6}$&\cite[Section~\ref{td:bottleneck-matching-section}]{Biniaz2015-hotd-CGTA}\\
	&cycle& 	$\kTD{5}$&\cite[Section~\ref{td:Hamiltonicity}]{Biniaz2015-hotd-CGTA}& $\kTD{7}$&\cite[Section~\ref{td:Hamiltonicity}]{Biniaz2015-hotd-CGTA}\\
	    \hline
    \end{tabular}
\end{table}

In order to show that $\kGG{9}$ and $\kTD{6}$ contain bottleneck matchings, we do the following.
For a matching $M$ we define the {\em length sequence} of $M$, \LS{M}, as the sequence containing the lengths of the edges of $M$ in non-increasing order. A matching $M_1$ is said to be less than a matching $M_2$ if \LS{M_1} is lexicographically smaller than \LS{M_2}. Then, we show that the matching with the minimum length sequence is contained in $\kGG{10}$ and $\kTD{6}$.

\subsection{Strong Matchings}
A geometric matching is {\em strong} with respect to a given shape $S$ if we can assign to each edge in the matching a scaled version of $S$ with the endpoints of the edge on the boundary, such that the shapes representing the edges are pairwise disjoint.

In~\cite[Chapter~\ref{ch:sm}]{Biniaz2015-strong} we considered the strong matching problem with respect to a given geometric object $S$. We proved that
if $S$ is a diametral-disk (a disk whose diameter is a line segment between two input points), then $P$ (and hence $GG(P)$) has a strong matching of size at least $\lceil\frac{n -1}{17}\rceil$, and if $S$ is a downward equilateral triangle, then $P$ (and hence $TD(P)$) has a strong matching of size at least $\lceil\frac{n - 1}{9}\rceil$. In case both downward and upward equilateral triangles are allowed, we compute a strong matching of size at least $\lceil\frac{n - 1}{4}\rceil$ in $P$ (and hence in $\Theta_6(P)$). If $S$ is an axis-aligned square, then we compute a strong matching of size at least $\lceil\frac{n - 1}{4}\rceil$ in $P$ (and hence in $L_\infty{\textrm -}DT(P)$); this improves the previous bound of $\lceil\frac{n}{5}\rceil$. The results are summarized in Table~\ref{table1-intro}.

\begin{figure}[H]
  \centering
$\begin{tabular}{ccc}
\multicolumn{1}{m{.31\textwidth}}{\centering\includegraphics[width=.22\textwidth]{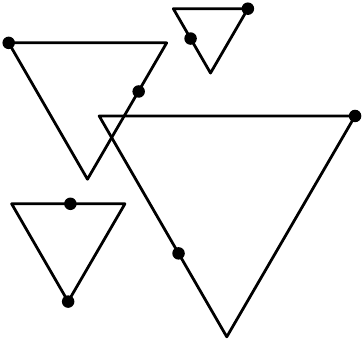}}
&\multicolumn{1}{m{.31\textwidth}}{\centering\includegraphics[width=.22\textwidth]{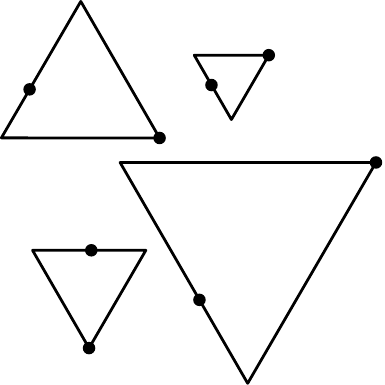}} &\multicolumn{1}{m{.31\textwidth}}{\centering\includegraphics[width=.22\textwidth]{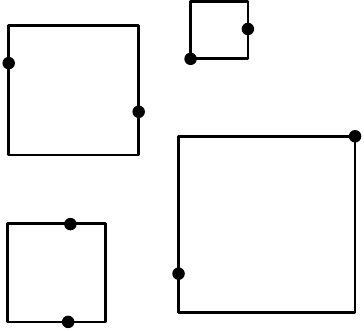}}
\\
(a) & (b)& (c)
\end{tabular}$
  \caption{Point set $P$ and (a) a non-strong perfect matching in $\G{\trids}{P}$ (two shapes overlap), (b) a perfect strong matching in $\GUD(P)$, and (c) a perfect strong matching in $\G{\sqr}{P}$.}
\label{sm:strong-example-intro}
\end{figure}

\begin{paragraph}{Our approach:}To compute a strong matching with diametral-disks and equilateral triangles we present the following algorithm that is depicted in Algorithm~\ref{sm:alg1-intro}. For a given point set $P$ and a geometric shape $S\in\{\ddisc, \trid\}$, the edge weighted geometric graph $\G{S}{P}$ on $P$ is defined to have an edge between two points of $P$ if and only if there exists an empty shape $S$ having the two points on its boundary. The weight of each edge $e$ is equal to the area of $S(e)$, where $S(e)$ is a smallest scaled copy of $S$ containing $e$. Let $T$ be a minimum spanning tree of $G_S(P)$. For each edge $e\in T$ we denote by $T(e^+)$ the set of all edges in $T$ whose weight is at least $w(e)$. Moreover, we define the {\em influence set} of $e$, as the set of all edges in $T(e^+)$ whose representing shapes overlap with $S(e)$, i.e.,
$$\Inf{e}=\{e': e'\in T(e^+), S(e')\cap S(e)\neq \emptyset\}.$$
\end{paragraph}
Note that $\Inf{e}$ is not empty, as $e\in \Inf{e}$. Consequently, we define the {\em influence number} of $T$ to be the maximum size of a set among the influence sets of edges in $T$, i.e.,
$$\Inf{T}=\max\{|\Inf{e}|: e\in T\}.$$

Algorithm~\ref{sm:alg1-intro} receives $\G{S}{P}$ as input and computes a strong matching in $\G{S}{P}$ as follows. The algorithm starts by computing a minimum spanning tree $T$ of $\G{S}{P}$, where the weight of each edge is equal to the area of its representing shape. Then it initializes a forest $F$ by $T$, and a matching $\mathcal{M}$ by an empty set. Afterwards, as long as $F$ is not empty, the algorithm adds to $\mathcal{M}$, the smallest edge $\emin$ in $F$, and removes the influence set of $e$ from $F$. Finally, it returns $\mathcal{M}$.
\begin{algorithm}[H]                  
\caption{\SMGG$(\G{S}{P})$}          
\label{sm:alg1-intro} 
\begin{algorithmic}[1]
      \State $T\gets \MST(G_S(P))$
      \State $F\gets T$
      \State $\mathcal{M}\gets \emptyset$
      \While {$F\neq \emptyset$}
	  \State $\emin\gets $ smallest edge in $F$
	  \State $\mathcal{M}\gets \mathcal{M}\cup \{\emin\}$
	  \State $F\gets F - \Inf{\emin}$
	  \EndWhile
    \State \Return $\mathcal{M}$
\end{algorithmic}
\end{algorithm}
\begin{theorem}
\label{sm:GS-thr-intro}
Given a set $P$ of $n$ points in the plane and a shape $S\in\{\ddisc, \trid\}$, Algorithm~\ref{sm:alg1-intro} computes a strong matching of size at least $\lceil\frac{n-1}{\emph{Inf}(T)}\rceil$ in $\G{S}{P}$, where $T$ is a minimum spanning tree of $\G{S}{P}$. 
\end{theorem}

In order to compute a strong matching with diametral-disks and equilateral triangles, we show that if $S$ is $\ddisc$ then $\Inf{T}\le 17$, and if $S$ is $\trid$, then $\Inf{T}\le 9$.

The bounds for the strong matching when $S$ is an square is proved by induction on the area of the smallest axis-aligned square containing $P$. The bounds for the strong matching when $S$ is allowed to be upward or downward equilateral triangle is proved similarly.
\subsection{Bottleneck Matchings}
A {\em bottleneck matching} in a graph $G$ is a maximum matching in which the length of the longest edge is minimized. A {\em plane bottleneck matching} is a bottleneck matching that is non-crossing.
 
\begin{table}[H]
\centering
\caption{Approximating plane bottleneck matching.}
\vspace{3pt}
\label{table3-intro}
    \begin{tabular}{|c|c|c|c|}
         \hline
             time complexity&plane bottleneck   & size of matching& Reference  \\ \hline\hline
             $O(n^{1.5}\sqrt{\log n})$  & $2\sqrt{10}\lambda^*$ & $n/2$& \cite{Abu-Affash2014} \\
             $O(n \log^2 n)$& $\lambda^*$ & $n/5$ &\cite[Section~\ref{bm:bottleneck-five-over-two}]{Abu-Affash2015-bottleneck}\\
             $O(n\log n)$ & $(\sqrt{2}+\sqrt{3})\lambda^*$ & $2n/5$&\cite[Section~\ref{bm:bottleneck-five-over-four}]{Abu-Affash2015-bottleneck}\\
         \hline
    \end{tabular}
\end{table}

Problems related to computing bottleneck plane matchings in geometric graphs are considered in \cite[Chapter~\ref{ch:bm}]{Abu-Affash2015-bottleneck} and \cite[Chapter~\ref{ch:bmb}]{Biniaz2014-bichromatic}. Computing a bottleneck plane matching in $K(P)$ is NP-hard. In~\cite{Abu-Affash2015-bottleneck}, we present an $O(n \log n)$-time algorithm that computes a plane matching of size at least $\frac{n-1}{5}$ in a connected disk graph. Using this algorithm we obtain a bottleneck plane matching of size at least $\frac{n}{5}$ in $K(P)$ with edges of length at most $\lambda^*$ in $O(n \log^2 n)$ time, where $\lambda^*$ is the length of the longest edge in an optimal bottleneck
plane perfect matching in $K(P)$. We also presented an $O(n \log n)$-time approximation algorithm that computes a plane matching of size at least $\frac{2n}{5}$ in $K(P)$ whose edges have length at most $(\sqrt{2}+\sqrt{3})\lambda^*$. Table~\ref{table3-intro} summarizes the results.
In~\cite{Biniaz2014-bichromatic} we considered the bottleneck plane matching problem in $K(R,B)$ with $|R|=|B|$. This problem is NP-hard; we provided polynomial-time algorithms that compute exact solutions for some special cases of the problem. When $R\cup B$ is in convex position we solve this problem in $O(n^3)$ time that improves upon the previous algorithm of~\cite{Carlsson2010} by a factor of $n\log n$. If $R\cup B$ is on the boundary of a circle we solve this problem in $O(n\log n)$ time. If the points in $R$ are on a line and the points in $B$ are on one side of the line, we solve this problem in $O(n^4)$ time. Table~\ref{table4-intro} summarizes the results.

\begin{paragraph}{Our approach:}In order to compute a bottleneck plane matching in $K(P)$ we do the following. First we present an $O(n \log n)$-time algorithm that computes a plane matching of size at least $\frac{n-1}{5}$ in any connected disk graph. Using this algorithm we obtain a bottleneck plane matching of size at least $\frac{n}{5}$ in $K(P)$ with edges of length at most $\lambda^*$ in $O(n \log^2 n)$ time.
\end{paragraph}

We show that every minimum spanning tree on $K(P)$ is a minimum spanning tree of $UDG(P)$. Monma et al. \cite{Monma1992} proved that every set of points in the plane admits a minimum spanning tree of degree at most five that can be computed in $O(n\log n)$ time. Now we present an algorithm that extracts a plane matching $M$ from a minimum spanning tree $T$ of $UDG(P)$ with vertices of degree at most five. We define the {\em skeleton tree}, ${T'}$, as the tree obtained from $T$ by removing all its leaves; see Figure \ref{bm:empty-skeleton-fig-intro}. Clearly ${T'} \subseteq T \subseteq UDG(P)$. For clarity we use $u$ and $v$ to refer to the leaves of $T$ and $T'$ respectively. In addition, let $v$ and $v'$, respectively, refer to the copies of a vertex $v$ in $T$ and $T'$. In each step, pick an arbitrary leaf $v'\in T'$. By the definition of ${T'}$, it is clear that the copy of $v'$ in $T$, i.e. $v$, is connected to vertices $u_1,\dots, u_k$, for some $1\le k \le 4$, that are leaves of $T$ (if $T'$ has one vertex then $k\le5$). Pick an arbitrary leaf $u_i$ and add $(v, u_i)$ as a matched pair to $M$. For the next step we update $T$ by removing $v$ and all its adjacent leaves. We also compute the new skeleton tree and repeat this process. 
In the last iteration, $T'$ is empty and we may be left with a tree $T$ consisting of one single vertex or one single edge. If $T$ consists of one single vertex, we disregard it, otherwise we add its only edge to $M$. $M$ has size at least $\frac{n-1}{5}$ and can be computed in $O(n\log n)$ time.
\begin{figure}[ht]
  \centering
    \includegraphics[width=0.6\textwidth]{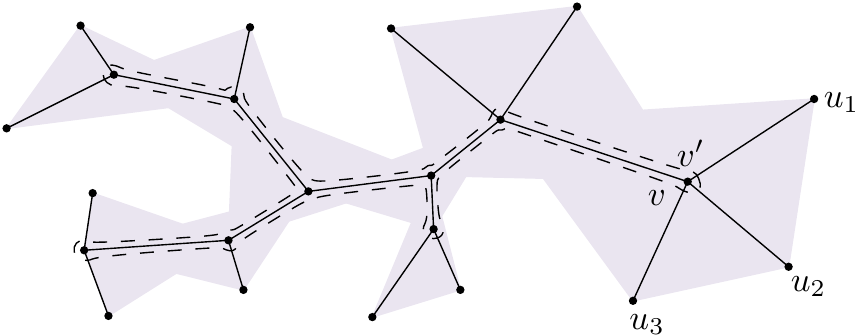}
  \caption{Minimum spanning tree $T$ with union of empty convex hulls. The skeleton tree $T'$ is surrounded by dashed line, and $v'$ is a leaf in $T'$. }
\label{bm:empty-skeleton-fig-intro}
\end{figure}

Then we compute a (possibly crossing) bottleneck perfect matching $\MC$ of $K(P)$ using the algorithm in \cite{Chang1992}. Let $\bt_{\MC}$ denote the length of the bottleneck edge in $\MC$. It is obvious that the bottleneck length of any plane perfect matching is not less than $\bt_{\MC}$. Therefore, $\btopt\ge\bt_{\MC}$. We consider a ``unit'' disk graph $DG(\bt_{\MC}, P)$ over $P$, in which there is an edge between two vertices $p$ and $q$ if $|pq|\le\bt_{\MC}$. Note that $DG(\bt_{\MC}, P)$ is not necessarily connected. Let $G_1,\dots,G_k$ be the connected components of $DG(\bt_{\MC}, P)$. For each component $G_i$, consider a minimum spanning tree $T_i$ of degree at most five. We show how to extract from $T_i$ a plane matching $M_i$ of proper size $\frac{n}{5}$ and length $\btopt$. We show how to compute $M$ in $O(n\log^2 n)$ time.

To compute a plane matching of size at least $\frac{2n}{5}$ in $K(P)$ with edges of length at most $(\sqrt{2}+\sqrt{3})\lambda^*$, we do the following.
Let $DT(P)$ denote the Delaunay triangulation of $P$. Let the edges of $DT(P)$ be, in sorted order of their lengths, $e_1, e_2, \dots$. Initialize a forest $F$ consisting of $n$ trees, each one being a single node for one point of $P$. Run Kruskal's algorithm on the edges of $DT(P)$ and terminate as soon as every tree in $F$ has an even number of nodes. Let $e_l$ be the last edge that is added by Kruskal's algorithm. Observe that $e_l$ is the longest edge in $F$. Denote the trees in $F$ by $T_1 ,\dots, T_k$ and for $1\le i \le k$, let $P_i$ be the vertex set of $T_i$ and let $n_i=|P_i|$. Then we prove the following lemma. 

\begin{lemma}
\label{longest-edge-intro}
 $\btopt \ge |e_l|.$
\end{lemma}

By Lemma \ref{longest-edge-intro} the length of the longest edge in $F$ is at most $\btopt$. For each $T_i\in F$, where $1\le i\le k$, we compute a plane matching $M_i$ of $P_i$ of size at least $\frac{2n_i}{5}$ with edges of length at most $(\sqrt{2}+\sqrt{3})\btopt$ and return $\bigcup_{i=1}^{k} M_i$. 
This gives a plane matching of $P$ of size at least $\frac{2}{5}n$ with bottleneck at most $(\sqrt{2}+\sqrt{3})\btopt$. This matching can be computed in $O(n\log n)$ time.

The plane bottleneck matchings for the special cases of $R\cup B$ in $K(R,B)$ are computed by dynamic programming.
\subsection{Plane Matchings}
There has been much research to extend the well-known ham-sandwich theorem\textemdash that partitions a two colored point set\textemdash to more colors; see~\cite{Bereg2015, Bereg2012, Kano2013}. For a color-balanced point set $P$ in the plane we showed how to compute a balanced cut in linear time~\cite[Chapter~\ref{ch:pm}]{Biniaz2015-RGB}. Moreover, by applying balanced cuts recursively, we computed a plane matching in $K(P_1,\dots,P_k)$ in $\Theta(n\log n)$ time. We have also extended this notion for points that are in the interior of a simple polygon~\cite{Biniaz2015-geodesic} (this extension is not included in the thesis).

\begin{table}
\centering
\caption{Plane colored matchings.}
\vspace{3pt}
\label{table4-intro}
\begin{tabular}{|c|c|c|c|}
\hline
{Problem}& Point set &{Time complexity}& {Reference} \\ \hline\hline 
\multirow{5}{*}{\begin{tabular}[c]{@{}c@{}}bottleneck 2-colored\\plane perfect matching\end{tabular} } 
&general position & NP-hard & \cite{Carlsson2010} \\ 
& convex position & $O(n^4\log n)$	 & \cite{Carlsson2010}\\
& convex position & $O(n^3)$ & \cite[Section~\ref{bmb:convex}]{Biniaz2014-bichromatic}\\ 
& on circle & $O(n\log n)$ & \cite[Section~\ref{bmb:circle}]{Biniaz2014-bichromatic}\\
& one color on a line & $O(n^4)$ & \cite[Section~\ref{bmb:line}]{Biniaz2014-bichromatic}\\  \hline\hline
\multirow{2}{*}{plane maximum matching}&2-colored 	& $\Theta(n\log n)$ & \cite{Hershberger1992} \\ 
& $k$-colored	& $\Theta(n\log n)$	 & \cite[Section~\ref{algorithm-section}]{Biniaz2015-RGB}\\ 
\hline
\end{tabular}
\end{table}

\begin{paragraph}{Our approach:}Let $\{R, B\}$ be a partition of $P$ such that $|R| = |B|$. Assume the points in $R$ are colored red and the points in $B$ are colored blue.
A plane perfect matching in $K_n(R,B)$ can be computed in $\Theta(n\log n)$ time by recursively applying the following Ham Sandwich Theorem.
\end{paragraph}

\begin{theorem}[Ham Sandwich Theorem]
\label{ham-sandwich-thr-intro}
 For a point set $P$ in general position in the plane that is partitioned into sets $R$ and $B$, there exists a line that simultaneously bisects $R$ and $B$.
\end{theorem}

Let $\{P_1,\dots, P_k\}$, with $k > 2$, be a partition of $P$. A necessary and sufficient for the existence of a plane perfect matching in $K_n(P_1,\dots,P_k)$ is obtained by the following theorem.

\begin{theorem}[Aichholzer et al.~\cite{Aichholzer2010}, and Kano et al.~\cite{Kano2013}]
\label{Aichholzer-thr-intro}
Let $k\ge 2$ and consider a partition $\{P_1,\dots,P_k\}$ of a point set $P$, where $|P|$ is even. Then, $K_n(P_1,\dots,P_k)$ has a plane colored perfect matching if and only if $P$ is color-balanced. 
\end{theorem}

Based on that, we first prove the existence of a balanced cut in a color-balanced point set.

\begin{figure}[htb]
  \centering
\includegraphics[width=.4\columnwidth]{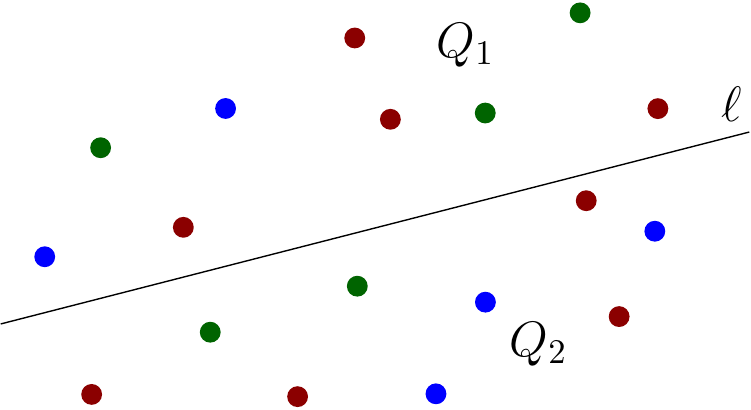}
\caption{Illustrating the balanced cut theorem. }
\label{balanced-cut-fig-intro}
\end{figure}

\begin{theorem}[Balanced Cut Theorem]
\label{balanced-cut-thr-intro}
Let $P$ be a color-balanced point set of $n\ge 4$ points in general position in the plane. In $O(n)$ time we can compute a line $\ell$ such that
\begin{enumerate}
  \item $\ell$ does not contain any point of $P$.
  \item $\ell$ partitions $P$ into two point sets $Q_1$ and $Q_2$, where
      \begin{enumerate}
	\item both $Q_1$ and $Q_2$ are color-balanced,
	\item both $Q_1$ and $Q_2$ contains at most $\frac{2}{3}n+1$ points.
	\item if $|P|$ is even, then both $|Q_1|$ and $|Q_2|$ are even.
      \end{enumerate}
\end{enumerate}
\end{theorem}

Then, by applying the Balanced Cut Theorem recursively, we can compute a plane perfect matching in matching in $K_n(P_1,\dots,P_k)$ in $\Theta(n\log n)$ time.
\subsection{Matching Packing}
Recall that in a matching packing problem, we ask for the largest number of matchings that can be packed into $K(P)$. A plane matching packing problem is to pack non-crossing matchings to $K(P)$.

In~\cite[Chapter~\ref{ch:mp}]{Biniaz2015-packing} we consider the problem of packing perfect matchings into $K(P)$. Let $n$ be the number of points in $P$, and assume that $n$ is an even number. We proved that if $P$ is in general position, then at least $\lfloor\log_2n\rfloor-1$ plane perfect matchings can be packed into $K(P)$. Moreover, we show that for some point set $P$ in general position, no more than $\lceil\frac{n}{3}\rceil$ can be packed into $K(P)$. If $P$ is in convex position we show that the maximum number of plane perfect matchings that can be packed to $K(P)$ is $\frac{n}{2}$, and $\frac{n}{2}-1$ when $P$ is in regular wheel configuration. As for matching persistencey, we showed that the matching persistency of $K(P)$ is $\frac{n}{2}$ if $n\equiv 2 \mod 4$, and $\frac{n}{2}+1$ if $n\equiv 0 \mod 4$. As for plane matchings we showed that if $P$ is in convex position then the plane matching persistency of $K(P)$ is 2. We also show the existence of a set of points in general position with plane matching persistence of at least 3. 

\begin{figure}[htb]
  \centering
\includegraphics[width=.33\columnwidth]{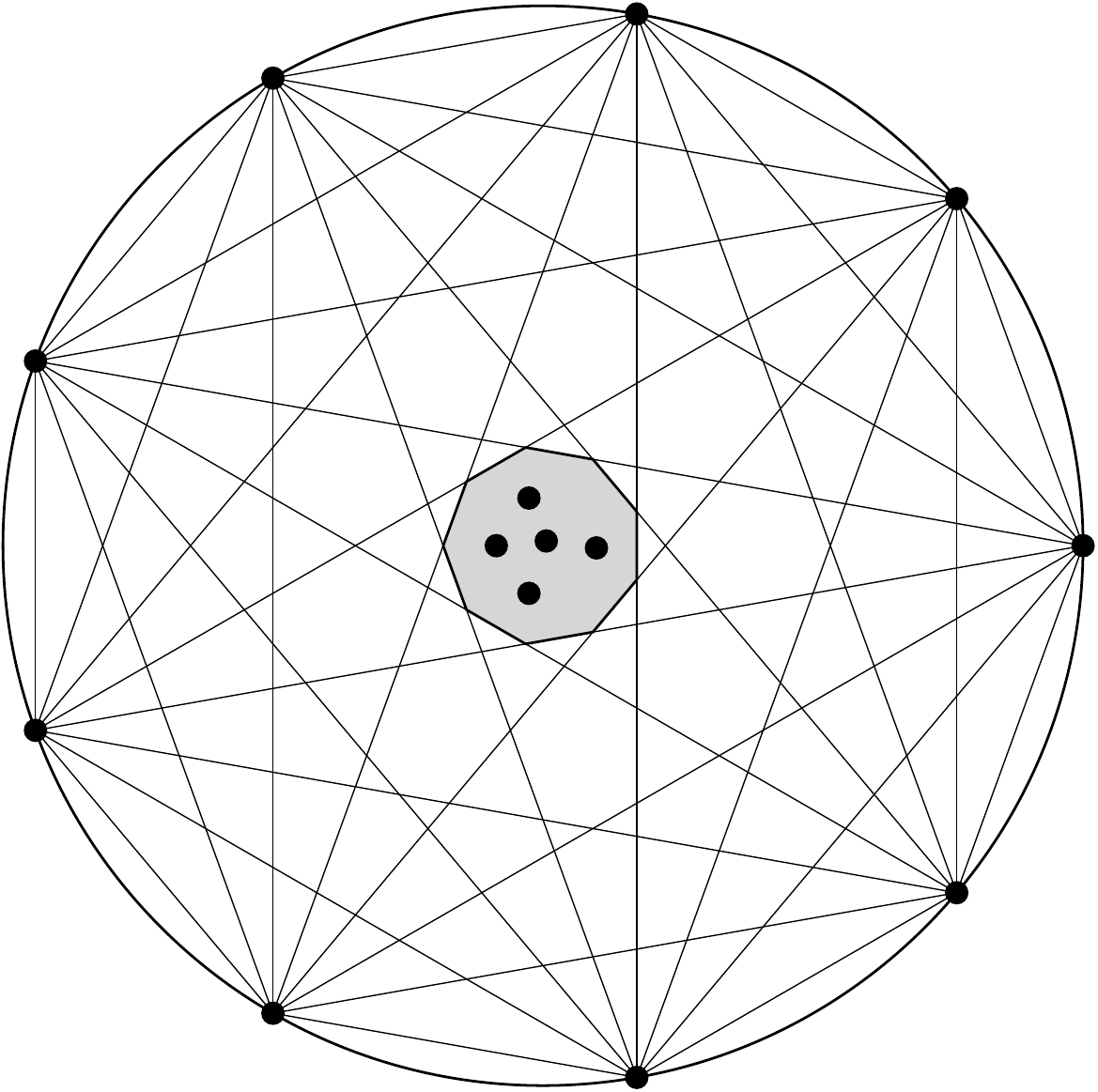}
  \caption{About $1/3$ of the points are in the middle and the boundary has an odd number of points.}
\label{n-over-3-fig-intro}
\end{figure}

\begin{paragraph}{Our approach:}In order to show the upper bound, we provide an example, as depicted in Figure~\ref{n-over-3-fig-intro}, that does not contain more than $\lceil\frac{n}{3}\rceil$ edge-disjoint plane perfect matchings.
\end{paragraph}

\begin{figure}[H]
  \centering
  \includegraphics[width=.6\columnwidth]{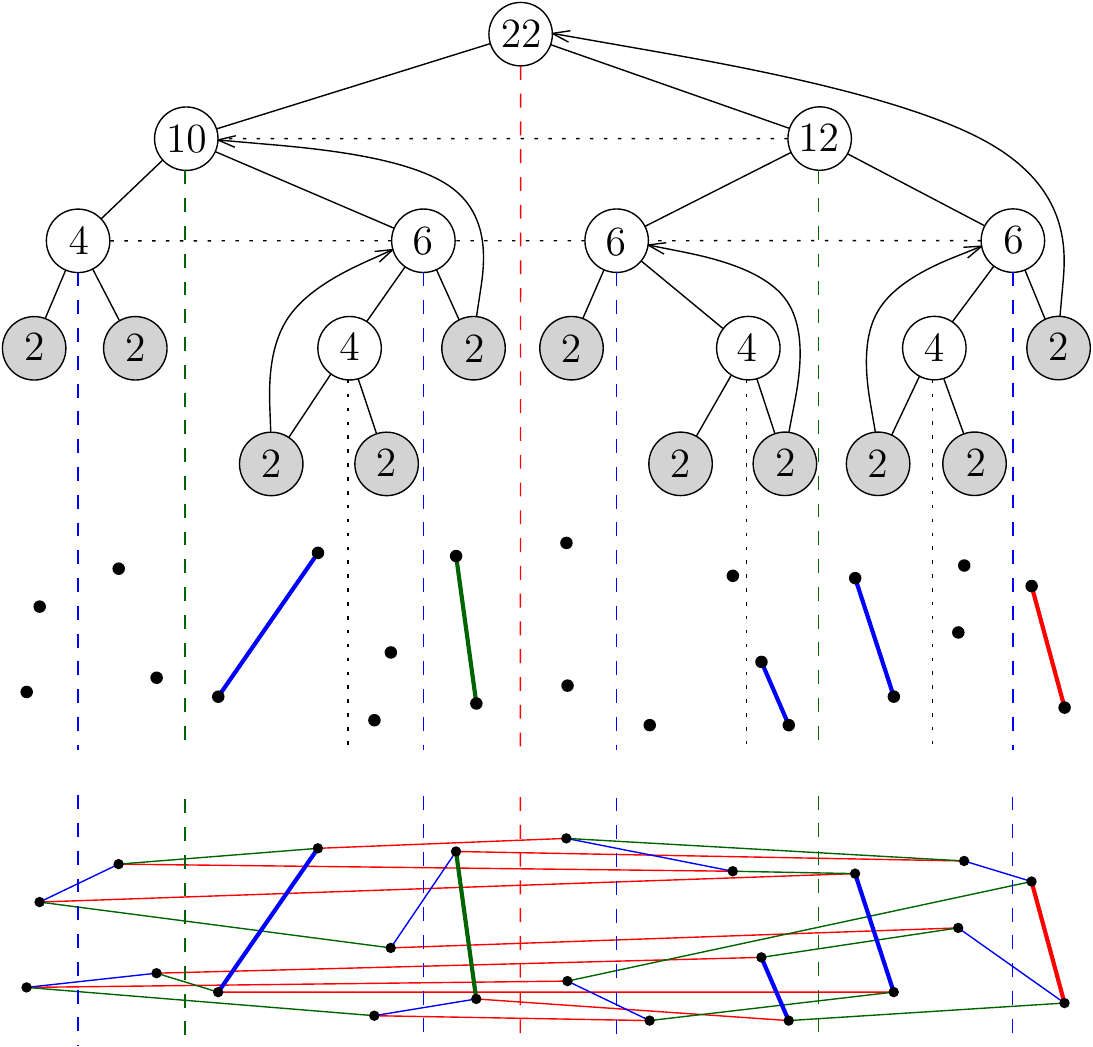}
 \caption{The points in $P$ are assigned, in pairs, to the leaves of $T$, from left to right. Internal nodes store the number of points in their subtree. The three edge disjoint plane perfect matchings are shown in red, green, and blue.}
  \label{matching-example-fig-intro}
\end{figure}

To prove the lower bound of $\lfloor\log_2n\rfloor-1$, we first build a binary tree $T$ on the point set $P$, then we assign the points of $P$ to the leaves of $T$, and then we extract a plane matching from the internal nodes in each level of $T$. See Figure~\ref{matching-example-fig-intro}.

\bibliographystyle{abbrv}
\bibliography{../thesis}

%% file: chapters/ch2-ggmatching.tex
\chapter{Matching in Gabriel Graphs}
\label{ch:gg}

Given a set $P$ of $n$ points in the plane, the order-$k$ Gabriel graph on $P$, denoted by \kGG{k}{}, has an edge between two points $p$ and $q$ if and only if the closed disk with diameter $pq$ contains at most $k$ points of $P$, excluding $p$ and $q$. We study matching problems in \kGG{k}{} graphs. We show that a Euclidean bottleneck matching of $P$ is contained in \kGG{9}{}, but \kGG{8}{} may not have any Euclidean bottleneck matching. In addition we show that \kGG{0}{} has a matching of size at least $\frac{n-1}{4}$ and this bound is tight. We also prove that \kGG{1}{} has a matching of size at least $\frac{2(n-1)}{5}$ and \kGG{2}{} has a perfect matching. Finally we consider the problem of blocking the edges of \kGG{k}{}.

\vspace{10pt}
This chapter is a combination of results that have been published in the journal of Theoretical Computer Science~\cite{Biniaz2015-ggmatching-TCS} and results that have been presented in the 32nd European Workshop on Computational Geometry (EuroCG'16)~\cite{Biniaz2016-ewcg}. 

\section{Introduction}
Let $P$ be a set of $n$ points in the plane. For any two points $p,q\in P$, let $\CD{p}{q}$ denote the closed disk which has the line segment $\overline{pq}$ as diameter. Let $|pq|$ be the Euclidean distance between $p$ and $q$.
The {\em Gabriel graph} on $P$, denoted by $GG(P)$, is defined to have an edge between two points $p$ and $q$ if $\CD{p}{q}$ is empty of points in $P\setminus\{p,q\}$. Let $C(p,q)$ denote the circle which has $\overline{pq}$ as diameter. Note that if there is a point of $P\setminus\{p,q\}$ on $C(p,q)$, then $(p,q)\notin GG(P)$. That is, $(p,q)$ is an edge of $GG(P)$ if and only if $$|pq|^2<|pr|^2+|rq|^2\quad\quad \forall r\in P,\quad\quad r\neq p,q.$$

Gabriel graphs were introduced by Gabriel and Sokal \cite{Gabriel1969} and can be computed in $O(n\log n)$ time \cite{Matula1980}. Every Gabriel graph has at most $3n-8$ edges, for $n\ge 5$, and this bound is tight \cite{Matula1980}. 

A {\em matching} in a graph $G$ is a set of edges without common vertices. A {\em maximum matching} in $G$ is a matching of maximum cardinality, i.e., maximum number of edges. A {\em perfect matching} is a matching which matches all the vertices of $G$. 
In the case that $G$ is an edge-weighted graph, a {\em bottleneck matching} is defined to be a perfect matching in $G$ in which the weight of the maximum-weight edge is minimized. For a perfect matching $M$, we denote the {\em bottleneck} of $M$, i.e., the length of the longest edge in $M$, by $\lambda(M)$. For a point set $P$, a {\em Euclidean bottleneck matching} is a perfect matching which minimizes the length of the longest edge. 

In this chapter we consider perfect matching and bottleneck matching admissibility of higher order Gabriel Graphs. The {\em order-$k$ Gabriel graph} on $P$, denoted by \kGG{k}{}, is the geometric graph which has an edge between two points $p$ and $q$ iff $\CD{p}{q}$ contains at most $k$ points of $P\setminus\{p,q\}$. The standard Gabriel graph, $GG(P)$, corresponds to \kGG{0}{}. It is obvious that \kGG{0}{} is plane, but \kGG{k}{} may not be plane for $k\ge 1$. Su and Chang \cite{Su1990} showed that \kGG{k}{} can be constructed in $O(k^2n\log n)$ time and contains $O(k(n-k))$ edges. In \cite{Bose2013}, the authors proved that \kGG{k}{} is $(k+1)$-connected.
\addtocontents{toc}{\protect\setcounter{tocdepth}{1}}
\subsection{Previous Work}
\addtocontents{toc}{\protect\setcounter{tocdepth}{2}}
It is well-known that a maximum matching in a graph with $n$ vertices and $m$ edges can be computed in $O(m\sqrt{n})$ time, e.g., by Edmonds algorithm (see \cite{Edmonds1965, Micali1980}). Any Gabriel graph is planar, and thus, has $O(n)$ edges. Therefore a maximum matching in a Gabriel graph can be computed in $O(n^{1.5})$ time. Mucha and Sankowski \cite{Mucha2006} showed that a maximum matching in a planar graph can be found in $O(n^{\omega/2})$ time, where $\omega$ is the exponent of matrix multiplication. Since $\omega<2.38$ (see \cite{Williams2012}) a maximum matching in a Gabriel graph can be computed in $O(n^{1.18})$ time.

For any two points $p$ and $q$ in $P$, the {\em lune} of $p$ and $q$, denoted by $L(p,q)$, is defined as the intersection of the open disks of radius $|pq|$ centred at $p$ and $q$.
The {\em order-$k$ Relative Neighborhood Graph} on $P$, denoted by \kRNG{k}{}, is the geometric graph which has an edge $(p,q)$ iff $L(p,q)$ contains at most $k$ points of $P$.  
The {\em order-$k$ Delaunay Graph} on $P$, denoted by \kDG{k}{}, is the geometric graph which has an edge $(p,q)$ iff there exists a circle through $p$ and $q$ which contains at most $k$ points of $P$ in its interior. 
It is obvious that $$\text{\kRNG{k}{}}\subseteq\text{\kGG{k}{}}\subseteq\text{\kDG{k}{}}.$$

The problem of determining whether a geometric graph has a (bottleneck) perfect matching is quite of interest. Dillencourt showed that the Delaunay triangulation (\kDG{0}{}) admits a perfect matching \cite{Dillencourt1990}. Chang et al. \cite{Chang1992} proved that a Euclidean bottleneck perfect matching of $P$ is contained in \kRNG{16}{}.\footnote{They defined \kRNG{k}{} in such a way that $L(p,q)$ contains at most $k-1$ points of $P$.} This implies that \kGG{16}{} and \kDG{16}{} contain a (bottleneck) perfect matching of $P$. In \cite{Abellanas2009} the authors showed that \kGG{15}{} is Hamiltonian. Recently, Kaiser et al.~\cite{Kaiser2015} improved the bound by showing that \kGG{10}{} is Hamiltonian. This implies that \kGG{10}{} has a perfect matching. 

Given a geometric graph $G(P)$ on a set $P$ of $n$ points, we say that a set $K$ of points {\em blocks} $G(P)$ if in $G(P\cup K)$ there is no edge connecting two points in $P$, in other words, $P$ is an independent set in $G(P\cup K)$.
Aichholzer et al.~\cite{Aichholzer2013-blocking} considered the problem of blocking the Delaunay triangulation (i.e. \kDG{0}{}) for $P$ in general position. They show that $\frac{3n}{2}$ points are sufficient to block DT($P$) and at least $n-1$ points are necessary. To block a Gabriel graph, $n-1$ points are sufficient, and $\frac{3}{4}n-o(n)$ points are sometimes necessary \cite{Aronov2013}.

\addtocontents{toc}{\protect\setcounter{tocdepth}{1}}
\subsection{Our Results}
\addtocontents{toc}{\protect\setcounter{tocdepth}{2}}
In this chapter we consider the following three problems: (a) for which values of $k$ does every \kGG{k}{} have a Euclidean bottleneck matching of $P$? (b) for a given value $k$, what is the size of a maximum matching in \kGG{k}{}? (c) how many points are sufficient/necessary to block a \kGG{k}{}? In Section~\ref{gg:preliminaries} we review and prove some graph-theoretic notions. In Section~\ref{gg:bottleneck-section} we consider the problem (a) and prove that a Euclidean bottleneck matching of $P$ is contained in \kGG{9}{}. In addition, we show that for some point sets, \kGG{8}{} does not have any Euclidean bottleneck matching. Moreover, we show that for some point sets, \kGG{7}{} does not have any Euclidean Hamiltonian cycle; this improves the previous bound of \kGG{5}{} that is obtained in~\cite{Kaiser2015}. In Section~\ref{gg:max-matching-section} we consider the problem (b) and give some lower bounds on the size of a maximum matching in \kGG{k}{}. We prove that \kGG{0}{} has a matching of size at least $\frac{n-1}{4}$, and this bound is tight. In addition we prove that \kGG{1}{} has a matching of size at least $\frac{2(n-1)}{5}$ and \kGG{2}{} has a perfect matching. In Section~\ref{gg:blocking-section} we consider the problem (c). We show that at least $\lceil\frac{n-1}{3}\rceil$ points are necessary to block a Gabriel graph and this bound is tight. We also show that at least $\lceil\frac{(k+1)(n-1)}{3}\rceil$ points are necessary and $(k+1)(n-1)$ points are sufficient to block a \kGG{k}{}. The open problems and concluding remarks are presented in Section~\ref{gg:conclusion}.

\section{Preliminaries}
\label{gg:preliminaries}
Let $G$ be an edge-weighted graph with vertex set $V$ and weight function $w:E\rightarrow\mathbb{R^+}$. Let $T$ be a minimum spanning tree of $G$, and let $w(T)$ be the total weight of $T$. 

\begin{lemma}
\label{gg:not-mst-edge}
Let $\delta(e)$ be a cycle in $G$ that contains an edge $e\in T$. Let $\delta'$ be the set of edges in $\delta(e)$ which do not belong to $T$ and let $e'_{max}$ be the largest edge in $\delta'$. Then, $w(e)\le w(e'_{max})$.
\end{lemma}

\begin{proof}
Let $e=(u,v)$ and let $T_u$ and $T_v$ be the two trees obtained by removing $e$ from $T$. Let $e'=(x,y)$ be an edge in $\delta'$ such that one of $x$ and $y$ belongs to $T_u$ and the other one belongs to $T_v$. By definition of $e'_{max}$, we have $w(e')\le w(e'_{max})$. Let $T'=T_u\cup T_v \cup\{(x,y)\}$. Clearly, $T'$ is a spanning tree of $G$. If $w(e')<w(e)$ then $w(T')<w(T)$; contradicting the minimality of $T$. Thus, $w(e)\le w(e')$, which completes the proof of the lemma. 
\end{proof}
 
For a graph $G=(V,E)$ and $S\subseteq V$, let $G-S$ be the subgraph obtained from $G$ by removing all vertices in $S$, and let $o(G-S)$ be the number of odd components in $G-S$, i.e., connected components with an odd number of vertices. The following theorem by Tutte~\cite{Tutte1947} gives a characterization of the graphs which have perfect matching: 

\begin{theorem}[Tutte~\cite{Tutte1947}] 
\label{gg:Tutte} 
$G$ has a perfect matching if and only if $o(G-S)\le |S|$ for all $S\subseteq V$.
\end{theorem}

Berge~\cite{Berge1958} extended Tutte's theorem to a formula (known as the Tutte-Berge formula) for the maximum size of a matching in a graph. In a graph $G$, the {\em deficiency}, $\text{def}_G(S)$, is $o(G-S)-|S|$. Let $\text{def}(G)=\max_{S\subseteq V}{\text{def}_G(S)}$.

\begin{theorem}[Tutte-Berge formula; Berge~\cite{Berge1958}] 
\label{gg:Berge} 
The size of a maximum matching in $G$ is $$\frac{1}{2}(n-\mathrm{def}(G)).$$
\end{theorem}

For an edge-weighted graph $G$ we define the {\em weight sequence} of $G$, \WS{G}, as the sequence containing the weights of the edges of $G$ in non-increasing order. A graph $G_1$ is said to be less than a graph $G_2$ if \WS{G_1} is lexicographically smaller than \WS{G_2}.

\section{Euclidean Bottleneck Matching}
\label{gg:bottleneck-section}
In Subsection~\ref{9GGsec} we prove the following theorem.

\begin{theorem}
 \label{9-GG-thr}
For every point set $P$ in the plane, \kGG{9}{} contains a Euclidean bottleneck matching of $P$.
\end{theorem}

In Subsection~\ref{8GGsec} we prove the following proposition.

\begin{proposition}
\label{8-GG-pro}
There exist point sets $P$ in the plane such that \kGG{8}{} does not contain any Euclidean bottleneck matching of $P$.
\end{proposition}

In Subsection~\ref{7GGsec} we prove the following proposition.

\begin{proposition}
\label{7-GG-pro}
There exist point sets $P$ in the plane such that \kGG{7}{} does not contain any Euclidean bottleneck Hamiltonian cycle of $P$.
\end{proposition} 
Proposition~\ref{7-GG-pro} improves the previous bound of \kGG{5}{} that is obtained by Kaiser~{\em et al.}~\cite{Kaiser2015}

\addtocontents{toc}{\protect\setcounter{tocdepth}{1}}
\subsection{Proof of Theorem~\ref{9-GG-thr}}
\addtocontents{toc}{\protect\setcounter{tocdepth}{2}}
\label{9GGsec}

In this section we prove Theorem~\ref{9-GG-thr}. The proofs for Lemmas~\ref{sab-lemma} and \ref{ss-lemma} are similar to the proofs in \cite{Chang1992} which are adjusted for Gabriel graphs. The proof of Lemma~\ref{distance-lemma} is based on a similar technique that is used in \cite{Kaiser2015} for the Hamiltonicity of Gabriel graphs. 

Let $\mathcal{M}$ be the set of all perfect matchings of the complete graph with vertex set $P$. For a matching $M\in \mathcal{M}$ we define the {\em weight sequence} of $M$, \WS{M}, as the sequence containing the weights of the edges of $M$ in non-increasing order. A matching $M_1$ is said to be less than a matching $M_2$ if \WS{M_1} is lexicographically smaller than \WS{M_2}. We define a total order on the elements of $\mathcal{M}$ by their weight sequence. If two elements have exactly the same weight sequence, break ties arbitrarily to get a total order.

Let $M^* = \{(a_1, b_1),\dots, (a_{\frac{n}{2}}, b_{\frac{n}{2}})\}$ be a matching in $\mathcal{M}$ with minimum weight sequence. Observe that $M^*$ is a Euclidean bottleneck matching for $P$. In order to prove Theorem~\ref{9-GG-thr}, we will show that all edges of $M^*$ are in \kGG{9}{}. Consider any edge $(a, b)$ in $M^*$. If $\CD{a}{b}$ contains no point of $P\setminus\{a,b\}$, then $(a,b)$ is an edge of \kGG{9}{}. Suppose that $\CD{a}{b}$ contains $k$ points of $P\setminus\{a,b\}$. We are going to prove that $k\le 9$. Let $R = \{r_1, r_2,\dots, r_k\}$ be the set of points of $P\setminus\{a,b\}$ that are in $\CD{a}{b}$. Let $S=\{s_1, s_2,\dots, s_k\}$ represent the points for which $(r_i,s_i)\in M^*$.  

Without loss of generality, we assume that $\CD{a}{b}$ has diameter 1 and is centered at the origin $o=(0,0)$, and $a=(-0.5,0)$ and $b=(0.5,0)$. For any point $p$ in the plane, let $\|p\|$ denote the distance of $p$ from $o$. Note that $|ab|=1$, and for any point $x\in \CD{a}{b}\setminus\{a,b\}$ we have $\max\{|xa|,|xb|\}<1$.

\begin{lemma}
\label{sab-lemma}
For each point $s_i\in S$, $\min\{|s_ia|, |s_ib|\}\ge 1$. 
\end{lemma}

\begin{proof}
The proof is by contradiction; suppose that $|s_ia|<1$. Let $M$ be the perfect matching obtained from $M^*$ by deleting $\{(a,b), (r_i,s_i)\}$ and adding $\{(s_i,a), (r_i,b)\}$. The lengths of the two new edges are smaller than 1, and hence both $(s_i,a)$ and $(r_i,b)$ are shorter than $(a,b)$. Thus, $\WS{M}<_{\mathrm{lex}}\WS{M^*}$, which contradicts the minimality of $M^*$.
\end{proof}

As a corollary of Lemma~\ref{sab-lemma}, $R$ and $S$ are disjoint.

\begin{lemma}
\label{ss-lemma}
For each pair of points $s_i,s_j\in S$, $|s_is_j|\ge \allowbreak\max\allowbreak\{|r_is_i|, \allowbreak |r_js_j|, \allowbreak 1\}$.
\end{lemma}

\begin{proof}
The proof is by contradiction; suppose that $|s_is_j|< \max\{|r_is_i|,\allowbreak  |r_js_j|,\allowbreak 1\}$. Let $M$ be the perfect matching obtained from $M^*$ by deleting $\{(a,b),\allowbreak (r_i,s_i),\allowbreak (r_j,s_j)\}$ and adding $\{(a, r_i), \allowbreak (b, r_j),(s_i,s_j)\}$. Note that $\allowbreak \max \{|ar_i|,\allowbreak |br_j|,\allowbreak |s_is_j|\}<\max\{|r_is_i|,\allowbreak |r_js_j|,\allowbreak |ab|\}$. Thus, we get $\WS{M} \allowbreak{<_{\mathrm{lex}}}\allowbreak \WS{M^*}$, which contradicts the minimality of $M^*$.
\end{proof}

Let $C(x,r)$ (resp.\ $D(x,r)$) be the circle (resp.\ closed disk) of radius $r$ that is centered at a point $x$ in the plane.
For $i\in\{1,\dots,k\}$, let $s'_i$ be the intersection point between $C(o,1.5)$ and the ray with origin at $o$ passing through $s_i$. 
Let the point $p_i$ be $s_i$, if $\|s_i\|< 1.5$, and $s'_i$, otherwise. See Figure~\ref{distance-fig}. Let $S'=\{a, b, p_1,\dots,p_k\}$. 

\begin{observation}
\label{obs}
Let $s_j$ be a point in $S$, where $\|s_j\|\ge 1.5$. Then, the disk $D(s_j, \|s_j\|-0.5)$ is contained in the disk $D(s_j,|s_jr_j|)$. Moreover, the disk $D(p_j,1)$ is contained in the disk $D(s_j, \|s_j\|-0.5)$. See Figure~\ref{distance-fig}.
\end{observation}

\begin{figure}[htb]
  \centering
  \includegraphics[width=.55\columnwidth]{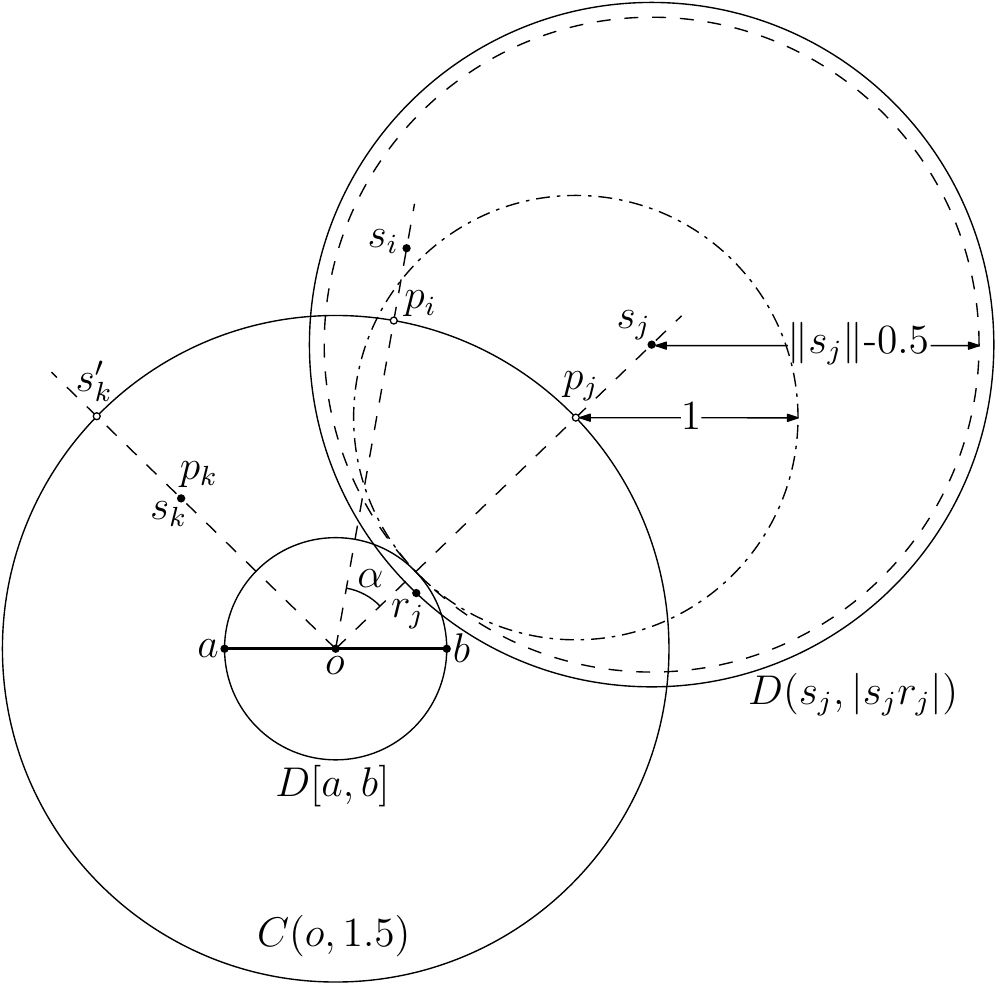}
 \caption{Proof of Lemma~\ref{distance-lemma}; $p_i=s'_i$, $p_j=s'_j$, and $p_k=s_k$.}
  \label{distance-fig}
\end{figure}

\begin{lemma}
\label{distance-lemma}
The distance between any pair of points in $S'$ is at least 1.
\end{lemma}

\begin{proof}
Let $x$ and $y$ be two points in $S'$. We are going to prove that $|xy|\ge 1$. We distinguish between the following three cases. 
\begin{itemize}
\item $\{x,y\}=\{a,b\}$. In this case the claim is trivial.
\item $x\in \{a,b\}, y\in \{p_1,\dots, p_k\}$. If $\|y\|=1.5$, then $y$ is on $C(o,1.5)$, and hence $|xy|\ge 1$. If $\|y\|<1.5$, then $y$ is a point in $S$. Therefore, by Lemma~\ref{sab-lemma}, $|xy|\ge 1$.

\item $x,y\in\{p_1,\dots,p_k\}$. Without loss of generality assume $x=p_i$ and $y= p_j$, where $1\le i<j\le k$. We differentiate between three cases:

\vspace{5pt}
{\em Case }($i$): $\|p_i\|< 1.5$ and $\|p_j\|<1.5$. In this case $p_i$ and $p_j$ are two points in $S$. Therefore, by Lemma~\ref{ss-lemma}, $|p_ip_j|\ge 1$.

\vspace{5pt}
{\em Case }($ii$): $\|p_i\|< 1.5$ and $\|p_j\|=1.5$. In this case $p_i$ is a point in $S$. By Observation~\ref{obs}, the disk $D(p_j, 1)$ is contained in the disk $D(s_j,|s_jr_j|)$, and by Lemma~\ref{ss-lemma}, $p_i$ is not in the interior of $D(s_j,|s_jr_j|)$. Therefore, $p_i$ is not in the interior of $D(p_j,1)$, which implies that $|p_ip_j|\ge 1$.

\vspace{5pt}
{\em Case }($iii$): $\|p_i\|= 1.5$ and $\|p_j\|=1.5$. In this case $\|s_i\|\ge 1.5$ and $\|s_j\|\ge 1.5$. Without loss of generality assume $\|s_i\|\le \|s_j\|$. For the sake of contradiction assume that $|p_ip_j|<1$; see Figure~\ref{distance-fig}. Then, for the angle $\alpha=\angle s_i o s_j$ we have $\sin(\alpha/2)< \frac{1}{3}$. Then, $\cos(\alpha)= 1-2\sin^2(\alpha/2)>\frac{7}{9}$. By the law of cosines in the triangle $\bigtriangleup s_ios_j$, we have
\begin{equation}
\label{ineq1}
|s_is_j|^2<\|s_i\|^2+\|s_j\|^2-\frac{14}{9}\|s_i\|\|s_j\|.
\end{equation}
By Observation~\ref{obs}, the disk $D(s_j,\|s_j\|-0.5)$ is contained in the disk $D(s_j,|s_jr_j|)$, and by Lemma~\ref{ss-lemma}, $s_i$ is not in the interior of $D(s_j,|s_jr_j|)$. Therefore, $s_i$ is not in the interior of $D(s_j,\|s_j\|-0.5)$. Thus, $|s_is_j|\ge \|s_j\|-0.5$. In combination with Inequality~(\ref{ineq1}), this implies
\begin{equation}
 \label{ineq2}
\|s_j\|\left(\frac{14}{9}\|s_i\|-1\right) < \|s_i\|^2-\frac{1}{4}.
\end{equation}

In combination with the assumption $\|s_i\| \le \|s_j\|$, Inequality~(\ref{ineq2}) implies
$$\frac{5}{9}\|s_i\|^2-\|s_i\|+\frac{1}{4}<0,$$
i.e., 
$$\frac{5}{9}\left(\|s_i\|-\frac{3}{10}\right)\left(\|s_i\|-\frac{3}{2}\right)<0.$$
This is a contradiction, because, since $\|s_i\|\ge 1.5$, the left-hand side is non-negative. Thus $|p_ip_j|\ge 1$, which completes the proof of the lemma.
\end{itemize}
\end{proof}
By Lemma~\ref{distance-lemma}, the points in $S'$ have mutual distance at least 1. Moreover, the points in $S'$ lie in $D(o,1.5)$.
Fodor~\cite{Fodor2000} proved that the smallest circle which contains 12 points with mutual distances at least 1 has radius 1.5148. 
Therefore, $S'$ contains at most $11$ points. Since $a,b\in S'$, this implies that $k\le 9$. Therefore, $S$, and consequently $R$, contains at most 9 points. Thus, $(a,b)$ is an edge in \kGG{9}{}. This completes the proof of Theorem~\ref{9-GG-thr}.
\addtocontents{toc}{\protect\setcounter{tocdepth}{1}}
\subsection{Proof of Proposition~\ref{8-GG-pro}}
\addtocontents{toc}{\protect\setcounter{tocdepth}{2}}
\label{8GGsec}

In this section we prove Proposition~\ref{8-GG-pro}. We show that for some point sets $P$, \kGG{8}{} does not contain any Euclidean bottleneck matching of $P$.

Consider Figure~\ref{gg:8-GG-fig} which shows a configuration of a set $P$ of 20 points. The closed disk $\CD{a}{b}$ is centred at $c$ and has diameter one, i.e., $|ab|=1$. $\CD{a}{b}$ contains 9 points $U=\{u_1, \dots, u_9\}$ which lie on a circle with radius $\frac{1}{2}-\epsilon$ which is centred at $c$. Nine points in $U'=\{r_1,\dots,r_9\}$ are placed on a circle with radius 1.5 which is centred at $c$ in such a way that $|r_ju_j|= 1+\epsilon$, $|r_ja|>1+\epsilon$, $|r_jb|>1+\epsilon$, and $|r_jr_k|>1+\epsilon$ for $1\le j, k\le 9$ and $j\neq k$. Consider a perfect matching $M=\{(a,b)\}\cup \{(r_i, u_i): i=1,\dots, 9\}$ where each point $r_i\in U'$ is matched to its closest point $u_i$. It is obvious that $\lambda(M)=1+\epsilon$, and hence the bottleneck of any bottleneck perfect matching is at most $1+\epsilon$. We will show that any Euclidean bottleneck matching of $P$ contains $(a,b)$. By contradiction, let $M^*$ be a Euclidean bottleneck matching which does not contain $(a,b)$. In $M^*$, $a$ is matched to a point $x\in U\cup U'$. If $x \in U'$, then $|ax|>1+\epsilon$. If $x\in U$, w.l.o.g. assume that $x = u_1$. Thus, in $M^*$ the point $r_1$ is matched to a point $y$ where $y\neq u_1$. Since $u_1$ is the closest point to $r_1$ and $|r_1u_1|=1+\epsilon$, $|r_1y|>1+\epsilon$. In both cases $\lambda(M^*)> 1+\epsilon$, which is a contradiction. Therefore, $M^*$ contains $(a,b)$. Since $\CD{a}{b}$ contains 9 points of $P\setminus\{a,b\}$, $(a,b)\notin\text{\kGG{8}{}}$. Therefore \kGG{8}{} does not contain any Euclidean bottleneck matching of $P$. 
\begin{figure}[htb]
  \centering
  \includegraphics[width=.7\columnwidth]{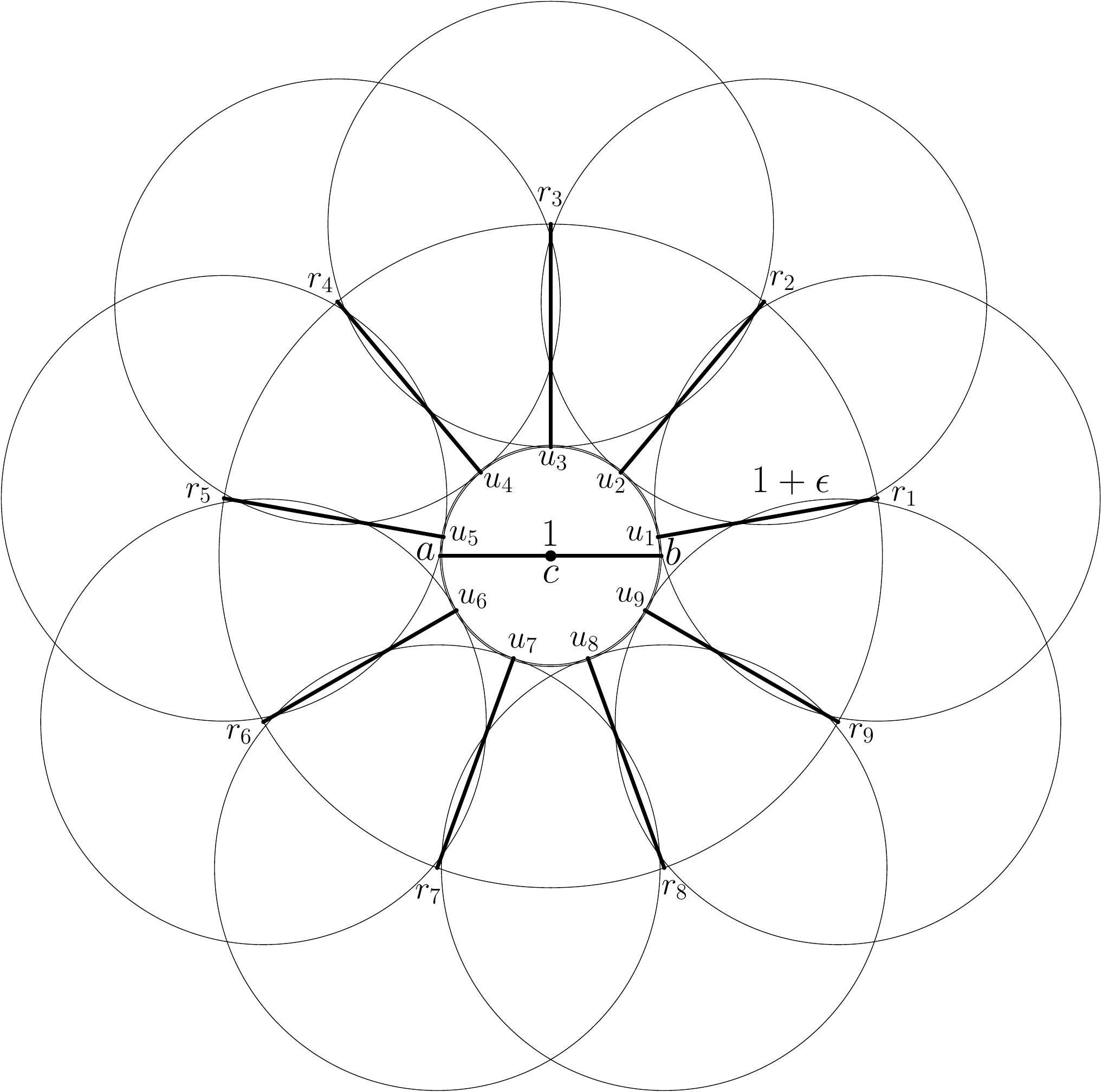}
 \caption{A set of 20 points such that \kGG{8}{} does not contain any Euclidean bottleneck matching.}
  \label{gg:8-GG-fig}
\end{figure}

\addtocontents{toc}{\protect\setcounter{tocdepth}{1}}
\subsection{Proof of Proposition~\ref{7-GG-pro}}
\addtocontents{toc}{\protect\setcounter{tocdepth}{2}}
\label{7GGsec}
In this section we prove Proposition~\ref{7-GG-pro}. We show that for some point sets $P$, \kGG{7}{} does not contain any Euclidean bottleneck Hamiltonian cycle of $P$.

Figure~\ref{7-GG-fig} shows a configuration of a multiset $P=\{a, b,x, r_1, \dots, r_8,\allowbreak s_1,\allowbreak \dots, s_7\}$ of 26 points, where $s_5$ is repeated nine times. The closed disk $\CD{a}{b}$ is centered at $o$ and has diameter one, i.e., $|ab|=1$. $\CD{a}{b}$ contains all 8 points of the set $R=\{r_1, \dots, r_8\}$; these points lie on the circle with radius $\frac{1}{2}-\epsilon$ that is centered at $o$; all points of $R$ are in the interior of $\CD{a}{b}$. Let $S=\{s_1,\dots, s_7\}$ be the multiset of 15 points, where $s_5$ is repeated nine times. The red circles have radius 1 and are centered at points in $S$. Each point in $S$ is connected to its first and second closest point (the black edges in Figure~\ref{7-GG-fig}). Let $B$ the chain formed by these edges. Note that $r_1$ and $r_8$ are the endpoints of $B$. Specifically, $|r_1s_1|=|r_8s_7|=1$, and for each point $r_i$, where $2\le i\le 7$, $|s_ia|>1$, $|s_ib|>1$, $|s_ix|>1$, and $|r_is_{i-1}|=|r_is_i|=1$ (here by $s_5$ we mean the first and last endpoints of the chain defined by points labeled $s_5$). Consider the Hamiltonian cycle $H=B\cup\{(r_1,a),(a,b),(b,x),(x,r_8)\}$. The longest edge in $H$ has length 1. Therefore, the length of the longest edge in any bottleneck Hamiltonian cycle for $P$ is at most 1. In the rest we will show\textemdash by contradiction\textemdash that any bottleneck Hamiltonian cycle of $P$ contains $(a,b)$. Since in $B$ each point of $S$ is connected to its first and second closest point, every bottleneck Hamiltonian cycle of $P$ contains $B$, because otherwise, one of the points in $S$ should be connected to a point that is farther than its second closest point, and hence that edge is longer than 1. Now we consider possible ways to construct a bottleneck Hamiltonian cycle, say $H^*$, using the edges in $B$ and the points $a, b,x$. Assume $(a,b)\notin H^*$. Then, in $H^*$, $a$ is connected to two points in $\{r_1,r_8,x\}$. We differentiate between two cases: 
\begin{itemize}
 \item $(a,x)\in H^*$. In this case $|ax|>1$, and hence the longest edge in $H^*$ is longer than 1, which is a contradiction.
  \item $(a,x)\notin H^*$. In this case $(a,r_1)\in H^*$ and $(a,r_8)\in H^*$. This means that $H^*$ does not contain $x$ and $b$, which is a contradiction. 
\end{itemize}

Therefore, we conclude that $H^*$, and consequently any bottleneck Hamiltonian cycle of $P$, contains $(a,b)$.
Since $\CD{a}{b}$ contains 8 points of $P\setminus\{a,b\}$, $(a,b)\notin\text{\kGG{7}{}}$. Therefore \kGG{7}{} does not contain any Euclidean bottleneck Hamiltonian cycle of $P$. 

\begin{figure*}[htb]
  \centering
  \includegraphics[width=.7\columnwidth]{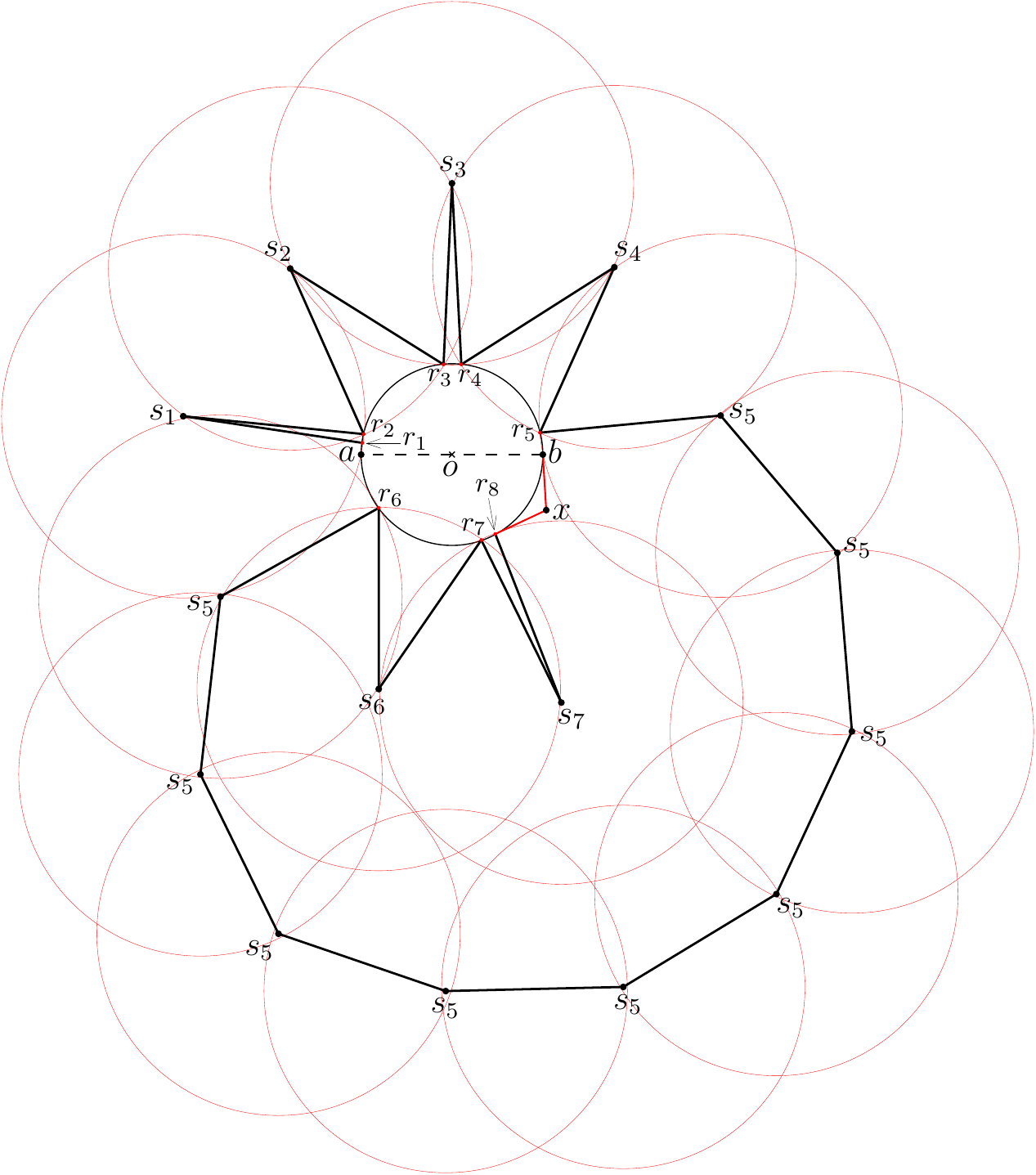}
 \caption{Proof of Proposition~\ref{7-GG-pro}. The bold-black edges belong to $B$. $\CD{a}{b}$ contains 8 points.}
  \label{7-GG-fig}
\end{figure*}

\section{Maximum Matching}
\label{gg:max-matching-section}
Let $P$ be a set of $n$ points in the plane. In this section we will prove that \kGG{0}{} has a matching of size at least $\frac{n-1}{4}$; this bound is tight. We also prove that \kGG{1}{} has a matching of size at least $\frac{2(n-1)}{5}$ and \kGG{2}{} has a perfect matching (when $n$ is even).

First we give a lower bound on the number of components that result after removing a set $S$ of vertices from \kGG{k}{}. Then we use Theorem~\ref{gg:Tutte} and Theorem~\ref{gg:Berge}, respectively presented by Tutte~\cite{Tutte1947} and Berge~\cite{Berge1958}, to prove a lower bound on the size of a maximum matching in \kGG{k}{}. 

\begin{figure}[htb]
  \centering
\setlength{\tabcolsep}{0in}
  $\begin{tabular}{cc}
 \multicolumn{1}{m{.5\columnwidth}}{\centering\includegraphics[width=.38\columnwidth]{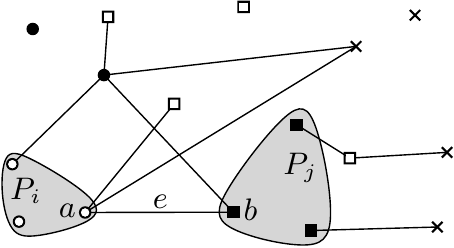}}
&\multicolumn{1}{m{.5\columnwidth}}{\centering\includegraphics[width=.38\columnwidth]{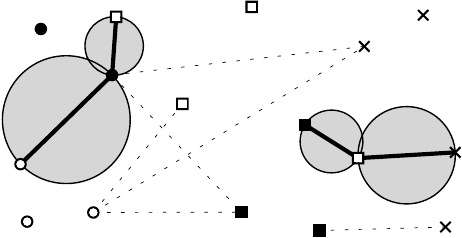}}\\
(a) & (b)
\end{tabular}$
  \caption{The point set $P$ of 16 points is partitioned into white/black disks, white/black squares, and crosses. (a) The graph $G(\mathcal{P})$, (b) The set $\mathcal{T}$ of straight-line edges corresponding to $MST(G(\mathcal{P}))$ is in bold, and the set $\mathcal{D}$ of their corresponding disks.}
\label{gg:partition-fig}
\end{figure}

Let $\mathcal{P}=\{P_1, P_2,\dots\}$ be a partition of the points in $P$. For two sets $P_i$ and $P_j$ in $\mathcal{P}$ define the distance $d(P_i,P_j)$ as the smallest Euclidean distance between a point in $P_i$ and a point in $P_j$, i.e., $d(P_i,P_j)=\min\{|ab|:a\in P_i, b\in P_j\}$.
Let $G(\mathcal{P})$ be the complete edge-weighted graph with vertex set $\mathcal{P}$. For each edge $e=(P_i,P_j)$ in $G(\mathcal{P})$, let $w(e)=d(P_i,P_j)$. This edge $e$ is defined by two points $a$ and $b$, where $a\in P_i$ and $b\in P_j$. Therefore, an edge $e\in G(\mathcal{P})$ corresponds to a straight line edge between two points $a,b\in P$; see Figure~\ref{gg:partition-fig}(a). Let $MST(G(\mathcal{P}))$ be a minimum spanning tree of $G(\mathcal{P})$. It is obvious that each edge $e$ in $MST(G(\mathcal{P}))$ corresponds to a straight line edge between $a,b\in P$. Let $\mathcal{T}$ be the set of all these straight line edges. Let $\mathcal{D}$ be the set of disks which have the edges of $\mathcal{T}$ as diameter, i.e., $\mathcal{D}=\{D[a,b]: (a,b)\in \mathcal{T}\}$. See Figure~\ref{gg:partition-fig}(b).

\begin{observation}
 \label{gg:T-plane}
$\mathcal{T}$ is a subgraph of a minimum spanning tree of $P$, and hence $\mathcal{T}$ is plane.
\end{observation}

\begin{lemma}
 \label{gg:D-empty}
A disk $\CD{a}{b}\in\mathcal{D}$ does not contain any point of $P\setminus\{a,b\}$.
\end{lemma}
\begin{proof}
By Observation~\ref{gg:T-plane}, $\mathcal{T}$ is a subgraph of a minimum spanning tree of $P$. It is well known that any minimum spanning tree of $P$ is a subgraph of \kGG{0}{(P)}. Thus, $\mathcal{T}$ is a subgraph of \kGG{0}{(P)}, and hence, any disk $\CD{a}{b}\in\mathcal{D}$\textemdash representing an edge of $\mathcal{T}$\textemdash does not contain any point of $P\setminus\{a,b\}$. 
\end{proof}

\begin{lemma}
\label{gg:center-in-lemma}
 For each pair $D_i$ and $D_j$ of disks in $\mathcal{D}$, $D_i$ (resp. $D_j$) does not contain the center of $D_j$ (resp $D_i$).
\end{lemma}

\begin{proof}
 Let $(a_i,b_i)$ and $(a_j,b_j)$ respectively be the edges of $\mathcal{T}$ which correspond to $D_i$ and $D_j$. Let $C_i$ and $C_j$ be the circles representing the boundary of $D_i$ and $D_j$. W.l.o.g. assume that $C_j$ is the bigger circle, i.e., $|a_ib_i|<|a_jb_j|$. By contradiction, suppose that $C_j$ contains the center $c_i$ of $C_i$. Let $x$ and $y$ denote the intersections of $C_i$ and $C_j$. Let $x_i$ (resp. $x_j$) be the intersection of $C_i$ (resp. $C_j$) with the line through $y$ and $c_i$ (resp. $c_j$). Similarly, let $y_i$ (resp. $y_j$) be the intersection of $C_i$ (resp. $C_j$) with the line through $x$ and $c_i$ (resp. $c_j$). 

\begin{figure}[htb]
  \centering
  \includegraphics[width=.45\columnwidth]{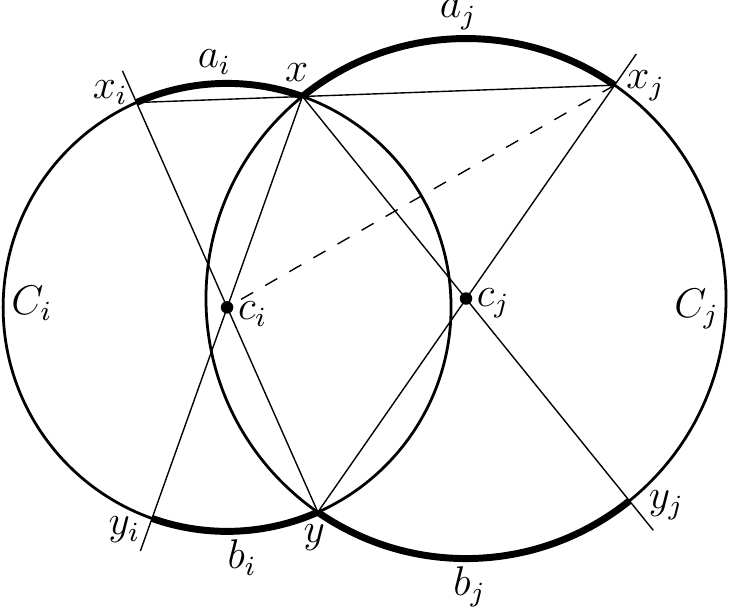}
 \caption{Illustration of Lemma~\ref{gg:center-in-lemma}: $C_i$ and $C_j$ intersect, and $C_j$ contains the center of $C_i$.}
  \label{gg:center-in-fig}
\end{figure}

As illustrated in Figure~\ref{gg:center-in-fig}, the arcs $\widehat{x_ix}$, $\widehat{y_iy}$, $\widehat{x_jx}$, and $\widehat{y_jy}$ are the potential positions for the points $a_i$, $b_i$, $a_j$, and $b_j$, respectively. First we will show that the line segment $x_ix_j$ passes through $x$ and $|a_ia_j|\leq|x_ix_j|$. The angles $\angle x_ixy$ and $\angle x_jxy$ are right angles, thus the line segment $x_ix_j$ goes through $x$. Since $\widehat{x_ix}<\pi$ (resp. $\widehat{x_jx}<\pi$), for any point $a_i\in \widehat{x_ix}, |a_ix|\leq|x_ix|$ (resp. $a_j\in \widehat{x_jx}, |a_jx|\leq|x_jx|$). Therefore, $$|a_ia_j|\leq|a_ix|+|xa_j|\leq|x_ix|+|xx_j|=|x_ix_j|.$$
Consider triangle $\bigtriangleup x_ix_jy$ which is partitioned by segment $c_ix_j$ into $t_1=\bigtriangleup x_ix_jc_i$ and $t_2=\bigtriangleup c_ix_jy$. Since $x_iy$ is a diameter of $C_i$ which passes through the center $c_i$, the length of the segment $x_ic_i$ of $t_1$ is equal to the length of the segment $c_iy$ of $t_2$. The segment $c_ix_j$ is shared by $t_1$ and $t_2$. Since $c_i$ is inside $C_j$ and $\widehat{yx_j}=\pi$, the angle $\angle yc_ix_j>\frac{\pi}{2}$. Thus, $\angle x_ic_ix_j$ in $t_1$ is smaller than $\frac{\pi}{2}$ (and hence smaller than $\angle yc_ix_j$ in $t_2$). Therefore, the segment $x_ix_j$ of $t_1$ is smaller than the segment $x_jy$ of $t_2$. Thus,

$$|a_ia_j|\leq|x_ix_j|<|x_jy|=|a_jb_j|.$$

By symmetry $|b_ib_j|<|a_jb_j|$. Therefore $\max\{|a_ia_j|,|b_ib_j|\}\allowbreak<\allowbreak\max\{\allowbreak|a_ib_i|,|a_jb_j|\}$. In addition $\delta=(a_i,a_j,b_j,b_i,a_i)$ is a cycle and at least one of $(a_i,a_j)$ and $(b_i,b_j)$ does not belong to $\mathcal{T}$. This contradicts Lemma~\ref{gg:not-mst-edge} (Note that by Observation~\ref{gg:T-plane}, $\mathcal{T}$ is a subgraph of a minimum spanning tree of $P$).
\end{proof}

Now we show that the intersection of every four disks in $\mathcal{D}$ is empty. In other words, every point in the plane cannot lie in more than three disks in $\mathcal{D}$. In Section~\ref{gg:proof-section} we prove the following theorem, and in Section~\ref{gg:lower-bounds-section} we present the lower bounds on the size of a maximum matching in \kGG{k}{}.

\begin{theorem}
 \label{gg:four-circle-theorem}
For every four disks $D_1,D_2,D_3,D_4\in\mathcal{D}$, $D_1\cap D_2\cap D_3\cap D_4=\emptyset$.
\end{theorem}
\addtocontents{toc}{\protect\setcounter{tocdepth}{1}}
\subsection{Proof of Theorem~\ref{gg:four-circle-theorem}}
\addtocontents{toc}{\protect\setcounter{tocdepth}{2}}
\label{gg:proof-section}
The proof is by contradiction. Let $D_1$, $D_2$, $D_3$, and $D_4$ be four disks in $\mathcal{D}$. Let $\mathcal{X}=D_1\cap D_2\cap D_3\cap D_4$  and let $x$ be a point in $\mathcal{X}$. For $i=1,2,3,4$, let $c_i$ be the center of $D_i$, let $C_i$ be the boundary of $D_i$, and let $(a_i,b_i)$ be the edge in $\mathcal{T}$ which corresponds to $D_i$. Denote the angle $\angle a_ixb_i$ by $\alpha_i$, for $i=1,2,3,4$. Since $(a_i,b_i)$ is a diameter of $D_i$ and $x$ lies in $D_i$, $\alpha_i \ge\frac{\pi}{2}$. First we prove the following observation.
\begin{observation}
\label{gg:inclusion-exclusion}
 For $i,j\in\{1,2,3,4\}$, where $i\neq j$, the angles $\alpha_i$ and $\alpha_j$ are either disjoint or one is completely contained in the other.
\end{observation}
\begin{proof}
The proof is by contradiction. Suppose that $\alpha_i$ and $\alpha_j$ are not disjoint and none of them is completely contained in the other. Thus $\alpha_i$ and $\alpha_j$ share some part and w.l.o.g. assume that $b_i$ is in the cone which is obtained by extending the edges of $\alpha_j$, and $b_j$ is in the cone which is obtained by extending the edges of $\alpha_i$. Three cases arise:
\begin{itemize}
 \item $b_i\in\bigtriangleup xa_jb_j$. In this case $b_i$ is inside $D_j$ which contradicts Lemma~\ref{gg:D-empty}.
 \item $b_j\in\bigtriangleup xa_ib_i$. In this case $b_j$ is inside $D_i$ which contradicts Lemma~\ref{gg:D-empty}.
 \item $b_i\notin\bigtriangleup xa_jb_j$ and $b_j\notin\bigtriangleup xa_ib_i$. In this case $(a_i,b_i)$ intersects $(a_j,b_j)$ which contradicts Observation~\ref{gg:T-plane}.
\end{itemize}
\end{proof}

We call $\alpha_i$ a {\em blocked angle} if $\alpha_i$ is contained in an angle $\alpha_j$ for some $j\in\{1,2,3,4\}$, where $j\neq i$. Otherwise, we call $\alpha_i$ a {\em free angle}.

\begin{lemma}
\label{gg:not-all-free-angles}
At least one $\alpha_i$, for $i\in\{1,2,3,4\}$, is blocked.
\end{lemma}
\begin{proof}
Suppose that all angles $\alpha_i$, where $i\in\{1,2,3,4\}$, are free. This implies that the $\alpha_i$s are pairwise disjoint and $\alpha=\sum_{i=1}^{4}{\alpha_i} \ge 2\pi$. If $\alpha > 2\pi$, we obtain a contradiction to the fact that the sum of the disjoint angles around $x$ is at most $2\pi$. If $\alpha = 2\pi$, then the four edges $(a_i,b_i)$ where $i\in\{1,2,3,4\}$, form a cycle which contradicts the fact that $\mathcal{T}$ is a subgraph of a minimum spanning tree of $P$.
\end{proof}

\begin{figure}[htb]
  \centering
\setlength{\tabcolsep}{0in}
  $\begin{tabular}{cc}
 \multicolumn{1}{m{.7\columnwidth}}{\centering\includegraphics[width=.38\columnwidth]{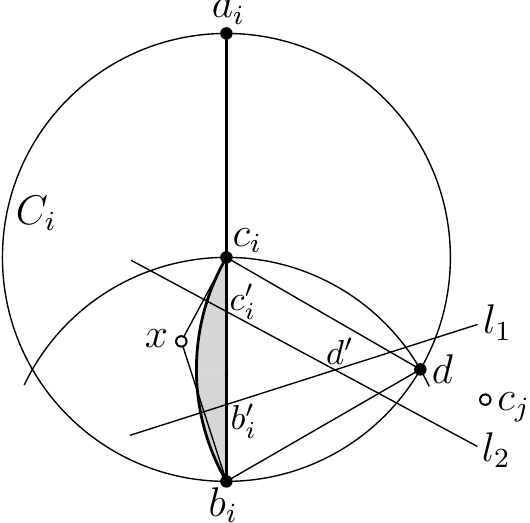}}
&\multicolumn{1}{m{.3\columnwidth}}{\centering\includegraphics[width=.056\columnwidth]{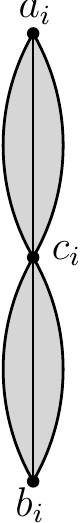}}\\
(a) & (b)
\end{tabular}$
  \caption{(a) The point $x$ should be inside the circle segment $\mathrm{sc}(c_i,b_i)$. (b) The $\mathrm{trap}(a_i,b_i)$ which consists of two lenses $\mathrm{lens}(a_i,c_i)$ and $\mathrm{lens}(b_i,c_i)$.}
\label{gg:trap-fig}
\end{figure}

By Lemma~\ref{gg:not-all-free-angles} at least one of the angles is blocked. Hereafter, assume that $\alpha_j$ is blocked by $\alpha_i$ where $1\le i,j\le 4$ and $i\neq j$. W.l.o.g. assume that $a_ib_i$ is a vertical line segment and the point $x$ (which belongs to $\mathcal{X}$) is to the left of $a_ib_i$. Thus, $a_jb_j$ and $c_j$ are to the right of $a_ib_i$. This implies that $a_ib_i\cap D_j\neq \emptyset$. See Figure~\ref{gg:trap-fig}(a). By Lemma~\ref{gg:center-in-lemma}, $c_i$ cannot be inside $D_j$, thus either $a_ic_i\cap D_j\neq \emptyset$ or $c_ib_i\cap D_j \neq\emptyset$, but not both. W.l.o.g. assume that $c_ib_i\cap D_j\neq \emptyset$. Let $C'$ be the circle with radius $|c_ib_i|$ which is centered at $b_i$. Let $d$ denote the intersection of $C'$ with $C_i$ which is to the right of $c_ib_i$. Consider the circle $C''$ with radius $|db_i|$ centered at $d$. Note that $C''$ goes through $b_i$ and $c_i$. Let $\mathrm{sc}(c_i,b_i)$ be the segment of the circle $C''$ which is between the chord $c_ib_i$ and the arc $\widehat{c_ib_i}$ as shown in Figure~\ref{gg:trap-fig}(a).

We show that $x$ cannot be outside $\mathrm{sc}(c_i,b_i)$. By contradiction suppose that $x$ is outside $\mathrm{sc}(c_i,b_i)$ (and to the left of $c_ib_i$). Let $l_1$ and $l_2$ respectively be the perpendicular bisectors of $xb_i$ and $xc_i$. Let $b'_i$ and $c'_i$ respectively be the intersection of $l_1$ and $l_2$ with $c_ib_i$ and let $d'$ be the intersection point of $l_1$ and $l_2$. Since $x$ is outside $\mathrm{sc}(c_i,b_i)$, the intersection point $d'$ is to the left of (the vertical line through) $d$ and inside triangle $\bigtriangleup b_ic_id$. If $c_j$ is below $l_1$ then $|c_jb_i|<|c_jx|$ and $D_j$ contains $b_i$ which contradicts Lemma~\ref{gg:center-in-lemma}. If $c_j$ is above $l_2$ then $|c_jb_i|<|c_jx|$ and $D_j$ contains $c_i$ which contradicts Lemma~\ref{gg:center-in-lemma}. Thus, $c_j$ is above $l_1$ and below $l_2$, and (by the initial assumption) to the right of $c_ib_i$. That is, $c_j$ is in triangle $\bigtriangleup b'_ic'_id'$. Since $\bigtriangleup b'_ic'_id'\subseteq \bigtriangleup b_ic_id\subseteq D_i$, $c_j$ lies inside $D_i$ which contradicts Lemma~\ref{gg:center-in-lemma}. Therefore, $x$ is contained in $\mathrm{sc}(c_i,b_i)$. 

By symmetry $D_j$ can intersect $a_ic_i$ and/or $c_j$ can be to the left of $a_ib_i$ as well. Therefore, if $\alpha_i$ blocks $\alpha_j$, the point $x$ can be in $\mathrm{sc}(c_i,b_i)$ or any of the symmetric segments of the circles. For an edge $a_ib_i$ we denote the union of these segments by $\mathrm{trap}(a_i,b_i)$ which is shown in Figure~\ref{gg:trap-fig}(b). For each disk $D_i$, let $\mathrm{trap}(D_i)= \mathrm{trap}(a_i,b_i)$ where $(a_i, b_i)$ is the edge in $\mathcal{T}$ corresponding to $D_i$. Therefore $x$ is contained in $\mathrm{trap}(D_i)$ which implies that $$\mathcal{X}\subseteq \mathrm{trap}(D_i).$$ Note that $\mathrm{trap}(D_i)$ consists of two symmetric lenses $\mathrm{lens}(a_i,c_i)$ and $\mathrm{lens}(b_i,c_i)$, i.e., $\mathrm{trap}(D_i)=\mathrm{lens}(a_i,c_i)\cup \mathrm{lens}(b_i,c_i)$.

\begin{lemma}
\label{gg:angle-in-trap}
 For any point $x\in \mathrm{trap}(a_i,b_i)$, $\angle a_ixb_i \ge 150^\circ$.
\end{lemma}
\begin{proof}
 See Figure~\ref{gg:trap-fig}(a). The angle $\angle b_idc_i = 60^\circ$, which implies that $\widehat{c_ib_i}= 60^\circ$. Thus, for any point $x'$ on the arc $\widehat{c_ib_i}$, $\angle x'c_ib_i + \angle x'b_ic_i = 30^\circ$, and hence for any point $x$ in the segment $\mathrm{sc}(c_i,b_i)$, $\angle xc_ib_i + \angle xb_ic_i \le 30^\circ$. This implies that in $\bigtriangleup xb_ic_i$, $\angle b_ixc_i \ge 150^\circ$. On the other hand $\angle b_ixc_i \le \angle b_ixa_i$, which proves the lemma.
\end{proof}

\begin{figure}[htb]
  \centering
\setlength{\tabcolsep}{0in}
  $\begin{tabular}{ccc}
 \multicolumn{1}{m{.33\columnwidth}}{\centering\includegraphics[width=.12\columnwidth]{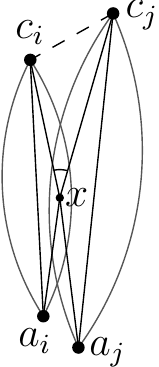}}
&\multicolumn{1}{m{.33\columnwidth}}{\centering\includegraphics[width=.12\columnwidth]{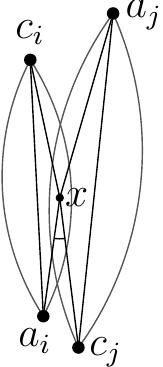}}
&\multicolumn{1}{m{.33\columnwidth}}{\centering\includegraphics[width=.12\columnwidth]{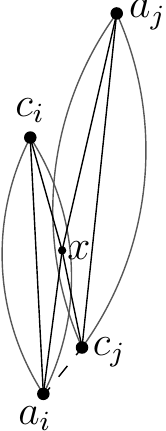}}\\
(a) & (b)& (c)
\end{tabular}$
  \caption{Illustration of Lemma~\ref{gg:intersecting-trap}.}
\label{gg:trap-intersection-fig}
\end{figure}

\begin{lemma}
\label{gg:intersecting-trap}
For any two disks $D_i$ and $D_j$ in $\mathcal{D}$, $\mathrm{trap}(D_i)\cap \mathrm{trap}(D_j)=\emptyset$.
\end{lemma}
\begin{proof}
 We prove this lemma by contradiction. Suppose $x\in \mathrm{trap}(D_i)\cap \mathrm{trap}(D_j)$ and w.l.o.g. assume that $x\in \mathrm{lens}(a_i,c_i)\cap \mathrm{lens}(a_j,c_j)$ as shown in Figure~\ref{gg:trap-intersection-fig}. Connect $x$ to $a_i$, $c_i$, $a_j$, and $c_j$ ($a_i$ may be identified with $a_j$). As shown in the proof of Lemma~\ref{gg:angle-in-trap}, $\min\{\angle a_ixc_i,\allowbreak \angle a_jxc_j\}\allowbreak >\allowbreak 150^\circ$. Two configurations may arise: 
\begin{itemize}
 \item $\angle c_ixc_j \le 60^{\circ}$. In this case $|c_ic_j|\le \max\{|xc_i|,|xc_j|\}$. W.l.o.g. assume that $|xc_i|\le|xc_j|$ which implies that $|c_ic_j|\le |xc_j|$; see Figure ~\ref{gg:trap-intersection-fig}(a). Clearly $|xc_j|<|c_ja_j|$, and hence $|c_ic_j|<|c_ja_j|$. Thus, $D_j$ contains $c_i$ which contradicts Lemma~\ref{gg:center-in-lemma}.
  \item $\angle c_ixc_j > 60^{\circ}$. In this case $\angle a_ixc_j \le 60^\circ$ and $\angle a_jxc_i \le 60^\circ$, hence $|a_ic_j|\le \max\{|a_ix|,|c_jx|\}$ and $|a_jc_i|\le \max\{|a_jx|,|c_ix|\}$. Three configurations arise:

\begin{itemize}
 \item $|a_ix|<|c_jx|$, in this case $|a_ic_j|< |c_jx|<|c_ja_j|$ and hence $D_j$ contains $a_i$. See Figure~\ref{gg:trap-intersection-fig}(b).
  \item $|a_jx|<|c_ix|$, in this case $|a_jc_i|< |c_ix|< |c_ia_i|$ and hence $D_i$ contains $a_j$. 
  \item $|a_ix|\ge|c_jx|$ and $|a_jx|\ge|c_ix|$, in this case w.l.o.g. assume that $|a_ix|\le|a_jx|$. Thus $|a_ic_j|\le |a_ix|\le |a_jx|<|a_jc_j|$ which implies that $D_j$ contains $a_i$. See Figure~\ref{gg:trap-intersection-fig}(b).
\end{itemize}
All cases contradict Lemma~\ref{gg:D-empty}. 
\end{itemize}
\end{proof}

Recall that each blocking angle represents a trap. Thus, by Lemma \ref{gg:not-all-free-angles} and Lemma~\ref{gg:intersecting-trap}, we have the following corollary:

\begin{corollary}
\label{gg:one-blocked-angle}
Exactly one $\alpha_i$, where $1\le i\le 4$, is blocked.
\end{corollary}
Recall that $\alpha_j$ is blocked by $\alpha_i$, $a_ib_i$ is vertical line segment, $c_j$ is to the right of $a_ib_i$, and $x\in\mathrm{sc}(c_i,b_i)$. As a direct consequence of Corollary~\ref{gg:one-blocked-angle}, $\alpha_i$, $\alpha_k$, and $\alpha_l$ are free angles, where $1\le\allowbreak i,j,\allowbreak k,\allowbreak l\le\allowbreak 4$ and $i\neq j\neq k\neq l$. In addition, $c_k$ and $c_l$ are to the left of $a_ib_i$. It is obvious that $$\mathcal{X}\subseteq \mathrm{trap}(D_i)\cap D_k \cap D_l.$$

\begin{figure}[htb]
  \centering
  \includegraphics[width=.42\columnwidth]{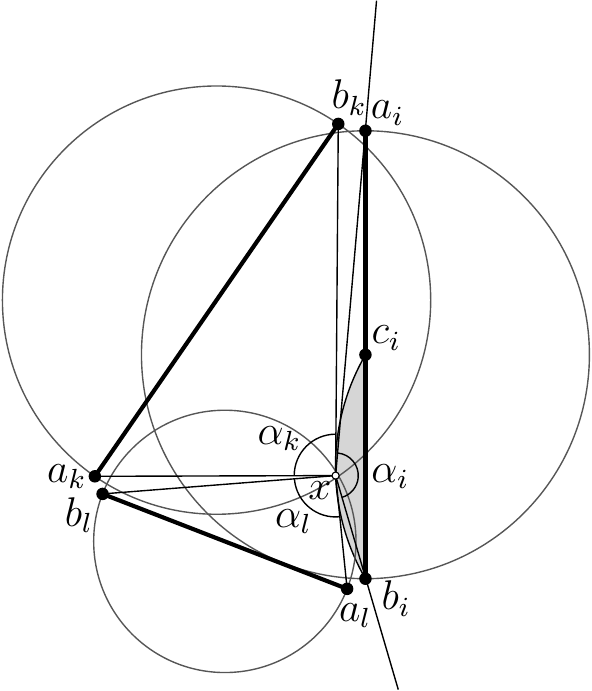}
 \caption{Illustration of Lemma~\ref{gg:intersecting-trap2}.}
  \label{gg:trap2-fig}
\end{figure}

\begin{lemma}
\label{gg:intersecting-trap2}
For a blocking angle $\alpha_i$ and free angles $\alpha_k$ and $\alpha_l$, $\mathrm{trap}(D_i)\allowbreak\cap D_k \allowbreak\cap D_l=\emptyset$. 
\end{lemma}
\begin{proof}
Since $\alpha_i$ is a blocking angle and $\alpha_k$, $\alpha_l$ are free angles, $c_k$ and $c_l$ are on the same side of $a_ib_i$.
 By contradiction, suppose that $x\in \mathrm{trap}(D_i)\cap D_j \cap D_k$. See Figure~\ref{gg:trap2-fig}. It is obvious that $\max\{|xa_i|, |xb_i|\}< |a_ib_i|$, $\max\{|xa_k|, |xb_k|\}< |a_kb_k|$, and $\max\{|xa_l|, |xb_l|\}<|a_lb_l|$. By Lemma~\ref{gg:angle-in-trap}, $\alpha_i \ge 150^\circ$. In addition $\alpha_k, \alpha_l \allowbreak\ge\allowbreak 90^\circ$. Thus, $\max\{\angle a_ixb_k,\allowbreak \angle a_kxb_l,\allowbreak \angle a_lxb_i\}\allowbreak\le\allowbreak 30^\circ$. Hence, $|a_ib_k|\allowbreak<\allowbreak\max\{|xa_i|,|xb_k|\}$, $|a_kb_l|<\max\{|xa_k|,|xb_l|\}$, and $|a_lb_i|<\max\{|xa_l|,\allowbreak|xb_i|\}$. Therefore, $\max\{|a_ib_k|,\allowbreak|a_kb_l|,\allowbreak |a_lb_i|\}\allowbreak<\max\allowbreak\{|a_ib_i|,\allowbreak|a_kb_k|,|a_lb_l|\}$. In addition $\delta=(a_i,\allowbreak b_i,a_l,\allowbreak b_l,a_k,\allowbreak b_k,a_i)$ is a cycle and at least one of $(a_i,b_k)$, $(a_k,b_l)$ and $(a_l,b_i)$ does not belong to $\mathcal{T}$. This contradicts Lemma~\ref{gg:not-mst-edge}. 
\end{proof}
Thus, $\mathcal{X}=\emptyset$; which contradicts the fact that $x\in \mathcal{X}$. This completes the proof of Theorem~\ref{gg:four-circle-theorem}.
\addtocontents{toc}{\protect\setcounter{tocdepth}{1}}
\subsection{Lower Bounds}
\addtocontents{toc}{\protect\setcounter{tocdepth}{2}}
\label{gg:lower-bounds-section}

In this section we present some lower bounds on the size of a maximum matching in \kGG{2}{}, \kGG{1}{}, and \kGG{0}{}.

\begin{theorem}
 \label{gg:matching-2GG}
For a set $P$ of an even number of points, \kGG{2}{} has a perfect matching.
\end{theorem}
\begin{proof}
First we show that by removing a set $S$ of points from \kGG{2}{}, at most $|S|+1$ components are generated. Then we show that at least one of these components must be even. Using Theorem~\ref{gg:Tutte}, we conclude that \kGG{2}{} has a perfect matching.

Let $S$ be a set vertices removed from \kGG{2}{}, and let $\mathcal{C}=\{C_1, \dots,\allowbreak C_{m(S)}\}$ be the resulting $m(S)$ components. Then, $\mathcal{P}=\{V(C_1),\dots, V(C_{m(S)})\}$ is a partition of the vertices in $P\setminus S$. 

\begin{paragraph}{Claim 1.}$m(S)\le |S|+1$.\end{paragraph} 

Let $G(\mathcal{P})$ be the complete graph with vertex set $\mathcal{P}$ which is constructed as described above. Let $\mathcal{T}$ be the set of all edges between points in $P$ corresponding to the edges of $MST(G(\mathcal{P}))$ and let $\mathcal{D}$ be the set of disks corresponding to the edges of $\mathcal{T}$. It is obvious that $\mathcal{T}$ contains $m(S)-1$ edges and hence $|\mathcal{D}|=m(S)-1$. Let $F=\{(p,D):p\in S, D\in \mathcal{D}, p\in D\}$ be the set of all (point, disk) pairs where $p\in S$, $D\in \mathcal{D}$, and $p$ is inside $D$. By Theorem~\ref{gg:four-circle-theorem} each point in $S$ can be inside at most three disks in $\mathcal{D}$. Thus, $|F|\le 3\cdot|S|$.
Now we show that each disk in $\mathcal{D}$ contains at least three points of $S$ in its interior.  
Consider any disk $D\in \mathcal{D}$ and let $e=(a,b)$ be the edge of $\mathcal{T}$ corresponding to $D$. By Lemma~\ref{gg:D-empty}, $D$ does not contain any point of $P\setminus S$. Therefore, $D$ contains at least three points of $S$, because otherwise $(a,b)$ is an edge in \kGG{2}{} which contradicts the fact that $a$ and $b$ belong to different components in $\mathcal{C}$. Thus, each disk in $\mathcal{D}$ has at least three points of $S$. That is, $3\cdot|\mathcal{D}|\le|F|$. Therefore, $3(m(S)-1)\le |F|\le 3|S|$, and hence $m(S)\le |S|+1$.

\begin{paragraph}{Claim 2.}$o(\mathcal{C})\le |S|$.\end{paragraph}

By Claim 1, $|\mathcal{C}|=m(S)\le |S|+1$. If $|\mathcal{C}|\le |S|$, then $o(\mathcal{C})\le |S|$. Assume that $|\mathcal{C}|=|S|+1$. Since $P=S\cup \{\bigcup^{|S|+1}_{i=1}{V(C_i)}\}$, the total number of vertices of $P$ is equal to $n=|S|+\sum_{i=1}^{|S|+1}{|V(C_i)|}$. Consider two cases where (i) $|S|$ is odd, (ii) $|S|$ is even. In both cases if all the components in $\mathcal{C}$ are odd, then $n$ is odd; contradicting our assumption that $P$ has an even number of vertices. Thus, $\mathcal{C}$ contains at least one even component, which implies that $o(\mathcal{C})\le |S|$.

Finally, by Claim 2 and Theorem~\ref{gg:Tutte}, we conclude that \kGG{2}{} has a perfect matching.
\end{proof}

\begin{theorem}
\label{gg:matching-1GG}
For every set $P$ of $n$ points, \kGG{1}{} has a matching of size at least $\frac{2(n-1)}{5}$.
\end{theorem}

\begin{proof}
Let $S$ be a set of vertices removed from \kGG{1}{}, and let $\mathcal{C}=\{C_1, \dots, C_{m(S)}\}$ be the resulting $m(S)$ components. Then, $\mathcal{P}=\{V(C_1),\allowbreak \dots, V(C_{m(S)})\}$ is a partition of the vertices in $P\setminus S$. Note that $o(\mathcal{C})\le m(S)$.
Let $M^*$ be a maximum matching in \kGG{1}{}. By Theorem~\ref{gg:Berge}, 

\begin{align}
\label{gg:align0}
|M^*|&= \frac{1}{2}(n-\text{def}(\text{\kGG{1}{}})),
\end{align}
where
\begin{align}
\label{gg:align1}
\text{def}(\text{\kGG{1}{}})&= \max\limits_{S\subseteq P}(o(\mathcal{C})-|S|)\le \max\limits_{S\subseteq P}(|\mathcal{C}|-|S|)= \max\limits_{0\le |S|\le n}(m(S)-|S|).
\end{align}
Define $G(\mathcal{P})$, $\mathcal{T}$, $\mathcal{D}$, and $F$ as in the proof of Theorem~\ref{gg:matching-2GG}. By Theorem~\ref{gg:four-circle-theorem}, $|F|\le 3\cdot|S|$.
By the same reasoning as in the proof of Theorem~\ref{gg:matching-2GG}, each disk in $\mathcal{D}$ has at least two points of $S$ in its interior. Thus, $2|\mathcal{D}|\le|F|$. Therefore, $2(m(S)-1)\le |F| \le 3|S|$, and hence
\begin{equation}
\label{gg:ineq1}
 m(S)\le\frac{3|S|}{2}+1.
\end{equation} 
In addition, $|S|+m(S)=|S|+|\mathcal{C}|\le |P|=n$, and hence
\begin{equation}
\label{gg:ineq2}               
m(S)\le n-|S|.
\end{equation}
By Inequalities~(\ref{gg:ineq1}) and ~(\ref{gg:ineq2}), 
\begin{equation}
\label{gg:ineq3}               
m(S)\le \min\left\{\frac{3|S|}{2}+1, n-|S|\right\}.
\end{equation}
Thus, by (\ref{gg:align1}) and (\ref{gg:ineq3})
\begin{align}
\label{gg:align2}
\text{def}(\text{\kGG{1}{}})&\le \max\limits_{0\le |S|\le n}(m(S)-|S|)\le \max\limits_{0\le |S|\le n}\left\{\min\left\{\frac{3|S|}{2}+1, n-|S|\right\}-|S|\right\}\nonumber\\
&= \max\limits_{0\le |S|\le n}\left\{\min\left\{\frac{|S|}{2}+1, n-2|S|\right\}\right\}= \frac{n+4}{5},
\end{align}

where the last equation is achieved by setting $\frac{|S|}{2}+1$ equal to $n-2|S|$, which implies $|S|=\frac{2(n-1)}{5}$. Finally by substituting (\ref{gg:align2}) in Equation (\ref{gg:align0}) we have
$$
|M^*|\ge \frac{2(n-1)}{5}.
$$
\end{proof}

By similar reasoning as in the proof of Theorem~\ref{gg:matching-1GG} we have the following Theorem.

\begin{figure}[htb]
  \centering
  \includegraphics[width=.56\columnwidth]{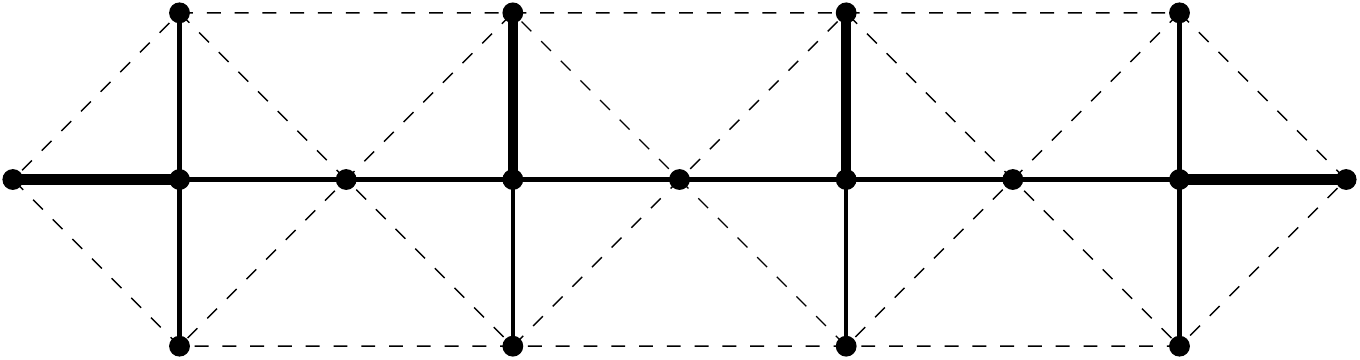}
 \caption{A \kGG{0}{} of  $n = 17$ points with a maximum matching of size $\frac{n-1}{4}=4$ (bold edges). The dashed edges do not belong to the graph because any of their corresponding closed disks has a point on its boundary.}
  \label{gg:tight-0GG}
\end{figure}

\begin{theorem}
\label{gg:matching-0GG}
For every set $P$ of $n$ points, \kGG{0}{} has a matching of size at least $\frac{n-1}{4}$.
\end{theorem}

The bound in Theorem~\ref{gg:matching-0GG} is tight, as can be seen from the graph in Figure~\ref{gg:tight-0GG}, for which the maximum matching has size $\frac{n-1}{4}$. Actually this is a Gabriel graph of maximum degree four which is a tree. The dashed edges do not belong to \kGG{0}{} because any closed disk which has one of these edges as diameter has a point on its boundary. Observe that each edge in any matching is adjacent to one of the vertices of degree four.

\begin{paragraph}{Note:}For a point set $P$, let $\nu_k(P)$ and $\alpha_k(P)$ respectively denote the size of a maximum matching and a maximum independent set in \kGG{k}{}. For every edge in the maximum matching, at most one of its endpoints can be in the maximum independent set. Thus,$$\alpha_k(P)\le |P| - \nu_k(P).$$
By combining this formula with the results of Theorems ~\ref{gg:matching-0GG}, \ref{gg:matching-1GG}, \ref{gg:matching-2GG}, respectively, we have $\alpha_0(P)\le \frac{3n+1}{4}$, $\alpha_1(P)\le \frac{3n+2}{5}$, and $\alpha_2(P)\le \lceil\frac{n}{2}\rceil$. The \kGG{0}{} graph in Figure~\ref{gg:tight-0GG} has an independent set of size $\frac{3n+1}{4}=13$, which shows that this bound is tight for \kGG{0}{}. On the other hand, \kGG{0}{} is planar and every planar graph is 4-colorable; which implies that $\alpha_0(P)\ge \lceil\frac{n}{4}\rceil$. There are some examples of \kGG{0}{} in \cite{Matula1980} such that $\alpha_0(P)= \lceil\frac{n}{4}\rceil$, which means that this bound is tight as well.
\end{paragraph}

\section{Blocking Higher-Order Gabriel Graphs}
\label{gg:blocking-section}
In this section we consider the problem of blocking higher-order Gabriel graphs. Recall that a point set $K$ blocks \kGG{k}{(P)} if in \kGG{k}{$(P\cup K)$} there is no edge connecting two points in $P$. 

\begin{theorem}
\label{gg:blocking-thr1}
For every set $P$ of $n$ points, at least $\lceil\frac{n-1}{3}\rceil$ points are necessary to block \kGG{0}{(P)}.
\end{theorem}
\begin{proof}
Let $K$ be a set of $m$ points which blocks \kGG{0}{(P)}. Let $G(\mathcal{P})$ be the complete graph with vertex set $\mathcal{P}=P$. Let $\mathcal{T}$ be a minimum spanning tree of $G(\mathcal{P})$ and let $\mathcal{D}$ be the set of closed disks corresponding to the edges of $\mathcal{T}$. Since $G(\mathcal{P})$ has $n$ vertices, $\mathcal{T}$ has $n-1$ edges. Thus, $|\mathcal{D}|=n-1$. By Lemma~\ref{gg:D-empty} each disk $\CD{a}{b}\in\mathcal{D}$ does not contain any point of $P\setminus\{a,b\}$, thus,  $\mathcal{T}\subseteq\text{\kGG{0}{(P)}}$. To block each edge of $\mathcal{T}$, corresponding to a disk in $\mathcal{D}$, at least one point is necessary. By Theorem~\ref{gg:four-circle-theorem} each point in $K$ can lie in at most three disks of $\mathcal{D}$. Therefore, $m\ge\lceil\frac{n-1}{3}\rceil$, which implies that at least $\lceil\frac{n-1}{3}\rceil$ points are necessary to block all the edges of $\mathcal{T}$ and hence \kGG{0}{(P)}.
\end{proof}
\begin{figure}[htb]
  \centering
\setlength{\tabcolsep}{0in}
  $\begin{tabular}{cc}
 \multicolumn{1}{m{.75\columnwidth}}{\centering\includegraphics[width=.6\columnwidth]{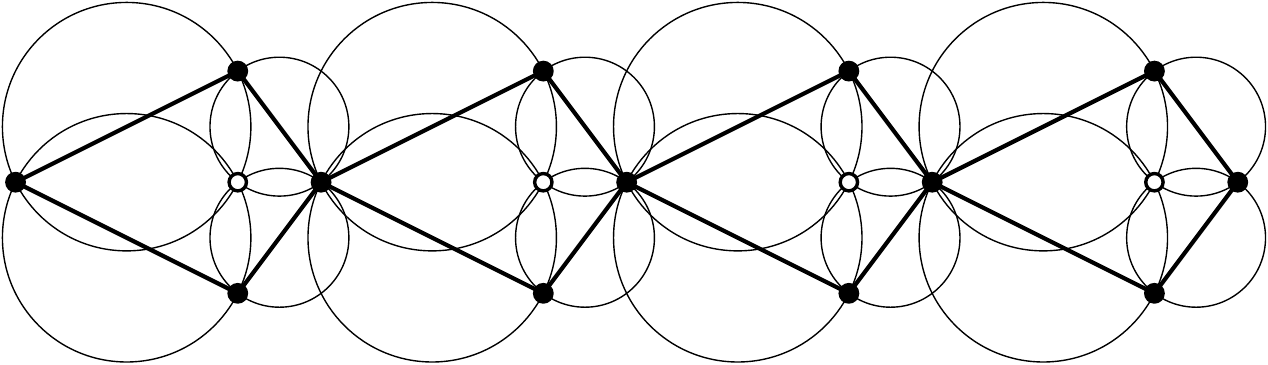}}
&\multicolumn{1}{m{.25\columnwidth}}{\centering\includegraphics[width=.22\columnwidth]{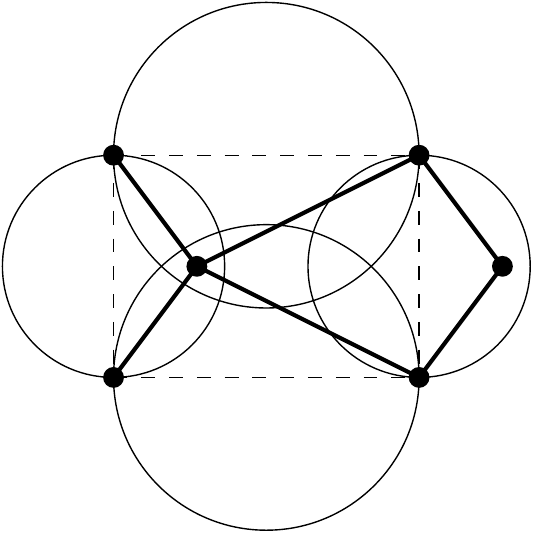}}\\
(a)&(b)
\end{tabular}$
  \caption{(a) \kGG{0}{} graph of $n=13$ points (in bold edges) which is blocked by $\lceil\frac{n-1}{3}\rceil=4$ white points, (b) dashed edges do not belomg to \kGG{0}{}.}
\label{gg:blocking-fig}
\end{figure}

Figure~\ref{gg:blocking-fig}(a) shows a \kGG{0}{} with $n=13$ (black) points which is blocked by $\lceil\frac{n-1}{3}\rceil=4$ (white) points. Note that all the disks, corresponding to the edges of every cycle, intersect at the same point in the plane (where we have placed the white points). As shown in Figure~\ref{gg:blocking-fig}(b), the dashed edges do not belong to \kGG{0}{}. Thus, the lower bound provided by Theorem~\ref{gg:blocking-thr1} is tight. It is easy to generalize the result of Theorem~\ref{gg:blocking-thr1} to higher-order Gabriel graphs. Since in a \kGG{k}{} we need at least $k+1$ points to block an edge of $\mathcal{T}$ and each point can be inside at most three disks in $\mathcal{D}$, we have the following corollary:

\begin{corollary}
For every set $P$ of $n$ points, at least $\lceil\frac{(k+1)(n-1)}{3}\rceil$ points are necessary to block \kGG{k}{(P)}.
\end{corollary}

In \cite{Aronov2013} the authors showed that every Gabriel graph can be blocked by a set $K$ of $n-1$ points by putting a point slightly to the right of each point of $P$, except for the rightmost one. Every disk with diameter determined by two points of $P$ will contain a point of $K$. Using a similar argument one can block a \kGG{k}{} by putting $k+1$ points slightly to the right of each point of $P$, except for the rightmost one. Thus,

\begin{corollary}
 For every set $P$ of $n$ points, there exists a set of $(k+1)(n-1)$ points that blocks \kGG{k}{(P)}.
\end{corollary}

Note that this upper bound is tight, because if the points of $P$ are on a line, the disks representing the minimum spanning tree are disjoint and each disk needs $k+1$ points to block the corresponding edge.

\section{Conclusions}
\label{gg:conclusion}
In this chapter, we considered the bottleneck and perfect matching admissibility of higher-order Gabriel graphs. We proved that
\begin{itemize}
  \item \kGG{9}{} contains a Euclidean bottleneck matching of $P$ and \kGG{8}{} may not have any.
  \item \kGG{0}{} has a matching of size at least $\frac{n-1}{4}$ and this bound is tight.
  \item \kGG{1}{} has a matching of size at least $\frac{2(n-1)}{5}$.
  \item \kGG{2}{} has a perfect matching.
    \item At least $\lceil\frac{n-1}{3}\rceil$ points are necessary to block \kGG{0}{} and this bound is tight.
  \item $\lceil\frac{(k+1)(n-1)}{3}\rceil$ points are necessary and $(k+1)(n-1)$ points are sufficient to block \kGG{k}{}.
\end{itemize}

\bibliographystyle{abbrv}
\bibliography{../thesis}

%% file: chapters/ch3-tdmatching.tex
\chapter{Matching in TD-Delaunay Graphs}
\label{ch:td}

We consider an extension of the triangular-distance Delaunay graphs (TD-Delaunay) on a set $P$ of points in general position in the plane. In TD-Delaunay, the convex distance is defined by a fixed-oriented equilateral triangle $\ConvexShape$, and there is an edge between two points in $P$ if and only if there is an empty homothet of $\ConvexShape$ having the two points on its boundary. We consider higher-order triangular-distance Delaunay graphs, namely \kTD{k}{}, which contains an edge between two points if the interior of the smallest homothet of $\ConvexShape$ having the two points on its boundary contains at most $k$ points of $P$. We consider the connectivity, Hamiltonicity and perfect-matching admissibility of \kTD{k}{}. Finally we consider the problem of blocking the edges of \kTD{k}{}.

\vspace{10pt}
This chapter was first published in the proceedings of the First International Conference on Algorithms and Discrete Applied Mathematics (CALDAM 2015)~\cite{Biniaz2015-hotd-CALDAM},
and was subsequently published in the journal of Computational Geometry: Theory and Applications~\cite{Biniaz2015-hotd-CGTA}. 

\vspace{10pt}
We also obtained lower and upper bounds on the size of maximum matching in \kTD{0}{} graphs, which are not included in this thesis. The results were first published in the proceedings of the 7th International Workshop on Algorithms and Computation (WALCOM 2013)~\cite{Babu2013-WALCOM},
and was subsequently invited and accepted to a special issue of Theoretical Computer Science~\cite{Babu2014}.
\section{Introduction}
The {\em triangular-distance Delaunay graph} of a point set $P$ in the plane, TD-Delaunay for short, was introduced by Chew \cite{Chew1989}. A TD-Delaunay is a graph whose convex distance function is defined by a fixed-oriented equilateral triangle. Let $\ConvexShape$ be a downward equilateral triangle whose barycenter is the origin and one of whose vertices is on the negative $y$-axis. A {\em homothet} of $\ConvexShape$ is obtained by scaling $\ConvexShape$ with respect to the origin by some factor $\mu\ge 0$, followed by a translation to a point $b$ in the plane: $b+\mu\ConvexShape=\{b+\mu a:a\in\ConvexShape\}$.
In the TD-Delaunay graph of $P$, there is a straight-line edge between two points $p$ and $q$ if and only if there exists a homothet of $\ConvexShape$ having $p$ and $q$ on its boundary and whose interior does not contain any point of $P$. In other words, $(p,q)$ is an edge of TD-Delaunay graph if and only if there exists an empty downward equilateral triangle having $p$ and $q$ on its boundary. In this case, we say that the edge $(p,q)$ has the {\em empty triangle property}. 

We say that $P$ is in general position if the line passing through any two points from $P$ does not make angles $0^\circ$, $60^\circ$, and $120^\circ$ with horizontal. In this chapter we consider point sets in general position and our results assume this pre-condition.
If $P$ is in general position, then the TD-Delaunay graph is a planar graph, see \cite{Bose2010}.
We define $t(p,q)$ as the smallest homothet of $\ConvexShape$ having $p$ and $q$ on its boundary. See Figure~\ref{td:TD}(a). Note that $t(p,q)$ has one of $p$ and $q$ at a vertex, and the other one on the opposite side. Thus,

\begin{observation}
\label{td:side-point-obs}
 Each side of $t(p,q)$ contains either $p$ or $q$.
\end{observation}

A graph $G$ is {\em connected} if there is a path between any pair of vertices in $G$. Moreover, $G$ is $k$-{\em connected} if there does not exist a set of at most $k-1$ vertices whose removal disconnects $G$. In case $k=2$, $G$ is called {\em biconnected}. In other words a graph $G$ is biconnected iff there is a simple cycle between any pair of its vertices. A {\em matching} in $G$ is a set of edges in $G$ without common vertices. A {\em perfect matching} is a matching which matches all the vertices of $G$. A {\em Hamiltonian cycle} in $G$ is a cycle (i.e., closed loop) through $G$ that visits each vertex of $G$ exactly once. For $H\subseteq G$ we denote the {\em bottleneck} of $H$, i.e., the length of the longest edge in $H$, by $\lambda(H)$.

Let $K_n(P)$ be a complete edge-weighted geometric graph on a point set $P$
which contains a straight-line edge between any pair of points in $P$. For an edge $(p,q)$ in $K_n(P)$ let $w(p,q)$ denote the weight of $(p,q)$.
A {\em bottleneck matching} (resp. {\em bottleneck Hamiltonian cycle}) in $P$ is defined to be a perfect matching (resp. Hamiltonian cycle) in $K_n(P)$, in which the weight of the maximum-weight edge is minimized. A {\em bottleneck biconnected spanning subgraph} of $P$ is a spanning subgraph, $G(P)$, of $K_n(P)$ which is biconnected and the weight of the longest edge in $G(P)$ is minimized. 

A tight lower bound on the size of a maximum matching in a TD-Delaunay graph, i.e. \kTD{0}{}, is presented in \cite{Babu2014}. In this chapter we study higher-order TD-Delaunay graphs. The {\em order-k TD-Delaunay graph} of a point set $P$, denoted by \kTD{k}{}, is a geometric graph which has an edge $(p,q)$ iff the interior of $t(p,q)$ contains at most $k$ points of $P$; see Figure~\ref{td:TD}(b). The standard TD-Delaunay graph corresponds to \kTD{0}{}. We consider graph-theoretic properties of higher-order TD-Delaunay graphs, such as connectivity, Hamiltonicity, and perfect-matching admissibility. We also consider the problem of blocking TD-Delaunay graphs.

\begin{figure}[htb]
  \centering
\setlength{\tabcolsep}{0in}
  $\begin{tabular}{ccc}
\multicolumn{1}{m{.33\columnwidth}}{\centering\includegraphics[width=.31\columnwidth]{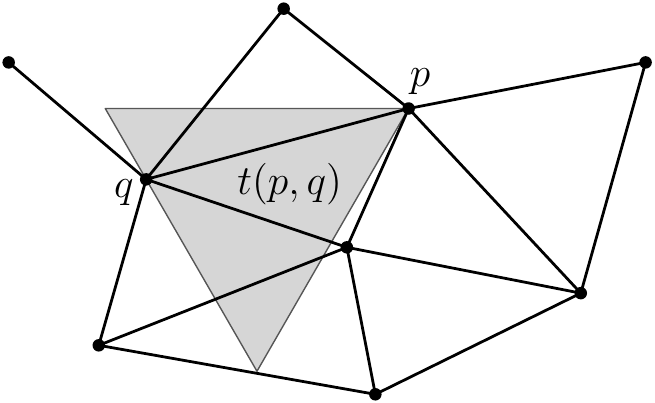}}
&\multicolumn{1}{m{.33\columnwidth}}{\centering\includegraphics[width=.31\columnwidth]{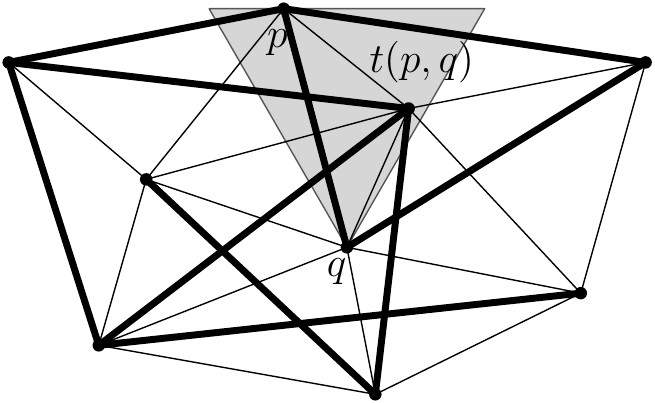}} &\multicolumn{1}{m{.33\columnwidth}}{\centering\includegraphics[width=.31\columnwidth]{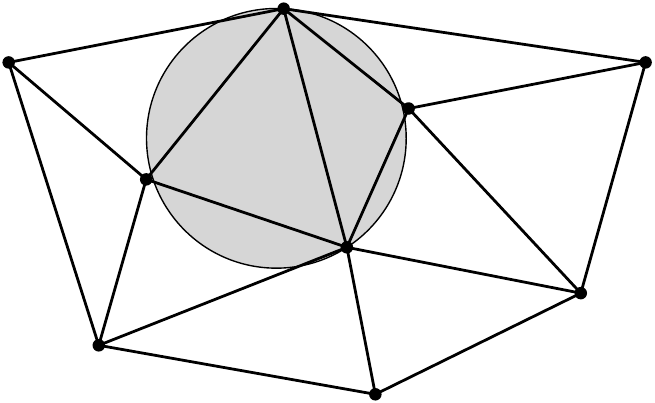}}
\\
(a) & (b)& (c)
\end{tabular}$
  \caption{(a) Triangular-distance Delaunay graph (\kTD{0}{}), (b) \kTD{1}{} graph, the light edges belong to \kTD{0}{} as well, and (c) Delaunay triangulation.}
\label{td:TD}
\end{figure}
\addtocontents{toc}{\protect\setcounter{tocdepth}{1}}
\subsection{Previous Work}
\addtocontents{toc}{\protect\setcounter{tocdepth}{2}}
\label{td:previous-work}
A {\em Delaunay triangulation} (DT) of $P$ (which does not have any four co-circular
points) is a graph whose distance function is defined by a fixed circle {\footnotesize $\bigcirc$} centered at the origin. DT has an edge between two points $p$ and $q$ iff there exists a homothet of {\footnotesize $\bigcirc$} having $p$ and $q$ on its boundary and whose interior does not contain any point of $P$; see Figure~\ref{td:TD}(c). In this case the edge $(p,q)$ is said to have the {\em empty circle property}. The {\em order-k
Delaunay Graph} on $P$, denoted by \kDT{k}{}, is defined to have an edge $(p, q)$ iff there exists a homothet of {\footnotesize $\bigcirc$} having $p$ and $q$ on its boundary and whose interior contains at most $k$ points of $P$. The standard Delaunay triangulation corresponds to \kDT{0}{}.  

For each pair of points $p,q\in P$ let $D[p,q]$ be the closed disk having $pq$ as diameter. A {\em Gabriel Graph} on $P$ is a geometric graph which has an edge between two points $p$ and $q$ iff $D[p,q]$ does not contain any point of $P\setminus\{p,q\}$. The {\em order-$k$ Gabriel Graph} on $P$, denoted by \kGG{k}{}, is defined to have an edge $(p,q)$ iff $D[p,q]$ contains at most $k$ points of $P\setminus\{p,q\}$.

For each pair of points $p,q\in P$, let $L(p,q)$ be the intersection of the two open disks with radius $|pq|$ centered at $p$ and $q$, where $|pq|$ is the Euclidean distance between $p$ and $q$. A {\em Relative Neighborhood Graph} on $P$ is a geometric graph which has an edge between two points $p$ and $q$ iff $L(p,q)$ does not contain any point of $P$. The {\em order-$k$ Relative Neighborhood Graph} on $P$, denoted by \kRNG{k}{}, is defined to have an edge $(p,q)$ iff $L(p,q)$ contains at most $k$ points of $P$. It is obvious that for a fixed point set, $\text{\kRNG{k}{}}$ is a subgraph of
$\text{\kGG{k}{}}$, and $\text{\kGG{k}{}}$ is a subgraph of \text{\kDT{k}{}}.

The problem of determining whether an order-$k$ geometric graph always has a (bottleneck) perfect matching or a (bottleneck) Hamiltonian cycle is of interest. In order to show the importance of this problem we provide the following example. Gabow and Tarjan~\cite{Gabow1988} showed that a bottleneck matching of maximum cardinality in a graph can be computed in $O(m\cdot(n\log n)^{0.5})$ time, where $m$ is the number of edges in the graph. Using their algorithm, a bottleneck perfect matching of a point set can be computed in $O(n^2\cdot(n\log n)^{0.5})$ time; note that the complete graph on $n$ points has $\Theta(n^2)$ edges. Chang et al.~\cite{Chang1992} showed that a bottleneck perfect matching of a point set is contained in \kDT{16}{}; this graph has $\Theta(n)$ edges and can be computed in $O(n\log n)$ time. Thus, by running the algorithm of Gabow and Tarjan on \kDT{16}{}, a bottleneck perfect matching of a point set can be computed in $O(n\cdot (n\log n)^{0.5})$ time.

If for each edge $(p,q)$ in $K_n(P)$, $w(p,q)$ is equal the Euclidean distance
between $p$ and $q$, then Chang et al. \cite{Chang1991, Chang1992b, Chang1992} proved that a bottleneck biconnected spanning graph, bottleneck perfect matching, and bottleneck Hamiltonian cycle of $P$ are contained in \kRNG{1}{}, \kRNG{16}{}, \kRNG{19}{}, respectively. This implies that \kRNG{16}{} has a perfect matching and \kRNG{19}{} is Hamiltonian. Since \kRNG{k}{} is a subgraph of \kGG{k}{}, the same results hold for \kGG{16}{} and \kGG{19}{}. It is known that \kGG{k}{} is $(k+1)$-connected \cite{Bose2013} and 
\kGG{10}{} (and hence \kDT{10}{}) is Hamiltonian~\cite{Kaiser2015}. Dillencourt showed that every Delaunay triangulation (\kDT{0}{}) admits a perfect matching \cite{Dillencourt1990} but it can fail to be Hamiltonian \cite{Dillencourt1987a}. 

Given a geometric graph $G(P)$ on a set $P$ of $n$ points, we say that a set $K$ of points {\em blocks} $G(P)$ if in $G(P\cup K)$ there is no edge connecting two points in $P$. Actually $P$ is an independent set in $G(P\cup K)$.
Aichholzer et al.~\cite{Aichholzer2013-blocking} considered the problem of blocking the Delaunay triangulation (i.e. \kDT{0}{}) for a given point set $P$ in which no four points are co-circular. They show that $\frac{3n}{2}$ points are sufficient to block \kDT{0}{} and $n-1$ points are necessary. To block a Gabriel graph, $n-1$ points are sufficient \cite{Aronov2013}.

\addtocontents{toc}{\protect\setcounter{tocdepth}{1}}
\subsection{Our Results}
\addtocontents{toc}{\protect\setcounter{tocdepth}{2}}
\label{td:our-results}
We consider some graph-theoretical properties of higher-order triangular distance Delaunay graphs on a given set $P$ of $n$ points in general position in the plane.  
We show for which values of $k$, \kTD{k}{} contains a bottleneck biconnected spanning graph, a bottleneck Hamiltonian cycle, and a (bottleneck) perfect-matching; for the bottleneck structures we assume that the weight of any edge $(p,q)$ in $K_n(P)$ is equal to the area of the smallest homothet of $\ConvexShape$ having $p$ and $q$ on its boundary. In Section~\ref{td:connectivity} we prove that every \kTD{k}{} graph is $(k+1)$-connected. In addition we show that a bottleneck biconnected spanning graph of $P$ is contained in \kTD{1}{}. Using a similar approach as in \cite{Abellanas2009, Chang1991}, in Section~\ref{td:Hamiltonicity} we show that a bottleneck Hamiltonian cycle of $P$ is contained in \kTD{7}{}. We also show a configuration of a point set $P$ such that \kTD{5}{} fails to have a bottleneck Hamiltonian cycle. In Section~\ref{td:matching} we prove that a bottleneck perfect matching of $P$ is contained in \kTD{6}{}, and we show that for some point set $P$, \kTD{5}{} does not have a bottleneck perfect matching. In Section~\ref{td:matching2} we prove that \kTD{2}{} has a perfect matching and \kTD{1}{} has a matching of size at least $\frac{2(n-1)}{5}$. In Section~\ref{td:blocking-section} we consider the problem of blocking \kTD{k}{}. We show that at least $\lceil\frac{n-1}{2}\rceil$ points are necessary and $n-1$ points are sufficient to block a \kTD{0}{}. The open problems and concluding remarks are presented in Section~\ref{td:conclusion}.

\section{Preliminaries}
\label{td:preliminaries}

Bonichon et al. \cite{Bonichon2010} showed that the half-$\Theta_6$ graph of a point set $P$ in the plane is equal to the TD-Delaunay graph of $P$. They also showed that every plane triangulation is TD-Delaunay realizable. 

\begin{figure}[htb]
  \begin{center}
\includegraphics[width=.38\textwidth]{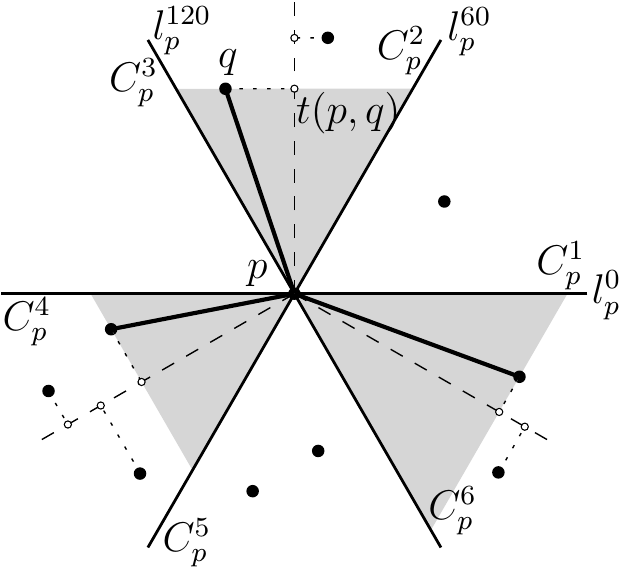}
  \end{center}
  \caption{The construction of the TD-Delaunay graph.}
\label{td:cones}
\end{figure}

The half-$\Theta_6$ graph (or equivalently a TD-Delaunay graph) on a point set $P$ can be constructed in the following way. For each point $p$ in $P$, let $l_p$ be the horizontal line through $p$. Define $l_p^{\gamma}$ as the line obtained by rotating $l_p$ by $\gamma$-degrees in counter-clockwise direction around $p$. Actually $l_p^0=l_p$. Consider three lines $l_p^{0}$, $l_p^{60}$, and $l_p^{120}$ which partition the plane into six disjoint cones with apex $p$. Let $C_p^1, \dots, C_p^6$ be the cones in counter-clockwise order around $p$ as shown in Figure~\ref{td:cones}. $C_p^1$, $C_p^3$, $C_p^5$ will be
referred to as {\em odd cones}, and $C_p^2$, $C_p^4$, $C_p^6$ will be referred to as {\em even cones}. For each even cone $C_p^i$, connect $p$ to the ``nearest'' point $q$ in $C_p^i$. The {\em distance} between $p$ and $q$, $d(p,q)$, is defined as the Euclidean distance between $p$ and the orthogonal projection of $q$ onto the bisector of $C_p^i$. See Figure~\ref{td:cones}. The resulting graph is the half-$\Theta_6$ graph which is defined by even cones \cite{Bonichon2010}. Moreover, the resulting graph is the TD-Delaunay graph defined with respect to homothets of $\ConvexShape$. By considering the odd cones, another half-$\Theta_6$ graph is obtained. The well-known $\Theta_6$ graph is the union of half-$\Theta_6$ graphs defined by odd and even cones. 

Recall that $t(p,q)$ is the smallest homothet of $\ConvexShape$ having $p$ and $q$ on its boundary. In other words, $t(p,q)$ is the smallest downward equilateral triangle through $p$ and $q$. Similarly we define $t'(p,q)$ as the smallest upward equilateral triangle having $p$ and $q$ on its boundary. It is obvious that the even cones correspond to downward triangles and odd cones correspond to upward triangles.   
We define an order on the equilateral triangles: for each two equilateral triangles $t_1$ and $t_2$ we say that $t_1\prec t_2$ if the area of $t_1$ is less than the area of $t_2$. Since the area of $t(p,q)$ is directly related to $d(p,q)$, 
$$d(p,q)<d(r,s) \quad\text{ if and only if }\quad t(p,q)\prec t(r,s).$$

\begin{figure}[htb]
  \centering
\setlength{\tabcolsep}{0in}
  $\begin{tabular}{ccc}
\multicolumn{1}{m{.33\columnwidth}}{\centering\includegraphics[width=.23\columnwidth]{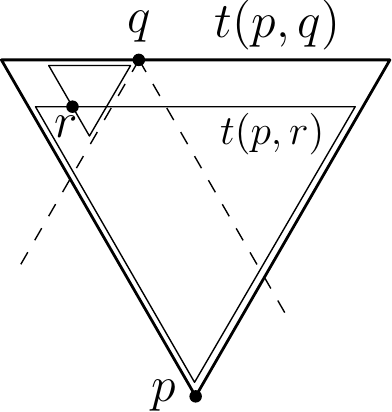}}
&\multicolumn{1}{m{.33\columnwidth}}{\centering\includegraphics[width=.23\columnwidth]{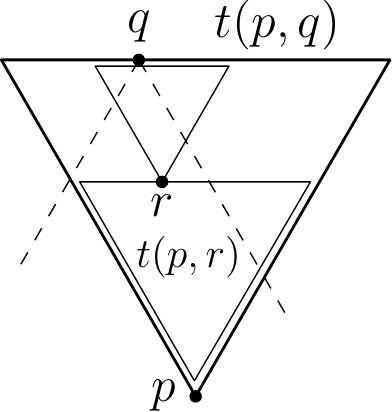}} &\multicolumn{1}{m{.33\columnwidth}}{\centering\includegraphics[width=.23\columnwidth]{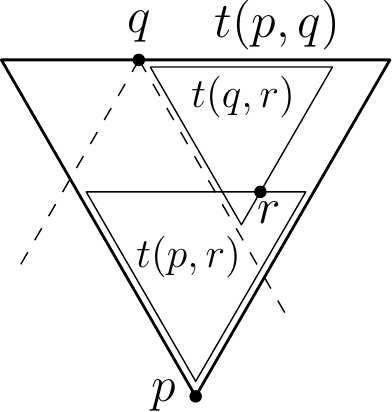}}
\\
(a) & (b)& (c)
\end{tabular}$
  \caption{Illustration of Observation~\ref{td:obs1}: the point $r$ is contained in $t(p,q)$. The triangles $t(p,r)$ and $t(q,r)$ are inside $t(p,q)$.}
\label{td:smaller-triangle-fig}
\end{figure}

For a set $\{t_1,\dots, t_m\}$ of equilateral triangles we define $\max\{t_1,\dots,\allowbreak t_m\}$ to be the triangle with the largest area. As shown in Figure~\ref{td:smaller-triangle-fig} we have the following observation:

\begin{observation}
\label{td:obs1}
 If $t(p,q)$ contains a point $r$, then $t(p,r)$ and $t(q,r)$ are contained in $t(p,q)$.
\end{observation}
As a direct consequence of Observation \ref{td:obs1}, if a point $r$ is contained in $t(p,q)$, then $\max\{t(p,r), \allowbreak t(q,r)\}\prec t(p,q)$. It is obvious that,

\begin{observation}
 \label{td:equal-triangles}
 For each two points $p,q\in P$, the area of $t(p,q)$ is equal to the area of $t'(p,q)$.
\end{observation}
Thus, we define $X(p,q)$ as a regular hexagon centered at $p$ which has $q$ on its boundary, and its sides are parallel to $l_p^0$, $l_p^{60}$, and $l_p^{120}$. 
\begin{observation}
\label{td:obs2}
 If $X(p,q)$ contains a point $r$, then $t(p,r)\prec t(p,q)$.
\end{observation}

We construct \kTD{k}{} as follows. For each point $p\in P$, imagine the six cones having their apex at $p$, as described earlier. Then connect $p$ to its $(k+1)$ nearest neighbors in each even cone around $p$. For each edge $(p,q)$ in \kTD{k}{} we define its {\em weight}, $w(p,q)$, to be equal to the area of $t(p,q)$. The resulting graphs is \kTD{k}{}, which has $O(kn)$ edges. The \kTD{k}{} can be constructed in $O(n\log n+kn\log \log n)$-time, using the algorithm introduced by Lukovszki~\cite{Lukovszki1999} for computing fault tolerant spanners.

For a graph $G=(V,E)$ and $K\subseteq V$, let $G-K$ be the subgraph obtained from $G$ by removing the vertices in $K$, and let $o(G-K)$ be the number of odd components in $G-K$. The following theorem by Tutte~\cite{Tutte1947} gives a characterization of the graphs which have a perfect matching: 

\begin{theorem}[Tutte~\cite{Tutte1947}] 
\label{td:Tutte} 
$G$ has a perfect matching if and only if $o(G-K)\le |K|$ for all $K\subseteq V$.
\end{theorem}

Berge~\cite{Berge1958} extended Tutte's theorem to a formula (known as Tutte-Berge formula) for the maximum size of a matching in a graph. In a graph $G$, the {\em deficiency}, $\text{def}_G(K)$, is $o(G-K)-|K|$. Let $\text{def}(G)=\max_{K\subseteq V}{\text{def}_G(K)}$.

\begin{theorem}[Tutte-Berge formula; Berge~\cite{Berge1958}] 
\label{td:Berge} 
The size of a maximum matching in $G$ is $$\frac{1}{2}(n-\mathrm{def}(G)).$$
\end{theorem}

For an edge-weighted graph $G$ we define the {\em weight sequence} of $G$, $\WS{G}$, as the sequence containing the weights of the edges of $G$ in non-increasing order. For two
graphs $G_1$ and $G_2$ we say that $\WS{G_1}\prec \WS{G_2}$ if $\WS{G_1}$ is lexicographically
smaller than $\WS{G_2}$. A graph $G_1$ is said to be less than a graph $G_2$ if $\WS{G_1}\prec \WS{G_2}$.

\section{Connectivity}
\label{td:connectivity}
In this section we consider the connectivity of higher-order triangular-distance Delaunay graphs.
\addtocontents{toc}{\protect\setcounter{tocdepth}{1}}
\subsection{$(k+1)$-connectivity}
\addtocontents{toc}{\protect\setcounter{tocdepth}{2}}
\label{td:connectivity-k-plus-1}
For a set $P$ of points in the plane, the TD-Delaunay graph, i.e., \kTD{0}{}, is not necessarily a triangulation \cite{Chew1989}, but it is connected and internally triangulated \cite{Babu2014}, i.e., all internal faces are triangles. As shown in Figure~\ref{td:TD}(a), \kTD{0}{} may not be biconnected. As a warm up exercise we show that every \kTD{k}{} is $(k+1)$-connected. 

\begin{theorem}
\label{td:k-connectivity-thr}
For every point set $P$ in general position in the plane, \kTD{k}{} is $(k+1)$-connected. In addition, for every $k$, there exists a point set $P$ such that \kTD{k}{} is not $(k+2)$-connected.
\end{theorem}
\begin{proof}
We prove the first part of this theorem by contradiction. Let $K$ be the set of (at most) $k$ vertices removed from \kTD{k}{}, and let $\mathcal{C}=\{C_1, C_2, \dots, C_m\}$, where $m>1$, be the resulting maximal connected components. Let $T$ be the set of all triangles defined by any pair of points belonging to different components, i.e., $T=\{t(a,b): a\in C_i, b\in C_j, i\neq j\}$. Consider the smallest triangle $\Tm\in T$. Assume that $\Tm$ is defined by two points $a$ and $b$, i.e., $\Tm=t(a,b)$, where $a\in C_i$, $b\in C_j$, and $i\neq j$.

{\em Claim 1}: $\Tm$ does not contain any point of $P\setminus K$ in its interior.
By contradiction, suppose that $\Tm$ contains a point $c\in P\setminus K$ in its interior. Three cases arise: (i) $c\in C_i$, (ii) $c\in C_j$, (iii) $c\in C_l$, where $l\neq i$ and $l\neq j$. In case (i) the triangle $t(c,b)$ between $C_i$ and $C_j$ is contained in $t(a,b)$. In case (ii) the triangle $t(a,c)$ between $C_i$ and $C_j$ is contained in $t(a,b)$. In case (iii) both triangles $t(a,c)$ and  $t(c,b)$ are contained in $t(a,b)$. All cases contradict the minimality of $t(a,b)=\Tm$. Thus, $\Tm$ contains no point of $P\setminus K$ in its interior, proving Claim 1.

By Claim 1, $\Tm=t(a,b)$ may only contain points of $K$. Since $|K|\le k$, there must be an edge between $a$ and $b$ in \kTD{k}{}. This contradicts that $a$ and $b$ belong to different components $C_i$ and $C_j$ in $\mathcal{C}$. Therefore, \kTD{k}{} is $(k+1)$-connected.

We present a constructive proof for the second part of theorem. Let $P=A\cup B\cup K$, where $|A|,|B|\ge 1$ and $|K|=k+1$. Place the points of $A$ in the plane. Let $C_A^{4}=\bigcap_{p\in A}{C_p^4}$. Place the points of $K$ in $C_A^4$. Let $C_K^4=\bigcap_{p\in K}{C_p^4}$. Place the points of $B$ in $C_K^4$. Consider any pair $(a,b)$ of points where $a\in A$ and $b\in B$. It is obvious that any path between $a$ and $b$ in \kTD{k}{} goes through the vertices in $K$. Thus by removing the vertices in $K$, $a$ and $b$ become disconnected. Therefore, \kTD{k}{} of $P$ is not $(k+2)$-connected. 
\end{proof}
\addtocontents{toc}{\protect\setcounter{tocdepth}{1}}
\subsection{Bottleneck Biconnected Spanning Graph}
\label{td:biconnected-section}
\addtocontents{toc}{\protect\setcounter{tocdepth}{2}}
As shown in Figure~\ref{td:TD}(a), \kTD{0}{} may not be biconnected. By Theorem~\ref{td:k-connectivity-thr}, \kTD{1}{} is biconnected. In this section we show that a bottleneck biconnected spanning graph of $P$ is contained in \kTD{1}{}. 

\begin{theorem}
\label{td:biconnectivity-thr}
 For every point set $P$ in general position in the plane, \kTD{1}{} contains a bottleneck biconnected spanning graph of $P$.
\end{theorem}

\begin{proof}                                                                                       
Let $\mathcal{G}$ be the set of all biconnected spanning graphs with vertex set $P$. We define a total order on the elements of $\mathcal{G}$ by their weight sequence. If two elements have the same weight sequence, we break the ties arbitrarily to get a total order.
Let $G^* = (P, E)$ be a graph in $\mathcal{G}$ with minimal weight sequence. Clearly, $G^*$ is a bottleneck biconnected spanning graph of $P$. We will show that all edges of $G^*$ are in \kTD{1}{}. By contradiction suppose that some edges in $E$ do not belong to \kTD{1}{}, and let $e = (a, b)$ be the longest one (by the area of the triangle $t(a,b)$). If the graph $G^*-\{e\}$ is biconnected, then by removing $e$, we obtain a biconnected spanning graph $G$ with $\WS{G}\prec \WS{G^*}$; this contradicts the minimality of $G^*$. Thus, there is a pair $\{p,q\}$ of points such that any cycle between $p$ and $q$ in $G^*$ goes through $e$. Since $(a,b)\notin$~\kTD{1}{}, $t(a,b)$ contains at least two points of $P$, say $x$ and $y$. Let $G$ be the graph obtained from $G^*$ by removing the edge $(a,b)$ and adding the edges $(a,x)$, $(b,x)$, $(a,y)$, $(b,y)$. We show that in $G$ there is a cycle $C$ between $p$ and $q$ which does not go through $e$. Consider a cycle $C^*$ in $G^*$ between two points $p$ and $q$ (which goes through $e$). If none of $x$ and $y$ belong to $C^*$, then $C=(C^*-\{(a,b)\})\allowbreak\cup\allowbreak\{(a,x),(b,x)\}$ is a cycle in $G$ between $p$ and $q$. If one of $x$ or $y$, say $x$, belongs to $C^*$, then $C=(C^*-\{(a,b)\})\allowbreak\cup\allowbreak\{(a,y),(b,y)\}$ is a cycle in $G$ between $p$ and $q$. If both $x$ and $y$ belong to $C^*$, w.l.o.g. assume that $x$ is between $b$ and $y$ in the path $C^*- \{(a,b)\}$. Consider the partition of $C^*$ into four parts: (a) edge $(a,b)$, (b) path $\delta_{bx}$ between $b$ and $x$, (c) path $\delta_{xy}$ between $x$ and $y$, and (d) path $\delta_{ya}$ between $y$ and $a$. 
There are four cases:
\begin{enumerate}
 \item None of $p$ and $q$ are on $\delta_{xy}$. Let $C= \delta_{bx}\cup \delta_{ya} \cup \{(a,x),(b,y)\}$.
 \item Both $p$ and $q$ are on $\delta_{xy}$. Let $C=\delta_{xy}\cup \{(a,x),(a,y)\}$.
 \item One of $p,q$ is on $\delta_{xy}$ while the other is on $\delta_{bx}$. Let $C=\delta_{bx}\cup\delta_{xy}\cup\{(b,y)\}$.
 \item One of $p,q$ is on $\delta_{xy}$ while the other is on $\delta_{ya}$. Let $C=\delta_{xy}\cup\allowbreak\delta_{ya}\cup\allowbreak\{(a,x)\}$.
\end{enumerate}

In all cases, $C$ is a cycle in $G$ between $p$ and $q$. Thus, between any pair of points in $G$ there exists a cycle, and hence $G$ is biconnected. Since $x$ and $y$ are inside $t(a,b)$, by Observation~\ref{td:obs1}, $\max\allowbreak\{t(a,x),\allowbreak t(a, y),\allowbreak t(b,x),\allowbreak t(b,y)\}\prec \allowbreak t(a,b)$. Therefore, $\WS{G}\prec \WS{G^*}$; this contradicts the minimality of $G^*$.  
\end{proof}

\section{Hamiltonicity}
\label{td:Hamiltonicity}

In this section we show that \kTD{7}{} contains a bottleneck Hamiltonian cycle. In addition, we will show that for some point sets, \kTD{5}{} does not contain any bottleneck Hamiltonian cycle.

\begin{theorem}
\label{td:hamiltonicity-thr}
For every point set $P$ in general position in the plane, \kTD{7}{} contains a bottleneck Hamiltonian cycle.
\end{theorem}
\begin{proof}

Let $\mathcal{H}$ be the set of all Hamiltonian cycles through the points of $P$. Define a total
order on the elements of $\mathcal{H}$ by their weight sequence. If two elements have exactly the same weight sequence, break ties arbitrarily to get a total order. 
Let $H^* = a_0, a_1,\dots, a_{n-1},a_0$ be a cycle in $\mathcal{H}$ with minimal weight sequence. It is obvious that $H^*$ is a bottleneck Hamiltonian cycle of $P$. We will show that all the edges of $H^*$ are in \kTD{7}{}. Consider any edge $e = (a_i, a_{i+1})$ in $H^*$ and let $t(a_i,a_{i+1})$ be the triangle corresponding to $e$ (all the index manipulations are modulo $n$).

\vspace{5pt}

{\em Claim 1}: None of the edges of $H^*$ can be completely in the interior $t(a_i,a_{i+1})$. Suppose there is an edge $f=(a_j, a_{j+1})$ inside $t(a_i,a_{i+1})$. Let $H$ be a cycle obtained from $H^*$ by deleting $e$ and $f$, and adding $(a_i, a_j)$ and $(a_{i+1}, a_{j+1})$. By Observation~\ref{td:obs1}, $t(a_i, a_{i+1}) \succ \max\{t(a_i, a_j), \allowbreak t(a_{i+1},\allowbreak a_{j+1})\}$, and hence $\WS{H}\prec \WS{H^*}$. This contradicts the minimality of $H^*$.

\vspace{5pt}

Therefore, we may assume that no edge of $H^*$ lies completely inside $t(a_i,a_{i+1})$. Suppose there are $w$ points of $P$ inside $t(a_i,a_{i+1})$. Let $U = u_1, u_2,\dots, u_w$ represent these points indexed in the order we would encounter them on $H^*$ starting from $a_i$. Let $R = \{r_1, r_2,\dots, r_w\}$ represent the vertices where $r_i$ is the vertex succeeding $u_i$ in the cycle. All the vertices in $R$, probably except $r_w$, are different from $a_i$ (and $a_{i+1}$).  
Without loss of generality assume that $a_i\in C_{a_{i+1}}^4$, and $t(a_i,a_{i+1})$ is anchored at $a_{i+1}$, as shown in Figure~\ref{td:hamiltonicity-fig1}. 

\vspace{5pt}

{\em Claim 2}: For each $r_j\in R$, $t(r_j, a_{i+1}) \succeq \max\allowbreak\{t(a_i, a_{i+1}), t(u_j, r_j)\}$. Suppose there is a point $r_j\in R$ such that $t(r_j, a_{i+1}) \prec  \max\allowbreak\{t(a_i, a_{i+1}),\allowbreak t(u_j, r_j)\}$. Construct a new cycle $H$ by removing the edges $(u_j, r_j)$, $(a_i, a_{i+1})$ and adding the edges $(a_{i+1}, r_j)$ and $(a_i, u_j)$. Since the two new edges have length strictly less than $\max\{t(a_i, a_{i+1}), t(u_j, r_j)\}$, $\WS{H} \prec  \WS{H^*}$; which is a contradiction.

\vspace{5pt}

{\em Claim 3}: For each $r_j,r_k\in R$, $t(r_j, r_k)\succeq \allowbreak\max\{t(a_i, a_{i+1}),\allowbreak t(u_j,r_j),\allowbreak t(u_k,r_k)\}$. Suppose there is a pair $r_j$ and $r_k$ such that $t(r_j, r_k)\prec  \allowbreak\max\{t(a_i, a_{i+1}),\allowbreak t(u_j,r_j),\allowbreak d(u_k,r_k)\}$.  Construct a cycle $H$ from $H^*$ by first deleting $(u_j,r_j)$, $(u_k,r_k)$,  $(a_i, a_{i+1})$. This results in three paths. One of the paths must contain both $a_i$ and either $r_j$ or $r_k$. W.l.o.g. suppose that $a_i$ and $r_k$ are on the same path. Add the edges $(a_i, u_j)$, $(a_{i+1}, u_k)$, $(r_j, r_k)$. Since $\max\{t(u_j,r_j),\allowbreak t(u_k,r_k),\allowbreak d(a_i, a_{i+1})\}\succ\max\allowbreak  \{\allowbreak t(a_i, u_j),\allowbreak t(a_{i+1}, u_k),\allowbreak t(r_j, r_k)\}$, $\WS{H} \prec \allowbreak \WS{H^*}$; we get a contradiction.

\vspace{5pt}
\begin{figure}[htb]
  \centering
  \includegraphics[width=.65\columnwidth]{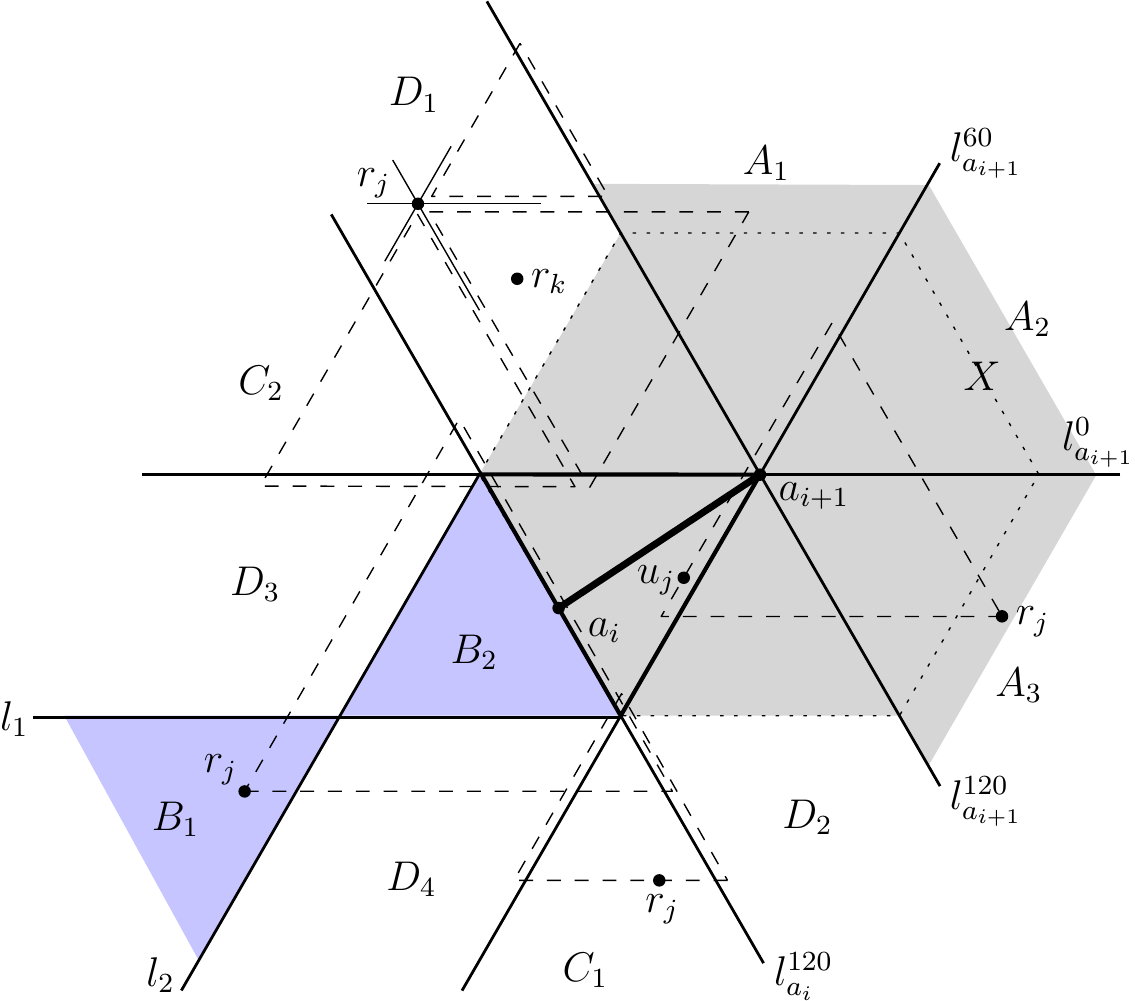}
 \caption{Illustration of Theorem~\ref{td:hamiltonicity-thr}.}
  \label{td:hamiltonicity-fig1}
\end{figure}

We use Claim 2 and Claim 3 to show that the size of $R$ is at most seven, and consequently $w\le 7$.
Consider the lines $l_{a_{i+1}}^0$, $l_{a_{i+1}}^{60}$, $l_{a_{i+1}}^{120}$, and $l_{a_{i}}^{120}$ as shown in Figure~\ref{td:hamiltonicity-fig1}. Let $l_1$ and $l_2$ be the rays starting at the corners of $t(a_i, a_{i+1})$ opposite to $a_{i+1}$ and parallel to $l_{a_{i+1}}^0$ and $l_{a_{i+1}}^{60}$ respectively. These lines and rays partition the plane into 12 regions, as shown in Figure~\ref{td:hamiltonicity-fig1}. We will show that each of the regions $D_1$, $D_2$, $D_3$, $D_4$, $C_1$, $C_2$, and $B=B_1\cup B_2$ contains at most one point of $R$, and the other regions do not contain any point of $R$. Consider the hexagon $X(a_{i+1},a_i)$. By Claim 2 and Observation~\ref{td:obs2}, no point of $R$ can be inside $X(a_{i+1},a_i)$. Moreover, no point of $R$ can be inside the cones $A_1$, $A_2$, or $A_3$, because if $r_j\in \{A_1 \cup A_2\cup A_3\}$, the (upward) triangle $t'(u_j, r_j)$ contains $a_{i+1}$. Then by Observation~\ref{td:obs2}, $t(r_j, a_{i+1}) \prec  t(u_j, r_j)$; which contradicts Claim 2.

 We show that each of the regions $D_1$, $D_2$, $D_3$, $D_4$ contains at most one point of $R$. Consider the region $D_1$; by similar reasoning we can prove the claim for $D_2$, $D_3$, $D_4$. Using contradiction, let $r_j$ and $r_k$ be two points in $D_1$, and w.l.o.g. assume that $r_j$ is the farthest to $l_{a_{i+1}}^{60}$. Then $r_k$ can lie inside any of the cones $C_{r_j}^1$, $C_{r_j}^5$, and $C_{r_j}^6$ (but not in $X$). If $r_k \in C_{r_j}^1$, then $t'(r_j, r_k)$ is smaller than $t'(a_i, a_{i+1})$ which means that $t(r_j, r_k)\prec  t(a_i,a_{i+1})$. If $r_k \in C_{r_j}^5$, then $t'(u_j,r_j)$ contains $r_k$, that is $t(r_j, r_k)\prec  t(u_j, r_j)$. If $r_k \in C_{r_j}^6$, then $t(u_j,r_j)$ contains $r_k$, that is $t(r_j, r_k)\prec  t(u_j, r_j)$. All cases contradict Claim 3. 

Now consider the region $C_1$ (or $C_2$). By contradiction assume that it contains two points $r_j$ and $r_k$. Let $r_j$ be the farthest from $l_{a_{i+1}}^{0}$. It is obvious that $t'(u_j, r_j)$ contains $r_k$, that is $t(r_j, r_k)\prec  t(u_j, r_j)$; which contradicts Claim 3. 

Consider the region $B=B_1\cup B_2$. Note that it is possible for $r_j$ or $r_k$ to be $a_i$. If both $r_j$ and $r_k$ belong to $B_2$, then $t'(r_j, r_k)$ is smaller that $t(a_i, a_{i+1})$. If $r_j\in B_1$ and $r_k\in B_2$, then $t'(u_j,r_j)$ contains $r_k$, and hence $t(r_j, r_k)\prec t(u_j, r_j)$. If both $r_j$ and $r_k$ belong to $B_1$, let $r_j$ be the farthest from $l_{a_i}^{120}$. Clearly, $t(u_j, r_j)$ contains $r_k$ and hence $t(r_j, r_k)\prec t(u_j, r_j)$. All cases contradict Claim 3.

Therefore, any of the regions $D_1$, $D_2$, $D_3$, $D_4$, $C_1$, $C_2$, and $B=B_1\cup B_2$ contains at most one point of $R$. Thus, $|R|\le 7$ and $w \le 7$, and $t(a_i, a_{i+1})$ contains at most 7 points of $P$. Therefore, $e=(a_i, a_{i+1})$ is an edge of \kTD{7}{}.
\end{proof}

As a direct consequence of Theorem~\ref{td:hamiltonicity-thr} we have shown that:
\begin{corollary}
 \kTD{7}{} is Hamiltonian.
\end{corollary}

\begin{figure}[htb]
  \centering
  \includegraphics[width=.8\columnwidth]{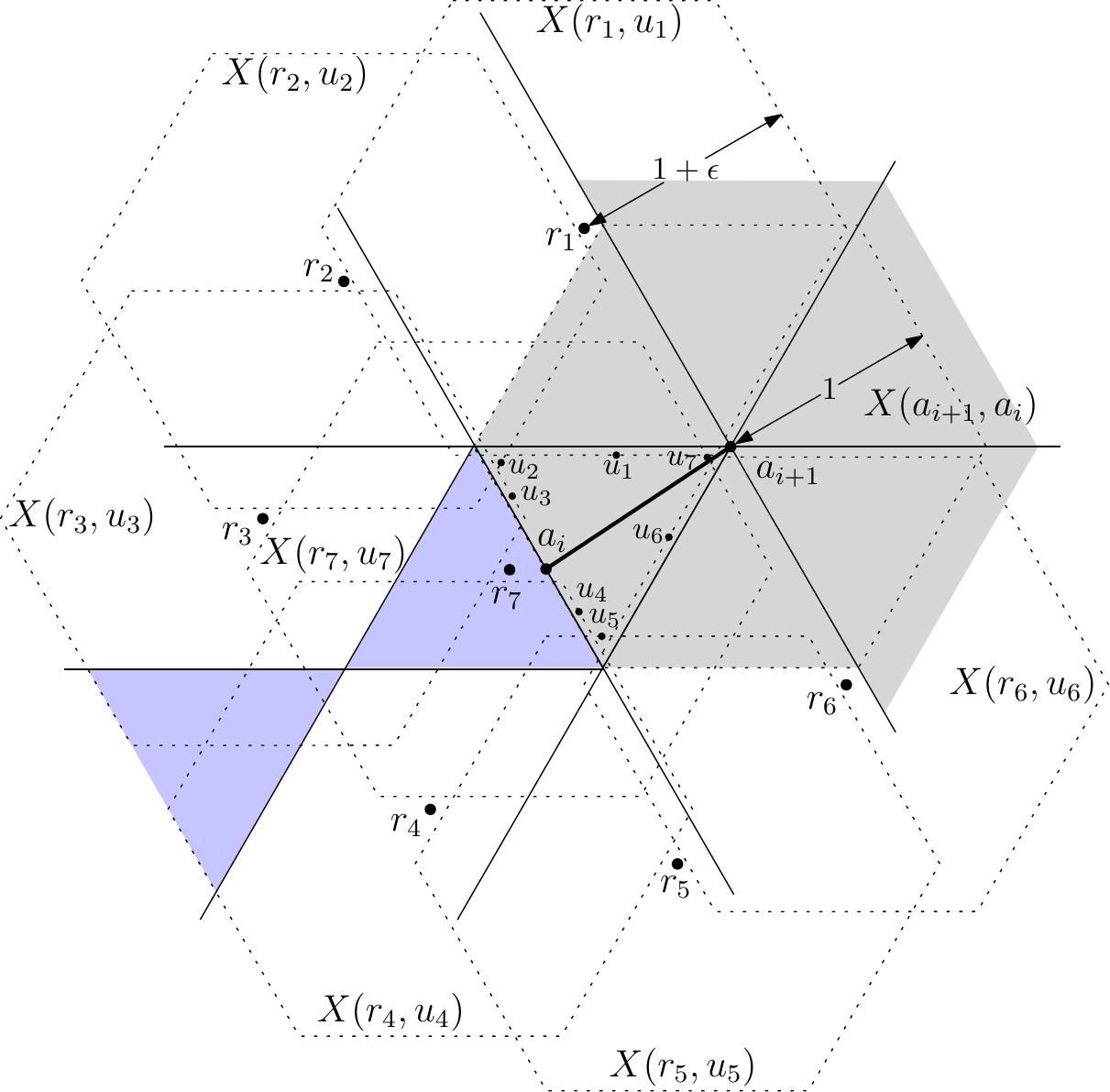}
 \caption{$t(a_i,a_{i+1})$ contains 7 points while the conditions in the proof of Theorem~\ref{td:hamiltonicity-thr} hold.}
  \label{td:hamiltonicity-fig2}
\end{figure}

An interesting question is to determine if \kTD{k}{} contains a bottleneck Hamiltonian cycle for $k<7$.
Figure~\ref{td:hamiltonicity-fig2} shows a configuration where $t(a_i, a_{i+1})$ contains 7 points while the conditions of Claim 1, Claim 2, and Claim 3 in the proof of Theorem~\ref{td:hamiltonicity-thr} hold. In Figure~\ref{td:hamiltonicity-fig2}, $d(a_i, a_{i+1})=1$, $d(r_i, u_i)= 1+\epsilon$, $d(r_i, r_j)>1+\epsilon$, $d(r_i, a_{i+1})> 1+\epsilon$ for $i,j=1,\dots 7$ and $i\neq j$.

\begin{figure}[htb]
  \centering
  \includegraphics[width=.8\columnwidth]{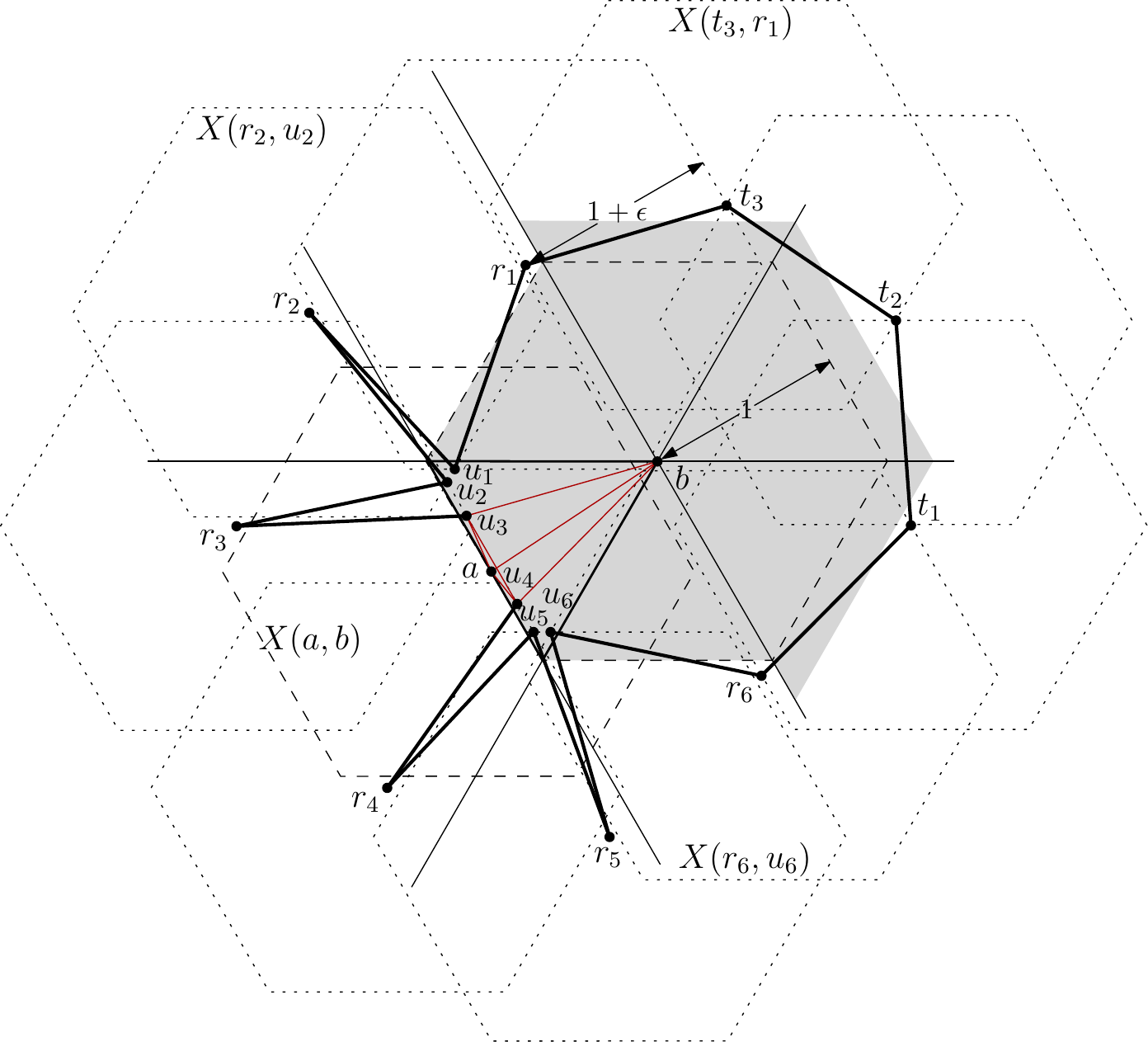}
 \caption{The points $\{r_1,\dots,r_6,t_1,t_2,t_3\}$ are connected to their first and second closest point (the bold edges). The edge $(a,b)$ should be in any bottleneck Hamiltonian cycle, while $t(a,b)$ contains 6 points.}
  \label{td:hamiltonicity-fig3}
\end{figure}

\begin{theorem}
There exists an arbitrary large point set such that its \kTD{5}{} does not contain any bottleneck Hamiltonian cycle.
\end{theorem}
\begin{proof}
In order to prove the theorem, we provide such a point set.
Figure~\ref{td:hamiltonicity-fig3} shows a configuration of $P$ with 17 points such that \kTD{5}{} does not contain a bottleneck Hamiltonian cycle. In Figure~\ref{td:hamiltonicity-fig3}, $d(a,b)=1$ and $t(a,b)$ contains 6 points $U=\{u_1, \dots, u_6\}$. In addition $d(r_i, u_i)= 1+\epsilon$, $d(r_i, r_j)>1+\epsilon$, $d(r_i, b)> 1+\epsilon$ for $i,j=1,\dots 6$ and $i\neq j$. Let $R=\{t_1, t_2,t_3, r_1,\dots, r_6\}$. The dashed hexagons are centered at $a$ and $b$ and have diameter 1. The dotted hexagons are centered at vertices in $R$ and have diameter $1+\epsilon$. Each point in $R$ is connected to its first and second closest points by edges of length $1+\epsilon$ (the bold edges). Let $B$ be the set of these edges. Let $H$ be a cycle formed by $B\cup\{(u_3,b),(b,a),(a,u_5)\}$, i.e., $H=(u_4,r_4,u_5,r_5,u_6,r_6,t_1,t_2,t_3,r_1,u_1,\allowbreak r_2,\allowbreak u_2,r_3,u_3,a,b,u_4)$. It is obvious that $H$ is a Hamiltonian cycle for $P$ and $\lambda(H)=1+\epsilon$. Thus, the bottleneck of any bottleneck Hamiltonian cycle for $P$ is at most $1+\epsilon$. We will show that any bottleneck Hamiltonian cycle for $P$ contains the edge $(a,b)$ which does not belong to \kTD{5}{}. By contradiction, let $H^*$ be a bottleneck Hamiltonian cycle which does not contain $(a,b)$. In $H^*$, $b$ is connected to two vertices $b_l$ and $b_r$, where $b_l\neq a$ and $b_r\neq a$. Since the distance between $b$ and any vertex in $R$ is strictly bigger than $1+\epsilon$ and $\lambda(H^*)\le 1+\epsilon$, $b_l\notin R$ and $b_r\notin R$. Thus $b_l$ and $b_r$ belong to $U$. Let $U'=\{u_1,u_2,u_5,u_6\}$. Consider two cases:

\begin{itemize}
 \item $b_l\in U'$ or $b_r\in U'$. W.l.o.g. assume that $b_l\in U'$ and $b_l=u_1$. Since $u_1$ is the first/second closest point of $r_1$ and $r_2$, in $H^*$ one of $r_1$ and $r_2$ must be connected by an edge $e$ to a point that is farther than its second closet point; $e$ has length strictly greater than $1+\epsilon$.
 \item $b_l\notin U'$ and $b_r\notin U'$. Thus, both $b_l$ and $b_r$ belong to $\{u_3,u_4\}$. That is, in $H^*$, $a$ should be connected to a point $c$ where $c\in R\cup U'$. If $c\in R$ then the edge $(a,c)$ has length more than $1+\epsilon$. If $c\in U'$, w.l.o.g. assume $c=u_1$; by the same argument as in the previous case, one of $r_1$ and $r_2$ must be connected by an edge $e$ to a point that is farther than its second closet point; $e$ has length strictly greater than $1+\epsilon$.
\end{itemize}

Since $e\in H^*$, both cases contradicts that $\lambda(H^*)\le 1+\epsilon$. Therefore, every bottleneck Hamiltonian cycle contains edge $(a,b)$. Since $(a,b)$ is not an edge in \kTD{5}{}, a bottleneck Hamiltonian cycle of $P$ is not contained in \kTD{5}{}.  
We can construct larger point sets by adding new points very close to $t_2$, and at distance at least $1+2\epsilon$ from $b$.
\end{proof}

\section{Perfect Matching Admissibility}
\label{td:matching}
In this section we consider the matching problem in higher-order triangular-distance Delaunay graphs. In Subsection~\ref{td:bottleneck-matching-section} we show that \kTD{6}{} contains a bottleneck perfect matching. We also show that for some point sets $P$, \kTD{5}{} does not contain any bottleneck perfect matching. In Subsection~\ref{td:matching2} we prove that every \kTD{2}{} has a perfect matching when $P$ has an even number of points, and \kTD{1}{} contains a matching of size at least $\frac{2(n-1)}{5}$.

\subsection{Bottleneck Perfect Matching}
\label{td:bottleneck-matching-section}
\begin{theorem}
\label{td:matching-thr}
 For a set $P$ of an even number of points in general position in the plane, \kTD{6}{} contains a bottleneck perfect matching.
\end{theorem}

\begin{proof}
Let $\mathcal{M}$ be the set of all perfect matchings through the points of $P$. Define a total order on the elements of $\mathcal{M}$ by their weight sequence. If two elements have exactly the same weight sequence, break ties arbitrarily to get a total order.
Let $M^* = \{(a_1, b_1),\dots, (a_{\frac{n}{2}}, b_{\frac{n}{2}})\}$ be a perfect matching in $\mathcal{M}$ with minimal weight sequence. It is obvious that $M^*$ is a bottleneck perfect matching for $P$. We will show that all edges of $M^*$ are in \kTD{6}{}. Consider any edge $e = (a_i, b_i)$ in $M^*$ and its corresponding triangle $t(a_i,b_i)$.

\begin{figure}[htb]
  \centering
  \includegraphics[width=.5\columnwidth]{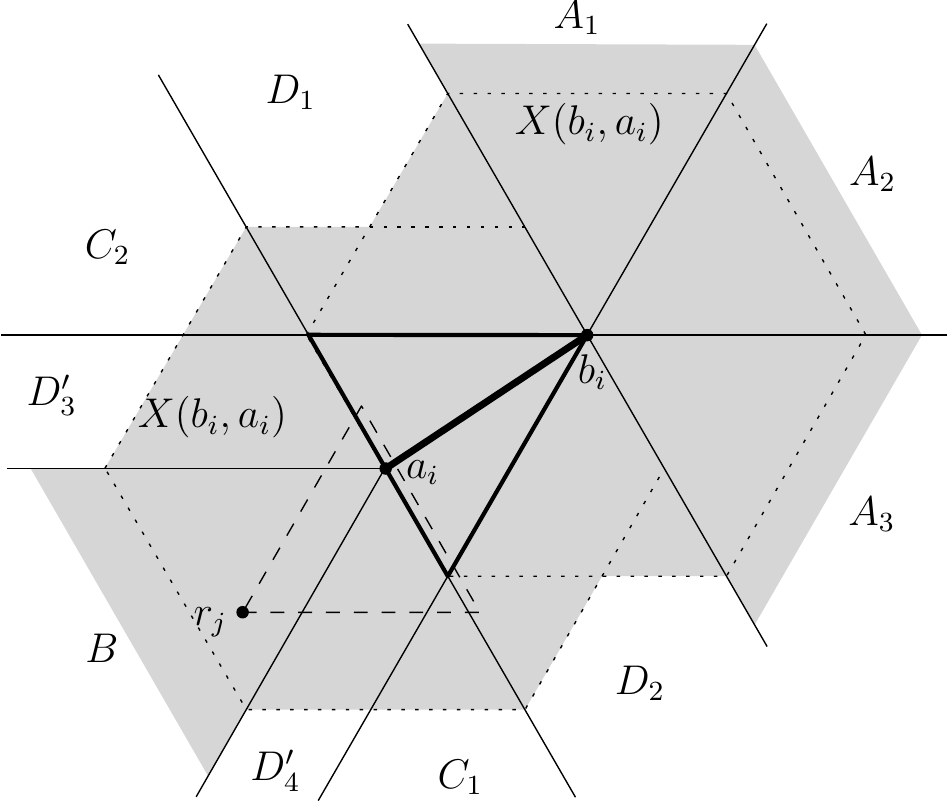}
 \caption{Proof of Theorem~\ref{td:matching-thr}.}
  \label{td:matching-fig1}
\end{figure}

{\em Claim 1}: None of the edges of $M^*$ can be inside $t(a_i,b_i)$. Suppose there is an edge $f=(a_j, b_j)$ inside $t(a_i,b_i)$. Let $M$ be a perfect matching obtained from $M^*$ by deleting $\{e, f\}$, and adding $\{(a_i, a_j), (b_i, b_j)\}$. By Observation~\ref{td:obs1}, the two new edges are smaller than the old ones. Thus, $\WS{M}\prec \WS{M^*}$ which contradicts the minimality of $M^*$.

Therefore, we may assume that no edge of $M^*$ lies completely inside $t(a_i,b_i)$. Suppose there are $w$ points of $P$ inside $t(a_i,b_i)$. Let $U = u_1, u_2,\dots, u_w$ represent the points inside $t(a_i, b_i)$, and $R=r_1, r_2,\dots, r_w$ represent the points where $(r_i,u_i)\in M^*$. W.l.o.g. assume that $a_i\in C^4_{b_i}$, and $t(a_i,b_i)$ is anchored at $b_i$ as shown in Figure~\ref{td:matching-fig1}.  

{\em Claim 2}: For each $r_j\in R$, $\min\{t(r_j, a_i),\allowbreak  t(r_j, b_i)\} \succeq \max\{t(a_i, b_i),\allowbreak  t(u_j, r_j)\}$. Otherwise, by a similar argument as in the proof of Claim 2 in Theorem \ref{td:hamiltonicity-thr} we can either match $r_j$ with $a_i$ or $b_i$ to obtain a smaller matching $M$; which is a contradiction.

{\em Claim 3}: For each pair $r_j$ and $r_k$ of points in $R$, $t(r_j, r_k)\succeq \max\{\allowbreak t(a_i, b_i),\allowbreak  t(r_j, u_j), t(r_k, u_k)\}$. The proof is similar to the proof of Claim 3 in Theorem \ref{td:hamiltonicity-thr}.

Consider Figure~\ref{td:matching-fig1} which partitions the plane into eleven regions. As a direct consequence of Claim 2, the hexagons $X(b_i, a_i)$ and $X(a_i,\allowbreak b_i)$ do not contain any point of $R$. By a similar argument as in the proof of Theorem \ref{td:hamiltonicity-thr}, the regions $A_1$, $A_2$, $A_3$ do not contain any point of $R$. In addition, the region $B$ does not contain any point $r_j$ of $R$, because otherwise $t'(r_j,u_j)$ contains $a_i$, that is $t(r_j, a_i)\prec  t(u_j, r_j)$ which contradicts Claim 2. As shown in the proof of Theorem \ref{td:hamiltonicity-thr} each of the regions $D_1$, $D_2$, $D'_3$, $D'_4$, $C_1$, and $C_2$ contains at most one point of $R$ (note that $D'_3\subset D_3$ and $D'_4\subset D_4$). Thus, $w \le 6$, and $t(a_i, b_i)$ contains at most 6 points of $P$. Therefore, $e=(a_i, b_i)$ is an edge of \kTD{6}{}.
\end{proof}

As a direct consequence of Theorem~\ref{td:matching-thr} we have shown that:
\begin{corollary}
 For a set $P$ of even number of points in general position in the plane, \kTD{6}{} has a perfect matching.
\end{corollary}

\begin{figure}[htb]
  \centering
  \includegraphics[width=.65\columnwidth]{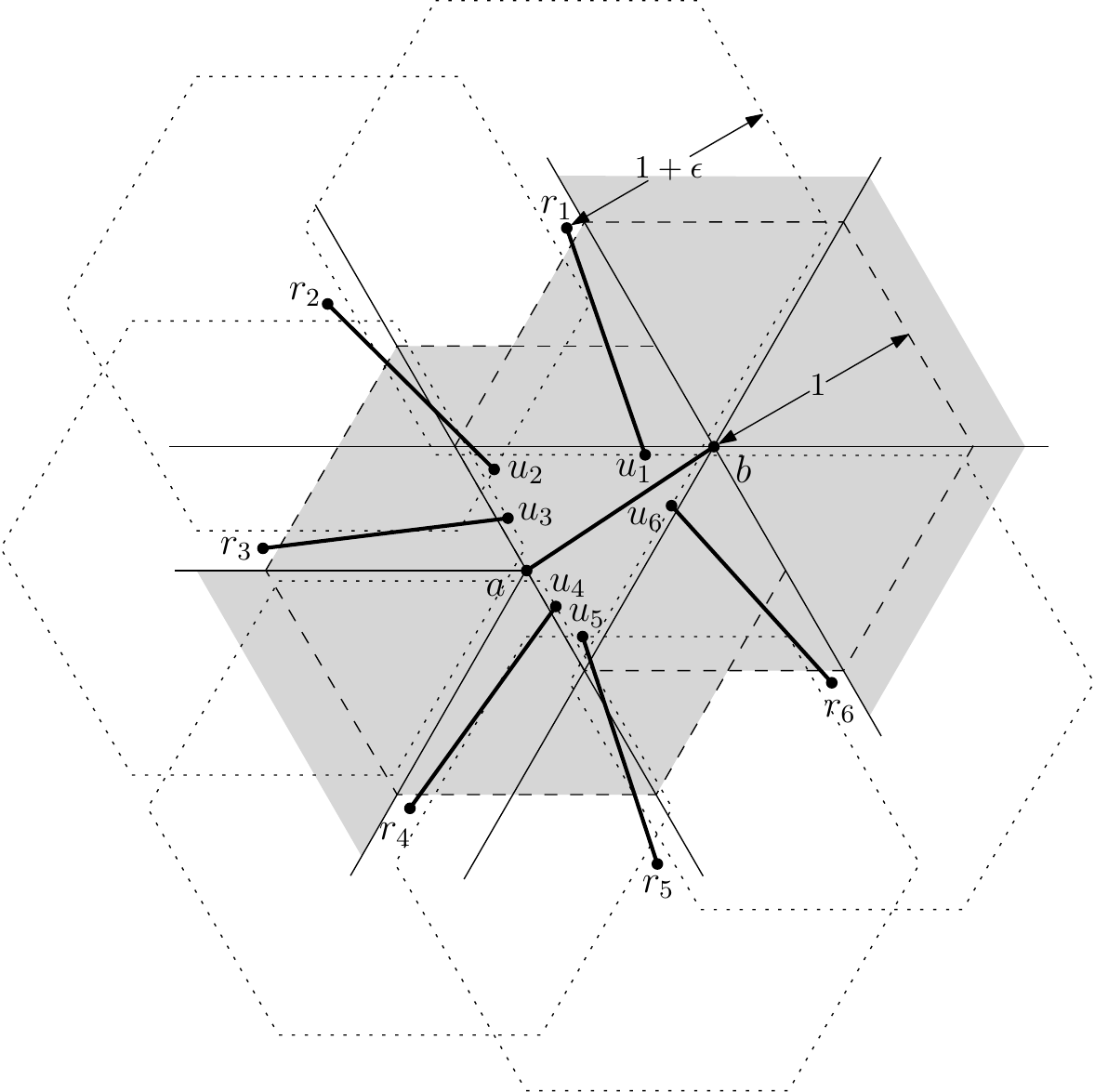}
 \caption{The points $\{r_1,\dots,r_6\}$ are matched to their closest point. The edge $(a, b)$ should be an edge in any bottleneck perfect matching, while $t(a, b)$ contains 6 points.}
  \label{td:matching-fig3}
\end{figure}

In the following theorem, we show that the bound $k=6$ proved in Theorem~\ref{td:matching-thr} is tight. 
\begin{theorem}
There exists an arbitrarily large point set such that its \kTD{5}{} does not contain any bottleneck perfect matching.
\end{theorem}
\begin{proof}
In order to prove the theorem, we provide such a point set.
Figure~\ref{td:matching-fig3} shows a configuration of a set $P$ with 14 points such that $d(a,b)=1$ and $t(a,b)$ contains six points $U=\{u_1, \dots, u_6\}$. In addition $d(r_i, u_i)= 1+\epsilon$, $d(r_i,x)>1+\epsilon$ where $x\neq u_i$, for $i=1,\dots 6$. Let $R=\{r_1,\dots, r_6\}$. In Figure~\ref{td:matching-fig3}, the dashed hexagons are centered at $a$ and $b$, each of diameter 1, and the dotted hexagons centered at vertices in $R$, each of diameter $1+\epsilon$. Consider a perfect matching $M=\{(a,b)\}\cup \{(r_i, u_i): i=1,\dots, 6\}$ where each point $r_i\in R$ is matched to its closest point $u_i$. It is obvious that $\lambda(M)=1+\epsilon$, and hence the bottleneck of any bottleneck perfect matching is at most $1+\epsilon$. We will show that any bottleneck perfect matching for $P$ contains the edge $(a,b)$ which does not belong to \kTD{5}{}. By contradiction, let $M^*$ be a bottleneck perfect matching which does not contain $(a,b)$. In $M^*$, $b$ is matched to a point $c\in R\cup U$. If $c \in R$, then $d(b,c)>1+\epsilon$. If $c\in U$, w.l.o.g. assume $c = u_1$. Thus, in $M^*$ the point $r_1$ is matched to a point $d$ where $d\neq u_1$. Since $u_1$ is the unique closest point to $r_1$ and $d(r_1,u_1)=1+\epsilon$, $d(r_1,d)>1+\epsilon$. Both cases contradicts that $\lambda(M^*)\le 1+\epsilon$. Therefore, every bottleneck perfect matching contains $(a,b)$. Since $(a,b)$ is not an edge in \kTD{5}{}, a bottleneck perfect matching of $P$ is not contained in \kTD{5}{}. 
We can construct larger point sets by adding new points\textemdash which are within distance $1+\epsilon$ from each other\textemdash at distance at least $1+2\epsilon$ from the current point set.
\end{proof}
\subsection{Perfect Matching}
\label{td:matching2}
In \cite{Babu2014} the authors proved a tight lower bound of $\lceil\frac{n-1}{3}\rceil$ on the size of a maximum matching in \kTD{0}{}. In this section we prove that \kTD{1}{} has a matching of size $\frac{2(n-1)}{5}$ and \kTD{2}{} has a perfect matching when $P$ has an even number of points.

\begin{figure}[htb]
  \centering
\setlength{\tabcolsep}{0in}
  $\begin{tabular}{cc}
 \multicolumn{1}{m{.5\columnwidth}}{\centering\includegraphics[width=.34\columnwidth]{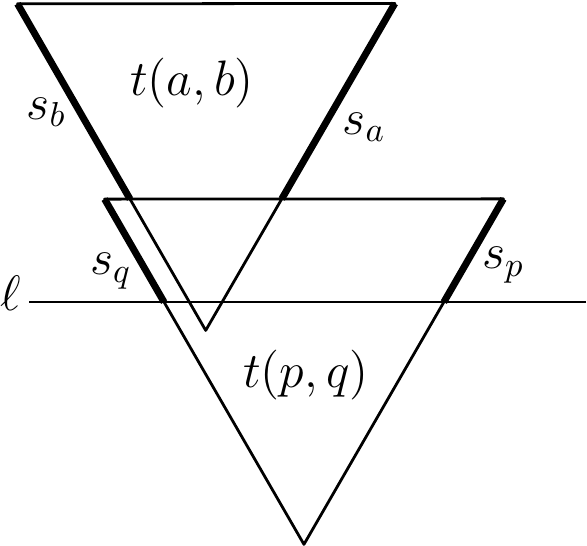}}
&\multicolumn{1}{m{.5\columnwidth}}{\centering\includegraphics[width=.36\columnwidth]{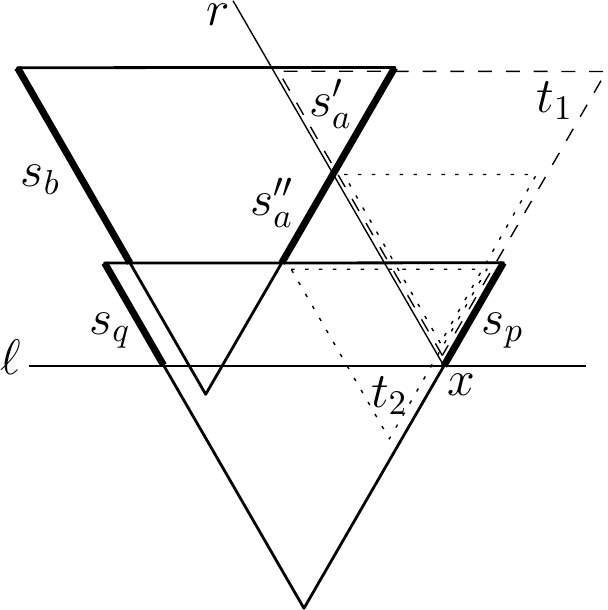}}\\
(a) & (b)
\end{tabular}$
  \caption{(a) Illustration of Lemma~\ref{td:triangle3}, and (b) proof of Lemma~\ref{td:triangle3}.}
\label{td:intersection-fig}
\end{figure}

For a triangle $t(a,b)$ through the points $a$ and $b$, let $top(a,b)$, $left(a,b)$, and $right(a,b)$ respectively denote the top, left, and right sides of $t(a,b)$. Refer to Figure~\ref{td:intersection-fig}(a) for the following lemma.
\begin{lemma}
\label{td:triangle3}
Let $t(a,b)$ and $t(p,q)$ intersect a horizontal line $\ell$, and $t(a,b)$ intersects $top(p,q)$ in such a way that $t(p,q)$ contains the lowest corner of $t(a,b)$. Let $a$ (resp. $p$) lie on $right(a,b)$ (resp. $right(p,q)$. If $a$ and $b$ lie above $top(p,q)$, and $p$ and $q$ lie above $\ell$, then, $\max\{t(a,p), t(b,q)\} \allowbreak \prec \allowbreak \max\{t(a,b), t(p,q)\}$.
\end{lemma}
\begin{proof}
Recall that $t(a,b)$ is the smallest downward triangle through $a$ and $b$. By Observation~\ref{td:side-point-obs} each side of $t(a,b)$ contains either $a$ or $b$. 
In Figure~\ref{td:intersection-fig}(a) the set of potential positions for point $a$ on the boundary of $t(a,b)$ is shown by the line segment $s_a$; and similarly by $s_b$, $s_p$, $s_q$ for $b$, $p$, $q$, respectively. We will show that $t(a,p)\prec \max\{t(a,b), t(p,q)\}$. By similar reasoning we can show that $t(b,q)\prec \max\{t(a,b), t(p,q)\}$. Let $x$ denote the intersection of $\ell$ and $right(p,q)$. Consider a ray $r$ initiated at $x$ and parallel to $left(p,q)$ which divides $s_a$ into (at most) two parts $s'_a$ and $s''_a$ as shown in Figure~\ref{td:intersection-fig}(b). Two cases may appear:

\begin{itemize}
 \item $a\in s'_a$. Let $t_1$ be a downward triangle anchored at $x$ which has its $top$ side on the line through $top(a,b)$ (the dashed triangle in Figure~\ref{td:intersection-fig}(b)). The top side of $t_1$ and $t(a,b)$ lie on the same horizontal line. The bottommost corner of $t_1$ is on $\ell$ while the bottommost corner of $t(a,b)$ is below $\ell$. Thus, $t_1\prec t(a,b)$. In addition, $t_1$ contains $s'_a$ and $s_p$, thus, for any two points $a\in s'_a$ and $p\in s_p$, $t(a,p)\preceq t_1$. Therefore, $t(a,p)\prec  t(a,b)$.
 \item $a \in s''_a$. Let $t_2$ be a downward triangle anchored at the intersection of $right(a,b)$ and $top(p,q)$ which has one side on the line through $right(p,q)$ (the dotted triangle in Figure~\ref{td:intersection-fig}(b)). This triangle is contained in $t(p,q)$, and has $s_p$ on its right side. If we slide $t_2$ upward while its top-left corner remains on $s''_a$, the segment $s_p$ remains on the right side of $t_2$. Thus, any triangle connecting a point $a\in s''_a$ to a point $p\in s_p$ has the same size as $t_2$. That is, $t(a,p)=t_2\prec t(p,q)$. 
\end{itemize}

Therefore, we have $t(a,p)\prec \max\{t(a,b), t(p,q)\}$. By similar argument we conclude that $t(b,q)\prec \max\{t(a,b), t(p,q)\}$.  
\end{proof}
Let $\mathcal{P}=\{P_1, P_2,\dots\}$ be a partition of the points in $P$.
Let $G(\mathcal{P})$ be the complete graph with vertex set $\mathcal{P}$. For each edge $e=(P_i,P_j)$ in $G(\mathcal{P})$, let $w(e)$ be equal to the area of the smallest triangle between a point in $P_i$ and a point in $P_j$, i.e. $w(e)=\min\{t(a,b):a\in P_i, b\in P_j\}$. That is, the weight of an edge $e\in G(\mathcal{P})$ corresponds to the size of the smallest triangle $t(e)$ defined by the endpoints of $e$. Let $\mathcal{T}$ be a minimum spanning tree of $G(\mathcal{P})$. Let $T$ be the set of triangles corresponding to the edges of $\mathcal{T}$, i.e. $T=\{t(e): e\in \mathcal{T}\}$. 

\begin{lemma}
 \label{td:empty-triangle-lemma}
The interior of any triangle in $T$ does not contain any point of $P$.
\end{lemma}
\begin{proof}
  By contradiction, suppose there is a triangle $\tau\in T$ which contains a point $c\in P$. Let $e=(P_i,P_j)$ be the edge in $\mathcal{T}$ which corresponds to $\tau$. Let $a$ and $b$ respectively be the points in $P_i$ and $P_j$ which define $\tau$, i.e. $\tau=t(a,b)$ and $w(e)=t(a,b)$. Three cases arise: (i) $c\in P_i$, (ii) $c\in P_j$, (iii) $c\in P_l$ where $l\neq i$ and $l\neq j$. In case (i) the triangle $t(c,b)$ between $c\in P_i$ and $b\in P_j$ is smaller than $t(a,b)$; contradicts that $w(e)=t(a,b)$ in $G(\mathcal{P})$.  In case (ii) the triangle $t(a,c)$ between $a\in P_i$ and $c\in P_j$ is smaller than $t(a,b)$; contradicts that $w(e)=t(a,b)$ in $G(\mathcal{P})$. In case (iii) the triangle $t(a,c)$ (resp. $t(c,b)$) between $P_i$ and $P_l$ (resp. $P_l$ and $P_j$) is smaller than $t(a,b)$; contradicts that $e$ is an edge in $\mathcal{T}$. 
\end{proof}

\begin{lemma}
 \label{td:intersection-lemma}
Each point in the plane can be in the interior of at most three triangles in $T$. 
\end{lemma}

\begin{proof}
For each $t(a,b)\in T$, the sides $top(a,b)$, $right(a,b)$, and $left(a,\allowbreak b)$ contains at least one of $a$ and $b$. In addition, by Lemma~\ref{td:empty-triangle-lemma}, $t(a,b)$ does not contain any point of $P$ in its interior. Thus, none of $top(a,b)$, $right(a,b)$, and $left(a,b)$ is completely inside the other triangles. Therefore, the only possible way that two triangles $t(a,b)$ and $t(p,q)$ can share a point is that one triangle, say $t(p,q)$, contains a corner of $t(a,b)$ in such a way that $a$ and $b$ are outside $t(p,q)$. In other words $t(a,b)$ intersects $t(p,q)$ through one of the sides $top(p,q)$, $right(p,q)$, or $left(p,q)$. If $t(a,b)$ intersects $t(p,q)$ through a direction $d\in \{top, right, left\}$ we say that $t(p,q)\prec_{d} t(a,b)$. 

By contradiction, suppose there is a point $c$ in the plane which is inside four triangles $\{t_1,t_2,\allowbreak t_3,\allowbreak t_4\}\subseteq T$. Out of these four, either (i) three of them are like $t_i\prec_d t_j \prec_d t_k$ or (ii) there is a triangle $t_l$ such that $t_l\prec_{top} t_i, t_l\prec_{right} t_j, t_l\prec_{left} t_k$, where $1\le i,j, k,l\le 4$ and $i\neq j \neq k \neq l$. Figure~\ref{td:configuration-fig} shows the two possible configurations (note that all other configurations obtained by changing the indices of triangles and/or the direction are symmetric to Figure~\ref{td:configuration-fig}(a) or Figure~\ref{td:configuration-fig}(b)).
\begin{figure}[htb]
  \centering
\setlength{\tabcolsep}{0in}
  $\begin{tabular}{cc}
 \multicolumn{1}{m{.5\columnwidth}}{\centering\includegraphics[width=.24\columnwidth]{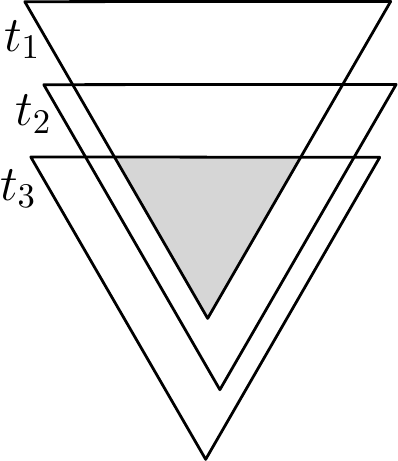}}
&\multicolumn{1}{m{.5\columnwidth}}{\centering\includegraphics[width=.34\columnwidth]{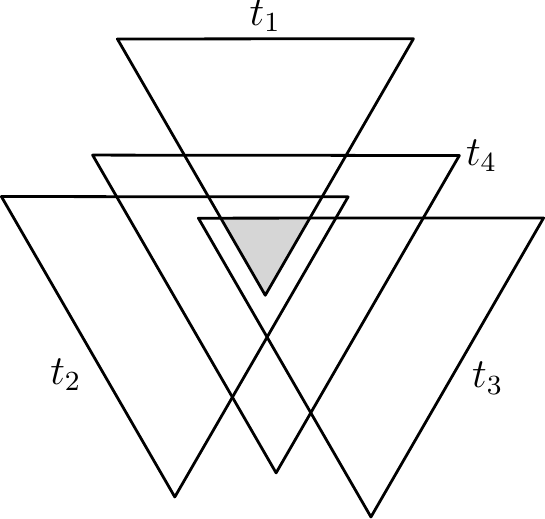}}\\
(a) & (b)
\end{tabular}$
  \caption{Two possible configurations: (a) $t_3 \prec_{top} t_2 \prec_{top} t_1$, (b) $t_4\prec_{top} t_1, t_4\prec_{left} t_2, t_4\prec_{right} t_3$.}
\label{td:configuration-fig}
\end{figure}

 Recall that each of $t_1,t_2,t_3,t_4$ corresponds to an edge in $\mathcal{T}$. In the configuration of Figure~\ref{td:configuration-fig}(a) consider $t_1$, $t_2$, and $top(t_3)$ which is shown in more detail in Figure~\ref{td:matching3-fig}(a). Suppose $t_1$ (resp. $t_2$) is defined by points $a$ and $b$ (resp. $p$ and $q$). By Lemma~\ref{td:empty-triangle-lemma}, $p$ and $q$ are above $top(t_3)$, $a$ and $b$ are above $top(t_2)$. By Lemma~\ref{td:triangle3}, $\max\{t(a,p),t(b,q)\} \prec \max\{t(a,b), t(p,q)\}$. This contradicts the fact that both of the edges representing $t(a,b)$ and $t(p,q)$ are in $\mathcal{T}$, because by replacing $\max\{t(a,b), t(p,q)\}$ with $t(a,p)$ or $t(b,q)$, we obtain a tree $\mathcal{T'}$ which is smaller than $\mathcal{T}$. In the configuration of Figure~\ref{td:configuration-fig}(b), consider all pairs of potential positions for two points defining $t_4$ which is shown in more detail in Figure~\ref{td:matching3-fig}(b). The pairs of potential positions on the boundary of $t_4$ are shown in red, green, and orange. Consider the red pair, and look at $t_2$, $t_4$, and $left(t_1)$. By Lemma~\ref{td:triangle3} and the same reasoning as for the previous configuration, we obtain a smaller tree $\mathcal{T'}$;  which contradicts the minimality of $\mathcal{T}$. By symmetry, the green and orange pairs lead to a contradiction.
Therefore, all configurations are invalid; which proves the lemma.

\begin{figure}[htb]
  \centering
\setlength{\tabcolsep}{0in}
  $\begin{tabular}{cc}
 \multicolumn{1}{m{.5\columnwidth}}{\centering\includegraphics[width=.34\columnwidth]{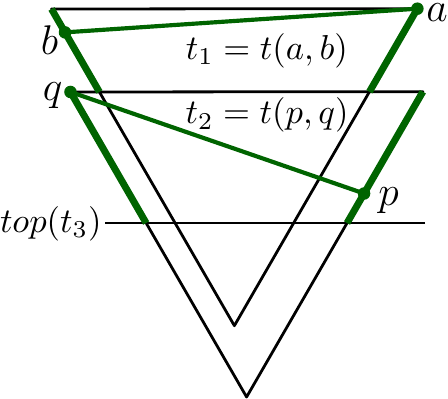}}
&\multicolumn{1}{m{.5\columnwidth}}{\centering\includegraphics[width=.4\columnwidth]{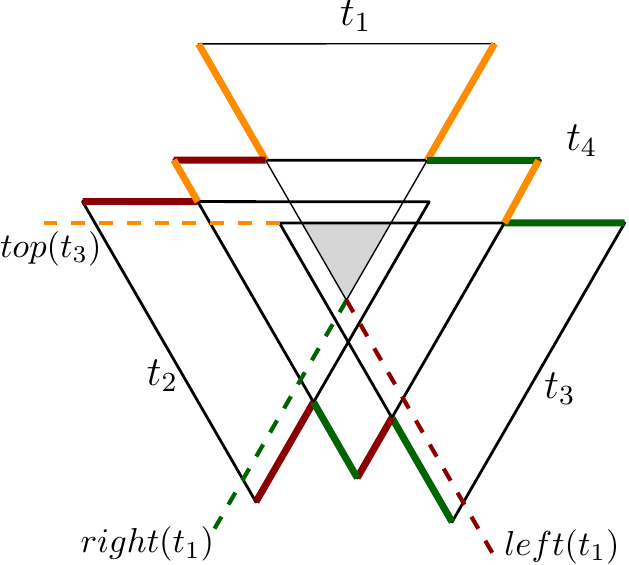}}\\
(a) & (b)
\end{tabular}$
  \caption{Illustration of Lemma~\ref{td:intersection-lemma}.}
\label{td:matching3-fig}
\end{figure}
\end{proof}

Our results in this section are based on Lemma~\ref{td:empty-triangle-lemma}, Lemma~\ref{td:intersection-lemma}, Theorem~\ref{td:Tutte}, and Theorem~\ref{td:Berge}. 

Now we prove that \kTD{2}{} has a perfect matching.

\begin{theorem}
 \label{td:mt-thr}
For a set $P$ of an even number of points in general position in the plane, \kTD{2}{} has a perfect matching.
\end{theorem}
\begin{proof}
First we show that by removing a set $K$ of points from \kTD{2}{}, at most $|K|+1$ components are generated. Then we show that at least one of these components must be even. Finally by Theorem~\ref{td:Tutte} we conclude that \kTD{2}{} has a perfect matching.

Let $K$ be a set of vertices removed from \kTD{2}{}, and let $\mathcal{C}=\{C_1, \dots,\allowbreak  C_{m(K)}\}$ be the resulting $m(K)$ components. Then, $\mathcal{P}=\{V(C_1),\dots, V(C_{m(K)})\}$ is a partition of the vertices in $P\setminus K$. 

\begin{paragraph}{Claim 1.}$m(K)\le |K|+1$.\end{paragraph}

Let $G(\mathcal{P})$ be the complete graph with vertex set $\mathcal{P}$ which is constructed as described above. Let $\mathcal{T}$ be a minimum spanning tree of $G(\mathcal{P})$ and let $T$ be the set of triangles corresponding to the edges of $\mathcal{T}$. It is obvious that $\mathcal{T}$ contains $m(K)-1$ edges and hence $|T|=m(K)-1$. Let $F=\{(p,t):p\in K, t\in T, p\in t\}$ be the set of all (point, triangle) pairs where $p\in K$, $t\in T$, and $p$ is inside $t$. By Lemma~\ref{td:intersection-lemma} each point in $K$ can be inside at most three triangles in $T$. Thus, $|F|\le 3\cdot|K|$.
Now we show that each triangle in $T$ contains at least three points of $K$.  
Consider any triangle $\tau\in T$. Let $e=(V(C_i),V(C_j))$ be the edge of $\mathcal{T}$ which is corresponding to $\tau$, and let $a\in V(C_i)$ and $b\in V(C_j)$ be the points defining $\tau$. By Lemma~\ref{td:empty-triangle-lemma}, $\tau$ does not contain any point of $P\setminus K$ in its interior. Therefore, $\tau$ contains at least three points of $K$, because otherwise $(a,b)$ is an edge in \kTD{2}{} which contradicts the fact that $a$ and $b$ belong to different components in $\mathcal{C}$. Thus, each triangle in $T$ contains at least three points of $K$ in its interior. That is, $3\cdot|T|\le|F|$. Therefore, $3(m(K)-1)\le |F|\le 3|K|$, and hence $m(K)\le |K|+1$.

\begin{paragraph}{Claim 2.}$o(\mathcal{C})\le |K|$.\end{paragraph}

By Claim 1, $|\mathcal{C}|=m(K)\le |K|+1$. If $|\mathcal{C}|\le |K|$, then $o(\mathcal{C})\le |K|$. Assume that $|\mathcal{C}|=|K|+1$. Since $P=K\cup \{\bigcup^{|K|+1}_{i=1}{V(C_i)}\}$, the total number of vertices of $P$ can be defined as $n=|K|+\sum_{i=1}^{|K|+1}{|V(C_i)|}$. Consider two cases where (i) $|K|$ is odd, (ii) $|K|$ is even. In both cases if all the components in $\mathcal{C}$ are odd, then $n$ is odd; this contradicts our assumption that $P$ has an even number of vertices. Thus, $\mathcal{C}$ contains at least one even component, which implies that $o(\mathcal{C})\le |K|$.

Finally, by Claim 2 and Theorem~\ref{td:Tutte}, we conclude that \kTD{2}{} has a perfect matching.
\end{proof}

\begin{theorem}
\label{td:matching-1TD}
 For every set $P$ of points in general position in the plane, \kTD{1}{} has a matching of size $\frac{2(n-1)}{5}$.
\end{theorem}

\begin{proof}
Let $K$ be a set of vertices removed from \kTD{1}{}, and let $\mathcal{C}=\{C_1, \dots, C_{m(K)}\}$ be the resulting $m(K)$ components. Then, $\mathcal{P}=\{V(C_1),\allowbreak \dots, V(C_{m(K)})\}$ is a partition of the vertices in $P\setminus K$. Note that $o(\mathcal{C})\le m(K)$.
Let $M^*$ be a maximum matching in \kTD{1}{}. By Theorem~\ref{td:Berge}, 

\begin{align}
\label{td:align0}
|M^*|&= \frac{1}{2}(n-\text{def}(\text{\kTD{1}{}})),
\end{align}
where
\begin{align}
\label{td:align1}
\text{def}(\text{\kTD{1}{}})&= \max\limits_{K\subseteq P}(o(\mathcal{C})-|K|)
 \le \max\limits_{K\subseteq P}(|\mathcal{C}|-|K|)= \max\limits_{0\le |K|\le n}(m(K)-|K|).
\end{align}
Define $G(\mathcal{P})$, $\mathcal{T}$, $T$, and $F$ as in the proof of Theorem~\ref{td:mt-thr}. Then, by Lemma~\ref{td:intersection-lemma}, $|F|\le 3\cdot|K|$.
By the same reasoning as in the proof of Theorem~\ref{td:mt-thr}, each triangle in $T$ has at least two points of $K$ in its interior. Thus, $2|T|\le|F|$. Therefore, $2(m(K)-1)\le |F| \le 3|K|$, and hence

\begin{equation}
\label{td:ineq1}
 m(K)\le\frac{3|K|}{2}+1.
\end{equation} 
In addition, $|K|+m(K)=|K|+|\mathcal{C}|\le |P|=n$, and hence
\begin{equation}
\label{td:ineq2}               
m(K)\le n-|K|.
\end{equation}
By Inequalities~(\ref{td:ineq1}) and ~(\ref{td:ineq2}), 
\begin{equation}
\label{td:ineq3}               
m(K)\le \min\left\{\frac{3|K|}{2}+1, n-|K|\right\}.
\end{equation}
Thus, by (\ref{td:align1}) and (\ref{td:ineq3})
\begin{align}
\label{td:align2}
\text{def}(\text{\kTD{1}{}})&\le \max\limits_{0\le |K|\le n}(m(K)-|K|)\le \max\limits_{0\le |K|\le n}\left\{\min\left\{\frac{3|K|}{2}+1, n-|K|\right\}-|K|\right\}\nonumber\\
&= \max\limits_{0\le |K|\le n}\left\{\min\left\{\frac{|K|}{2}+1, n-2|K|\right\}\right\}= \frac{n+4}{5},
\end{align}

where the last equation is achieved by setting $\frac{|K|}{2}+1$ equal to $n-2|K|$, which implies $|K|=\frac{2(n-1)}{5}$. Finally by substituting (\ref{td:align2}) in Equation (\ref{td:align0}) we have
$$
|M^*|\ge \frac{2(n-1)}{5}.
$$
\end{proof}
\section{Blocking TD-Delaunay graphs}
\label{td:blocking-section}
In this section we consider the problem of blocking TD-Delaunay graphs. Let $P$ be a set of $n$ points in the plane such that no pair of points of $P$ is collinear in the $l^{0}$, $l^{60}$, and $l^{120}$ directions. Recall that a point set $K$ blocks \kTD{k}{$(P)$} if in \kTD{k}{$(P\cup K)$} there is no edge connecting two points in $P$. That is, $P$ is an independent set in \kTD{k}{$(P\cup K)$}.

\begin{theorem}
\label{td:blocking-thr1}
At least $\lceil\frac{(k+1)(n-1)}{3}\rceil$ points are necessary to block \kTD{k}{$(P)$}.
\end{theorem}
\begin{proof}
Let $K$ be a set of $m$ points which blocks \kTD{k}{$(P)$}. Let $G(\mathcal{P})$ be the complete graph with vertex set $\mathcal{P}=P$. Let $\mathcal{T}$ be a minimum spanning tree of $G(\mathcal{P})$ and let $T$ be the set of triangles corresponding to the edges of $\mathcal{T}$. It is obvious that $|T|=n-1$. By Lemma~\ref{td:empty-triangle-lemma} the triangles in $T$ are empty, thus, the edges of $\mathcal{T}$ belong to any \kTD{k}{$(P)$} where $k\ge 0$. To block each edge, corresponding to a triangle in $T$, at least $k+1$ points are necessary. By Lemma~\ref{td:intersection-lemma} each point in $K$ can lie in at most three triangles of $T$. Therefore, $m\ge\lceil\frac{(k+1)(n-1)}{3}\rceil$, which implies that at least $\lceil\frac{(k+1)(n-1)}{3}\rceil$ points are necessary to block all the edges of $\mathcal{T}$ and hence \kTD{k}{$(P)$}.
\end{proof}
\begin{figure}[htb]
  \centering
\setlength{\tabcolsep}{0in}
  $\begin{tabular}{cc}
 \multicolumn{1}{m{.5\columnwidth}}{\centering\includegraphics[width=.4\columnwidth]{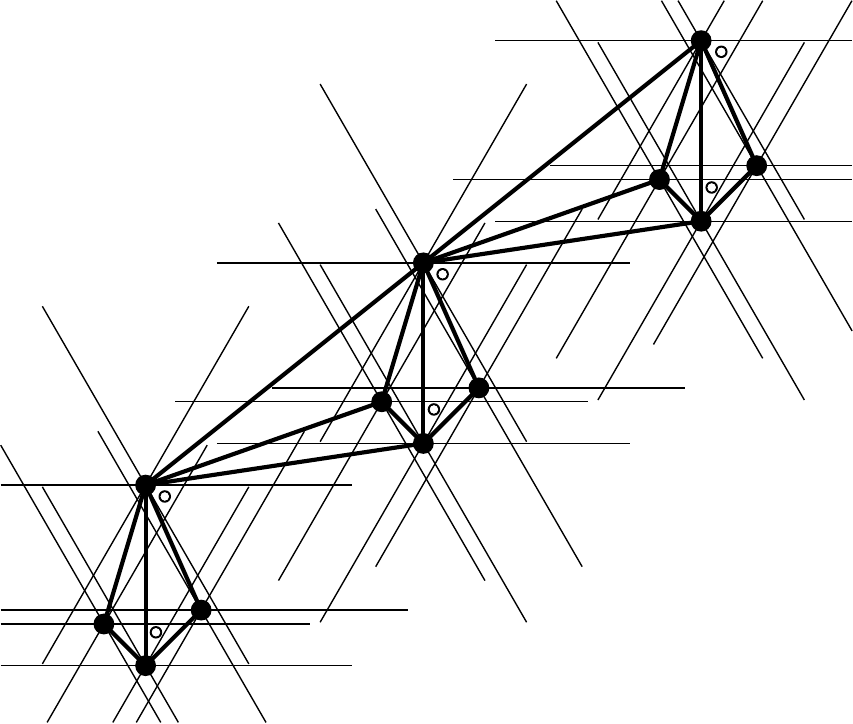}}
&\multicolumn{1}{m{.5\columnwidth}}{\centering\includegraphics[width=.4\columnwidth]{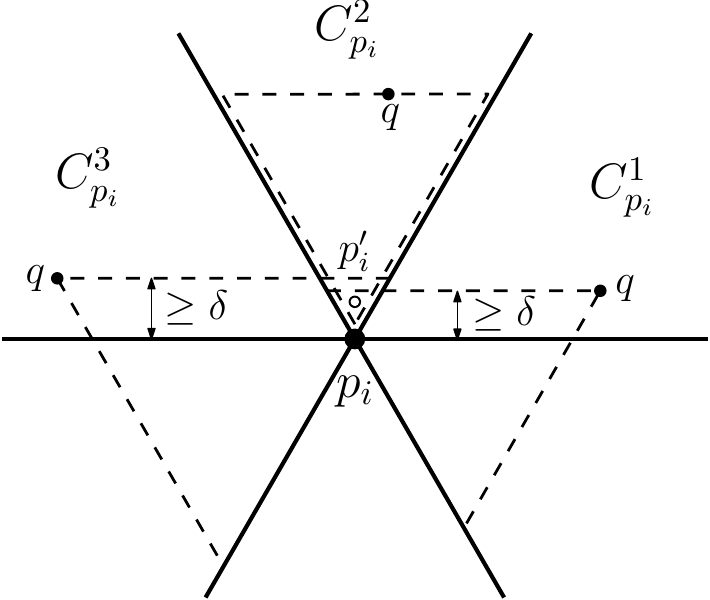}}\\
(a) & (b)
\end{tabular}$
  \caption{(a) A \kTD{0}{} graph which is shown in bold edges is blocked by $\lceil\frac{n-1}{2}\rceil$ white points, (b) $p'_i$ blocks all the edges connecting $p_i$ to the vertices above $l^0_{p_i}$.}
\label{td:blocking-fig}
\end{figure}

Theorem~\ref{td:blocking-thr1} gives a lower bound on the number of points that are necessary to block a TD-Delaunay graph. By this theorem, at least $\lceil\frac{n-1}{3}\rceil$, $\lceil\frac{2(n-1)}{3}\rceil$, $n-1$ points are necessary to block $0\text{-}$, $1\text{-}$, \kTD{2}{$(P)$} respectively. Now we introduce another formula which gives a better lower bound for \kTD{0}{}. For a point set $P$, let $\nu_k(P)$ and $\alpha_k(P)$ respectively denote the size of a maximum matching and a maximum independent set in \kTD{k}{$(P)$}. For every edge in the maximum matching, at most one of its endpoints can be in the maximum independent set. Thus,
\begin{equation}
\label{td:matching-independent}
 \alpha_k(P)\le |P| - \nu_k(P).
\end{equation}
Let $K$ be a set of $m$ points which blocks \kTD{k}{$(P)$}. By definition there is no edge between points of $P$ in \kTD{k}{$(P\cup K)$}. That is, $P$ is an independent set in \kTD{k}{$(P\cup K)$}. Thus, 
\begin{equation}
\label{td:ineq4}
 n\le \alpha_k(P\cup K).
\end{equation}
By (\ref{td:matching-independent}) and (\ref{td:ineq4}) we have
\begin{equation}
\label{td:ineq5}
 n\le \alpha_k(P\cup K)\le (n+m)-\nu_k(P\cup K).
\end{equation}
\begin{theorem}
\label{td:blocking-0TD-thr}
  At least $\lceil\frac{n-1}{2}\rceil$ points are necessary to block \kTD{0}{$(P)$}.
\end{theorem}
\begin{proof}
Let $K$ be a set of $m$ points which blocks \kTD{0}{$(P)$}. Consider \kTD{0}{$(P\cup K)$}. It is known that $\nu_0(P\cup K) \ge\lceil\frac{n+m-1}{3}\rceil$; see~\cite{Babu2014}. By Inequality~(\ref{td:ineq5}), $$n\le (n+m)-\left\lceil\frac{n+m-1}{3}\right\rceil\le \frac{2(n+m)+1}{3},$$ and consequently $m\ge \lceil\frac{n-1}{2}\rceil$ (note that $m$ is an integer number).
\end{proof}

Figure~\ref{td:blocking-fig}(a) shows a \kTD{0}{} graph on a set of 12 points which is blocked by 6 points. By removing the topmost point we obtain a set with odd number of points which can be blocked by 5 points. Thus, the lower bound provided by Theorem~\ref{td:blocking-0TD-thr} is tight. 

Now let $k=1$. By Theorem~\ref{td:matching-1TD} we have $\nu_1(P\cup K)\ge \frac{2((n+m)-1)}{5}$, and by Inequality~(\ref{td:ineq5}) $$n\le (n+m)-\frac{2((n+m)-1)}{5}=\frac{3(n+m)+2}{5},$$ and consequently $m\ge \lceil\frac{2(n-1)}{3}\rceil$; the same lower bound as in Theorem~\ref{td:blocking-thr1}. 

Now let $k=2$. By Theorem~\ref{td:mt-thr} we have $\nu_2(P\cup K)= \lfloor\frac{n+m}{2}\rfloor$ (note that $n+m$ may be odd). By Inequality~(\ref{td:ineq5}) $$n\le (n+m)-\left\lfloor\frac{n+m}{2}\right\rfloor=\left\lceil\frac{n+m}{2}\right\rceil,$$ and consequently $m\ge n$, where $n+m$ is even, and $m\ge n-1$, where $n+m$ is odd.  

\begin{theorem}
\label{td:blocking-thr2}
 There exists a set $K$ of $n-1$ points that blocks \kTD{0}{$(P)$}.
\end{theorem}

\begin{proof}
Let $d^0(p,q)$ be the Euclidean distance between $l^0_p$ and $l^0_q$. Let $\delta = \min\{d^0(p,q): p,q\in P\}$.
 For each point $p\in P$ let $p(x)$ and $p(y)$ respectively denote the $x$ and $y$ coordinates of $p$ in the plane. Let $p_1, \dots, p_n$ be the points of $P$ in the increasing order of their $y$-coordinate. Let $K=\{p'_i: p'_i(x)=p_i(x), p'_i(y)=p_i(y)+\epsilon, \epsilon<\delta, 1\le i\le n-1\}$. See Figure~\ref{td:blocking-fig}(b). For each point $p_i$, let $E_{p_i}$ (resp. $\overline{E_{p_i}}$) denote the edges of \kTD{0}{$(P)$} between $p_i$ and the points above $l^0_{p_i}$ (resp. below $l^0_{p_i}$). It is easy to see that the downward triangle between $p_i$ and any point $q$ above $l^0_{p_i}$ (i.e. any point $q\in C^1_{p_i}\cup C^2_{p_i}\cup C^3_{p_i}$) contains $p'_i$. Thus, $p'_i$ blocks all the edges in $E_{p_i}$. In addition, the edges in $\overline{E_{p_i}}$ are blocked by $p'_1, \dots, p'_{i-1}$. Therefore, all the edges of \kTD{0}{$(P)$} are blocked by the $n-1$ points in $K$.
\end{proof}

We can extend the result of Theorem~\ref{td:blocking-thr2} to \kTD{k}{$(P)$} where $k\ge 1$. For each point $p_i$ we put $k+1$ copies of $p'_i$ very close to $p_i$. Thus, 

\begin{theorem}
 There exists a set $K$ of $(k+1)(n-1)$ points that blocks \kTD{k}{$(P)$}.
\end{theorem}

This bound is tight. Consider the case where $k=0$. In this case \kTD{0}{$(P)$} can be a path representing $n-1$ disjoint triangles and for each triangle we need at least one point to block its corresponding edge.

\section{Conclusions}
\label{td:conclusion}
In this chapter, we considered some combinatorial properties of higher order triangular-distance Delaunay graphs of a point set $P$. We proved that
\begin{itemize}

  \item \kTD{k}{} is $(k+1)$ connected.
  \item \kTD{1}{} contains a bottleneck biconnected spanning graph of $P$.
  \item \kTD{7}{} contains a bottleneck Hamiltonian cycle and \kTD{5}{} may not have any.
  \item \kTD{6}{} contains a bottleneck perfect matching and \kTD{5}{} may not have any.
  \item \kTD{1}{} has a matching of size at least $\frac{2(n-1)}{5}$.
  \item \kTD{2}{} has a perfect matching when $P$ has an even number of points.
  \item $\lceil\frac{n-1}{2}\rceil$ points are necessary to block \kTD{0}{}.
  \item $\lceil\frac{(k+1)(n-1)}{3}\rceil$ points are necessary and $(k+1)(n-1)$ points are sufficient to block \kTD{k}{}.
\end{itemize}

We leave a number of open problems:
\begin{itemize}

  \item What is a tight lower bound for the size of maximum matching in \kTD{1}{}?
  \item Does \kTD{6}{} contain a bottleneck Hamiltonian cycle?
 \item As shown in Figure~\ref{td:TD}(a) \kTD{0}{} may not have a Hamiltonian cycle. For which values of $k=1,\dots, 6$, is the graph \kTD{k}{} Hamiltonian?
\end{itemize}

\bibliographystyle{abbrv}
\bibliography{../thesis}

%% file: chapters/ch7-strongmatching.tex
\chapter{Strong Matching of Points with Geometric Shapes}
\label{ch:sm}

Let $P$ be a set of $n$ points in general position in the plane. Given a convex geometric shape $S$, a geometric graph $\G{S}{P}$ on $P$ is defined to have an edge between two points if and only if there exists an empty homothet of $S$ having the two points on its boundary. A matching in $\G{S}{P}$ is said to be {\em strong}, if the homothests of $S$ representing the edges of the matching, are pairwise disjoint, i.e., do not share any point in the plane. We consider the problem of computing a strong matching in $\G{S}{P}$, where $S$ is a diametral-disk, an equilateral-triangle, or a square. We present an algorithm which computes a strong matching in $\G{S}{P}$; if $S$ is a diametral-disk, then it computes a strong matching of size at least $\lceil \frac{n-1}{17} \rceil$, and if $S$ is an equilateral-triangle, then it computes a strong matching of size at least $\lceil \frac{n-1}{9} \rceil$. If $S$ can be a downward or an upward equilateral-triangle, we compute a strong matching of size at least $\lceil \frac{n-1}{4} \rceil$ in $\G{S}{P}$. When $S$ is an axis-aligned square we compute a strong matching of size $\lceil \frac{n-1}{4} \rceil$ in $\G{S}{P}$, which improves the previous lower bound of $\lceil \frac{n}{5} \rceil$. 

\vspace{10pt}
The results that are presented in this chapter are accepted to be published in the journal of Computational Geometry: Theory and Applications, special issue in memoriam: Ferran Hurtado~\cite{Biniaz2015-strong}. 

\section{Introduction}
\label{sm:intro}

Let $S$ be a compact and convex set in the plane that contains the origin in its interior. A {\em homothet} of $S$ is obtained by scaling $S$ with respect to the origin by some factor $\mu\ge 0$, followed by a translation to a point $b$ in the plane: $b+\mu S=\{b+\mu a: a\in S\}$.
For a point set $P$ in the plane, we define $\G{S}{P}$ as the geometric graph on $P$ which has an straight-line edge between two points $p$ and $q$ if and only if there exists a homothet of $S$ having $p$ and $q$ on its boundary and whose interior does not contain any point of $P$. If $P$ is in ``general position'', i.e., no four points of $P$ lie on the boundary of any homothet of $S$, then $\G{S}{P}$ is plane~\cite{Bose2010}. Hereafter, we assume that $P$ is a set of $n$ points in the plane, which is in general position with respect to $S$ (we will define the general position in Section~\ref{sm:preliminaries}). 
If $S$ is a disk $\disc$ whose center is the origin, then $\G{\discs}{P}$ is the Delaunay triangulation of $P$. If $S$ is an equilateral triangle $\trid$ whose barycenter is the origin, then $\G{\trids}{P}$ is the triangular-distance Delaunay graph of $P$ which is introduced by Chew~\cite{Chew1989}.

A {\em matching} in a graph $G$ is a set of edges which do not share any vertices. A {\em maximum matching} is a matching with maximum cardinality. A {\em perfect matching} is a matching which matches all the vertices of $G$. Let $\mathcal{M}$ be a matching in $\G{S}{P}$. $\mathcal{M}$ is referred to as a {\em matching of points with shape $S$}, e.g., a matching in $\G{\discs}{P}$ is a matching of points with with disks. Let $\mathcal{S}_{\mathcal{M}}$ be a set of homothets of $S$ representing the edges of $\mathcal{M}$. $\mathcal{M}$ is called a {\em strong matching} if there exists a set $\mathcal{S}_{\mathcal{M}}$ whose elements are pairwise disjoint, i.e., the objects in $\mathcal{S}_{\mathcal{M}}$ do not share any point in the plane. Otherwise, $\mathcal{M}$ is a {\em weak matching}. See Figure~\ref{sm:strong-example}. To be consistent with the definition of the matching in the graph theory, we use the term ``matching'' to refer to a weak matching. Given a point set $P$ in the plane and a shape $S$, the {\em (strong) matching problem} is to compute a (strong) matching of maximum cardinality in $\G{S}{P}$.
In this chapter we consider the strong matching problem of points in general position in the plane with respect to a given shape $S\in \{\ddisc,\trid, \sqr\}$ (see Section~\ref{sm:preliminaries} for the definition), where by $\ddisc$ we mean the line segment between the two points on the boundary of the disk is a diameter of that disk.

\begin{figure}[htb]
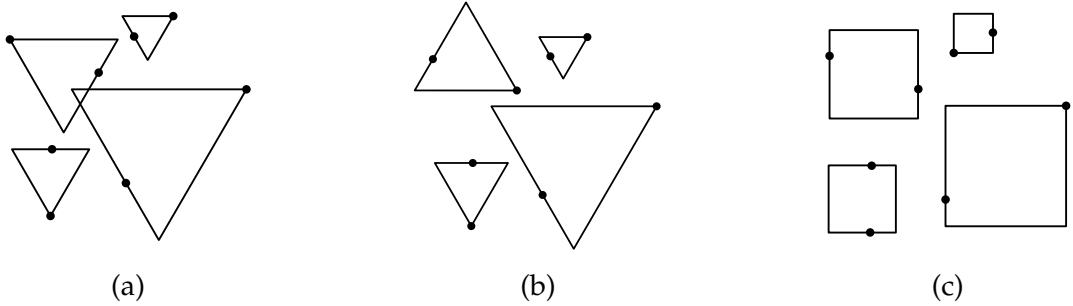

  \centering
\setlength{\tabcolsep}{0in}
  $\begin{tabular}{ccc}
\multicolumn{1}{m{.33\columnwidth}}{\centering\includegraphics[width=.2\columnwidth]{sm/strong-triangle-example.pdf}}
&\multicolumn{1}{m{.33\columnwidth}}{\centering\includegraphics[width=.2\columnwidth]{sm/strong-triangleup-example.pdf}} &\multicolumn{1}{m{.33\columnwidth}}{\centering\includegraphics[width=.2\columnwidth]{sm/strong-square-example.pdf}}
\\
(a) & (b)& (c)
\end{tabular}$
  \caption{Point set $P$ and (a) a perfect weak matching in $\G{\trids}{P}$, (b) a perfect strong matching in $\GUD(P)$, and (c) a perfect strong matching in $
\G{\sqr}{P}$.}
\label{sm:strong-example}
\end{figure}
\addtocontents{toc}{\protect\setcounter{tocdepth}{1}}
\subsection{Previous Work}
\addtocontents{toc}{\protect\setcounter{tocdepth}{2}}
\label{sm:previous-work}
The problem of computing a maximum matching in $\G{S}{P}$ is one of the fundamental problems in computational geometry and graph theory \cite{Abrego2004, Abrego2009, Babu2014, Bereg2009, Biniaz2015-ggmatching-TCS, Biniaz2015-hotd-CGTA, Dillencourt1990}. 
Dillencourt~\cite{Dillencourt1990} and \'{A}brego et al. \cite{Abrego2004} considered the problem of matching points with disks. Let $S$ be a closed disk $\disc$ whose center is the origin, and let $P$ be a set of $n$ points in the plane which is in general position with respect to $\disc$. Then, $\G{\discs}{P}$ is the graph which has an edge between two points $p,q\in P$ if there exists a homothet of $\disc$ having $p$ and $q$ on its boundary and does not contain any point of $P\setminus\{p,q\}$. $\G{\discs}{P}$ is equal to the Delaunay triangulation on $P$, $DT(P)$. Dillencourt~\cite{Dillencourt1990} proved that $\G{\discs}{P}$ contains a perfect (weak) matching. \'{A}brego et al. \cite{Abrego2004} proved that $\G{\discs}{P}$ has a strong matching of size at least $\lceil(n-1)/8\rceil$. They also showed that there exists a set $P$ of $n$ points in the plane with arbitrarily large $n$, such that $\G{\discs}{P}$ does not contain a strong matching of size more than $\frac{36}{73}n$.

For two points $p$ and $q$, the disk which has the line segment $pq$ as its diameter is called the diametral-disk between $p$ and $q$. We denote a diametral-disk by $\ddisc$. Let $\G{\ddiscs}{P}$ be the graph which has an edge between two points $p,q\in P$ if the diametral-disk between $p$ and $q$ does not contain any point of $P\setminus\{p,q\}$. $\G{\ddiscs}{P}$ is equal to the Gabriel graph on $P$, $GG(P)$. Biniaz et al.~\cite{Biniaz2015-ggmatching-TCS} proved that $\G{\ddiscs}{P}$ has a matching of size at least $\lceil(n-1)/4\rceil$, and this bound is tight.

The problem of matching of points with equilateral triangles has been considered by Babu et al.~\cite{Babu2014}.
Let $S$ be a downward equilateral triangle $\trid$ whose barycenter is the origin and one of its vertices is on the negative $y$-axis. Let $P$ be a set of $n$ points in the plane which is in general position with respect to $\trid$. Let $\G{\trids}{P}$ be the graph which has an edge between two points $p,q\in P$ if there exists a homothet of $\trid$ having $p$ and $q$ on its boundary and does not contain any point of $P\setminus\{p,q\}$. $\G{\trids}{P}$ is equal to the triangular-distance Delaunay graph on $P$, which was introduced by Chew~\cite{Chew1989}. Bonichon et al.~\cite{Bonichon2010} showed that $\G{\trids}{P}$ is equal to the half-theta six graph on $P$, $\frac{1}{2}\Theta_6(P)$. Babu et al.~\cite{Babu2014} proved that $\G{\trids}{P}$ has a matching of size at least $\lceil(n-1)/3\rceil$, and this bound is tight. If we consider an upward triangle $\triu$, then $\G{\trius}{P}$ is defined similarly. Let $\GUD(P)$ be the graph on $P$ which is the union of $\G{\trids}{P}$ and $\G{\trius}{P}$. Bonichon et al.~\cite{Bonichon2010} showed that $\GUD(P)$ is equal to the theta six graph on $P$, $\Theta_6(P)$. Since $\G{\trids}{P}$ is a subgraph of $\GUD(P)$, the lower bound of $\lceil(n-1)/3\rceil$ on the size of maximum matching in $\G{\trids}{P}$ holds for $\GUD(P)$.  

The problem of strong matching of points with axis-aligned rectangles is trivial. An obvious algorithm is to repeatedly match the two leftmost points. The problem of matching points with axis-aligned squares was considered by \'{A}brego et al. \cite{Abrego2009}.
Let $S$ be an axis-aligned square $\sqr$ whose center is the origin. Let $P$ be a set of $n$ points in the plane which is in general position with respect to $\sqr$. Let $\G{\sqrs}{P}$ be the graph which has an edge between two points $p,q\in P$ if there exists a homothet of $\sqr$ having $p$ and $q$ on its boundary and does not contain any point of $P\setminus\{p,q\}$. $\G{\sqrs}{P}$ is equal to the $L_\infty$-Delaunay graph on $P$. \'{A}brego et al. \cite{Abrego2004, Abrego2009} proved that $\G{\sqrs}{P}$ has a perfect (weak) matching and a strong matching of size at least $\lceil n/5\rceil$. Further, they showed that there exists a set $P$ of $n$ points in the plane with arbitrarily large $n$, such that $\G{\sqrs}{P}$ does not contain a strong matching of size more than $\frac{5}{11}n$. Table~\ref{sm:table1} summarizes the results.

Bereg et al.~\cite{Bereg2009} concentrated on matching points of $P$ with axis-aligned rectangles and squares, where $P$ is not necessarily in general position. 
They proved that any set of $n$ points in the plane has a strong rectangle matching of size at least $\lfloor\frac{n}{3}\rfloor$, and such a matching can be computed in $O(n \log n)$ time. As for squares, they presented a $\Theta(n\log n)$ time algorithm that decides whether a given matching has a weak square realization, 
and an $O(n^2\log n)$ time algorithm for the strong square matching realization. They also proved that it is NP-hard to decide whether a given point set has a perfect strong square-matching. 

\begin{table}
\centering
\caption{Lower bounds on the size of weak and strong matchings in $\G{S}{P}$.}
\label{sm:table1}
    \begin{tabular}{|c|c|c||c|c|}
         \hline
             $S$ 	& Weak match. & Ref.& Strong match. &Ref. \\  \hline  \hline
	      {$\disc$}& 	${\lfloor \frac{n}{2}\rfloor}$&\cite{Dillencourt1990}& $\lceil\frac{n-1}{8}\rceil$&\cite{Abrego2004}\\
	      {$\ominus$} &${\lceil \frac{n-1}{4}\rceil}$&\cite{Biniaz2015-ggmatching-TCS} &$\lceil\frac{n-1}{17}\rceil$&Theorem~\ref{sm:Gabriel-thr}\\  
             {$\bigtriangledown$} &${\lceil \frac{n-1}{3}\rceil}$&\cite{Babu2014} &$\lceil\frac{n-1}{9}\rceil$&Theorem~\ref{sm:half-theta-six-thr}\\  
	    {$\trid$ or $\triu$} &$\lceil \frac{n-1}{3}\rceil$& \cite{Babu2014} &$\lceil\frac{n-1}{4}\rceil$&Theorem~\ref{sm:theta-six-thr}\\ \hline\hline
	    \multirow{2}{*}{$\Square$}& \multirow{2}{*}{$\lfloor \frac{n}{2}\rfloor$} &\multirow{2}{*}{\cite{Abrego2004, Abrego2009}}&
	    $\lceil\frac{n}{5}\rceil$ & \cite{Abrego2004, Abrego2009}       \\ 
		& &	& $\lceil\frac{n-1}{4}\rceil$& Theorem~\ref{sm:infty-Delaunay-thr} \\ \hline
    \end{tabular}
\end{table}
\addtocontents{toc}{\protect\setcounter{tocdepth}{1}}
\subsection{Our results}
\addtocontents{toc}{\protect\setcounter{tocdepth}{2}}
\label{sm:our-results}
In this chapter we consider the problem of computing a strong matching in $\G{S}{P}$, where $S\in \{\ddisc,\trid, \sqr\}$. In Section~\ref{sm:preliminaries}, we provide some observations and prove necessary Lemmas. Given a point set $P$ in which is in general position with respect to a given shape $S$, in Section~\ref{sm:algorithm-section}, we present an algorithm which computes a strong matching in $\G{S}{P}$. In Section~\ref{sm:Gabriel-section}, we prove that if $S$ is a diametral-disk, then the algorithm of Section~\ref{sm:algorithm-section} computes a strong matching of size at least $\lceil(n-1)/17\rceil$ in $\G{\ddisc}{P}$. In Section~\ref{sm:half-theta-six-section}, we prove that if $S$ is an equilateral triangle, then the algorithm of Section~\ref{sm:algorithm-section} computes a strong matching of size at least $\lceil(n-1)/9\rceil$ in $\G{\trids}{P}$. In Section~\ref{sm:theta-six-section}, we compute a strong matching of size at least $\lceil (n-1)/4\rceil$ in $\GUD(P)$. In Section~\ref{sm:infty-Delaunay-section}, we compute a strong matching of size at least $\lceil(n-1)/4\rceil$ in $\G{\sqrs}{P}$; this improves the previous lower bound of $\lceil n/5\rceil$. A summary of the results is given in Table~\ref{sm:table1}. In Section~\ref{sm:conjecture-section} we discuss a possible way to further improve upon the result obtained for diametral-disks in Section~\ref{sm:Gabriel-section}. Concluding remarks and open problems are given in Section~\ref{sm:conclusion}.

\section{Preliminaries}
\label{sm:preliminaries}

Let $S\in\{\ddisc, \trid, \sqr\}$, and let $S_1$ and $S_2$ be two homothets of $S$. We say that $S_1$ is {\em smaller then} $S_2$ if the area of $S_1$ is smaller than the area of $S_2$. For two points $p,q\in P$, let $S(p,q)$ be a smallest homothet of $S$ having $p$ and $q$ on its boundary. If $S$ is a diametral-disk, a downward equilateral-triangle, or a square, then we denote $S(p,q)$ by $D(p,q)$, $t(p,q)$, or $Q(p,q)$, respectively. If $S$ is a diametral-disk, then $D(p,q)$ is uniquely defined by $p$ and $q$. If $S$ is an equilateral-triangle or a square, then $S$ has the {\em shrinkability} property: if there exists a homothet $S'$ of $S$ that contains two points $p$ and $q$, then there exists a homothet $S''$ of $S$ such that $S''\subseteq S'$, and $p$ and $q$ are on the boundary of $S''$. If $S$ is an equilateral-triangle, then we can
shrink $S''$ further, such that each side of $S''$ contains either $p$ or $q$. If $S$ is a square, then we can
shrink $S''$ further, such that $p$ and $q$ are on opposite sides of $S''$. Thus, we have the following observation:

\begin{observation}
\label{sm:shrink-triangle-obs}
For two points $p,q\in P$,
\begin{itemize}
 \item $D(p,q)$ is uniquely defined by $p$ and $q$, and it has the line segment $pq$ as a diameter.
\item $t(p, q)$ is uniquely defined by $p$ and $q$, and it has one of $p$ and $q$ on a corner and the other point is
on the side opposite to that corner.
\item $Q(p,q)$ has $p$ and $q$ on opposite sides.
\end{itemize}
\end{observation}

\begin{figure}[htb]
  \centering
\setlength{\tabcolsep}{0in}
  $\begin{tabular}{ccc}
\multicolumn{1}{m{.33\columnwidth}}{\centering\includegraphics[width=.25\columnwidth]{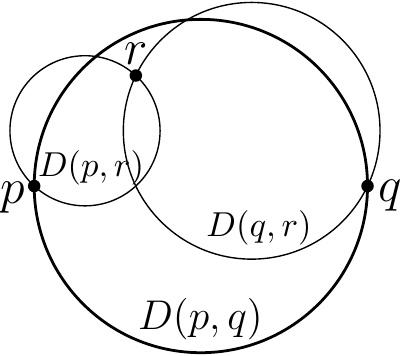}}
&\multicolumn{1}{m{.33\columnwidth}}{\centering\includegraphics[width=.22\columnwidth]{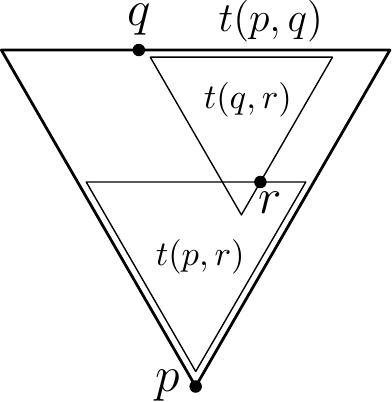}} &\multicolumn{1}{m{.33\columnwidth}}{\centering\includegraphics[width=.23\columnwidth]{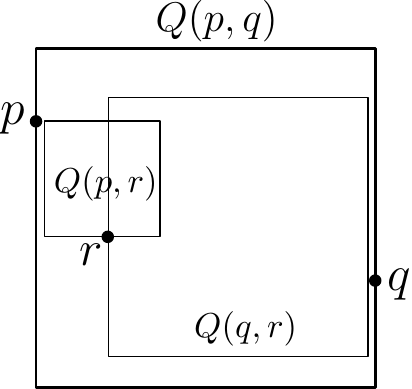}}
\end{tabular}$
  \caption{Illustration of Observation~\ref{sm:obs1}.}
\label{sm:mst-in-GS-fig}
\end{figure}
Given a shape $S\in\{\ddisc, \trid, \sqr\}$, we define an order on the homothets of $S$. Let $S_1$ and $S_2$ be two homothets of $S$. We say that $S_1\prec S_2$ if the area of $S_1$ is less than the area of $S_2$. Similarly, $S_1\preceq S_2$ if the area of $S_1$ is less than or equal to the area of $S_2$. We denote the homothet with the larger area by $\max\{S_1, S_2\}$. As illustrated in Figure~\ref{sm:mst-in-GS-fig}, if $S(p,q)$ contains a point $r$, then both $S(p,r)$ and $S(q,r)$ have smaller area than $S(p,q)$. Thus, we have the following observation:
 
\begin{observation}
\label{sm:obs1}
 If $S(p,q)$ contains a point $r$, then $\max\{S(p,r), \allowbreak S(q,r)\}\allowbreak \prec\allowbreak  S(p,q)$.
\end{observation}

\begin{definition}
 Given a point set $P$ and a shape $S\in\{\ddisc, \trid, \sqr\}$, we say that $P$ is in ``general position'' with respect to $S$ if
\begin{description}
 \item[$S=\ddisc$:] no four points of $P$ lie on the boundary of any diametral disk defined by any two points of $P$.
 \item[$S=\trid$:] the line passing through any two points of $P$ does not make angles $0^\circ$, $60^\circ$, or $120^\circ$ with the horizontal. This implies that no four points of $P$ are on the boundary of any homothet of $\trid$.
 \item[$S=\sqr$:] (i) no two points in $P$ have the same $x$-coordinate or the same $y$-coordinate, and (ii) no four points of $P$ lie on the boundary of any homothet of $\sqr$.
\end{description}
\end{definition}

Given a point set $P$ which is in general position with respect to a given shape $S\in\{\ddisc, \trid, \sqr\}$, let $K_S(P)$ be the complete edge-weighted geometric graph on $P$. For each edge $e=(p,q)$ in $K_S(P)$, we define $S(e)$ to be the shape $S(p,q)$, i.e., a smallest homothet of $S$ having $p$ and $q$ on its boundary. We say that $S(e)$ {\em represents} $e$, and vice versa. Furthermore, we assume that the weight $w(e)$ (resp. $w(p,q)$) of $e$ is equal to the area of $S(e)$. Thus,
$$w(p,q)<w(r,s) \quad\text{ if and only if }\quad S(p,q)\prec S(r,s).$$
Note that $\G{S}{P}$ is a subgraph of $K_S(P)$, and has an edge $(p,q)$ iff $S(p,q)$ does not contain any point of $P\setminus\{p,q\}$.

\begin{lemma}
\label{sm:mst-in-GS}
Let $P$ be a set of $n$ points in the plane which is in general position with respect to a given shape $S\in\{\ddisc, \trid, \sqr\}$. Then, any minimum spanning tree of $K_S(P)$ is a subgraph of $\G{S}{P}$. 
\end{lemma}
\begin{proof}
 The proof is by contradiction. Assume there exists an edge $e=(p,q)$ in a minimum spanning tree $T$ of $K_S(P)$ such that $e\notin \G{S}{P}$. Since $(p,q)$ is not an edge in $\G{S}{P}$, $S(p,q)$ contains a point $r$ such that $r\in P\setminus\{p,q\}$. By Observation~\ref{sm:obs1}, $\max\{S(p,r),S(q,r)\}\prec S(p,q)$. Thus, $w(p,r)<w(p,q)$ and $w(q,r)<w(p,q)$. By replacing the edge $(p,q)$ in $T$ with either $(p,r)$ or $(q,r)$, we obtain a spanning tree in $K_S(P)$ which is smaller than $T$. This contradicts the minimality of $T$.
\end{proof}

\begin{lemma}
\label{sm:cycle-lemma}
Let $G$ be an edge-weighted graph with edge set $E$ and edge-weight function $w:E\rightarrow\mathbb{R^+}$. For any cycle $C$ in $G$, if the maximum-weight edge in $C$ is unique, then that edge is not in any minimum spanning tree of $G$.
\end{lemma}

\begin{proof}
The proof is by contradiction. Let $e=(u,v)$ be the unique maximum-weight edge in a cycle $C$ in $G$, such that $e$ is in a minimum spanning tree $T$ of $G$. Let $T_u$ and $T_v$ be the two trees obtained by removing $e$ from $T$. Let $e'=(x,y)$ be an edge in $C$ which connects a vertex $x\in T_u$ to a vertex $y\in T_v$. By assumption, $w(e')< w(e)$. Thus, in $T$, by replacing $e$ with $e'$, we obtain a tree $T'=T_u\cup T_v \cup\{(x,y)\}$ in $G$ such that $w(T')<w(T)$. This contradicts the minimality of $T$. 
\end{proof}

Recall that $t(p,q)$ is the smallest homothet of $\trid$ which has $p$ and $q$ on its boundary. Similarly, let $t'(p, q)$ denote the smallest upward equilateral-triangle $\triu$ having $p$ and $q$ on its boundary. Note that $t'(p, q)$ is uniquely defined by $p$ and $q$, and it has one of $p$ and $q$ on a corner and the other point is on the side opposite to that corner. In addition the area of $t'(p, q)$ is equal to the area of $t(p, q)$. 

\begin{figure}[htb]
 \begin{center}
\includegraphics[width=.38\textwidth]{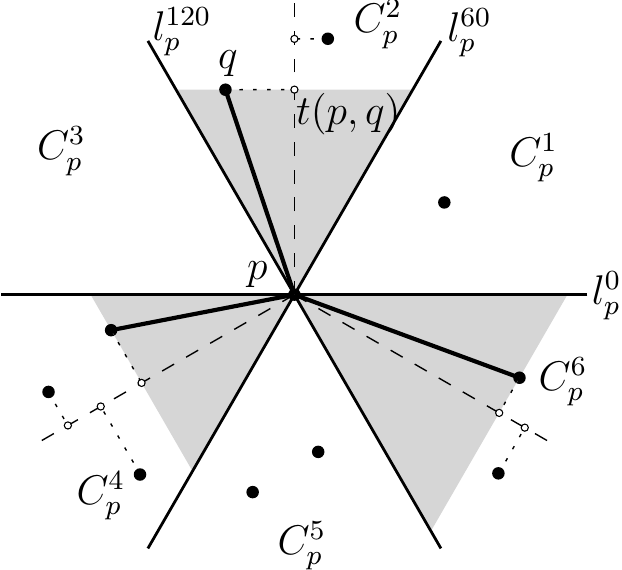}
  \end{center}
  \caption{The construction of $\G{\trid}{P}$.}
\label{sm:cones}
\end{figure}

$\G{\trids}{P}$ is equal to the triangular-distance Delaunay graph $\text{\em TD-DG}(P)$, which is in turn equal to a half theta-six graph $\frac{1}{2}\Theta_6(P)$~\cite{Bonichon2010}. 
A half theta-six graph on $P$, and equivalently $\G{\trids}{P}$, can be constructed in the following way. For each point $p$ in $P$, let $l_p$ be the horizontal line through $p$. Define $l_p^{\gamma}$ as the line obtained by rotating $l_p$ by $\gamma$-degrees in counter-clockwise direction around $p$. Thus, $l_p^0=l_p$. Consider three lines $l_p^{0}$, $l_p^{60}$, and $l_p^{120}$ which partition the plane into six disjoint cones with apex $p$. Let $C_p^1, \dots, C_p^6$ be the cones in counter-clockwise order around $p$ as shown in Figure~\ref{sm:cones}. $C_p^1,C_p^3,C_p^5$ will be referred to as {\em odd cones}, and $C_p^2,C_p^4,C_p^6$ will be referred to as {\em even cones}. For each even cone $C_p^i$, connect $p$ to the ``nearest'' point $q$ in $C_p^i$. The {\em distance} between $p$ and $q$, is defined as the Euclidean distance between $p$ and the orthogonal projection of $q$ onto the bisector of $C_p^i$. See Figure~\ref{sm:cones}. In other words, the nearest point to $P$ in $\cone{i}{p}$ is a point $q$ in $\cone{i}{p}$ which minimizes the area of $t(p,q)$. The resulting graph is the half theta-six graph which is defined by even cones \cite{Bonichon2010}. Moreover, the resulting graph is $\G{\trids}{P}$ which is defined with respect to the homothets of $\trid$. By considering the odd cones, $\G{\trius}{P}$ is obtained. By considering the odd cones and the even cones, $\GUD(P)$\textemdash which is equal to $\Theta_6(P)$\textemdash is obtained. Note that $\GUD(P)$ is the union of $\G{\trids}{P}$ and $\G{\trius}{P}$. 

Let $X(p,q)$ be the regular hexagon centered at $p$ which has $q$ on its boundary, and its sides are parallel to $l_p^0$, $l_p^{60}$, and $l_p^{120}$. Then, we have the following observation:
\begin{observation}
\label{sm:obs2}
If $X(p,q)$ contains a point $r$, then $t(p,r)\prec t(p,q)$.
\end{observation}

\section{Strong Matching in $\G{S}{P}$}
\label{sm:algorithm-section}

Given a point set $P$ in the plane which is in general position with respect to a given shape $S\in\{\ddisc, \trid, \sqr\}$, in this section we present an algorithm which computes a strong matching in $\G{S}{P}$. Recall that $K_S(P)$ is the complete edge-weighted graph on $P$ with the weight of each edge $e$ is equal to the area of $S(e)$, where $S(e)$ is a smallest homothet of $S$ representing $e$. Let $T$ be a minimum spanning tree of $K_S(P)$. By Lemma~\ref{sm:mst-in-GS}, $T$ is a subgraph of $\G{S}{P}$. For each edge $e\in T$ we denote by $T(e^+)$ the set of all edges in $T$ whose weight is at least $w(e)$. Moreover, we define the {\em influence set} of $e$, as the set of all edges in $T(e^+)$ whose representing shapes overlap with $S(e)$, i.e.,
$$\Inf{e}=\{e': e'\in T(e^+), S(e')\cap S(e)\neq \emptyset\}.$$

Note that $\Inf{e}$ is not empty, as $e\in \Inf{e}$. Consequently, we define the {\em influence number} of $T$ to be the maximum size of a set among the influence sets of edges in $T$, i.e.,
$$\Inf{T}=\max\{|\Inf{e}|: e\in T\}.$$

Algorithm~\ref{sm:alg1} receives $\G{S}{P}$ as input and computes a strong matching in $\G{S}{P}$ as follows. The algorithm starts by computing a minimum spanning tree $T$ of $\G{S}{P}$, where the weight of each edge is equal to the area of its representing shape. Then it initializes a forest $F$ by $T$, and a matching $\mathcal{M}$ by an empty set. Afterwards, as long as $F$ is not empty, the algorithm adds to $\mathcal{M}$, the smallest edge $\emin$ in $F$, and removes the influence set of $e$ from $F$. Finally, it returns $\mathcal{M}$.
\begin{algorithm}                      
\caption{\SMGG$(\G{S}{P})$}          
\label{sm:alg1} 
\begin{algorithmic}[1]
      \State $T\gets \MST(G_S(P))$
      \State $F\gets T$
      \State $\mathcal{M}\gets \emptyset$
      \While {$F\neq \emptyset$}
	  \State $\emin\gets $ smallest edge in $F$
	  \State $\mathcal{M}\gets \mathcal{M}\cup \{\emin\}$
	  \State $F\gets F - \Inf{\emin}$
	  \EndWhile
    \State \Return $\mathcal{M}$
\end{algorithmic}
\end{algorithm}
\begin{theorem}
\label{sm:GS-thr}
Given a set $P$ of $n$ points in the plane and a shape $S\in\{\ddisc, \trid, \sqr\}$, Algorithm~\ref{sm:alg1} computes a strong matching of size at least $\lceil\frac{n-1}{\emph{Inf}(T)}\rceil$ in $\G{S}{P}$, where $T$ is a minimum spanning tree of $\G{S}{P}$. 
\end{theorem}
\begin{proof}
Let $\mathcal{M}$ be the matching returned by Algorithm~\ref{sm:alg1}. First we show that $\mathcal{M}$ is a strong matching. If $\mathcal{M}$ contains one edge, then trivially, $\mathcal{M}$ is a strong matching. Consider any two edges $e_1$ and $e_2$ in $\mathcal{M}$. Without loss of generality assume that $e_1$ is considered before $e_2$ in the {\sf while} loop. At the time $e_1$ is added to $\mathcal{M}$, the algorithm removes from $F$, the edges in $\Inf{e_1}$, i.e., all the edges whose representing shapes intersect $S(e_1)$. Since $e_2$ remains in $F$ after the removal of $\Inf{e_1}$, $e_2\notin\Inf{e_1}$. This implies that $S(e_1)\cap S(e_2)=\emptyset$, and hence $\mathcal{M}$ is a strong matching.

In each iteration of the {\sf while} loop we select $\emin$ as the smallest edge in $F$, where $F$ is a subgraph of $T$. Then, all edges in $F$ have weight at least $w(e)$. Thus, $F\subseteq T(\emin^+)$; which implies that the set of edges in $F$ whose representing shapes intersect $S(\emin)$ is a subset of $\Inf{\emin}$. Therefore, in each iteration of the {\sf while} loop, out of at most $|\Inf{e}|$-many edges of $T$, we add one edge to $\mathcal{M}$. Since $|\Inf{\emin}|\le \Inf{T}$ and $T$ has $n-1$ edges, we conclude that $|\mathcal{M}|\ge\lceil\frac{n-1}{\Inf{T}}\rceil$.
\end{proof}

\paragraph{Remark}
Let $T$ be the minimum spanning tree computed by Algorithm~\ref{sm:alg1}. Let $e=(u,v)$ be an edge in $T$. Recall that $T(e^+)$ contains all the edges of $T$ whose weight is at least $w(e)$. We define the {\em degree} of $e$ as $\dg{e}=\dg{u}+\dg{v}-1$, where $\dg{u}$ and $\dg{v}$ are the number of edges incident on $u$ and $v$ in $T(e^+)$, respectively. Note that all the edges incident on $u$ or $v$ in $T(e^+)$ are in the influence set of $e$. Thus, $|\Inf{e}|\ge \dg{e}$, and consequently $\Inf{T}\ge \dg{e}$.

\section{Strong Matching in $\G{\ddiscs}{P}$}
\label{sm:Gabriel-section}
In this section we consider the case where $S$ is a diametral-disk $\ddisc$. Recall that $\G{\ddiscs}{P}$ is an edge-weighted geometric graph, where the weight of an edge $(p,q)$ is equal to the area of $D(p,q)$. $\G{\ddiscs}{P}$ is equal to the Gabriel graph, $GG(P)$. We prove that $\G{\ddiscs}{P}$, and consequently $GG(P)$, has a strong diametral-disk matching of size at least $\lceil\frac{n-1}{17}\rceil$. 

We run Algorithm~\ref{sm:alg1} on $\G{\ddiscs}{P}$ to compute a matching $\mathcal{M}$. By Theorem~\ref{sm:GS-thr}, $\mathcal{M}$ is a strong matching of size at least $\lceil\frac{n-1}{\Inf{T}}\rceil$, where $T$ is a minimum spanning tree in $\G{\ddiscs}{P}$. By Lemma~\ref{sm:mst-in-GS}, $T$ is a minimum spanning tree of the complete graph $K_{\ddiscs}(P)$. Observe that $T$ is a Euclidean minimum spanning tree for $P$ as well. In order to prove the desired lower bound, we show that $\Inf{T}\le 17$. Since $\Inf{T}$ is the maximum size of a set among the
influence sets of edges in $T$, it suffices to show that for every edge $e$ in $T$, the influence set of $e$ contains at most 17 edges. 
\begin{lemma}
\label{sm:disk-inf-lemma}
Let $T$ be a minimum spanning tree of $\G{\ddiscs}{P}$, and let $e$ be any edge in $T$. Then, $|\emph{Inf}(e)|\le 17$.
\end{lemma}
We will prove this lemma in the rest of this section. Recall that, for each two points $p,q\in P$, $D(p,q)$ is the closed diametral-disk with diameter $pq$. Let $\mathcal{D}$ denote the set of diametral-disks representing the edges in $T$. Since $T$ is a subgraph of $\G{\ddiscs}{P}$, we have the following observation:

\begin{observation}
\label{sm:no-point-in-circle-obs}
 Each disk in $\mathcal{D}$ does not contain any point of $P$ in its interior.
\end{observation}

We have proved the following lemma in Chapter~\ref{ch:gg}. 
\begin{lemma}
\label{sm:center-in-lemma}
 For each pair $D_i$ and $D_j$ of disks in $\mathcal{D}$, $D_i$ (resp. $D_j$) does not contain the center of $D_j$ (resp $D_i$).
\end{lemma}

Let $e=(u,v)$ be an edge in $T$. Without loss of generality, we suppose that $D(u,v)$ has radius 1 and centered at the origin $o=(0,0)$ such that $u=(-1,0)$ and $v=(1,0)$. For any point $p$ in the plane, let $\|p\|$ denote the distance of $p$ from $o$. Let $\mathcal{D}(e^+)$ be the disks in $\mathcal{D}$ representing the edges of $T(e^+)$. Recall that $T(e^+)$ contains the edges of $T$ whose weight is at least $w(e)$, where $w(e)$ is equal to the area of $\cmin$. Since the area of any circle is directly related to its radius, we have the following observation:

\begin{observation}
 \label{sm:radius-one}
The disks in $\mathcal{D}(e^+)$ have radius at least $1$.
\end{observation}

Let $C(x,r)$ (resp. $D(x,r)$) be the circle (resp. closed disk) of radius $r$ which is centered at a point $x$ in the plane. 
Let $\mathcal{I}(e^+)=\{D_1,\dots, D_k\}$ be the set of disks in $\mathcal{D}(e^+)\setminus\{D(u,v)\}$ intersecting $D(u,v)$. We show that $\mathcal{I}(e^+)$ contains at most sixteen disks, i.e., $k\le 16$.

For $i\in\{1,\dots,k\}$, let $c_i$ denote the center of the disk $D_i$. 
In addition, let $c'_i$ be the intersection point between $C(o,2)$ and the ray with origin at $o$ which passing through $c_i$. Let the point $p_i$ be $c_i$, if $\|c_i\|< 2$, and $c'_i$, otherwise. See Figure~\ref{sm:distance-fig}. Finally, let $P'=\{o, u, v, p_1,\dots,p_k\}$. 

\begin{observation}
\label{sm:obs}
Let $c_j$ be the center of a disk $D_j$ in $\mathcal{I}(e^+)$, where $\|c_j\|\ge 2$. Then, the disk $D(c_j, \|c_j\|-1)$ is contained in the disk $D_j$. Moreover, the disk $D(p_j,1)$ is contained in the disk $D(c_j, \|c_j\|-1)$. See Figure~\ref{sm:distance-fig}.
\end{observation}

\begin{figure}[htb]
  \centering
  \includegraphics[width=.5\columnwidth]{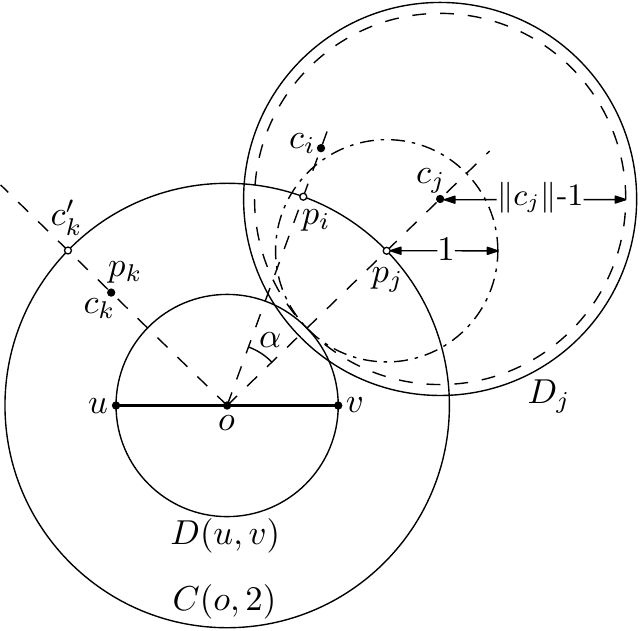}
 \caption{Proof of Lemma~\ref{sm:distance-lemma}; $p_i=c'_i$, $p_j=c'_j$, and $p_k=c_k$.}
  \label{sm:distance-fig}
\end{figure}
The main idea in the proof of the following lemma is similar to the main idea in the proof of Lemma~\ref{distance-lemma} in Chapter~\ref{ch:gg}. Since the number of points and the constants involved in the calculations are different, for the sake of completeness we present a full proof for the following lemma. 
\begin{lemma}
\label{sm:distance-lemma}
The distance between any pair of points in $P'$ is at least 1.
\end{lemma}
\begin{proof}
Let $x$ and $y$ be two points in $P'$. We are going to prove that $|xy|\ge 1$. We distinguish between the following three cases. 
\begin{itemize}
 \item $x,y \in\{o,u,v\}$. In this case the claim is trivial.
\item $x\in \{o,u,v\}, y\in \{p_1,\dots, p_k\}$. If $\|y\|=2$, then $y$ is on $C(o,2)$, and hence $|xy|\ge 1$. If $\|y\|<2$, then $y$ is the center of a disk $D_i$ in $\mathcal{I}(e^+)$. By Observation~\ref{sm:no-point-in-circle-obs}, $D_i$ does not contain $u$ and $v$, and by Lemma~\ref{sm:center-in-lemma}, $D_i$ does not contain $o$. Since $D_i$ has radius at least 1, we conclude that $|xy|\ge 1$.

\item $x,y\in\{p_1,\dots,p_k\}$. Without loss of generality assume $x=p_i$ and $y= p_j$, where $1\le i<j\le k$. We differentiate between three cases:
\begin{itemize}
 \item $\|p_i\|< 2$ and $\|p_j\|<2$. In this case $p_i$ and $p_j$ are the centers of $D_i$ and $D_j$, respectively. By Lemma~\ref{sm:center-in-lemma} and Observation~\ref{sm:radius-one}, we conclude that $|p_ip_j|\ge 1$.
\item $\|p_i\|< 2$ and $\|p_j\|=2$. By Observation~\ref{sm:obs} the disk $D(p_j, 1)$ is contained in the disk $D_j$. By Lemma~\ref{sm:center-in-lemma}, $p_i$ is not in the interior of $D_j$, and consequently, it is not in the interior of $D(p_j,1)$. Therefore, $|p_ip_j|\ge 1$.
\item $\|p_i\|= 2$ and $\|p_j\|=2$. Recall that $c_i$ and $c_j$ are the centers of $D_i$ and $D_j$, such that $\|c_i\|\ge 2$ and $\|c_j\|\ge2$. Without loss of generality assume $\|c_i\|\le \|c_j\|$. For the sake of contradiction assume that $|p_ip_j|<1$. Then, for the angle $\alpha=\angle c_i o c_j$ we have $\sin(\alpha/2)< \frac{1}{4}$. Then, $\cos(\alpha)> 1-2\sin^2(\alpha/2)=\frac{7}{8}$. By the law of cosines in the triangle $\bigtriangleup c_ioc_j$, we have
\begin{equation}
\label{sm:ineq1}
|c_ic_j|^2<\|c_i\|^2+\|c_j\|^2-\frac{14}{8}\|c_i\|\|c_j\|.
\end{equation}
By Observation~\ref{sm:obs} the disk $D(c_j,\|c_j\|-1)$ is contained in $D_j$; see Figure~\ref{sm:distance-fig}. By Lemma~\ref{sm:center-in-lemma}, $c_i$ is not in the interior of $D_j$, and consequently, is not in the interior of $D(c_j,\|c_j\|-1)$. Thus, $|c_ic_j|\ge \|c_j\|-1$. In combination with Inequality~(\ref{sm:ineq1}), this gives
\begin{equation}
 \label{sm:ineq2}
\|c_j\|\left(\frac{14}{8}\|c_i\|-2\right) < \|c_i\|^2-1.
\end{equation}
In combination with the assumption that $\|c_i\| \le \|c_j\|$, Inequality~(\ref{sm:ineq2}) gives
$$\frac{6}{8}\|c_i\|^2-2\|c_i\|+1<0.$$

To satisfy this inequality, we should have $\|c_i\|<2$, contradicting the fact
that $\|c_i\| \ge 2$. This completes the proof.
\end{itemize}
\end{itemize}
\end{proof}

By Lemma~\ref{sm:distance-lemma}, the points in $P'$ has mutual distance 1. Moreover, the points in $P'$ lie in (including the boundary) $C(o,2)$.
Bateman and Erd\H{o}s~\cite{Bateman1951} proved that it is impossible to have 20 points in (including the boundary) a circle of radius 2 such that one of the points is at the center and all of the mutual distances are at least 1.
Therefore, $P'$ contains at most $19$ points, including $o$, $u$, and $v$. This implies that $k\le 16$, and hence $\mathcal{I}(e^+)$ contains at most sixteen edges. This completes the proof of Lemma~\ref{sm:disk-inf-lemma}.

\begin{theorem}
 \label{sm:Gabriel-thr}
Algorithm~\ref{sm:alg1} computes a strong matching of size at least $\lceil\frac{n-1}{17}\rceil$ in $\G{\ominus}{P}$.
\end{theorem}

\section{Strong Matching in $\G{\trids}{P}$}
\label{sm:half-theta-six-section}
In this section we consider the case where $S$ is a downward equilateral triangle $\trid$, whose barycenter is the origin and one of its vertices is on the negative $y$-axis. In this section we assume that $P$ is in general position, i.e., for each point $p\in P$, there is no point of $P\setminus \{p\}$ on $l_p^0$, $l_p^{60}$, and $l_p^{120}$. In combination with Observation~\ref{sm:shrink-triangle-obs}, this implies that for two points $p,q\in P$, no point of $P\setminus\{p,q\}$ are on the boundary of $t(p,q)$ (resp. $t'(p,q)$). Recall that $t(p,q)$ is the smallest homothet of $\trid$ having of $p$ and $q$ on a corner and the other point on the side opposite to that corner. We prove that $\G{\trids}{P}$, and consequently $\frac{1}{2}\Theta_6(P)$, has a strong triangle matching of size at least $\lceil\frac{n-1}{9}\rceil$. 

We run Algorithm~\ref{sm:alg1} on $\G{\trids}{P}$ to compute a matching $\mathcal{M}$. Recall that $\G{\trids}{P}$ is an edge-weighted graph with the weight of each edge $(p,q)$ is equal to the area of $t(p,q)$. By Theorem~\ref{sm:GS-thr}, $\mathcal{M}$ is a strong matching of size at least $\lceil\frac{n-1}{\Inf{T}}\rceil$, where $T$ is a minimum spanning tree in $\G{\trids}{P}$. In order to prove the desired lower bound, we show that $\Inf{T}\le 9$. Since $\Inf{T}$ is the maximum size of a set among the
influence sets of edges in $T$, it suffices to show that for every edge $e$ in $T$, the influence set of $e$ has at most nine edges. 
\begin{lemma}
\label{sm:triangle-inf-lemma}
Let $T$ be a minimum spanning tree of $\G{\trids}{P}$, and let $e$ be any edge in $T$. Then, $|\emph{Inf}(e)|\le 9$.
\end{lemma}

 \begin{figure}[htb]
  \centering
\setlength{\tabcolsep}{0in}
  $\begin{tabular}{ccc}
\multicolumn{1}{m{.33\columnwidth}}{\centering\includegraphics[width=.25\columnwidth]{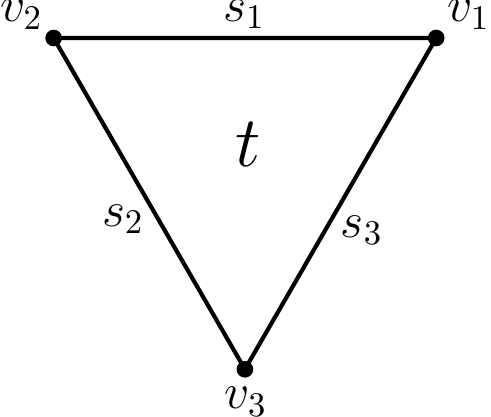}}
&\multicolumn{1}{m{.33\columnwidth}}{\centering\includegraphics[width=.25\columnwidth]{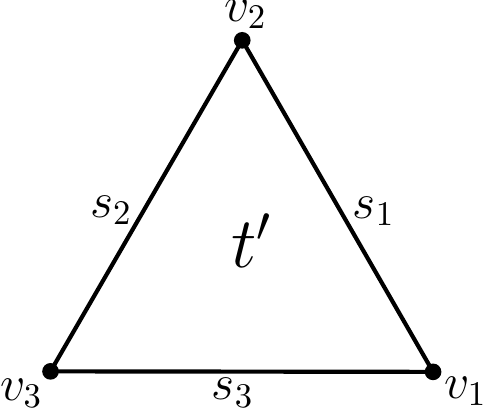}} 
&\multicolumn{1}{m{.33\columnwidth}}{\centering\includegraphics[width=.27\columnwidth]{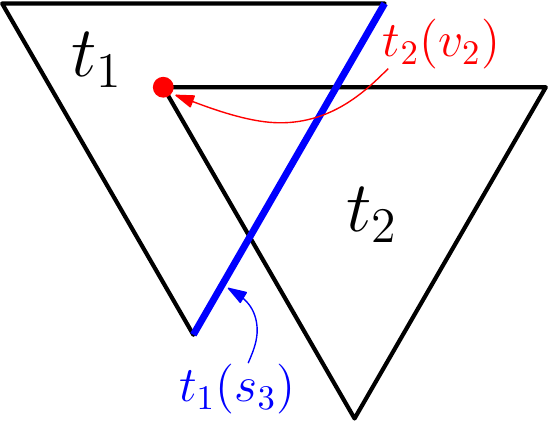}}\\
(a)&(b)&(c)
\end{tabular}$
  \caption{(a) Labeling the vertices and the sides of a downward triangle. (b) Labeling the vertices and the sides of an upward triangle. (c) Two intersecting triangles.}
  \label{sm:triangle-fig}
\end{figure}

We will prove this lemma in the rest of this section. We label the vertices and the sides of a downward equilateral-triangle, $t$, and an upward equilateral-triangle, $t'$, as depicted in Figures~\ref{sm:triangle-fig}(a) and ~\ref{sm:triangle-fig}(b). We refer to a vertex $v_i$ and a side $s_i$ of a triangle $t$ by $t(v_i)$ and $t(s_i)$, respectively.

Recall that $F$ is a subgraph of the minimum spanning tree $T$ in $\G{\trids}{P}$. In each iteration of the {\sf while} loop in Algorithm~\ref{sm:alg1}, let $\mathcal{T}$ denote the set of triangles representing the edges in $F$. By Lemma~\ref{sm:mst-in-GS} and the general position assumption we have

\begin{observation}
\label{sm:no-point-in-triangle-obs}
Each triangle $t(p,q)$ in $\mathcal{T}$ does not contain any point of $P\setminus \{p,q\}$ in its interior or on its boundary.
\end{observation}

Consider two intersecting triangles $t_1(p_1,q_1)$ and $t_2(p_2,q_2)$ in $\mathcal{T}$. By Observation~\ref{sm:shrink-triangle-obs}, each side of $t_1$ contains either $p_1$ or $q_1$, and each side of $t_2$ contains either $p_2$ or $q_2$. Thus, by Observation~\ref{sm:no-point-in-triangle-obs}, we argue that no side of $t_1$ is completely in the interior of $t_2$, and vice versa. Therefore, either exactly one vertex (corner) of $t_1$ is in the interior of $t_2$, or exactly one vertex of $t_2$ is in the interior of $t_1$. Without loss of generality assume that a corner of $t_2$ is in the interior of $t_1$, as shown in Figure~\ref{sm:triangle-fig}(c). In this case we say that $t_1$ intersects $t_2$ through the vertex $t_2(v_2)$, or symmetrically, $t_2$ intersects $t_1$ through the side $t_1(s_3)$.

The following two lemmas have been proved by Biniaz et al.~\cite{Biniaz2015-hotd-CGTA}:
\begin{lemma}[Biniaz et al.~\cite{Biniaz2015-hotd-CGTA}]
\label{sm:triangle3}
Let $t_1$ be a downward triangle which intersects a downward triangle $t_2$ through $\tra{t_2}{s_1}$, and let a horizontal line $\ell$ intersects both $t_1$ and $t_2$. Let $p_1$ and $q_1$ be two points on $t_1(s_2)$ and $t_1(s_3)$, respectively, which are above $t_2(s_1)$. Let $p_2$ and $q_2$ be two points on $t_2(s_2)$ and $t_2(s_3)$, respectively, which are above $\ell$. Then, $\max\{t(p_1,p_2), t(q_1,q_2)\} \allowbreak \prec\allowbreak  \max\{t_1,t_2\}$. See Figure~\ref{sm:triangle-intersection-fig}(b).
\end{lemma}

\begin{lemma}[Biniaz et al.~\cite{Biniaz2015-hotd-CGTA}]
\label{sm:intersection-lemma}
For every four triangles $t_1,t_2,t_3,t_4\in \mathcal{T}$, $t_1\cap t_2\cap t_3\cap t_4 =\emptyset$. 
\end{lemma}

As a consequence of Lemma~\ref{sm:triangle3}, we have the following corollary:
\begin{corollary}
\label{sm:biniaz-cor}
 Let $t_1, t_2, t_3$ be three triangles in $\mathcal{T}$. Then $t_1$, $t_2$, and $t_3$ cannot make a chain configuration, such that $t_2$ intersects $t_3$ through $t_3(s_1)$ and $t_1$ intersects both $t_2$ and $t_3$ through $t_2(s_1)$ and $t_3(s_1)$. See Figure~\ref{sm:triangle-intersection-fig}(b).
\end{corollary}

\begin{figure}[htb]
  \centering
\setlength{\tabcolsep}{0in}
  $\begin{tabular}{cc}
\multicolumn{1}{m{.5\columnwidth}}{\centering\includegraphics[width=.38\columnwidth]{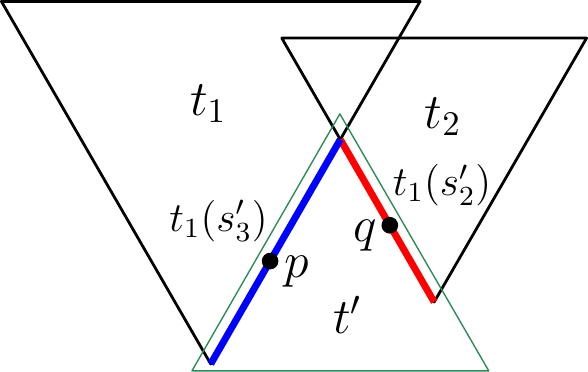}}
&\multicolumn{1}{m{.5\columnwidth}}{\centering\includegraphics[width=.34\columnwidth]{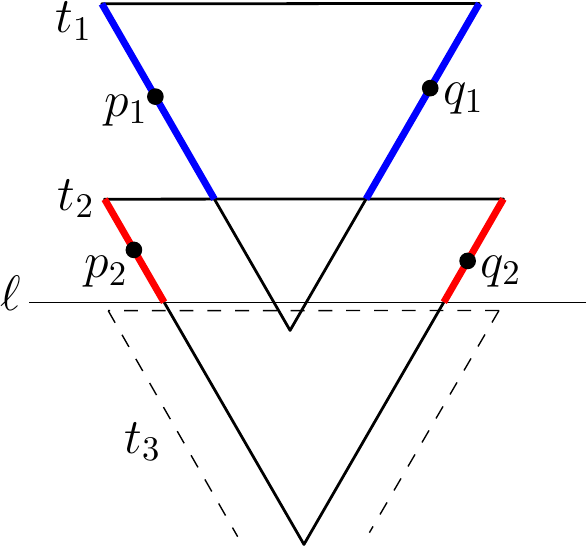}} 
\\(a)&(b)
\end{tabular}$
  \caption{(a) Illustration of Lemma~\ref{sm:triangle-intersection-lemma}. (b) Illustration of Lemma~\ref{sm:triangle3}.}
  \label{sm:triangle-intersection-fig}
\end{figure}

\begin{lemma}
 \label{sm:triangle-intersection-lemma}
Let $t_1$ be a downward triangle which intersects a downward triangle $t_2$ through $t_2(v_2)$. Let $p$ be a point on $t_1(s_3)$ and to the left of $t_2(s_2)$, and let $q$ be a point on $t_2(s_2)$ and to the right of $t_1(s_3)$. Then, $\tr{p,q}\prec\max\{t_1,t_2\}$.
\end{lemma}
\begin{proof}
 Refer to Figure~\ref{sm:triangle-intersection-fig}(a). Let $t_1(s'_3)$ be the part of the line segment $t_1(s_3)$ which is to the left of $t_2(s_2)$, and let $t_2(s'_2)$ be the part of the line segment $t_2(s_2)$ which is to the right of $t_1(s_3)$. Without loss of generality assume that $t_1(s'_3)$ is larger than $t_2(s'_2)$. Let $t'$ be an upward triangle having $t_1(s'_3)$ as its left side. Then, $t'\prec t_1$, which implies that $t'\prec\max\{t_1,t_2\}$. Since $t'$ has both $p$ and $q$ on its boundary, the area of the downward triangle $t(p,q)$ is smaller than the area of $t'$. Therefore, $\tr{p,q}\preceq t'$; which completes the proof.
\end{proof}

Because of the symmetry, the statement of Lemma~\ref{sm:triangle-intersection-lemma} holds even if $p$ is above $t_2(s_1)$ and $q$ is on $t_2(s_1)$.
Consider the six cones with apex at $p$, as shown in Figure~\ref{sm:cones}.
\begin{lemma}
\label{sm:deg-six-half}
Let $T$ be a minimum spanning tree in $\G{\trids}{P}$. Then, in $T$, every point $p$ is adjacent to at most one point in each cone $\cone{i}{p}$, where $1\le i\le 6$.
\end{lemma}
\begin{proof}
If $i$ is even, then by the construction of $\G{\trids}{P}$, which is given in Section~\ref{sm:preliminaries}, $p$ is adjacent to at most one point in $\cone{i}{p}$. Assume $i$ is odd. For the sake of contradiction, assume in $T$, the point $p$ is adjacent to two points $q$ and $r$ in a cone $\cone{i}{p}$. Then, $t(p,q)$ has $q$ on a corner, and $t(p,r)$ has $r$ on a corner. Without loss of generality assume $t(p,r)\prec t(q,r)$. Then, the hexagon $\hex{q}{p}$ has $r$ in its interior. Thus, $\tr{q,r}\prec \tr{p,q}$. Then the cycle $r,p,q,r$ contradicts Lemma~\ref{sm:cycle-lemma}. Therefore, $p$ is adjacent to at most one point in each of the six cones.
\end{proof}

In Algorithm~\ref{sm:alg1}, in each iteration of the {\sf while} loop, let $\mathcal{T}(e^+)$ be the triangles representing the edges of $F$. Recall that $e$ is the smallest edge in $F$, and hence, $t(e)$ is a smallest triangle in $\mathcal{T}(e^+)$.
Let $e=(p,q)$ and let $\mathcal{I}(e^+)$ be the set of triangles in $\mathcal{T}(e^+)$ (excluding $t(e)$) which intersect $t(e)$. We show that $\mathcal{I}(e^+)$ contains at most eight triangles.
We partition the triangles in $\mathcal{I}(e^+)$ into $\{\mathcal{I}_1,\mathcal{I}_2\}$, such that every triangle $\tau\in\mathcal{I}_1$ shares only $p$ or $q$ with $t=t(e)$, i.e., $\mathcal{I}_1=\{\tau: \tau\in\mathcal{I}(e^+),\tau\cap t\in \{p,q\}\}$, and every triangle $\tau\in\mathcal{I}_2$ intersects $t$ either through a side or through corner which is not $p$ nor $q$.

\begin{figure}[htb]
 \begin{center}
\includegraphics[width=.4\textwidth]{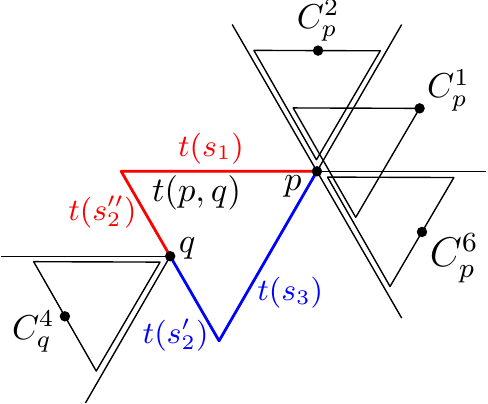}
  \end{center}
  \caption{Illustration of the triangles in $\mathcal{I}_1$.}
\label{sm:cones-pq}
\end{figure}

By Observation~\ref{sm:shrink-triangle-obs}, for each triangle $t(p,q)$, one of $p$ and $q$ is on a corner of $t(p,q)$ and the other one is on the side opposite to that corner. Without loss of generality assume that $p$ is on the corner $t(v_1)$, and hence, $q$ is on the side $t(s_2)$. See Figure~\ref{sm:cones-pq}. Note that the other cases, where $p$ is on $t(v_2)$ or on $t(v_3)$ are similar.
Since the intersection of $t$ with any triangle $\tau\in\mathcal{I}_1$ is either $p$ or $q$, $\tau$ has either $p$ or $q$ on its boundary. In combination with Observation~\ref{sm:no-point-in-triangle-obs}, this implies that $\tau$ represent an edge $e'$ in $T$, and hence, either $p$ or $q$ is an endpoint of $e'$. As illustrated in Figure~\ref{sm:cones-pq}, the other endpoint of $e'$ can be either in $\cone{1}{p}$, $\cone{2}{p}$, $\cone{6}{p}$, or in $\cone{4}{q}$, because otherwise $\tau\cap t\notin \{p,q\}$. By Lemma~\ref{sm:deg-six-half}, $p$ has at most one neighbor in each of $\cone{1}{p}$, $\cone{2}{p}$, $\cone{6}{p}$, and $q$ has at most one neighbor in $\cone{4}{q}$. Therefore, $\mathcal{I}_1$ contains at most four triangles. We are going to show that $\mathcal{I}_2$ also contains at most four triangles. 

The point $q$ divides $t(s_2)$ into two parts. Let $\tra{t}{s'_2}$ and $\tra{t}{s''_2}$ be the parts of $\tra{t}{s_2}$ which are below and above $q$, respectively; see Figure~\ref{sm:cones-pq}. The triangles in $\mathcal{I}_2$ intersect $t$ either through $t(s_1)\cup t(s''_2)$ or through $t(s_3)\cup t(s'_2)$; which are shown by red and blue polylines in Figure~\ref{sm:cones-pq}. We show that most two triangles in $\mathcal{I}_2$ intersect $t$ through each of $t(s_1)\cup t(s''_2)$ or $t(s_3)\cup t(s'_2)$. Because of symmetry, we only prove for $t(s_3)\cup t(s'_2)$. When a triangle $t'$ intersects $t$ through both $t(s_3)$ and $t(s'_2)$ we say $t'$ intersects $t$ through $t(v_3)$. In the next lemma, we prove that at most one triangle in $\mathcal{I}_2$ intersects $t$ through each of $\tra{t}{s_3}$, $\tra{t}{s'_2}$. Again, because of symmetry, we only prove for $\tra{t}{s_3}$. 

\begin{figure}[htb]
  \centering
\setlength{\tabcolsep}{0in}
  $\begin{tabular}{cc}
\multicolumn{1}{m{.5\columnwidth}}{\centering\includegraphics[width=.34\columnwidth]{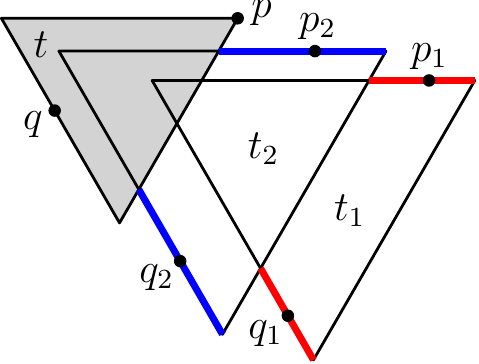}}
&\multicolumn{1}{m{.5\columnwidth}}{\centering\includegraphics[width=.34\columnwidth]{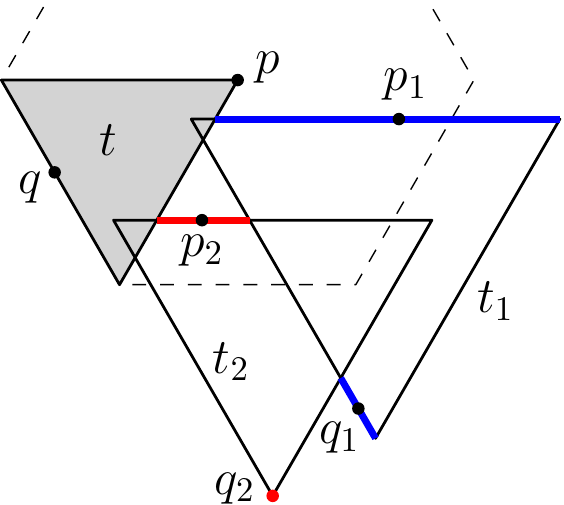}} \\
(a) & (b)
\end{tabular}$
  \caption{Illustration of Lemma~\ref{sm:side-intersection}: (a) $t_1(v_2)\in t_2$. (b) $t_1(v_2)\notin t_2$ and $t_2(v_2)\notin t_1$.}
  \label{sm:two-triangles-fig}
\end{figure}

\begin{lemma}
\label{sm:side-intersection}
At most one triangle in $\mathcal{I}_2$ intersects $\trmin$ through $\tra{\trmin}{s_3}$.
\end{lemma}
\begin{proof}
The proof is by contradiction. Assume two triangles $t_1(p_1,q_1)$ and $t_2(p_2,q_2)$ in $\mathcal{I}_2$ intersect $\trmin$ through $\tra{\trmin}{s_3}$. Without loss of generality assume that $p_i$ is on $t_i(s_1)$ and $q_i$ is on $t_i(s_2)$ for $i=1,2$. Recall that the area of $t_1$ and the area of $t_2$ are at least the area of $\trmin$. If $\tra{t_1}{v_2}$ is in the interior of $t_2$ (as shown in Figure~\ref{sm:two-triangles-fig}(a)) or $\tra{t_2}{v_2}$ is in the interior of $t_1$, then we get a contradiction to Corollary~\ref{sm:biniaz-cor}. Thus, assume that $\tra{t_1}{v_2}\notin t_2$ and $\tra{t_2}{v_2} \notin t_1$. 

Without loss of generality assume that $\tra{t_1}{s_1}$ is above $\tra{t_2}{s_1}$; see Figure~\ref{sm:two-triangles-fig}(b). By Lemma~\ref{sm:triangle-intersection-lemma}, we have $t(p,p_1)\prec \max\{t,t_1\}\preceq t_1$. If $q_1$ is in $\hex{p}{q}$, then by Observation~\ref{sm:obs2}, $t(p,q_1)\prec t$. Then, the cycle $p,p_1,q_1,p$ contradicts Lemma~\ref{sm:cycle-lemma}. Thus, assume that $q_1\notin \hex{p,q}$. In this case $\tra{t_2}{s_3}$ is to the left of $\tra{t_1}{s_3}$, because otherwise $q_1$ lies in $t_2$ which contradicts Observation~\ref{sm:no-point-in-triangle-obs}.
Since both $t_1$ and $t_2$ are larger than $t$, $t_2$ intersects $t_1$ through $t_1(s_2)$, and hence $\tra{t_2}{v_1}$ is in the interior of $t_1$. This implies that $q_2$ is on $\tra{t_2}{v_3}$. In addition, $p_2$ is on the part of $\tra{t_2}{s_1}$ which lies in the interior of $\hex{p}{q}$. By Observation~\ref{sm:obs2} and Lemma~\ref{sm:triangle-intersection-lemma}, we have $t(p,p_2)\prec t$ and $t(q_1,q_2)\prec \max\{t_1,t_2\}$, respectively. Thus, the cycle $p,p_1,q_1,q_2,p_2,p$ contradicts Lemma~\ref{sm:cycle-lemma}. 
\end{proof}

\begin{lemma}
\label{sm:vertex-intersection}
 At most two triangles in $\mathcal{I}_2$ intersect $\trmin$ through $\tra{\trmin}{v_3}$.
\end{lemma}
\begin{proof}
For the sake of contradiction assume three triangles $t_1,t_2, t_3\in \mathcal{I}_2$ intersect $\trmin$ through $\tra{\trmin}{v_3}$. This implies that $\tra{\trmin}{v_3}$ belongs to four triangles $t,t_1,t_2,t_3$, which contradicts Lemma~\ref{sm:intersection-lemma}. \end{proof}

\begin{figure}[htb]
  \centering
\setlength{\tabcolsep}{0in}
  $\begin{tabular}{cc}
\multicolumn{1}{m{.5\columnwidth}}{\centering\includegraphics[width=.36\columnwidth]{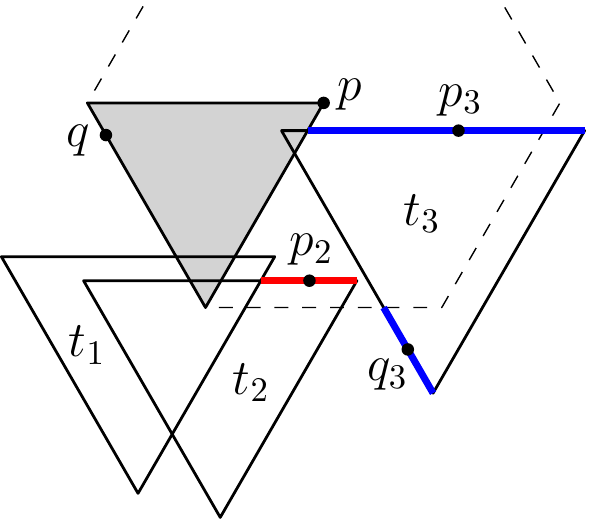}}
&\multicolumn{1}{m{.5\columnwidth}}{\centering\includegraphics[width=.36\columnwidth]{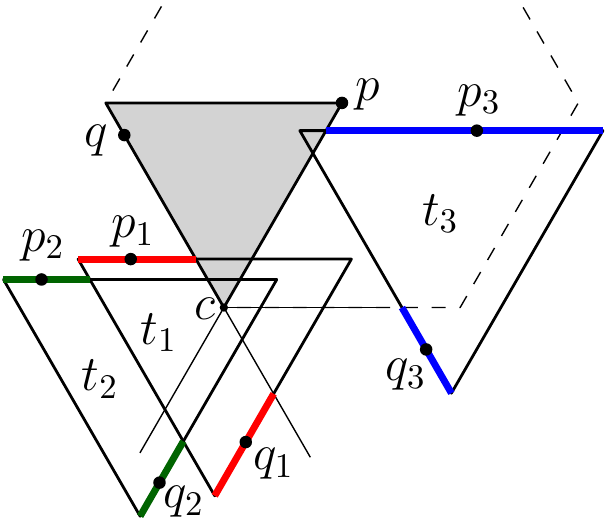}} \\
(a) & (b)\\\multicolumn{1}{m{.5\columnwidth}}{\centering\includegraphics[width=.36\columnwidth]{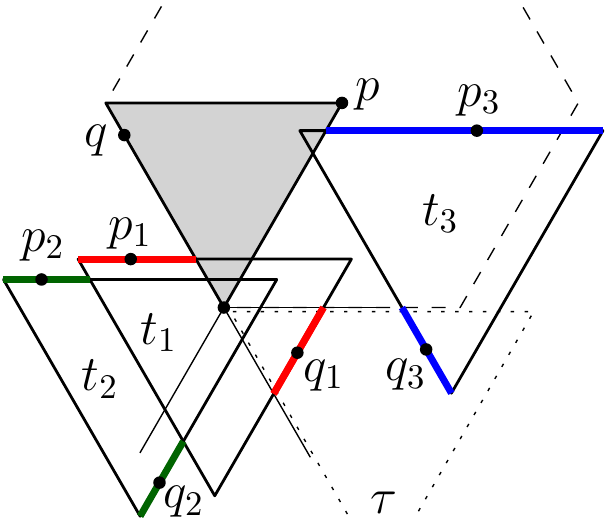}}
&\multicolumn{1}{m{.5\columnwidth}}{\centering\includegraphics[width=.36\columnwidth]{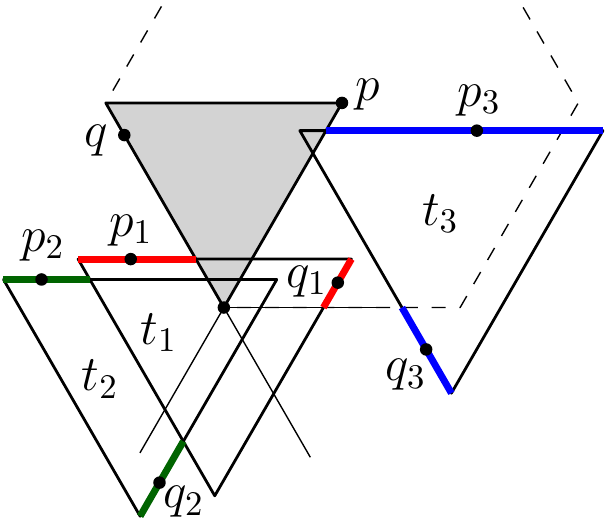}} \\
(c) & (d)
\end{tabular}$
  \caption{Illustration of Lemma~\ref{sm:vertex-side-intersection-2}: (a) $p_2$ is to the right of $t_1(s_3)$, (b) $q_1\in C^5_{\tra{t}{v_3}}$, (c) $q_1\in C^6_{\tra{t}{v_3}}$, and (d) $q_1\in C^1_{\tra{t}{v_3}}$.}
\label{sm:three-triangle-fig2}
\end{figure}
\begin{lemma}
\label{sm:vertex-side-intersection-2}
If two triangles in $\mathcal{I}_2$ intersect $\trmin$ through $\tra{\trmin}{v_3}$, then no other triangle in $\mathcal{I}_2$ intersects $\trmin$ through $\tra{\trmin}{s_3}$ or through $\tra{\trmin}{s'_2}$. 
\end{lemma}

\begin{proof}
The proof is by contradiction. Assume two triangles $t_1(p_1,q_1)$ and $t_2(p_2,q_2)$ in $\mathcal{I}_2$ intersect $\trmin$ through $\tra{\trmin}{v_3}$, and a triangle $t_3(p_3,q_3)$ in $\mathcal{I}_2$ intersects $\trmin$ through $\tra{\trmin}{s_3}$ or $\tra{\trmin}{s'_2}$. Let $p_i$ be the point which lies on $t_i(s_1)$ for $i=1,2,3$. By Lemma~\ref{sm:vertex-intersection}, $t_3$ cannot intersect both $\tra{\trmin}{s_3}$ and $\tra{\trmin}{s'_2}$. Thus, $t_3$ intersects $t$ either through $\tra{\trmin}{s_3}$ or through $\tra{\trmin}{s'_2}$. We prove the former case; the proof for the latter case is similar. Assume that $t_3$ intersects $t$ through $\tra{\trmin}{s_3}$. By Lemma~\ref{sm:triangle-intersection-lemma}, $t(p,p_3)\prec t_3$. See Figure~\ref{sm:three-triangle-fig2}. In addition, both $t_1(s_3)$ and $t_2(s_3)$ are to the left of $t_3(s_3)$, because otherwise $q_3$ lies in $t_1\cup t_2\cup \hex{p}{q}$. If $q_3\in t_1\cup t_2$ we get a contradiction to Observation~\ref{sm:no-point-in-triangle-obs}. If $q_3\in \hex{p}{q}$ then by Observation~\ref{sm:obs2}, we have $t(p,q_3)\prec t$, and hence, the cycle $p,p_3,q_3,p$ contradicts Lemma~\ref{sm:cycle-lemma}.

Without loss of generality assume that $\tra{t_1}{s_1}$ is above $\tra{t_2}{s_1}$; see Figure~\ref{sm:three-triangle-fig2}. If $\tra{t_1}{v_3}$ is in $t_2$ or $\tra{t_2}{v_3}$ is in $t_1$, then we get a contradiction to Corollary~\ref{sm:biniaz-cor}. Thus, assume that $\tra{t_1}{v_3}\notin t_2$ and $\tra{t_2}{v_3} \notin t_1$. This implies that either (i) $\tra{t_2}{s_3}$ is to the right of $\tra{t_1}{s_3}$ or (ii) $\tra{t_2}{s_2}$ is to the left of $\tra{t_1}{s_2}$. We show that both cases lead to a contradiction.

In case (i), $p_2$ lies in the interior of $\hex{p,q}$, and then by Observation~\ref{sm:obs2}, we have $t(p,p_2)\prec t$; see Figure~\ref{sm:three-triangle-fig2}(a). In addition, Lemma \ref{sm:triangle-intersection-lemma} implies that $t(p_2,q_3)\prec \max\{t,t_3\}\preceq t_3$. Thus, the cycle $p,p_3,\allowbreak q_3,\allowbreak p_2,p$ contradicts Lemma~\ref{sm:cycle-lemma}.

Now consider case (ii) where $\tra{t_1}{s_1}$ is above $\tra{t_2}{s_1}$ and $\tra{t_2}{s_2}$ is to the left of $\tra{t_1}{s_2}$. If $p_1$ is to the right of $t$, then as in case (i), the cycle $p,p_3,\allowbreak q_3,\allowbreak p_1,p$ contradicts Lemma~\ref{sm:cycle-lemma}. Thus, assume that $p_1$ is to the left of $t$, as shown in Figure~\ref{sm:three-triangle-fig2}(b). By Lemma~\ref{sm:triangle-intersection-lemma}, we have $t(q,p_1)\prec \max\{t,t_1\}\preceq t_1$. Each side of $t_1$ contains either $p_1$ or $q_1$, while $p_1$ is on the part of $t_1(s_1)$ which is to the left of $t$, thus, $q_1$ is on $\tra{t_1}{s_3}$. Consider the six cones around $\tra{t}{v_3}$; see Figure~\ref{sm:three-triangle-fig2}(b). We have three cases: (a) $q_1\in C^5_{\tra{t}{v_3}}$, (b) $q_1\in C^6_{\tra{t}{v_3}}$ or (c) $q_1\in C^1_{\tra{t}{v_3}}$. 

In case (a), which is shown in Figure~\ref{sm:three-triangle-fig2}(b), by Lemma~\ref{sm:triangle3}, we have $\max\{t(p_1,p_2),t(q_1,q_2)\}\prec\max\{t_1,t_2\}$. Thus, the cycle $p_1,p_2,q_2,\allowbreak q_1,\allowbreak p_1$ contradicts Lemma~\ref{sm:cycle-lemma}. In Case (b), which is shown in Figure~\ref{sm:three-triangle-fig2}(c), we have $t(q_1,q_3)\prec t_3$, because if we map $t_3$ to a downward triangle $\tau$\textemdash of area equal to the area of $t_3$\textemdash which has $\tau(v_2)$ on $t(v_3)$, then $\tau$ contains both $q_1$ and $q_3$. Therefore, the cycle $p,p_3,q_3,q_1,\allowbreak p_1,\allowbreak q,p$ contradicts Lemma~\ref{sm:cycle-lemma}. In Case (c), which is shown in Figure~\ref{sm:three-triangle-fig2}(d), by Observation~\ref{sm:obs2}, $t(p,q_1)\prec t$, and then, the cycle $p,q_1,p_1,q,p$ contradicts Lemma~\ref{sm:cycle-lemma}.
\end{proof}

\begin{figure}[htb]
  \centering
\setlength{\tabcolsep}{0in}
  $\begin{tabular}{cc}
\multicolumn{1}{m{.5\columnwidth}}{\centering\includegraphics[width=.43\columnwidth]{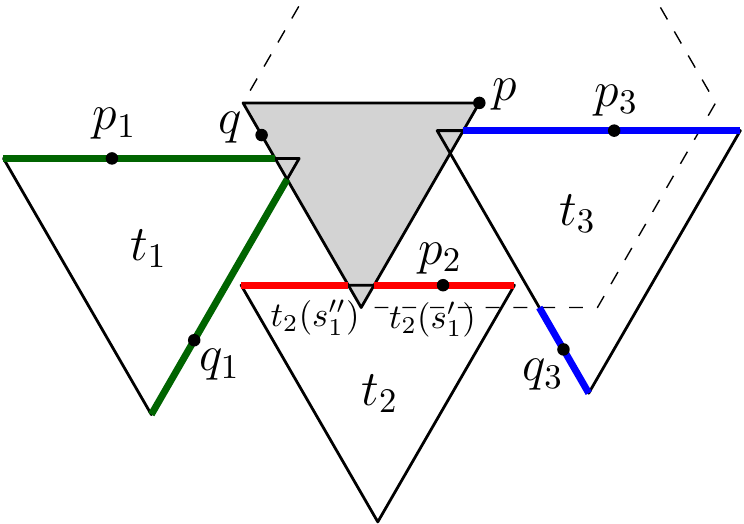}}
&\multicolumn{1}{m{.5\columnwidth}}{\centering\includegraphics[width=.43\columnwidth]{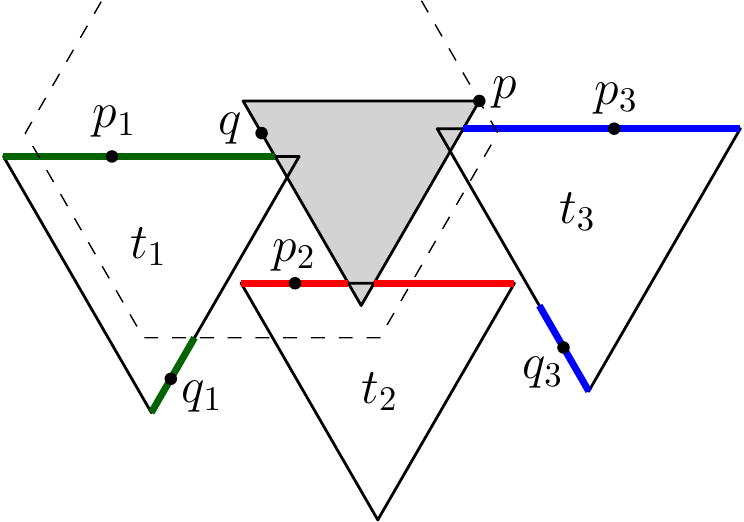}} \\
(a) & (b)
\end{tabular}$
  \caption{Illustration of Lemma~\ref{sm:vertex-side-intersection-1}: (a) $p_2\in t_2(s'_1)$, and (b) $p_2\in t_2(s''_1)$.}
\label{sm:three-triangle-fig}
\end{figure}

\begin{lemma}
\label{sm:vertex-side-intersection-1}
If three triangles intersect $\trmin$ through $\tra{\trmin}{s'_2}, \tra{\trmin}{v_3}$ and $\tra{\trmin}{s_3}$. Then, at least one of the three triangles is not in $\mathcal{I}_2$. 
\end{lemma}
\begin{proof}
The proof is by contradiction. Assume that three triangles $t_1(p_1,\allowbreak q_1),\allowbreak  t_2(p_2,q_2),\allowbreak t_3(p_3,q_3)$ in $\mathcal{I}_2$ intersect $\trmin$ through $\tra{\trmin}{s'_2}, \tra{\trmin}{v_3}, \tra{\trmin}{s_3}$, respectively. Let $p_i$ be the point which lies on $t_i(s_1)$ for $i=1,2,3$. See Figure~\ref{sm:three-triangle-fig}(a). By Lemma \ref{sm:triangle-intersection-lemma}, we have $t(p,p_3)\prec t_3$ and $t(q,p_1)\prec t_1$. If $q_3$ is in the interior of $\hex{p}{q}$, then by Observation~\ref{sm:obs2}, $t(p,q_3)\prec t$, and hence, the cycle $p,p_3,q_3,p$ contradicts Lemma~\ref{sm:cycle-lemma}. If $q_1$ is in $\hex{q}{p}$, then by Observation~\ref{sm:obs2}, $t(q,q_1)\prec t$, and hence, the cycle $q,q_1,p_1,q$ contradicts Lemma~\ref{sm:cycle-lemma}; see Figure~\ref{sm:three-triangle-fig}(b). Thus, assume that $q_3\notin \hex{p}{q}$ and $q_1\notin \hex{q}{p}$. Let $\tra{t_2}{s'_1}$ and $\tra{t_2}{s''_1}$ be the parts of $\tra{t_2}{s_1}$ which are to the right of $t(s_3)$ and to the left of $t(s_2)$, respectively. Consider the point $p_2$ which lies on $\tra{t_2}{s_1}$. 
If $p_2\in\tra{t_2}{s'_1}$, then $p_2\in \hex{p}{q}$ and by Observation~\ref{sm:obs2}, $t(p,p_2)\prec t$. In addition, Lemma ~\ref{sm:triangle-intersection-lemma} implies that $t(p_2,q_3)\prec t_3$. Thus, the cycle $p,p_3,q_3,p_2,p$ contradicts Lemma~\ref{sm:cycle-lemma}; see Figure~\ref{sm:three-triangle-fig}(a).
If $p_2\in\tra{t_2}{s''_1}$, then $p_2\in \hex{q}{p}$ and by Observation~\ref{sm:obs2}, $t(q,p_2)\prec t$. In addition, Lemma ~\ref{sm:triangle-intersection-lemma} implies that $t(p_2,q_1)\prec t_2$. Thus, the cycle $q,p_2,q_1,p_1,q$ contradicts Lemma~\ref{sm:cycle-lemma}; see Figure~\ref{sm:three-triangle-fig}(b).
\end{proof}

Putting Lemmas~\ref{sm:side-intersection}, \ref{sm:vertex-intersection}, \ref{sm:vertex-side-intersection-2}, and \ref{sm:vertex-side-intersection-1} together, implies that at most two triangles in $\mathcal{I}_2$ intersect $t$ through $t(s_3)\cup t(s'_2)$,  and consequently, at most two triangles in $\mathcal{I}_2$ intersect $t$ through $t(s_1)\cup t(s''_2)$. Thus, $\mathcal{I}_2$ contains at most four triangles. Recall that $\mathcal{I}_1$ contains at most four triangles. Then, $\mathcal{I}(e^+)$ has at most eight triangles. Therefore, the influence set of $e$, contains at most 9 edges (including $e$ itself). This completes the proof of Lemma~\ref{sm:triangle-inf-lemma}. 

\begin{theorem}
\label{sm:half-theta-six-thr}
Algorithm~\ref{sm:alg1} computes a strong matching of size at least $\lceil\frac{n-1}{9}\rceil$ in $\G{\trids}{P}$.
\end{theorem}

\begin{figure}[htb]
  \centering
\includegraphics[width=.43\columnwidth]{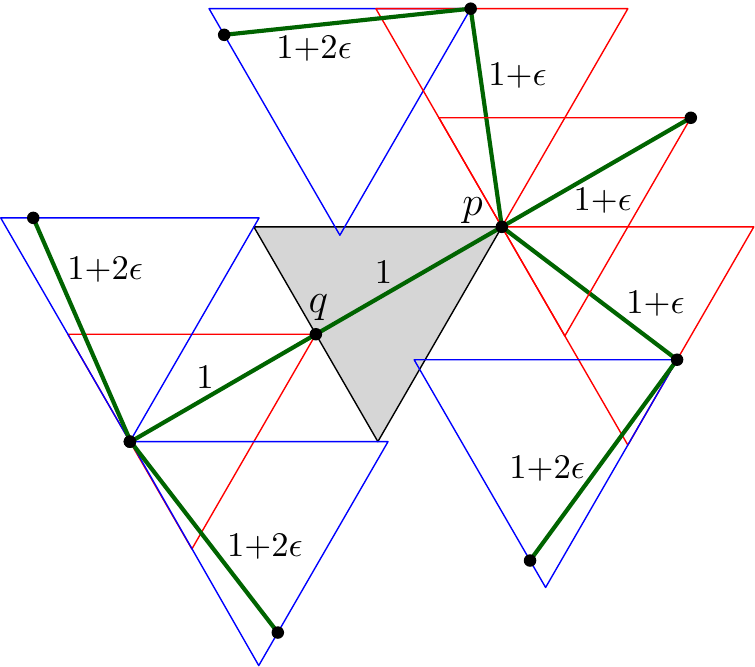}
  \caption{Four triangles in $\mathcal{I}_1$ (in red) and four triangles in $\mathcal{I}_2$ (in blue) intersect with $t(p,q)$.}
\label{sm:five-fig}
\end{figure}

The bound obtained by Lemma~\ref{sm:triangle-inf-lemma} is tight. Figure~\ref{sm:five-fig} shows a configuration of 10 points in general position such that the influence set of a minimal edge is 9. In Figure~\ref{sm:five-fig}, $t=t(p,q)$ represents a smallest edge of weight 1; the minimum spanning tree is shown in bold-green line segments. The weight of all edges\textemdash the area of the triangles representing these edges\textemdash is at least 1. The red triangles are in $\mathcal{I}_1$ and share either $p$ or $q$ with $t$. The blue triangles are in $\mathcal{I}_2$ and intersect $t$ through $t(s_1)\cup t(s''_2)$ or through $t(s_3)\cup t(s'_2)$; as show in Figure~\ref{sm:five-fig}, two of them share only the points $t(v_2)$ and $t(v_3)$.

\section{Strong Matching in {\normalsize $G_{\bigtriangledown \hspace*{-8pt} \bigtriangleup}(P)$}}
\label{sm:theta-six-section}
In this section we consider the problem of computing a strong matching in $\GUD(P)$. Recall that $\GUD(P)$ is the union of $\G{\trids}{P}$ and $\G{\trius}{P}$, and is equal to the graph $\Theta_6(P)$. We assume that $P$ is in general position, i.e., for each point $p\in P$, there is no point of $P\setminus \{p\}$ on $l_p^0$, $l_p^{60}$, and $l_p^{120}$. A matching $\mathcal{M}$ in $\GUD(P)$ is a strong matching if for each edge $e$ in $\mathcal{M}$ there is a homothet of $\trid$ or a homothet of $\triu$ representing $e$, such that these homothets are pairwise disjoint. See Figure~\ref{sm:strong-example}(b). Using a similar approach as in~\cite{Abrego2009}, we prove the following theorem:

\begin{theorem}
\label{sm:theta-six-thr}
Let $P$ be a set of $n$ points in general position in the plane. Let $S$ be an upward or a downward equilateral-triangle that contains $P$. Then, it is possible to find a strong matching of size at least $\lceil\frac{n-1}{4}\rceil$ for $\GUD(P)$ in $S$.
\end{theorem}

\begin{proof}
The proof is by induction. Assume that any point set of size $n'\le n-1$ in a triangle $S'$, has a strong matching of size $\lceil \frac{n'-1}{4}\rceil$ in $S'$. Without loss of generality, assume $S$ is an upward equilateral-triangle. If $n$ is $0$ or $1$, then there is no matching in $S$, and if $n\in\{2, 3, 4, 5\}$, then by shrinking $S$, it is possible to find a strongly matched pair; the statement of the theorem holds. Suppose that $n\ge 6$, and $n=4m+r$, where $r\in\{0,1,2,3\}$. If $r\in\{0, 1,3\}$, then 
$\lceil \frac{n-1}{4}\rceil = \lceil \frac{(n-1)-1}{4}\rceil$, and by
induction we are done. Suppose that $n=4m+2$, for some $m\ge 1$. We prove that there are $\lceil\frac{n-1}{4}\rceil=m+1$ disjoint equilateral-triangles (upward or downward) in $S$,
each of them matches a pair of points in $P$. Partition $S$ into four equal area equilateral triangles $S_1, S_2, S_3, S_4$ containing $n_1, n_2,n_3, n_4$ points, respectively; see Figure~\ref{sm:Theta-six-fig}(a). Let $n_i=4m_i+r_i$, where $r_i\in\{0,1,2,3\}$. 
By induction, in $S_1\cup S_2\cup S_3\cup S_4$, we have a strong matching of size at least
\begin{equation}
\label{sm:A-eq}
A=\left\lceil\frac{n_1-1}{4}\right\rceil + \left\lceil\frac{n_2-1}{4}\right\rceil +\left\lceil\frac{n_3-1}{4}\right\rceil+\left\lceil \frac{n_4-1}{4}\right\rceil.
\end{equation}

\begin{paragraph}{Claim 1.}$A\ge m$.\end{paragraph}
\begin{proof}
By Equation~(\ref{sm:A-eq}), we have
\begin{align*}
A&= \sum_{i=1}^{4}{\left\lceil\frac{n_i-1}{4}\right\rceil} \ge 
\sum_{i=1}^{4}\frac{n_i-1}{4} =\frac{n}{4} -1=\frac{4m+2}{4} -1=m-\frac{1}{2}.
\end{align*}
Since $A$ is an integer, we argue that $A\ge m$.
\end{proof}

If $A> m$, then we are done. Assume that $A=m$; in fact, by the induction hypothesis we have an strong matching of size $m$ for $P$. In order to complete the proof, we have to get one more strongly matched pair. Let $R$ be the multiset $\{r_1,r_2,r_3,r_4\}$.

\begin{paragraph}{Claim 2.}{\em If $A=m$, then either (i) one element in $R$ is equal to $3$ and the other elements are equal to $1$, or (ii) two elements in $R$ are equal to $0$ and the other elements are equal to $1$.}\end{paragraph}
\begin{proof}
Let $\alpha=r_1+r_2+r_3+r_4$, where $0\le r_i\le 3$. Then $n=4(m_1+m_2+m_3+m_4)+\alpha$. Since $n=4m+2$, $\alpha=4k+2$, for some $0\le k\le 2$. Thus, $n=4(m_1+m_2+m_3+m_4+k)+2$, where $m=m_1+m_2+m_3+m_4+k$.

By induction, in $S_i$, we get a matching of size at least $\lceil \frac{(4m_i+r_i)-1}{4}\rceil=m_i+\lceil \frac{r_i-1}{4}\rceil$. Hence, in $S_1\cup S_2\cup S_3\cup S_4$, we get a matching of size at least

\begin{equation*}
A=m_1+m_2+m_3+m_4+\left\lceil\frac{r_1-1}{4}\right\rceil + \left\lceil\frac{r_2-1}{4}\right\rceil +\left\lceil\frac{r_3-1}{4}\right\rceil+\left\lceil \frac{r_4-1}{4}\right\rceil.
\end{equation*}

Since $A=m$ and $m=m_1+m_2+m_3+m_4+k$, we have 

\begin{equation}
\label{sm:k-eq}
k=\left\lceil\frac{r_1-1}{4}\right\rceil + \left\lceil\frac{r_2-1}{4}\right\rceil +\left\lceil\frac{r_3-1}{4}\right\rceil+\left\lceil \frac{r_4-1}{4}\right\rceil.
\end{equation}

Note that $0\le k\le 2$.
We go through some case analysis: (i) $k=0$, (ii) $k=1$, (iii) $k=2$. In case (i), we have $\alpha =4k+2=r_1+r_2+r_3+r_4=2$. In order to have $k$ equal to 0 in Equation~(\ref{sm:k-eq}), no element in $R$ should be more than 1; this happens only if two elements in $R$ are equal to 0 and the other two elements are equal to 1. In case (ii), we have $\alpha =r_1+r_2+r_3+r_4=6$. In order to have $k$ equal to 1 in Equation~(\ref{sm:k-eq}), at most one element in $R$ should be greater than 1; this happens only if three elements in $R$ are equal to 1 and the other element is equal to 3 (note that all elements in $R$ are smaller than 4). In case (iii), we have $\alpha =r_1+r_2+r_3+r_4=10$. In order to have $k$ equal to 2 in Equation~(\ref{sm:k-eq}), at most two elements in $R$ should be greater than 1; which is not possible.
\end{proof} 
We show how to find one more matched pair in each case of Claim 2.

We define $\SM{$x$}{1}$ as the smallest upward equilateral-triangle contained in $S_1$ and anchored at the top corner of $S_1$, which contains all the points in $S_1$ except $x$ points. If $S_1$ contains less than $x$ points, then the area of $\SM{$x$}{1}$ is zero. We also define $\SP{$x$}{1}$ as the smallest upward equilateral-triangle that contains $S_1$ and anchored at the top corner of $S_1$, which has all the points in $S_1$ plus $x$ other points of $P$. Similarly we define upward triangles $\SM{$x$}{2}$ and $\SP{$x$}{2}$ which are anchored at the left corner of $S_2$. Moreover, we define upward triangles $\SM{$x$}{4}$ and $\SP{$x$}{4}$ which are anchored at the right corner of $S_4$. We define downward triangles $\SM{$x$}{3l}$, $\SM{$x$}{3r}$, $\SM{$x$}{3b}$ which are anchored at the top-left corner, top-right corner, and bottom corner of $S_3$, respectively. See Figure~\ref{sm:Theta-six-fig}(a). 

{Case 1:} {\em One element in $R$ is equal to 3 and the other elements are equal to 1.}

In this case, we have $m=m_1+m_2+m_3+m_4+1$. Because of the symmetry, we have two cases: (i) $r_3=3$, (ii) $r_j=3$ for some $j\in\{1,2,4\}$.

\begin{itemize}
 
\begin{figure}[h!]
  \centering
\setlength{\tabcolsep}{0in}
  $\begin{tabular}{cc}
\multicolumn{1}{m{.5\columnwidth}}{\centering\includegraphics[width=.4\columnwidth]{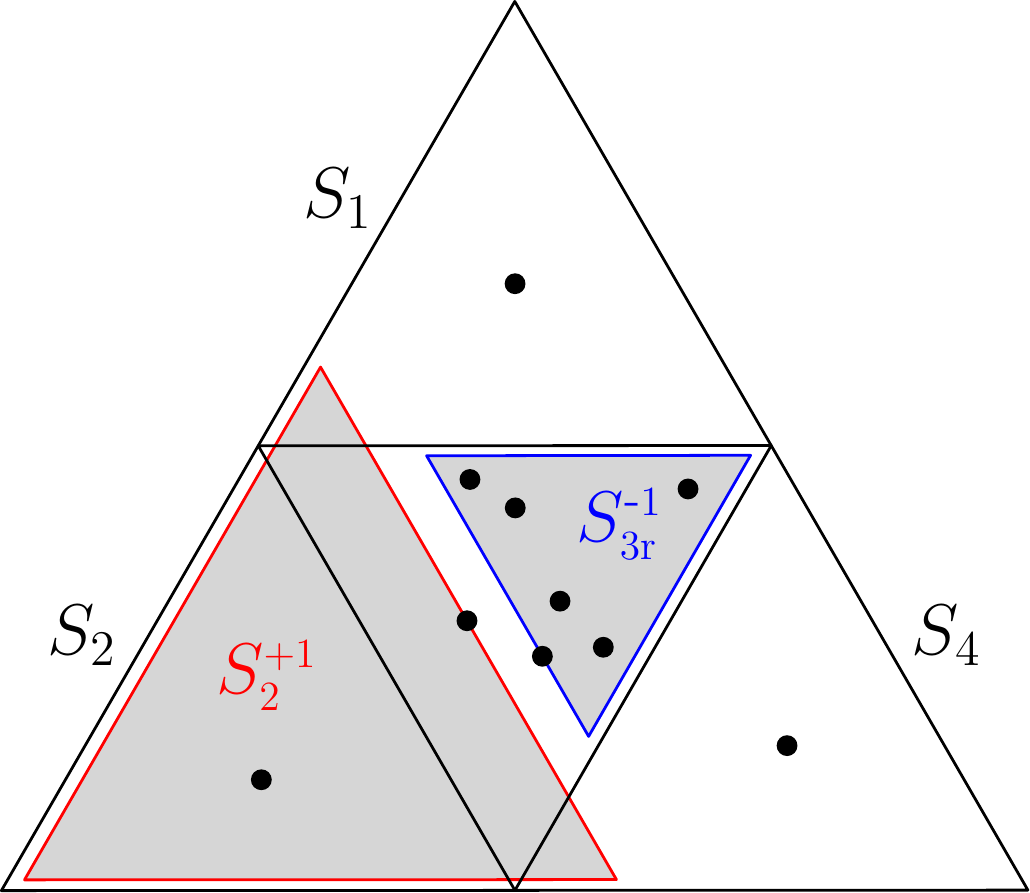}}
&\multicolumn{1}{m{.5\columnwidth}}{\centering\includegraphics[width=.4\columnwidth]{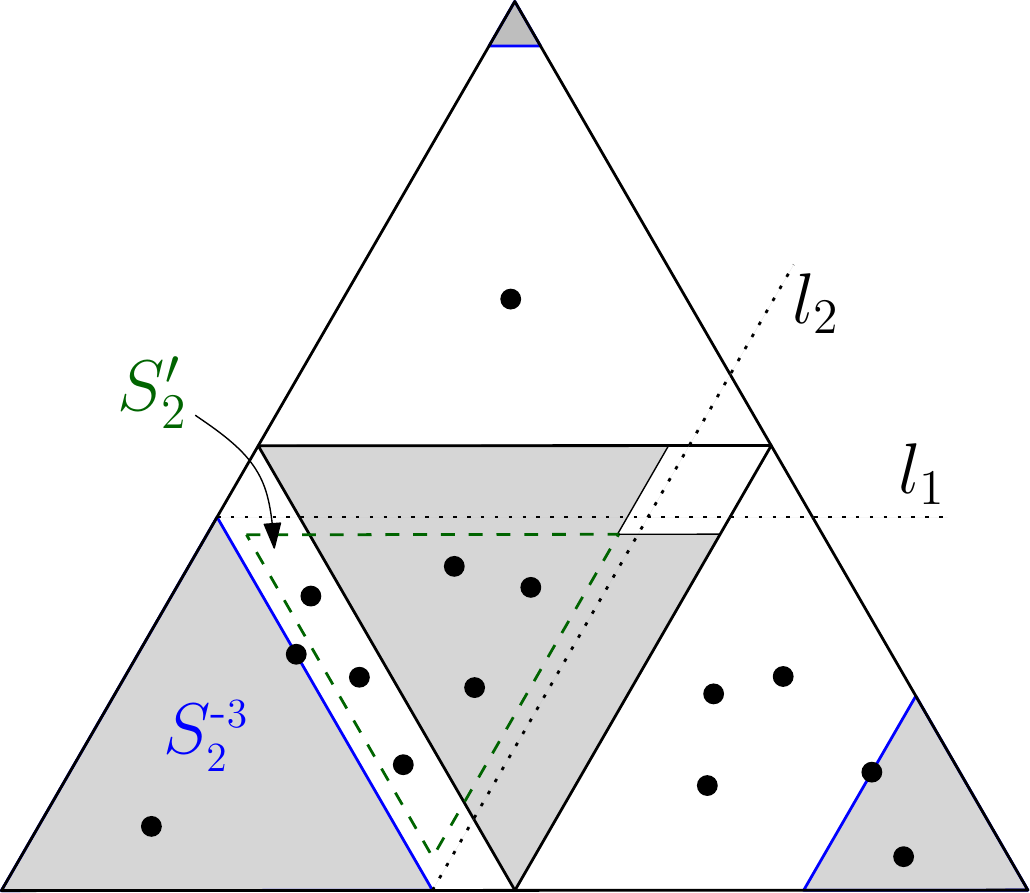}} \\
(a) & (b)
\end{tabular}$
  \caption{(a) Split $S$ into four equal area triangles. (b) $\SM{3}{2}$ is larger than $\SM{3}{1}$ and $\SM{3}{4}$.}
\label{sm:Theta-six-fig}
\end{figure}

 \item {$r_3=3$.}

In this case $n_3=4m_3+3$. We differentiate between two cases, where all the elements of the multiset $\{m_1, m_2, m_4\}$ are equal to zero, or some of them are greater than zero.

\begin{itemize}
 \item {\em All elements of $\{m_1, m_2, m_4\}$ are equal zero.} In this case, we have $m=m_3+1$. Consider the triangles $\SP{1}{2}$ and $\SM{1}{3r}$. See Figure~\ref{sm:Theta-six-fig}(a). Note that $\SP{1}{2}$ and $\SM{1}{3r}$ are disjoint, $\SP{1}{2}$ contains two points, and $\SM{1}{3r}$ contains $4m_3+2$ points. By induction, we get a matched pair in $\SP{1}{2}$ and a matching of size at least $m_3+1$ in $\SM{1}{3r}$. Thus, in total, we get a matching of size at least $1+(m_3+1)=m+1$ in $S$.

 \item {\em Some elements of $\{m_1, m_2, m_4\}$ are greater than zero.} Consider the triangles $\SM{3}{1}$, $\SM{3}{2}$, and $\SM{3}{4}$. Note that the area of some of these triangles\textemdash but not all\textemdash may be equal to zero. See Figure~\ref{sm:Theta-six-fig}(b). By induction, we get matchings of size $m_1$, $m_2$, and $m_4$ in $\SM{3}{1}$, $\SM{3}{2}$, and $\SM{3}{4}$, respectively. Without loss of generality, assume $\SM{3}{2}$, is larger than $\SM{3}{1}$ and $\SM{3}{4}$. Consider the half-lines $l_1$ and $l_2$ which are parallel to $l^0$ and $l^{60}$ axis, and have their endpoints on the top corner and right corner of $\SM{3}{2}$, respectively. We define $S'_2$ as the downward equilateral-triangle which is bounded by $l_1$, $l_2$, and the right side of $\SM{3}{2}$; the dashed triangle in Figure~\ref{sm:Theta-six-fig}(b). Note that $l_1$ and $l_2$ do not intersect $\SM{3}{1}$ and $\SM{3}{4}$. In addition, $\SM{3}{1}$, $\SM{3}{2}$, $\SM{3}{4}$, and $S'_2$ are pairwise disjoint. If any point of $S_1\cup S_2\cup S_3$ is to the right of $l_2$, then consider $\SP{1}{4}$ and $\SM{1}{3l}$. By induction, we get a matching of size $m_1+m_2+(m_3+1)+(m_4+1)$ in $\SM{3}{1}\cup \SM{3}{2}\cup\SM{1}{3l}\cup \SP{1}{4}$, and hence a matching of size $m+1$ in $S$. If any point of $S_2\cup S_3\cup S_4$ is above $l_1$, then consider $\SP{1}{1}$ and $\SM{1}{3b}$. By induction, we get a matching of size $(m_1+1)+m_2+(m_3+1)+ m_4$ in $\SP{1}{1}\cup \SM{3}{2}\cup\SM{1}{3b}\cup \SM{3}{4}$, and hence a matching of size $m+1$ in $S$. Otherwise, $S'_2$ contains $n_3+3=4(m_3+1)+2$ points. Thus, by induction, we get a matching of size $m_1+m_2+(m_3+2)+ m_4$ in $S_1\cup \SM{3}{2}\cup S'_2\cup S_4$, and hence a matching of size $m+1$ in $S$.
\end{itemize}

\item {\em $r_j=3$, for some $j\in\{1,2,4\}$.}

Without loss of generality, assume that $r_j=r_2$. Then, $n_2=4m_2+3$. Consider the triangles $\SM{3}{1}$, $\SM{1}{2}$, and $\SM{3}{4}$. See Figure~\ref{sm:Theta-six-fig2}(a). By induction, we get matchings of size $m_1$, $m_2+1$, and $m_4$ in $\SM{3}{1}$, $\SM{1}{2}$, and $\SM{3}{4}$, respectively. 
Now we consider the largest triangle among $\SM{3}{1}$, $\SM{1}{2}$, and $\SM{3}{4}$. Because of the symmetry, we have two cases: (i) $\SM{1}{2}$ is the largest, or (ii) $\SM{3}{4}$ is the largest.
\begin{figure}[h!]
  \centering
\setlength{\tabcolsep}{0in}
  $\begin{tabular}{cc}
\multicolumn{1}{m{.5\columnwidth}}{\centering\includegraphics[width=.4\columnwidth]{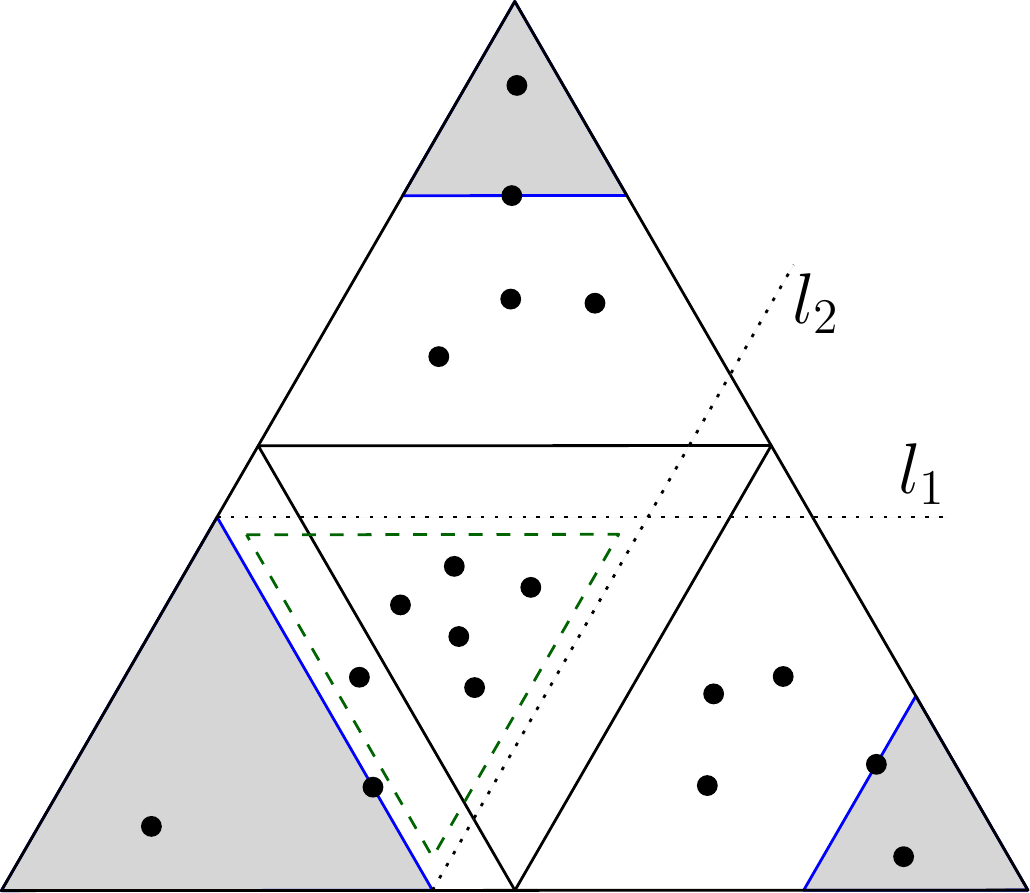}}
&\multicolumn{1}{m{.5\columnwidth}}{\centering\includegraphics[width=.4\columnwidth]{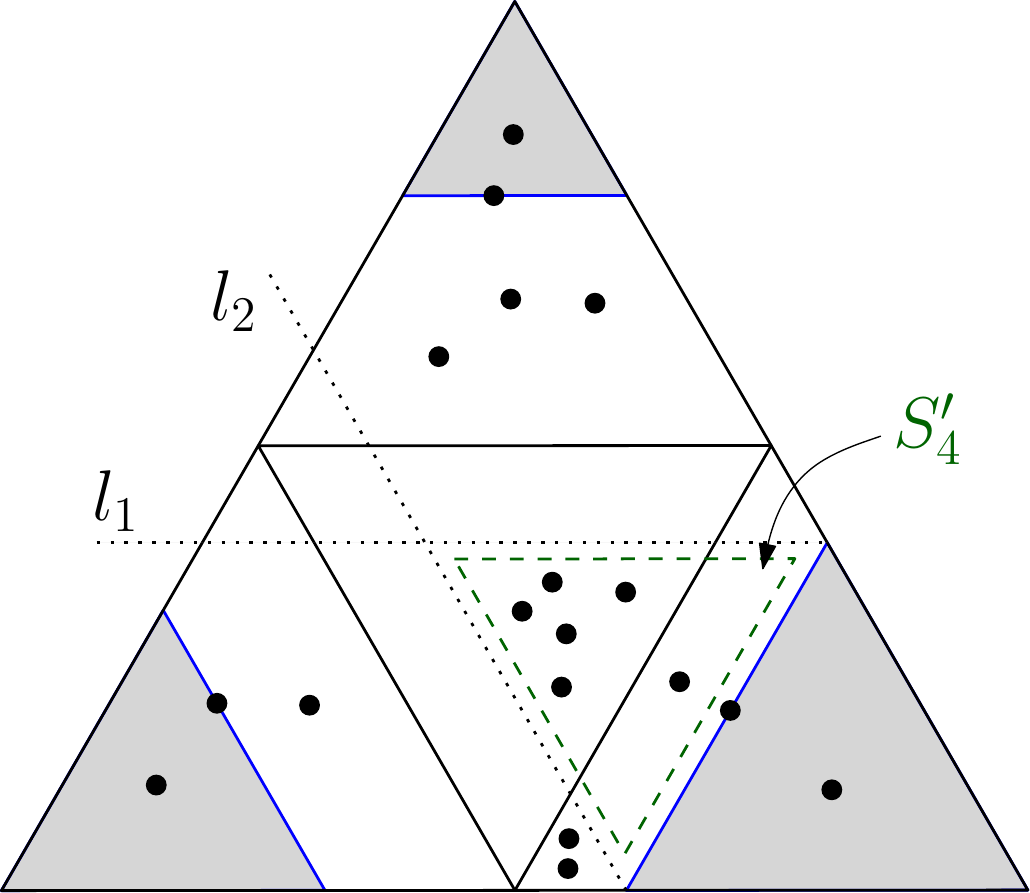}}
\\
(a) & (b)
\end{tabular}$
  \caption{(a) $\SM{1}{2}$ is larger than $\SM{3}{1}$ and $\SM{3}{4}$. (b) $\SM{3}{4}$ is larger than $\SM{3}{1}$ and $\SM{1}{2}$.}
\label{sm:Theta-six-fig2}
\end{figure}
\begin{itemize}
 \item {\em $\SM{1}{2}$ is larger than $\SM{3}{1}$ and $\SM{3}{4}$.}
Define the half-lines $l_1$, $l_2$, and the triangle $S'_2$ as in the previous case. See Figure~\ref{sm:Theta-six-fig2}(a). If any point of $S_1\cup S_2\cup S_3$ is to the right of $l_2$, then consider $\SP{1}{4}$ and $\SM{1}{3l}$. By induction, we get a matching of size $m_1+(m_2+1)+m_3+(m_4+1)$ in $\SM{3}{1}\cup \SM{1}{2}\cup\SM{1}{3l}\cup \SP{1}{4}$. If any point of $S_2\cup S_3\cup S_4$ is above $l_1$, then consider $\SP{1}{1}$ and $\SM{1}{3b}$. By induction, we get a matching of size $(m_1+1)+(m_2+1)+m_3+m_4$ in $\SP{1}{1}\cup \SM{1}{2}\cup\SM{1}{3b}\cup \SM{3}{4}$. Otherwise, $S'_2$ contains $n_3+1=4m_3+2$ points. Thus, by induction, we get a matching of size $m_1+(m_2+1)+(m_3+1)+ m_4$ in $S_1\cup \SM{1}{2}\cup S'_2\cup S_4$. As a result, in all cases we get a matching of size $m+1$ in $S$.

\item {\em $\SM{3}{4}$ is larger than $\SM{3}{1}$ and $\SM{1}{2}$.}
Define the half-lines $l_1$, $l_2$, and the triangle $S'_4$ as in Figure~\ref{sm:Theta-six-fig2}(b). If any point of $S_1\cup S_3\cup S_4$ is above $l_1$, then by induction, we get a matching of size $(m_1+1)+(m_2+1)+m_3+m_4$ in $\SP{1}{1}\cup \SM{1}{2}\cup\SM{1}{3b}\cup \SP{3}{4}$. If at least three points of $S_1\cup S_3\cup S_4$ are to the left of $l_2$, then consider $\SP{3}{2}$ and $\SM{3}{3r}$. Note that $\SP{3}{2}$ contains $n_2+3=4(m_2+1)+2$ points. By induction, we get a matching of size $m_1+(m_2+2)+m_3+m_4$ in $\SM{3}{1}\cup \SP{3}{2}\cup\SM{3}{3r}\cup \SM{3}{4}$. Otherwise, $S'_4$ contains at least $n_3+1=4m_3+2$ points. Thus, by induction, we get a matching of size $m_1+(m_2+1)+(m_3+1)+ m_4$ in $S_1\cup S_2\cup S'_4\cup \SM{3}{4}$. As a result, in all cases we get a matching of size $m+1$ in $S$.
\end{itemize}
\end{itemize}

{Case 2:} {\em Two elements in $R$ are equal to 0 and the other elements are equal to 1.}

In this case, we have $m=m_1+m_2+m_3+m_4$. Again, because of the symmetry, we have two cases: (i) $r_3=0$, (ii) $r_3\neq 0$.

\begin{itemize}
 \item $r_3=0.$

Without loss of generality assume that $r_2=0$ and $r_1=r_4=1$. Thus, $n_1=4m_1+1$, $n_2=4m_2$, $n_3=4m_3$, and $n_4=4m_4+1$. If all elements of $\{m_1,m_2,m_4\}$ are equal to zero, then we have $m=m_3$, where $m_3\ge 1$. Consider the triangles $\SP{1}{4}$ and $\SM{1}{3l}$, which are disjoint. By induction, we get a matched pair in $\SP{1}{4}$ and a matching of size at least $m_3$ in $\SM{1}{3l}$. Thus, in total, we get a matching of size at least $1+m_3=m+1$ in $S$. Assume some elements in $\{m_1,m_2,m_4\}$ are greater than zero. Consider the triangles $\SM{3}{1}$, $\SM{2}{2}$, and $\SM{3}{4}$. See Figure~\ref{sm:Theta-six-fig3}(a). By induction, we get a matching of size $m_1$, $m_2$, and $m_4$ in $\SM{3}{1}$, $\SM{2}{2}$, and $\SM{3}{4}$, respectively. 
Now we consider the largest triangle among $\SM{3}{1}$, $\SM{2}{2}$, and $\SM{3}{4}$. Because of the symmetry, we have two cases: (i) $\SM{2}{2}$ is the largest, or (ii) $\SM{3}{4}$ is the largest.

\begin{figure}[h!]
  \centering
\setlength{\tabcolsep}{0in}
  $\begin{tabular}{cc}
\multicolumn{1}{m{.5\columnwidth}}{\centering\includegraphics[width=.4\columnwidth]{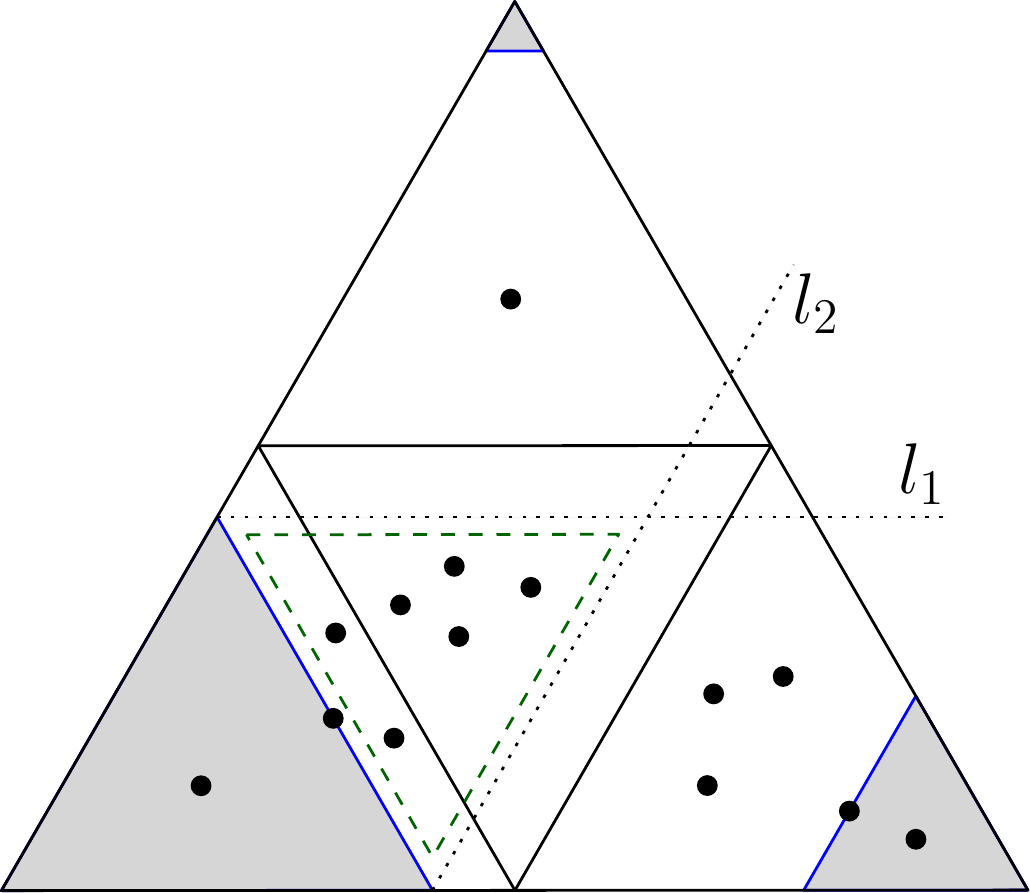}}
&\multicolumn{1}{m{.5\columnwidth}}{\centering\includegraphics[width=.4\columnwidth]{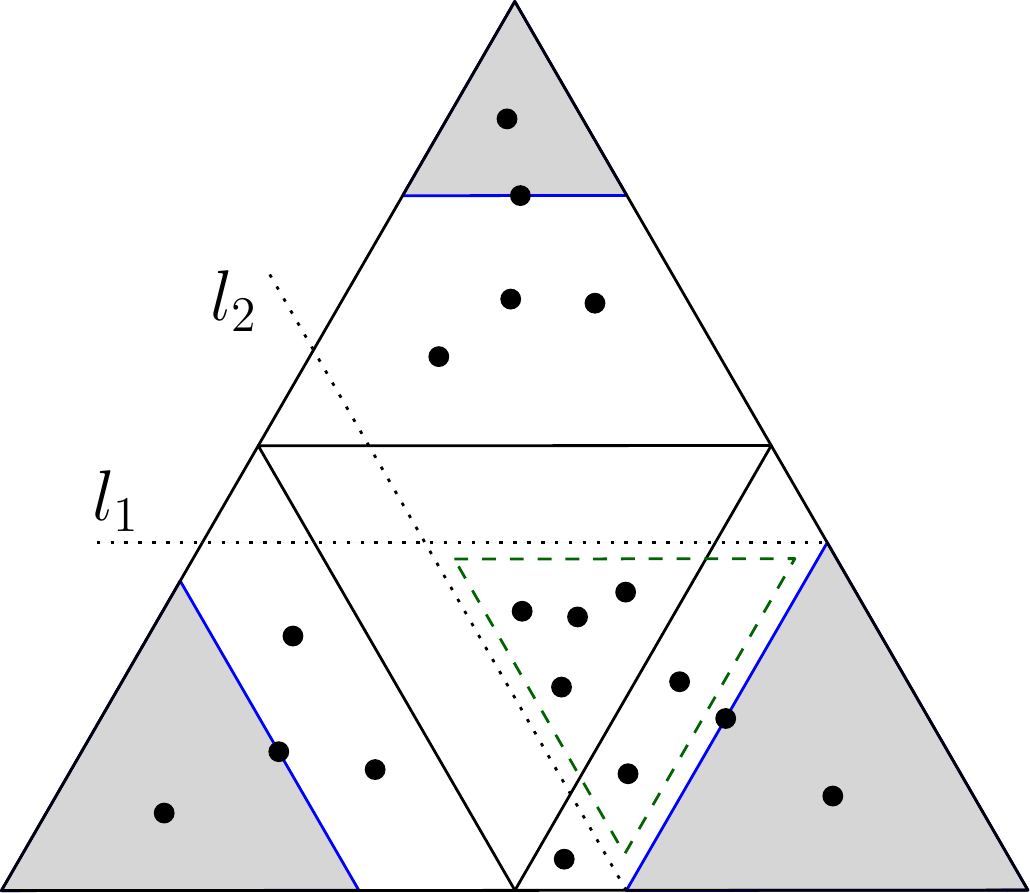}}
\\
(a) & (b)
\end{tabular}$
  \caption{(a) $\SM{2}{2}$ is larger than $\SM{3}{1}$ and $\SM{3}{4}$. (b) $\SM{3}{4}$ is larger than $\SM{3}{1}$ and $\SM{2}{2}$.}
\label{sm:Theta-six-fig3}
\end{figure}
\begin{itemize}
 \item {\em $\SM{2}{2}$ is larger than $\SM{3}{1}$ and $\SM{3}{4}$.}
Define $l_1$, $l_2$, $S'_2$ as in Figure~\ref{sm:Theta-six-fig3}(a). If any point of $S_1\cup S_2\cup S_3$ is to the right of $l_2$, then by induction, we get a matching of size $m_1+m_2+m_3+(m_4+1)$ in $\SM{3}{1}\cup \SM{2}{2}\cup\SM{1}{3l}\cup \SP{1}{4}$. If any point of $S_2\cup S_3\cup S_4$ is above $l_1$, then by induction, we get a matching of size $(m_1+1)+m_2+m_3+m_4$ in $\SP{1}{1}\cup \SM{2}{2}\cup\SM{1}{3b}\cup \SM{3}{4}$. Otherwise, $S'_2$ contains $n_3+2=4m_3+2$ points. Thus, by induction, we get a matching of size $m_1+m_2+(m_3+1)+ m_4$ in $S_1\cup \SM{2}{2}\cup S'_2\cup S_4$. In all cases we get a matching of size $m+1$ in $S$.

\item {\em $\SM{3}{4}$ is larger than $\SM{3}{1}$ and $\SM{2}{2}$.}
Define $l_1$, $l_2$, $S'_4$ as in Figure~\ref{sm:Theta-six-fig3}(b). If any point of $S_1\cup S_3\cup S_4$ is above $l_1$, then by induction, we get a matching of size $(m_1+1)+m_2+m_3+m_4$ in $\SP{1}{1}\cup \SM{2}{2}\cup\SM{1}{3b}\cup \SP{3}{4}$. If at least two points of $S_1\cup S_3\cup S_4$ are to the left of $l_2$, then by induction, we get a matching of size $m_1+(m_2+1)+m_3+m_4$ in $\SM{3}{1}\cup \SP{2}{2}\cup\SM{2}{3r}\cup \SM{3}{4}$. Otherwise, $S'_4$ contains at least $n_3+2=4m_3+2$ points. Thus, by induction, we get a matching of size $m_1+m_2+(m_3+1)+ m_4$ in $S_1\cup S_2\cup S'_4\cup \SM{3}{4}$. In all cases we get a matching of size $m+1$ in $S$.
\end{itemize}
  \item $r_3\neq 0.$

In this case $r_3=1$, and without loss of generality, assume that $r_2=1$; which means $r_1=r_4=0$. Thus, $n_1=4m_1$, $n_2=4m_2+1$, $n_3=4m_3+1$, and $n_4=4m_4$. If all elements of $\{m_1,m_2,m_4\}$ are equal to zero, then we have $m=m_3$, where $m_3\ge 1$. Consider the triangles $\SP{1}{2}$ and $\SM{1}{3r}$, which are disjoint. By induction, we get a matched pair in $\SP{1}{2}$ and a matching of size at least $m_3$ in $\SM{1}{3r}$. Thus, in total, we get a matching of size at least $1+m_3=m+1$ in $S$. Assume some elements in $\{m_1,m_2,m_4\}$ are greater than zero. Consider the triangles $\SM{2}{1}$, $\SM{3}{2}$, and $\SM{2}{4}$. See Figure~\ref{sm:Theta-six-fig4}(a). By induction, we get matchings of size $m_1$, $m_2$, and $m_4$ in $\SM{2}{1}$, $\SM{3}{2}$, and $\SM{2}{4}$, respectively. 
Now we consider the largest triangle among $\SM{2}{1}$, $\SM{3}{2}$, and $\SM{2}{4}$. Because of symmetry, we have two cases: (i) $\SM{3}{2}$ is the largest, or (ii) $\SM{2}{4}$ is the largest.

\begin{figure}[h!]
  \centering
\setlength{\tabcolsep}{0in}
  $\begin{tabular}{cc}
\multicolumn{1}{m{.5\columnwidth}}{\centering\includegraphics[width=.4\columnwidth]{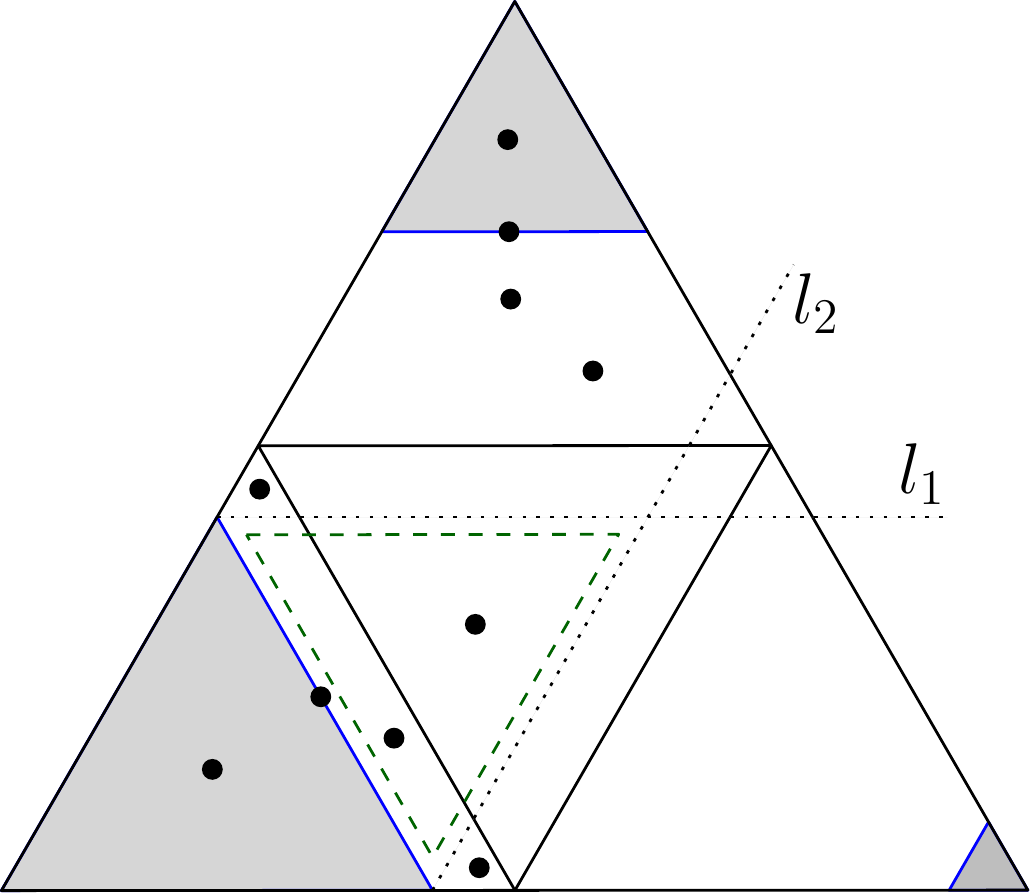}}
&\multicolumn{1}{m{.5\columnwidth}}{\centering\includegraphics[width=.4\columnwidth]{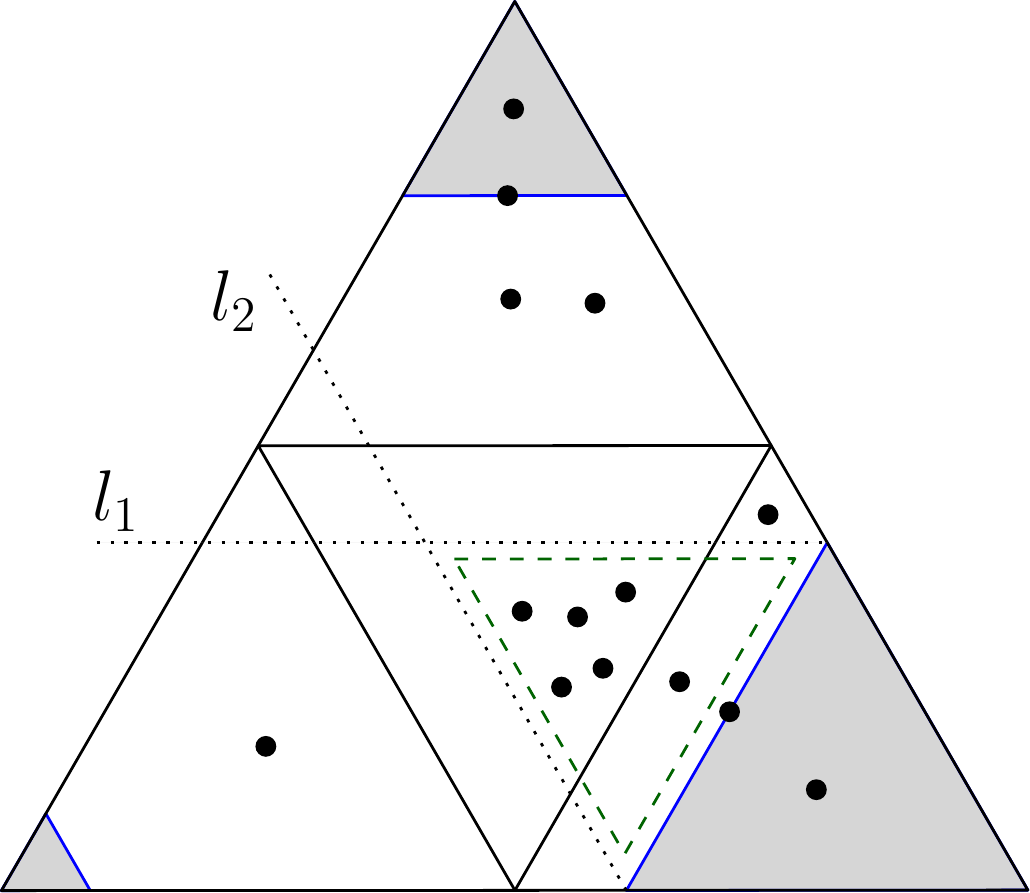}}
\\
(a) & (b)
\end{tabular}$
  \caption{(a) $\SM{3}{2}$ is larger than $\SM{2}{1}$ and $\SM{2}{4}$. (b) $\SM{2}{4}$ is larger than $\SM{2}{1}$ and $\SM{3}{2}$.}
\label{sm:Theta-six-fig4}
\end{figure}
\begin{itemize}
 \item {\em $\SM{3}{2}$ is larger than $\SM{2}{1}$ and $\SM{2}{4}$.}
Define $l_1$, $l_2$, $S'_2$ as in Figure~\ref{sm:Theta-six-fig4}(a). If at least two points of $S_1\cup S_2\cup S_3$ are to the right of $l_2$, then by induction, we get a matching of size $m_1+m_2+m_3+(m_4+1)$ in $\SM{2}{1}\cup \SM{3}{2}\cup\SM{2}{3l}\cup \SP{2}{4}$. If at least two points of $S_2\cup S_3\cup S_4$ are above $l_1$, then by induction, we get a matching of size $(m_1+1)+m_2+m_3+m_4$ in $\SP{2}{1}\cup \SM{3}{2}\cup\SM{2}{3b}\cup \SM{2}{4}$. Otherwise, $S'_2$ contains $n_3+1=4m_3+2$ points, and we get a matching of size $m_1+m_2+(m_3+1)+ m_4$ in $S_1\cup \SM{3}{2}\cup S'_2\cup S_4$. In all cases we get a matching of size $m+1$ in $S$.

\item {\em $\SM{2}{4}$ is larger than $\SM{2}{1}$ and $\SM{3}{2}$.}
Define $l_1$, $l_2$, $S'_4$ as in Figure~\ref{sm:Theta-six-fig4}(b). If at least two points of $S_2\cup S_3\cup S_4$ are above $l_1$, then by induction, we get a matching of size $(m_1+1)+m_2+m_3+m_4$ in $\SP{2}{1}\cup \SM{3}{2}\cup\SM{2}{3b}\cup \SM{2}{4}$. If any point of $S_1\cup S_3\cup S_4$ is to the left of $l_2$, then by induction, we get a matching of size $m_1+(m_2+1)+m_3+m_4$ in $\SM{2}{1}\cup \SP{1}{2}\cup\SM{1}{3r}\cup \SM{2}{4}$. Otherwise, $S'_4$ contains at least $n_3+1=4m_3+2$ points, and we get a matching of size $m_1+m_2+(m_3+1)+ m_4$ in $S_1\cup S_2\cup S'_4\cup \SM{2}{4}$. In all cases we get a matching of size $m+1$ in $S$. 
\end{itemize}
\end{itemize}
\end{proof}
\section{Strong Matching in 
$\G{\sqrs}{P}$}
\label{sm:infty-Delaunay-section}
In this section we consider the problem of computing a strong matching in $\G{\sqrs}{P}$, where $\sqr$ is an axis-aligned square whose center is the origin. We assume that $P$ is in general position, i.e., (i) no two points have the same $x$-coordinate or the same $y$-coordinate, and (ii) no four points are on the boundary of any homothet of $\sqr$. Recall that $\G{\sqrs}{P}$ is equal to the $L_\infty$-Delaunay graph on $P$. \'{A}brego et al. \cite{Abrego2004, Abrego2009} proved that $\G{\sqrs}{P}$ has a strong matching of size at least $\lceil n/5\rceil$. Using a similar approach as in Section~\ref{sm:theta-six-section}, we prove that $\G{\sqrs}{P}$ has a strong  matching of size at least $\lceil\frac{n-1}{4}\rceil$.

\begin{theorem}
\label{sm:infty-Delaunay-thr}
Let $P$ be a set of $n$ points in general position in the plane. Let $S$ be an axis-parallel square that contains $P$. Then, it is possible to find a strong matching of size at least $\lceil\frac{n-1}{4}\rceil$ for $\G{\sqrs}{P}$ in $S$.
\end{theorem}

\begin{proof}
The proof is by induction. Assume that any point set of size $n'\le n-1$ in an axis-parallel square $S'$, has a strong matching of size $\lceil \frac{n'-1}{4}\rceil$ in $S'$. If $n$ is $0$ or $1$, then there is no matching in $S$, and if $n\in\{2, 3, 4, 5\}$, then by shrinking $S$, it is possible to find a strongly matched pair. Suppose that $n\ge 6$, and $n=4m+r$, where $r\in\{0,1,2,3\}$. If $r\in\{0, 1,3\}$, then 
$\lceil \frac{n-1}{4}\rceil = \lceil \frac{(n-1)-1}{4}\rceil$, and by
induction we are done. Suppose that $n=4m+2$, for some $m\ge 1$. We prove that there are $\lceil\frac{n-1}{4}\rceil=m+1$ disjoint squares in $S$,
each of them matches a pair of points in $P$. Partition $S$ into four equal area squares $S_1, S_2, S_3, S_4$ which contain $n_1, n_2,n_3, n_4$ points, respectively; see Figure~\ref{sm:square-fig1}(a). Let $n_i=4m_i+r_i$ for $1\le i\le 4$, where $r_i\in\{0,1,2,3\}$. Let $R$ be the multiset $\{r_1,r_2,r_3,r_4\}$. 
By induction, in $S_1\cup S_2\cup S_3\cup S_4$, we have a strong matching of size at least
\begin{equation}
A=\left\lceil\frac{n_1-1}{4}\right\rceil + \left\lceil\frac{n_2-1}{4}\right\rceil +\left\lceil\frac{n_3-1}{4}\right\rceil+\left\lceil \frac{n_4-1}{4}\right\rceil.\nonumber
\end{equation} 

In the proof of Theorem~\ref{sm:theta-six-thr}, we have shown the following two claims:

\begin{paragraph}{Claim 1.}{$A\ge m$.}\end{paragraph}

\begin{paragraph}{Claim 2.}{\em If $A=m$, then either (i) one element in $R$ is equal to $3$ and the other elements are equal to $1$, or (ii) two elements in $R$ are equal to $0$ and the other elements are equal to $1$.}
\end{paragraph}

\vspace{8pt}
If $A> m$, then we are done. Assume that $A=m$; in fact, by the induction hypothesis we have an strong matching of size $m$ in $S$. 
We show how to find one more strongly matched pair in each case of Claim 2.

We define $\SM{$x$}{1}$ as the smallest axis-parallel square contained in $S_1$ and anchored at the top-left corner of $S_1$, which contains all the points in $S_1$ except $x$ points. If $S_1$ contains less than $x$ points, then the area of $\SM{$x$}{1}$ is zero. We also define $\SP{$x$}{1}$ as the smallest axis-parallel square that contains $S_1$ and anchored at the top-left corner of $S_1$, which has all the points in $S_1$ plus $x$ other points of $P$. See Figure~\ref{sm:square-fig1}(a). Similarly we define the squares $\SM{$x$}{2}$, $\SP{$x$}{2}$, and $\SM{$x$}{3}$, $\SP{$x$}{3}$, and $\SM{$x$}{4}$, $\SP{$x$}{4}$ which are anchored at the top-right corner of $S_2$, and the bottom-left corner of $S_3$, and the bottom-right corner of $S_4$, respectively.

{Case 1:} {\em One element in $R$ is equal to 3 and the other elements are equal to 1.}

In this case, we have $m = m_1 + m_2 + m_3 + m_4 + 1$. Without loss of generality, assume that $r_1=3$ and $r_2=r_3=r_4=1$. Consider the squares $\SM{1}{1}$, $\SM{3}{2}$, $\SM{3}{3}$, and $\SM{3}{4}$. Note that the area of some of these squares\textemdash but not all\textemdash may be
equal to zero. See Figure~\ref{sm:square-fig1}(b). By induction, we get matchings of size $m_1+1$, $m_2$, $m_3$, and $m_4$, in 
$\SM{1}{1}$, $\SM{3}{2}$, $\SM{3}{3}$, and $\SM{3}{4}$, respectively. Now consider the largest square among $\SM{1}{1}$, $\SM{3}{2}$, $\SM{3}{3}$, and $\SM{3}{4}$. Because of the symmetry, we have only three cases: (i) $\SM{1}{1}$ is the largest, (ii) $\SM{3}{2}$ is the largest, and (iii) $\SM{3}{4}$ is the largest.
\begin{figure}[htb]
  \centering
\setlength{\tabcolsep}{0in}
  $\begin{tabular}{ccc}
\multicolumn{1}{m{.33\columnwidth}}{\centering\includegraphics[width=.28\columnwidth]{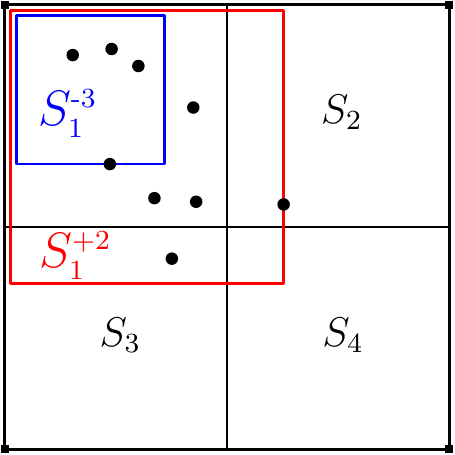}}
&\multicolumn{1}{m{.33\columnwidth}}{\centering\includegraphics[width=.28\columnwidth]{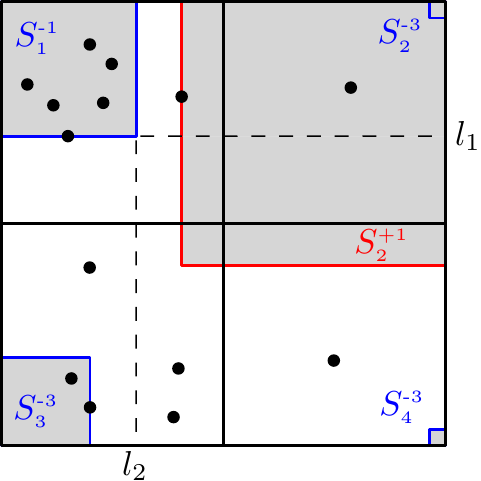}} &\multicolumn{1}{m{.33\columnwidth}}{\centering\includegraphics[width=.28\columnwidth]{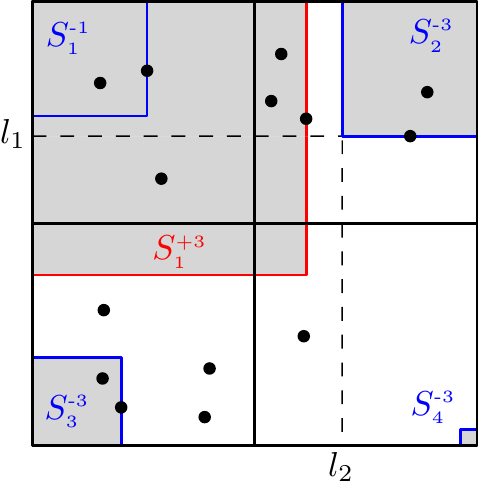}}
\\
(a) & (b)& (c)
\end{tabular}$
  \caption{(a) Split $S$ into four equal area squares. (b) $\SM{1}{1}$ is larger than $\SM{3}{2}$, $\SM{3}{3}$, and $\SM{3}{4}$. (c) $\SM{3}{2}$ is larger than $\SM{1}{1}$, $\SM{3}{3}$, and $\SM{3}{4}$.}
\label{sm:square-fig1}
\end{figure}
\begin{itemize}
\item {\em $\SM{1}{1}$ is the largest square.}
Consider the lines $l_1$ and $l_2$ which contain the bottom side and right side of $\SM{1}{1}$, respectively; the dashed lines in Figure~\ref{sm:square-fig1}(b). Note that $l_1$ and $l_2$ do not intersect any of $\SM{3}{2}$, $\SM{3}{3}$, and $\SM{3}{4}$. If any point of $S_1$ is to the right of $l_2$, then by induction, we get a matching of size $(m_1+1)+(m_2+1)+m_3+m_4$ in $\SM{1}{1}\cup \SP{1}{2}\cup\SM{3}{3}\cup \SM{3}{4}$. Otherwise, by induction, we get a matching of size $(m_1+1)+m_2+(m_3+1)+ m_4$ in $\SM{1}{1}\cup \SM{3}{2}\cup\SP{1}{3}\cup \SM{3}{4}$. In all cases we get a matching of size $m+1$ in $S$. 
\item {\em $\SM{3}{2}$ is the largest square.}
Consider the lines $l_1$ and $l_2$ which contain the bottom side and left side of $\SM{3}{2}$, respectively; the dashed lines in Figure~\ref{sm:square-fig1}(c). Note that $l_1$ and $l_2$ do not intersect any of $\SM{1}{1}$, $\SM{3}{3}$, and $\SM{3}{4}$. If any point of $S_2$ is below $l_1$, then by induction, we get a matching of size $(m_1+1)+m_2+m_3+(m_4+1)$ in $\SM{1}{1}\cup \SM{3}{2}\cup\SM{3}{3}\cup \SP{1}{4}$. Otherwise, by induction, we get a matching of size $(m_1+2)+m_2+m_3+ m_4$ in $\SP{3}{1}\cup \SM{3}{2}\cup\SM{3}{3}\cup \SM{3}{4}$; see Figure~\ref{sm:square-fig1}(c). In all cases we get a matching of size $m+1$ in $S$. 

\item {\em $\SM{3}{4}$ is the largest square.}
Consider the lines $l_1$ and $l_2$ which contain the top side and left side of $\SM{3}{4}$, respectively. If any point of $S_4$ is above $l_1$, then by induction, we get a matching of size $(m_1+1)+(m_2+1)+m_3+m_4$ in $\SM{1}{1}\cup \SP{1}{2}\cup\SM{3}{3}\cup \SM{3}{4}$. Otherwise, by induction, we get a matching of size $(m_1+1)+m_2+(m_3+1)+ m_4$ in $\SM{1}{1}\cup \SM{3}{2}\cup\SP{1}{3}\cup \SM{3}{4}$. In all cases we get a matching of size $m+1$ in $S$. 
\end{itemize}

{Case 2:} {\em Two elements in $R$ are equal to 0 and two elements are equal to 1.}

In this case, we have $m = m_1 + m_2 + m_3 + m_4$. Because of the symmetry, only two cases may arise: (i) $r_1=r_2=1$ and $r_3=r_4=0$, (ii) $r_1=r_4=1$ and $r_2=r_3=0$. 
\begin{itemize}
 \item {\em $r_1=r_2=1$ and $r_3=r_4=0$.}
Consider the squares $\SM{3}{1}$, $\SM{3}{2}$, $\SM{2}{3}$, and $\SM{2}{4}$. By induction, we get matchings of size $m_1$, $m_2$, $m_3$, and $m_4$, in $\SM{3}{1}$, $\SM{3}{2}$, $\SM{2}{3}$, and $\SM{2}{4}$, respectively. Now consider the largest square among $\SM{3}{1}$, $\SM{3}{2}$, $\SM{2}{3}$, and $\SM{2}{4}$. Because of the symmetry, we have only two cases: (a) $\SM{3}{1}$ is the largest, (b) $\SM{2}{3}$ is the largest. In case (a) we get one more matched pair either in $\SP{1}{2}$ or in $\SP{2}{3}$. In case (b) we get one more matched pair either in $\SP{1}{1}$ or in $\SP{2}{4}$.

 \item {\em $r_1=r_4=1$ and $r_2=r_3=0$.}
Consider the squares $\SM{3}{1}$, $\SM{2}{2}$, $\SM{2}{3}$, and $\SM{3}{4}$. By induction, we get matchings of size $m_1$, $m_2$, $m_3$, and $m_4$, in $\SM{3}{1}$, $\SM{2}{2}$, $\SM{2}{3}$, and $\SM{3}{4}$, respectively. Now consider the largest square among $\SM{3}{1}$, $\SM{2}{2}$, $\SM{2}{3}$, and $\SM{3}{4}$. Because of the symmetry, we have only two cases: (a) $\SM{3}{1}$ is the largest, (b) $\SM{2}{2}$ is the largest. In case (a) we get one more matched pair either in $\SP{2}{2}$ or in $\SP{2}{3}$. In case (b) we get one more matched pair either in $\SP{1}{1}$ or in $\SP{1}{4}$.
\end{itemize}
\end{proof}

\section{A Conjecture on Strong Matching in $\G{\ddiscs}{P}$}
\label{sm:conjecture-section}
In this section, we discuss a possible way to further improve upon Theorem~\ref{sm:Gabriel-thr}, as well as
a construction leading to the conjecture that Algorithm~\ref{sm:alg1} computes a strong matching of size at least $\lceil\frac{n-1}{8}\rceil$; unfortunately we are not able to prove this. 

In Section~\ref{sm:Gabriel-section} we proved that $\mathcal{I}(e^+)$ contains at most 16 edges. In order to achieve this upper bound we used the fact that the centers of the disks in $\mathcal{I}(e^+)$ should be far apart. We did not consider the endpoints of the edges representing these disks. By Observation~\ref{sm:no-point-in-circle-obs}, the disks representing the edges in $\mathcal{I}(e^+)$ cannot contain any of the endpoints. We applied this observation only on $u$ and $v$. Unfortunately, our attempts to apply this observation on the endpoints of edges in $\mathcal{I}(e^+)$ have been so far unsuccessful.

Recall that $T$ is a Euclidean minimum spanning tree of $P$, and for every edge $e=(u,v)$ in $T$, $\dg{e}$ is the degree of $e$ in $T(e^+)$, where $T(e^+)$ is the set of all edges of $T$ with weight at least $w(e)$. Note that $w(e)$ is directly related to the Euclidean distance between $u$ and $v$. Observe that the discs representing the edges adjacent to $e$ intersect $D(u,v)$. Thus, these edges are in $\Inf{e}$. 
We call an edge $e$ in $T$ a {\em minimal edge} if $e$ is not longer than any of its adjacent edges. We observed that the maximum degree of a minimal edge is an upper bound for $\Inf{e}$. We conjecture that,

\begin{conjecture}
{\em \Inf{$T$}} is at most the maximum degree of a minimal edge.
\end{conjecture}

Monma and Suri~\cite{Monma1992} showed that for every point set $P$ there exists a Euclidean minimum spanning tree, $MST(P)$, of maximum vertex degree five. Thus, the maximum edge degree in $MST(P)$ is 9. We show that for every point set $P$, there exists a Euclidean minimum spanning tree, $MST(P)$, such that the degree of each node is at most five and the degree of each minimal edge is at most eight. This implies the conjecture that $\Inf{MST(P)}\le 8$. That is, Algorithm~\ref{sm:alg1} returns a strong matching of size at least $\lceil\frac{n-1}{8}\rceil$.

\begin{figure}[ht]
  \centering
    \includegraphics[width=0.4\textwidth]{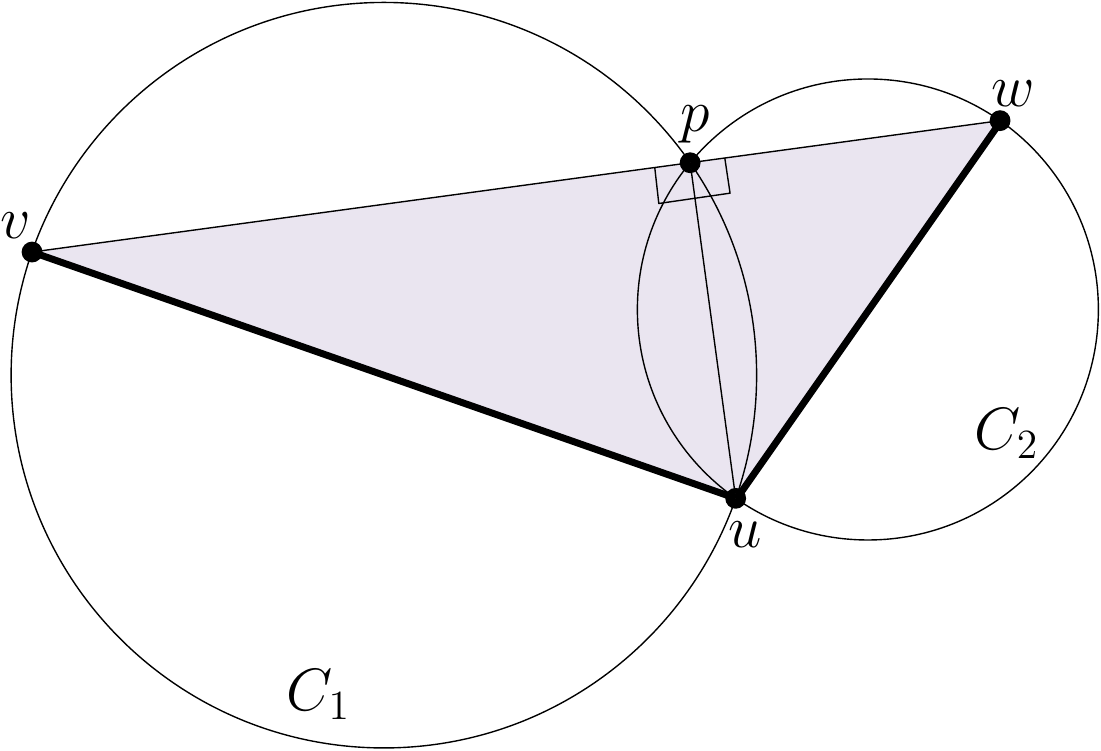}
  \caption{In $MST(P)$, the triangle $\bigtriangleup uvw$ formed by two adjacent edges $uv$ and $uw$, is empty.}
\label{sm:empty-triangle-figure}
\end{figure}
\begin{lemma}
\label{sm:empty-triangle-lemma}
If $uv$ and $uw$ are two adjacent edges in $MST(P)$, then the triangle $\bigtriangleup uvw$ has no point of $P\setminus\{u, v, w\}$ in its interior or on its boundary.
\end{lemma}
\begin{proof}
If the angle between $uv$ and $uw$ is equal to $\pi$, then there is no other point of $P$ on $uv$ and $uw$. Assume that $\angle vuw < \pi$. Refer to Figure \ref{sm:empty-triangle-figure}. Since $MST(P)$ is a subgraph of the Gabriel graph, the circles $C_1$ and $C_2$ with diameters $uv$ and $uw$ are empty. Since $\angle vuw < \pi$, $C_1$ and $C_2$ intersect each other at two points, say $u$ and $p$. Connect $u$, $v$ and $w$ to $p$. Since $uv$ and $uw$ are the diameters of $C_1$ and $C_2$, $\angle upv=\angle wpu=\pi/2$.
This means that $vw$ is a straight line segment. Since $C_1$ and $C_2$ are empty and $\bigtriangleup uvw \subset C_1 \cup C_2$, it follows that $\bigtriangleup uvw \cap P = \{u, v, w\}$.
\end{proof}

\begin{figure}[htb]
  \centering
  \includegraphics[width=.38\textwidth]{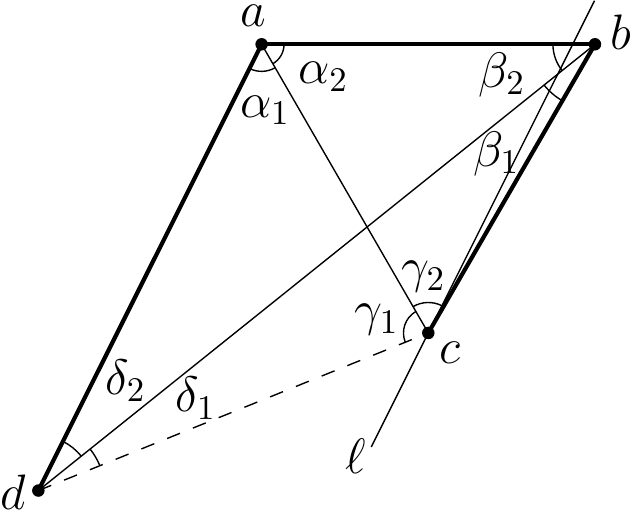}
  \caption{Illustration of Lemma~\ref{sm:convex-quadrilateral-lemma}: $|ab|\le|bc|\le|ad|$, $\angle abc\ge \pi/3$, $\angle bad\ge \pi/3$, and $\angle abc + \angle bad \le \pi$.}
\label{sm:convex-quadrilateral-fig}
\end{figure}

\begin{lemma}
\label{sm:convex-quadrilateral-lemma}
Follow Figure~\ref{sm:convex-quadrilateral-fig}. For a convex-quadrilateral $Q=a,b,c,d$ with $|ab|\le|bc|\le|ad|$, if  $\min\{\angle abc, \angle bad\}\ge \pi/3$ and $\angle abc + \angle bad \le \pi$, then $|cd|\le |ad|$.
\end{lemma}
\begin{proof}
 Let $\alpha_1=\angle cad$, $\alpha_2=\angle bac$, $\beta_1=\angle cbd$, $\beta_2=\angle abd$, $\gamma_1=\angle acd$, $\gamma_2=\angle acb$, $\delta_1=\angle bdc$, and $\delta_2=\angle adb$; see Figure~\ref{sm:convex-quadrilateral-fig}. Since $|ab|\le|bc|\le |ad|$, $$\gamma_2\le \alpha_2\text{ and }\delta_2\le \beta_2.$$ Let $\ell$ be a line passing through $c$ which is parallel to $ad$. Since $\angle abc + \angle bad \le \pi$, $\ell$ intersects the line segment $ab$. This implies that $\alpha_1\le \gamma_2$. If $\beta_1<\delta_1$, then $|cd|<|bc|$, and hence $|cd|<|ad|$ and we are done. Assume that $\delta_1\le \beta_1$. In this case, $\delta\le \beta$. Now consider the two triangles $\bigtriangleup abc$ and $\bigtriangleup acd$. Since $\delta\le \beta$ and $\alpha_1\le\gamma_2$, $\alpha_2\le \gamma_1$. Then we have $$\alpha_1\le \gamma_2\le \alpha_2\le\gamma_1.$$

Since $\alpha_1\le\gamma_1$, $|cd|\le |ad|$, where the equality holds only if $\alpha_1= \gamma_2= \alpha_2=\gamma_1$, i.e., $Q$ is a diamond. This completes the proof.
\end{proof}

\begin{figure}[htb]
  \centering
  \includegraphics[width=.6\textwidth]{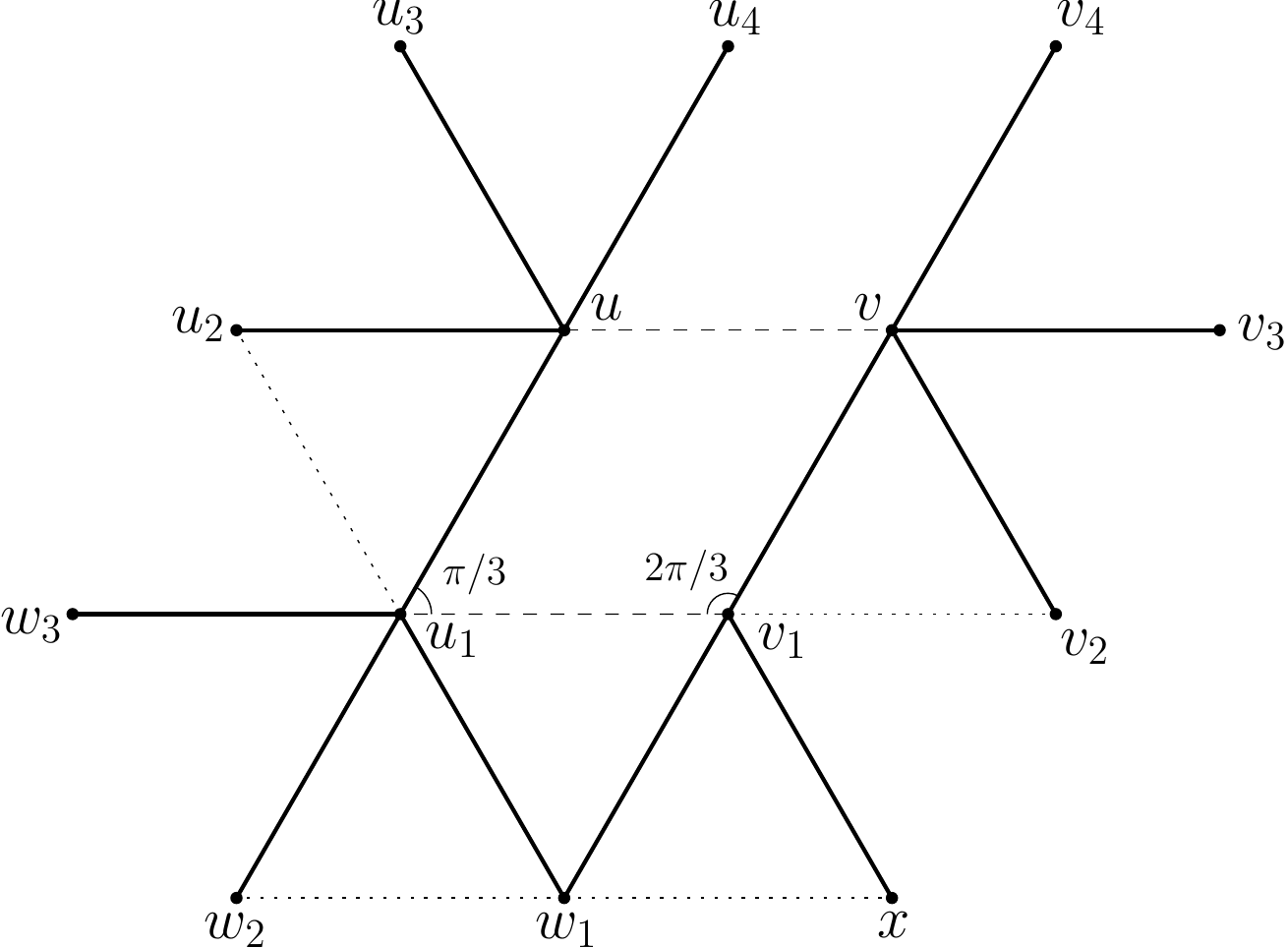}
  \caption{Solid segments represent the edges of $MST(P)$. Dashed segments represent the swapped edges. Dotted segments represent the edges which cannot exist.}
\label{sm:degree8-fig}
\end{figure}

\begin{lemma} 
\label{sm:edge-degree-lemma}
Every finite set of points $P$ in the plane admits a minimum spanning tree whose node degree is at most five and whose minimal-edge degree is at most nine.
\end{lemma}
\begin{proof}
Consider a minimum spanning tree, $MST(P)$, of maximum vertex degree 5. The maximum edge degree in $MST(P)$ is 9. Consider any minimal edge, $uv$. If the degree of $uv$ is 8, then $MST(P)$ satisfies the statement of the lemma. Assume that the degree of $uv$ is 9. Let $u_1, u_2,u_3,u_4$ and $v_1, v_2,v_3,v_4$ be the the neighbors of $u$ and $v$ in clockwise and counterclockwise orders, respectively. See Figure~\ref{sm:degree8-fig}. In $MST(P)$, the angles between two adjacent edges are at least $\pi/3$. Since $\angle u_iuu_{i+1}\ge \pi/3$ and $\angle v_ivv_{i+1}\ge \pi/3$ for $i=1,2,3$, either $\angle vuu_1+\angle uvv_1\le \pi$ or $\angle vuu_4+\angle uvv_4\le \pi$. Without loss of generality assume that $\angle vuu_1+\angle uvv_1\le \angle vuu_4+\angle uvv_4$ and $\angle vuu_1+\angle uvv_1\le \pi$. We prove that the spanning tree
obtained by swapping the edge $uv$ with $u_1v_1$ is also a minimum spanning tree, and it has one fewer minimal-edge of degree 9. By repeating this procedure at each minimal-edge of degree 9, we obtain a minimum spanning tree which satisfies the statement of the lemma. 
Let $Q=u,v,v_1,u_1$. By Lemma~\ref{sm:empty-triangle-lemma}, $v_1$ is outside the triangle $\bigtriangleup u_1uv$, and $u_1$ is outside the triangle $\bigtriangleup uvv_1$. In addition, $u_1$ and $v_1$ are on the same side of the line subtended from $uv$. Thus, $Q$ is a convex quadrilateral.
Without loss of generality assume that $|vv_1|\le |uu_1|$. By Lemma~\ref{sm:convex-quadrilateral-lemma}, $|u_1v_1|\le |uu_1|$. If $|u_1v_1|< |uu_1|$, we get a contradiction to Lemma~\ref{sm:cycle-lemma}. Thus, assume that $|u_1v_1|= |uu_1|$. As shown in the proof of Lemma~\ref{sm:convex-quadrilateral-lemma}, this case happens only when $Q$ is a diamond. This implies that $\angle vuu_1+\angle uvv_1=\pi$, and consequently $\angle vuu_4+\angle uvv_4= \pi$. In addition, $\angle u_iuu_{i+1}= \pi/3$ and $\angle v_ivv_{i+1}= \pi/3$ for $i=1,2,3$. To establish the validity of our edge-swap, observe that the nine edges incident to $u$ and $v$ are all equal in length. Therefore, swapping $uv$ with $u_1v_1$ does not change the cost of the spanning tree and, furthermore, the resulting tree is a valid spanning tree
since $u_1v_1$ is not an edge of the original spanning tree $MST(P)$; otherwise $u,v,v_1$, and $u_1$ would form a cycle. We have removed a minimal edge $uv$ of degree 9, but it remains to show that the degree of $u_1$ and $v_1$ does not increase to six and new minimal edge of degree 9 is not generated. Note that $u_1u_2$ and $v_1v_2$ are not the edges of $MST(P)$, and hence, $\dg{u_1}$ and $\dg{v_1}$ are still less than six. In order to show that no new minimal edge is generated, we differentiate between two cases:

\begin{itemize}
 \item $\min\{\angle vv_1u_1, \angle v_1u_1u\} > \pi/3$. Since $\angle v_1u_1u>\pi/3$ and $\angle uu_1u_2= \pi/3$, $u_1$ can be adjacent to at most two vertices other than $u$ and $v_1$, and hence $\dg{u_1}\le 4$; similarly $\dg{v_1}\le 4$. Thus, $u$, $v$, $u_1$, and $v_1$ are of degree at most four, and hence no new minimal edge of degree 9 is generated.

  \item $\min\{\angle vv_1u_1, \angle v_1u_1u\} = \pi/3$. W.l.o.g. assume that $\angle v_1u_1u=\pi/3$. This implies that $\angle vv_1u_1= 2\pi/3$. Since $\angle v_1u_1u=\pi/3$ and $\angle uu_1u_2= \pi/3$, $u_1$ is adjacent to at most three vertices other than $u$ and $v_1$. Let $u,v_1, w_1,w_2,w_3$ be the neighbors of $u_1$ in clockwise order. Note that $v_1$ is not adjacent to $u$, $v_2$ nor $w_1$. But $v_1$ can be connected to another vertex, say $x$, which implies that $\dg{v_1}\le 3$. We prove that the spanning tree obtained by swapping the edge $u_1v_1$ with $v_1w_1$ is also a minimum spanning tree of node degree at most five, which has one fewer minimal edge of degree 9. The new tree is a legal minimum spanning tree for $P$, because $|v_1w_1|=|v_1u_1|$. In addition, $\dg{u_1}\le 4$ and $\dg{v_1}\le 4$. Since $w_1w_2$ and $w_1x$ are illegal edges, $\dg{w_1}\le 4$. Thus, $u$, $v$, $u_1$, $v_1$, and $w_1$ are of degree at most four and no new minimal edge of degree 9 is generated. This completes the proof that our edge-swap reduces the number of minimal-edges of degree nine by one.
\end{itemize}
\end{proof}

\section{Conclusions}
\label{sm:conclusion}

Given a set of $n$ points in general position in the plane, we considered the problem of strong matching of points with convex geometric shapes. A matching is strong if the objects representing whose edges are pairwise disjoint. In this chapter we presented algorithms which compute strong matchings of points with diametral-disks, equilateral-triangles, and squares. Specifically we showed that:
\begin{itemize}
 \item There exists a strong matching of points with diametral-disks of size at least $\lceil\frac{n-1}{17}\rceil$.
 \item There exists a strong matching of points with downward equilat\-eral-triangles of size at least $\lceil\frac{n-1}{9}\rceil$.
\item There exists a strong matching of points with downward/upward equilateral-triangles of size at least $\lceil\frac{n-1}{4}\rceil$.
\item There exists a strong matching of points with axis-parallel squa\-res of size at least $\lceil\frac{n-1}{4}\rceil$.
\end{itemize}
 
The existence of a downward/upward equilateral-triangle matching of size at least $\lceil\frac{n-1}{4}\rceil$, implies the existence of either a downward equilateral-triangle matching of size at least $\lceil\frac{n-1}{8}\rceil$ or an upward equilateral-triangle matching of size at least $\lceil\frac{n-1}{8}\rceil$. This does not imply a lower bound better than $\lceil\frac{n-1}{9}\rceil$ for downward equilateral-triangle matching (or any fixed oriented equilateral-triangle).

A natural open problem is to improve any of the provided lower bounds, or extend these results for other convex shapes. A specific open problem is to prove that Algorithm~\ref{sm:alg1} computes a strong matching of points with diametral-disks of size at least $\lceil\frac{n-1}{8}\rceil$ as discussed in Section~\ref{sm:conjecture-section}.
An alternative for Algorithm~\ref{sm:alg1} is to pick an edge with the smallest influence set in each iteration. It remains open to derive lower bounds on the size of matchings reported by this new algorithm.

\bibliographystyle{abbrv}
\bibliography{../thesis}

%% file: chapters/ch4-bottleneckmatching.tex
\chapter{Bottleneck Plane Matchings in a Point Set}
\label{ch:bm}

A bottleneck plane perfect matching of a set of $n$ points in $\mathbb{R}^2$ is defined to be a perfect non-crossing matching that minimizes the length of the longest edge; the length of this longest edge is known as {\em bottleneck}. The problem of computing a bottleneck plane perfect matching has been proved to be NP-hard. 
We present an algorithm that computes a bottleneck plane matching of size at least $\frac{n}{5}$ in $O(n \log^2 n)$-time.
Then we extend our idea toward an $O(n\log n)$-time approximation algorithm which computes a plane matching of size at least $\frac{2n}{5}$ whose edges have length at most $\sqrt{2}+\sqrt{3}$ times the bottleneck.

\vspace{10pt}
This chapter is published in the journal of Computational Geometry: Theory and Applications~\cite{Abu-Affash2015-bottleneck}. 
               
\section{Introduction}
We study the problem of computing a bottleneck non-crossing matching of points in the plane.
For a given set $P$ of $n$ points in the plane, where $n$ is even, let $K(P)$ denote the complete Euclidean graph with vertex set $P$. The {\em bottleneck plane matching} problem is to find a perfect non-crossing matching of $K(P)$ that minimizes the length of the longest edge. We denote such a matching by $\MOPT$. The bottleneck, $\btopt$, is the length of the longest edge in $\MOPT$. The problem of computing $\MOPT$ 
has been proved to be NP-hard \cite{Abu-Affash2014}. 
Figure 1 in \cite{Abu-Affash2014} and \cite{Carlsson2010} shows that the longest edge in the minimum weight matching (which is planar) can be unbounded with respect to $\btopt$. On the other hand the weight of the bottleneck matching can be unbounded with respect to the weight of the minimum weight matching, see Figure \ref{bm:weight}.

\begin{figure}[ht]
  \centering
\setlength{\tabcolsep}{0in}
  $\begin{tabular}{cc}
 \multicolumn{1}{m{.5\columnwidth}}{\centering\includegraphics[width=.4\columnwidth]{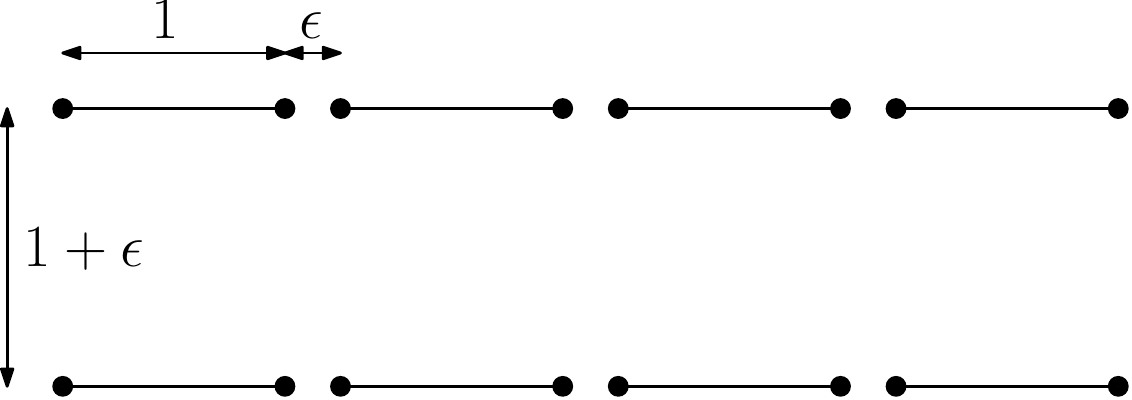}}
&\multicolumn{1}{m{.5\columnwidth}}{\centering\includegraphics[width=.35\columnwidth]{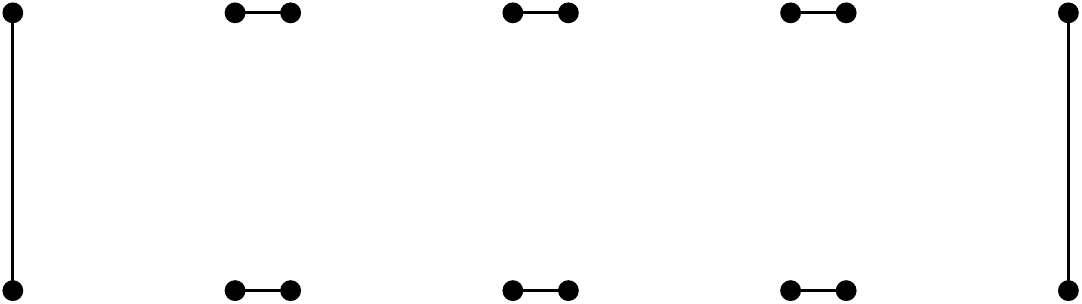}}\\
(a) & (b)
  \end{tabular}$
  \caption{(a) bottleneck matching, (b) minimum weight matching.}
\label{bm:weight}
\end{figure}

Matching and bottleneck matching problems play an important role in graph theory, and thus, they have been studied extensively,  e.g.,~\cite{Abu-Affash2014, Aloupis2013, Chang1992, Efrat2001, Efrat2000, Vaidya1989, Varadarajan1998}. Self-crossing configurations are often undesirable and may even imply an error condition; for example, a potential collision between moving objects, or inconsistency in a layout of a circuit.
In particular, non-crossing matchings are especially important in the context of VLSI circuit layouts~\cite{Lengauer1990} and operations research.

\addtocontents{toc}{\protect\setcounter{tocdepth}{1}}
\subsection{Previous Work}
\addtocontents{toc}{\protect\setcounter{tocdepth}{2}}

It is desirable to compute a perfect matching of a point set in the plane, which is optimal with respect to some criterion such as: (a) {\em minimum-cost matching} which minimizes the sum of the lengths of all edges; also known as {\em minimum-weight matching} or {\em min-sum matching}, and (b) {\em bottleneck matching} which minimizes the length of the longest edge; also known as {\em min-max matching} \cite{Efrat2000}. For the minimum-cost matching, Vaidya \cite{Vaidya1989} presented an $O(n^{2.5}\log^4 n)$ time algorithm, which was improved to $O(n^{1.5}\log^5 n)$ by Varadarajan \cite{Varadarajan1998}. As for bottleneck matching, Chang et al. \cite{Chang1992} proved that such kind of matching is a subset of 17-RNG (relative neighborhood graph). They presented an algorithm, running in $O(n^{1.5}\sqrt{\log n})$-time to compute a bottleneck matching of maximum cardinality. The matching computed by their algorithm may be crossing.
Efrat and Katz \cite{Efrat2000} extended the result of Chang et al. \cite{Chang1992} to higher dimensions. They proved that a bottleneck matching in any constant dimension can be computed in $O(n^{1.5}\sqrt{\log n})$-time under the $L_\infty$-norm. 

Note that a plane perfect matching of a point set can be computed in $O(n\log n)$-time, e.g., by matching the two leftmost points recursively.

Abu-Affash et al.~\cite{Abu-Affash2014} showed that the bottleneck plane perfect matching problem is NP-hard and presented an algorithm that computes a plane perfect matching whose edges have length at most $2\sqrt{10}$ times the bottleneck, i.e., $2\sqrt{10}\btopt$. They also showed that this problem does not
admit a PTAS (Polynomial Time Approximation Scheme), unless P=NP. Carlsson et al.~\cite{Carlsson2010} showed that the bottleneck plane perfect matching problem for a Euclidean bipartite complete graph is also NP-hard.
\addtocontents{toc}{\protect\setcounter{tocdepth}{1}}
\subsection{Our Results}
\addtocontents{toc}{\protect\setcounter{tocdepth}{2}}
The main results of this chapter are summarized in Table \ref{bm:table1}. We use the unit disk graph as a tool for our approximations. First, we present an $O(n\log n)$-time algorithm in Section \ref{bm:UDG}, that computes a plane matching of size at least $\frac{n-1}{5}$ in a connected unit disk graph. Then in Section \ref{bm:bottleneck-five-over-two} we describe how one can use this algorithm to obtain a bottleneck plane matching of size at least $\frac{n}{5}$ with edges of length at most $\btopt$ in $O(n\log^2 n)$-time.
In Section \ref{bm:bottleneck-five-over-four} we present an $O(n\log n)$-time approximation algorithm that computes a plane matching of size at least $\frac{2n}{5}$ whose edges have length at most $(\sqrt{2}+\sqrt{3})\btopt$. Finally we conclude this chapter in Section \ref{bm:conclusion}.

\begin{table}
\centering
\caption{Summary of results.}
\label{bm:table1}
    \begin{tabular}{|l|c|c|c|}
         \hline
             Algorithm &running time&bottleneck ($\bt$)    & match. size  \\ \hline
             A.-Affash et al. \cite{Abu-Affash2014}&$O(n^{1.5}\sqrt{\log n})$  & $2\sqrt{10}\btopt$ & $n/2$ \\
             Section \ref{bm:bottleneck-five-over-two} &$O(n \log^2 n)$& $\btopt$ & $n/5$ \\
             Section \ref{bm:bottleneck-five-over-four}&$O(n\log n)$ & $(\sqrt{2}+\sqrt{3})\btopt$ & $2n/5$\\
         \hline
    \end{tabular}
\end{table}

\section{Preliminaries}
\label{bm:preliminaries}

Let $P$ denote a set of $n$ points in the plane, where $n$ is even, and let $K(P)$ denote the complete Euclidean graph over $P$. A {\em matching}, $M$, is a subset of edges of $K(P)$ without common vertices. Let $|M|$ denote the {\em cardinality} of $M$, which is the number of edges in $M$. $M$ is a {\em perfect matching} if it covers all the vertices of $P$, i.e., $|M|=\frac{n}{2}$. The {\em bottleneck} of $M$ is defined as the longest edge in $M$. We denote its length by $\bt_M$. A {\em bottleneck perfect matching} is a perfect matching that minimizes the bottleneck length. A {\em plane matching} is a matching with non-crossing edges. 
We denote a plane matching by $M_=$ and a crossing matching by $M_\times$. 
The {\em bottleneck plane perfect matching}, $\MOPT$, is a perfect plane matching which minimizes the length of the longest edge. Let $\btopt$ denote the length of the bottleneck edge in $\MOPT$. In this chapter we consider the problem of computing a bottleneck plane matching of $P$. 

The Unit Disk Graph, $UDG(P)$, is defined to have the points of $P$ as its vertices and two vertices $p$ and $q$ are connected by an edge if their Euclidean distance $|pq|$ is at most 1. The {\em maximum plane matching} of $UDG(P)$ is the maximum cardinality matching of $UDG(P)$, which has no pair of crossing edges. The following lemma states a folklore result, however, for the sake of completeness we provide a simple proof. 

\begin{lemma}
 If the maximum plane matching in unit disk graphs can be computed in polynomial time, then the bottleneck plane perfect matching problem for point sets can also be solved in polynomial time.
\end{lemma}
\begin{proof}
Let $D=\{|pq|:p,q\in P\}$ be the set of all distances determined by pairs of points in $P$. Note that $\btopt\in D$. For each $\bt\in D$, define the ``unit'' disk graph $DG(\bt,P)$, in which two points $p$ and $q$ are connected by an edge if and only if $|pq|\le\bt$. Then $\btopt$ is the minimum $\bt$ in $D$ such that $DG(\bt,P)$ has a plane matching of size $\frac{n}{2}$.
\end{proof}

The Gabriel Graph, $GG(P)$, is defined to have the points of $P$ as its vertices and two vertices $p$ and $q$ are connected by an edge if the disk with diameter $pq$ does not contain any point of $P\setminus\{p,q\}$ in its interior and on its boundary.

\begin{lemma}
\label{bm:MST-UDG}
If the unit disk graph $UDG(P)$ of a point set $P$ is connected, then $UDG(P)$ and $K(P)$ have the same minimum spanning tree.
\end{lemma}
\begin{proof}
 By running Kruskal's algorithm on $UDG(P)$, we get a minimum spanning tree, say $T$. All the edges of $T$ have length at most one, and the edges of $K(P)$ which do not belong to $UDG(P)$ all have length greater than one. Hence, $T$ is also a minimum spanning tree of $K(P)$.
\end{proof}

As a direct consequence of Lemma \ref{bm:MST-UDG} we have the following corollary:

\begin{corollary}
\label{bm:MST-DEL}
Consider the unit disk graph $UDG(P)$ of a point set $P$. We can compute the minimum spanning forest of $UDG(P)$, by first computing the minimum spanning tree of $P$ and then removing the edges whose length is more than one.
\end{corollary}

\begin{lemma}
\label{bm:same-component}
 For each pair of crossing edges $(u,v)$ and $(x,y)$ in $UDG(P)$, the four endpoints $u$, $v$, $x$, and $y$ are in the same component of $UDG(P)$. 
\end{lemma}
\begin{proof}
Note that the quadrilateral $Q$ formed by the end points $u$, $v$, $x$, and $y$ is convex. W.l.o.g. assume that the angle $\angle xuy$ is the largest angle in $Q$. Clearly $\angle xuy \ge \pi/2$, and hence, in triangle $\bigtriangleup xuy$, the angles $\angle yxu$ and $\angle uyx$ are both less than $\pi/2$. Thus, the edges $(u,x)$ and $(u,y)$ are both less than $(x,y)$. This means that $(u,x)$ and $(u,y)$ are also edges of $UDG(P)$, thus, their four endpoints belong to the same component.  
\end{proof}

As a direct consequence of Lemma \ref{bm:same-component} we have the following corollary: 

\begin{corollary}
\label{bm:non-crossing-edges}
Any two edges that belong to different components of $UDG(P)$ do not cross.
\end{corollary}

Let $MST(P)$ denote the Euclidean minimum spanning tree of $P$.
We have prove the following lemma in Chapter~\ref{ch:sm}.
\begin{lemma}
\label{bm:empty-triangle-lemma}
If $(u,v)$ and $(u,w)$ are two adjacent edges in $MST(P)$, then the triangle $\bigtriangleup uvw$ has no point of $P\setminus\{u, v, w\}$ inside or on its boundary.
\end{lemma}

\begin{corollary}
\label{bm:empty-convex-hull}
Consider $MST(P)$ and let $N_v$ be the set of neighbors of a vertex $v$ in $MST(P)$. Then the convex hull of $N_v$ contains no point of $P$ except $v$ and the set $N_v$.
\end{corollary}

The shaded area in Figure \ref{bm:empty-skeleton-fig} shows the union of all these convex hulls.

\section{Plane Matching in Unit Disk Graphs}
\label{bm:UDG}

In this section we present two approximation algorithms for computing a maximum plane matching in a unit disk graph $UDG(P)$. In Section \ref{bm:three-approximation} we present a straight-forward $\frac{1}{3}$-approximation algorithm. For the case when $UDG(P)$ is connected, we present a $\frac{2}{5}$-approximation algorithm in Section \ref{bm:five-over-two-approximation}.

\subsection{$\frac{1}{3}$-approximation algorithm}
\label{bm:three-approximation}

Given a possibly disconnected unit disk graph $UDG(P)$, we start by computing a (possibly crossing) maximum matching $\MC$ of $UDG(P)$ using Edmonds algorithm \cite{Edmonds1965}. Then we transform $\MC$ to another (possibly crossing) matching $\MC'$ with some properties, and then pick at least one-third of its edges which satisfy the non-crossing property. Consider a pair of crossing edges $(p,q)$ and $(r,s)$ in $\MC$, and let $c$ denote the intersection point. If their smallest intersection angle is at most $\pi/3$, we replace these two edges with new ones. For example if $\angle pcr \le \pi/3$, we replace $(p,q)$ and $(r,s)$ with new edges $(p,r)$ and $(q,s)$. Since the angle between them is at most $\pi/3$, the new edges are not longer than the older ones, i.e. $\max\{|pr|,|qs|\}\le\max\{|pq|,|rs|\}$, and hence the new edges belong to the unit disk graph. On the other hand the total length of the new edges is strictly less than the older ones; i.e. $|pr|+|qs|<|pq|+|rs|$. For each pair of intersecting edges in $\MC$, with angle at most $\pi/3$, we apply this replacement. We continue this process until we have a matching $\MC'$ with the property that if two matched edges intersect, each of the angles incident on $c$ is larger than $\pi/3$.

For each edge in $M'_{\times}$, consider the counter clockwise angle it makes with the positive $x$-axis; this angle is in the range $[0,\pi)$. Using these angles, we partition the edges of $M'_{\times}$ into three subsets, one subset for the angles $[0,\pi/3)$, one subset for the angles $[\pi/3,2\pi/3)$, and
one subset for the angles $[2\pi/3,\pi)$. Observe that edges within one subset are non-crossing. Therefore, if we output the largest subset, we obtain a non-crossing matching of size at least
$|M'_{\times}|/3 = |M_{\times}|/3$.

Since in each step (each replacement) the total length of the matched edges decreases, the replacement process converges and the algorithm will stop. Bonnet ant Miltzow~\cite{Bonnet2016-ewcg} showed that this process stops after $O(n^3)$ steps, where $n$ is the number of points in $P$. Thus the running time of this algorithm is polynomial.

\subsection{$\frac{2}{5}$-approximation algorithm for connected unit disk graphs}
\label{bm:five-over-two-approximation}

In this section we assume that the unit disk graph $UDG(P)$ is connected. Monma et al. \cite{Monma1992} proved that every set of points in the plane admits a minimum spanning tree of degree at most five which can be computed in $O(n\log n)$ time. By Lemma \ref{bm:MST-UDG}, the same claim holds for $UDG(P)$. Here we present an algorithm which extracts a plane matching $M$ from $MST(P)$. Consider a minimum spanning tree $T$ of $UDG(P)$ with vertices of degree at most five. We define the {\em skeleton tree}, ${T'}$, as the tree obtained from $T$ by removing all its leaves; see Figure \ref{bm:empty-skeleton-fig}. Clearly ${T'} \subseteq T \subseteq UDG(P)$. For clarity we use $u$ and $v$ to refer to the leaves of $T$ and $T'$ respectively. In addition, let $v$ and $v'$, respectively, refer to the copies of a vertex $v$ in $T$ and $T'$. In each step, pick an arbitrary leaf $v'\in T'$. By the definition of ${T'}$, it is clear that the copy of $v'$ in $T$, i.e. $v$, is connected to vertices $u_1,\dots, u_k$, for some $1\le k \le 4$, which are leaves of $T$ (if $T'$ has one vertex then $k\le5$). Pick an arbitrary leaf $u_i$ and add $(v, u_i)$ as a matched pair to $M$. For the next step we update $T$ by removing $v$ and all its adjacent leaves. We also compute the new skeleton tree and repeat this process. 
In the last iteration, $T'$ is empty and we may be left with a tree $T$ consisting of one single vertex or one single edge. If $T$ consists of one single vertex, we disregard it, otherwise we add its only edge to $M$.

\begin{figure}[ht]
  \centering
    \includegraphics[width=0.6\textwidth]{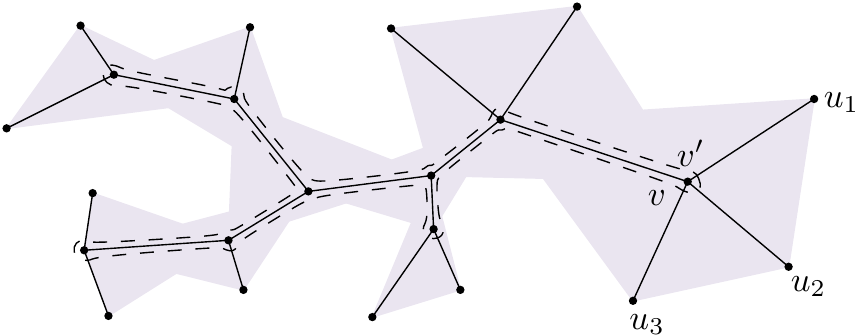}
  \caption{Minimum spanning tree $T$ with union of empty convex hulls. The skeleton tree $T'$ is surrounded by dashed line, and $v'$ is a leaf in $T'$. }
\label{bm:empty-skeleton-fig}
\end{figure}

The formal algorithm is given as {\scshape PlaneMatching}, which receives a point set $P$---whose unit disk graph is connected---as input and returns a matching $M$ as output. The function MST5$(P)$ returns a Euclidean minimum spanning tree of $P$ with degree at most five, and the function Neighbor$(v', T')$ returns the only neighbor of leaf $v'$ in $T'$.

\begin{algorithm}                      
\caption{{\scshape PlaneMatching}$(P)$}          
\label{bm:alg1} 
\require{set $P$ of $n$ points in the plane, such that $UDG(P)$ is connected.}\\
\ensure{plane matching $M$ of $MST(P)$ with $|M|\ge\frac{n-1}{5}$.}
\begin{algorithmic}[1]
    \State $M \gets \emptyset$
    \State $T \gets $MST5$(P)$
    \State $T' \gets $SkeletonTree$(T)$
    \While {$T' \neq \emptyset$}
	\State $v' \gets$ {a leaf of } $T'$
	\State $L_v \gets$ {set of leaves  connected to } $v$ { in } $T$
        \State $u \gets$ {an element of } $L_v$
	\State $M \gets M\cup\{(v,u)\}$	
	\State $T\gets T\setminus (\{v\}\cup L_v)$
	\If {$deg($Neighbor$(v', T'))=2$}
	    \State $T'\gets T'\setminus \{v',${ Neighbor}$(v', T')\}$
	\Else
	    \State $T'\gets T'\setminus \{v'\}$
	\EndIf
    \EndWhile
    \If {$T$ {consists of one single edge}}
	\State $M \gets M\cup T$
    \EndIf
    \State \Return $M$
\end{algorithmic}
\end{algorithm}

\begin{lemma}
\label{bm:n-minus-one-lemma}
Given a set $P$ of $n$ points in the plane such that $UDG(P)$ is connected, algorithm {\scshape PlaneMatching} returns a plane matching $M$ of $MST(P)$ of size $|M|\ge \frac{n-1}{5}$. Furthermore, $M$ can be computed in $O(n\log n)$ time.
\end{lemma}
\begin{proof}
In each iteration an edge $(v, u_i)\in T$ is added to $M$. Since $T$ is plane, $M$ is also plane and $M$ is a matching of $MST(P)$.

Line 5 picks $v'\in T'$ which is a leaf, so its analogous vertex $v\in T$ is connected to at least one leaf. In each iteration we select an edge incident to one of the leaves and add it to $M$, then disregard all other edges connected to $v$ (line 9). So for the next iteration $T$ looses at most five edges. Since $T$ has $n-1$ edges initially and we add one edge to $M$ out of each five edges of $T$, we have $|M|\ge\frac{n-1}{5}$. 

According to Corollary \ref{bm:MST-DEL} and by \cite{Monma1992}, line 2 takes $O(n\log n)$ time. The while-loop is iterated $O(n)$ times and in each iteration, lines 5-13 take constant time. So, the total running time of algorithm {\scshape PlaneMatching} is $O(n\log n)$. 
\end{proof}

The size of a maximum matching can be at most $\frac{n}{2}$. Therefore, algorithm {\scshape PlaneMatching} computes a matching of size at least $\frac{2(n-1)}{5n}$ times the size of a perfect matching, and hence, when $n$ is large enough, {\scshape PlaneMatching} is a $\frac{2}{5}$-approximation algorithm. On the other hand there are unit disk graphs whose maximum matchings have size $\frac{n-1}{5}$; see Figure \ref{bm:worst-case-optimal}. In this case {\scshape PlaneMatching} returns a maximum matching. In addition, when $UDG(P)$ is a tree or a cycle, {\scshape PlaneMatching} returns a maximum matching. 

\begin{figure}[ht]
  \centering
    \includegraphics[width=0.7\textwidth]{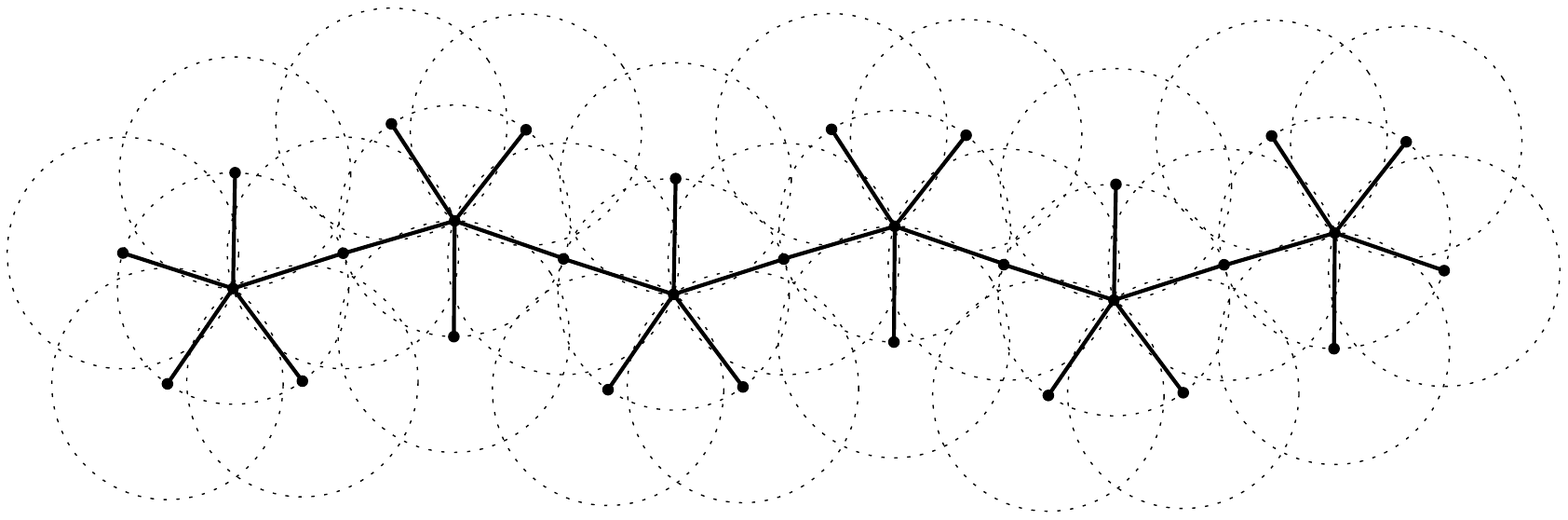}
  \caption{Unit disk graph with all edges of unit length, and maximum matching of size $\frac{n-1}{5}$.}
\label{bm:worst-case-optimal}
\end{figure}

In Section \ref{bm:bottleneck-five-over-two} we will show how one can use a modified version of algorithm {\scshape PlaneMatching} to compute a bottleneck plane matching of size at least $\frac{n}{5}$ with bottleneck length at most $\btopt$. Recall that $\btopt$ is the length of the bottleneck edge in the bottleneck plane perfect matching $\MOPT$. Section \ref{bm:bottleneck-five-over-four} extends this idea to an algorithm which computes a plane matching of size at least $\frac{2n}{5}$ with edges of length at most $(\sqrt{2}+\sqrt{3})\btopt$.

\section{Approximating Bottleneck Plane Perfect Matching}
\label{bm:bottleneck}

The general approach of our algorithms is to first compute a (possibly crossing) bottleneck perfect matching $\MC$ of $K(P)$ using the algorithm in \cite{Chang1992}. Let $\bt_{\MC}$ denote the length of the bottleneck edge in $\MC$. It is obvious that the bottleneck length of any plane perfect matching is not less than $\bt_{\MC}$. Therefore, $\btopt\ge\bt_{\MC}$. We consider a ``unit'' disk graph $DG(\bt_{\MC}, P)$ over $P$, in which there is an edge between two vertices $p$ and $q$ if $|pq|\le\bt_{\MC}$. Note that $DG(\bt_{\MC}, P)$ is not necessarily connected. Let $G_1,\dots,G_k$ be the connected components of $DG(\bt_{\MC}, P)$. For each component $G_i$, consider a minimum spanning tree $T_i$ of degree at most five. We show how to extract from $T_i$ a plane matching $M_i$ of proper size and appropriate edge lengths.

\begin{lemma}
 Each component of $DG(\bt_{\MC}, P)$ has an even number of vertices.
\end{lemma}
\begin{proof}
This follows from the facts that $\MC$ is a perfect matching and both end points of each edge in ${\MC}$ belong to the same component of $DG(\bt_{\MC}, P)$.
\end{proof}

\subsection{First Approximation Algorithm}
\label{bm:bottleneck-five-over-two}
In this section we describe the process of computing a plane matching $M$ of size
$|M|\ge\frac{1}{5}n$ with bottleneck length at most $\btopt$. Consider the minimum spanning trees $T_1,\dots,T_k$ of the $k$ components of $DG(\bt_{\MC}, P)$. For $1\le i\le k$, let $P_i$ denote the set of vertices in $T_i$ and $n_i$ denote the number of vertices of $P_i$. Our approximation algorithm runs in two steps:

\begin{paragraph}{Step 1:} 
We start by running algorithm {\scshape PlaneMatching} on each of the point sets $P_i$. Let $M_i$ be the output. Recall that algorithm {\scshape PlaneMatching}, from Section \ref{bm:five-over-two-approximation}, picks a leaf $v'\in T_i'$, corresponding to a vertex $v \in T_i$, matches it to one of its neighboring leaves in $T_i$ and disregards the other edges connected to $v$. According to Lemma \ref{bm:n-minus-one-lemma}, this gives us a plane matching $M_i$ of size at least $\frac{n_i-1}{5}$. However, we are looking for a matching of size at least $\frac{n_i}{5}$. 

 The total number of edges of $T_i$ is $n_i-1$ and in each of the iterations, the algorithm picks one edge out of at most five candidates. If in at least one of the iterations of the while-loop, $v$ has degree at most four (in $T_i$), then in that iteration algorithm {\scshape PlaneMatching} picks one edge out of at most four candidates. Therefore, the size of $M_i$ satisfies 
$$
 |M_i| \ge 1 + \frac{(n_i-1)-4}{5}=\frac{n_i}{5}.
$$
 If in all the iterations of the while-loop, $v$ has degree five, we look at the angles between the consecutive leaves connected to $v$. Recall that in $MST(P)$ all the angles are greater than or equal to $\pi/3$. If in at least one of the iterations, $v$ is connected to two consecutive leaves $u_j$ and $u_{j+1}$ for $1\le j\le 3$, such that $\angle u_jvu_{j+1} = \pi/3$, we change $M_i$ as follow. Remove from $M_i$ the edge incident to $v$ and add to $M_i$ the edges $(u_j,u_{j+1})$ and $(v,u_s)$, where $u_s, s\notin\{j, j+1\}$, is one of the leaves connected to $v$. Clearly $\bigtriangleup vu_ju_{j+1}$ is equilateral and $|u_ju_{j+1}| = |u_jv|=|u_{j+1}v|\le \bt_{\MC}$, and by Lemma \ref{bm:empty-triangle-lemma}, $(u_j, u_{j+1})$ does not cross other edges. In this case, the size of $M_i$ satisfies 
$$
 |M_i| = 2+\frac{(n_i-1)-5}{5}=\frac{n_i+4}{5}\ge\frac{n_i}{5}.
$$

\end{paragraph}

\begin{paragraph}{Step 2:}
In this step we deal with the case that in all the iterations of the while-loop, $v$ has degree five and the angle between any pair of consecutive leaves connected to $v$ is greater than $\pi/3$. Recall that $\MC$ is a perfect matching and both end points of each
edge in $\MC$ belong to the same $T_i$.

\begin{figure}[ht]
  \centering
    \includegraphics[width=0.5\textwidth]{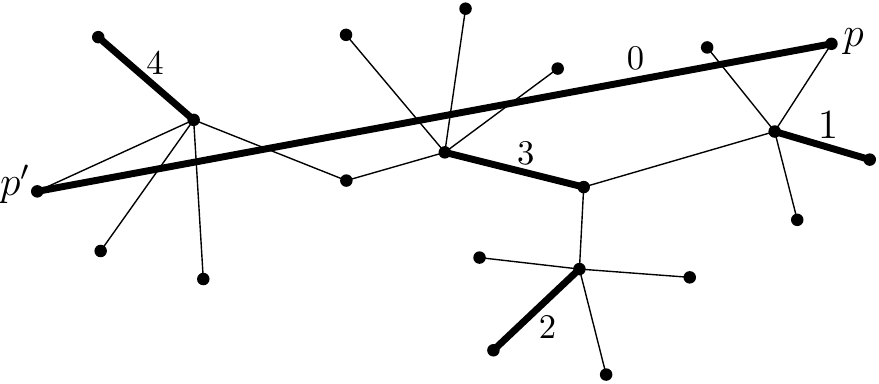}
  \caption{Resulted matching $M_i$ by modified {\scshape PlaneMatching}. The numbers show the order in which the edges (bold edges) are added to $M_i$.}
\label{bm:two-leaves-fig}
\end{figure}

\begin{lemma}
\label{bm:two-leaves-lemma}
In Step 2, at least two leaves of $T_i$ are matched in $\MC$.
\end{lemma}
\begin{proof}
 Let $m_i$ and $m_i'$ denote the number of external (leaves) and internal nodes of $T_i$, respectively. Clearly $m_i'$ is equal to the number of vertices of $T_i'$ and $m_i+m_i'=n_i$. Consider the reverse process in {\scshape PlaneMatching}. Start with a 5-star tree $t_i$, i.e. $t_i=K_{1,5}$, and in each iteration append a new $K_{1,5}$ to $t_i$ until $t_i=T_i$. In the first step $m_i=5$ and $m_i'=1$. In each iteration a leaf of the appended $K_{1,5}$ is identified with a leaf of $t_i$; the resulting vertex becomes an internal node. On the other hand, the \textquotedblleft center\textquotedblright of the new $K_{1,5}$ becomes an internal node of $t_i$ and its other four neighbors become leaves of $t_i$. So in each iteration, the number of leaves $m_i$ increases by three, and the number of internal nodes $m_i'$ increases by two. Hence, in all iterations (including the first step) we have $m_i\ge m'_i+4$.

Again consider $\MC$. In the worst case if all $m_i'$ internal vertices of $T_i$ are matched to leaves, we still have four leaves which have to be matched together.   
\end{proof}

According to Lemma \ref{bm:two-leaves-lemma} there is an edge $(p,p')\in \MC$ where $p$ and $p'$ are leaves in $T_i$. We can find $p$ and $p'$ for all $T_i$'s by checking all the edges of $\MC$ once. We remove all the edges of $M_i$ and initialize $M=\{(p,p')\}$. Again we run a modified version of {\scshape PlaneMatching} in such a way that in each iteration, in line 7 it selects the leaf $u_i$ adjacent to $v$ such that $(v, u_i)$ is not intersected by $(p,p')$. In each iteration $v$ has degree five and is connected to at least four leaf edges with angles greater than $\pi/3$. Thus, $(p,p')$ can intersect at most three of the leaf edges and such kind of $(v, u_i)$ exists. See Figure \ref{bm:two-leaves-fig}. In this case, $M_i$ has size
$$
 |M_i| = 1+\frac{n_i-1}{5}=\frac{n_i+4}{5}\ge\frac{n_i}{5}.
$$

\end{paragraph}

We run the above algorithm on each $T_i$ and for each of them we compute a plane matching $M_i$. The final matching of point set $P$ will be $M=\bigcup_{i=1}^{k} M_i$.

\begin{theorem}
\label{bm:two-over-five-theorem}
Let $P$ be a set of $n$ points in the plane, where $n$ is even, and let $\btopt$ be the minimum bottleneck length of any plane perfect matching of $P$. In $O(n^{1.5}\sqrt{\log n})$ time, a plane matching of $P$ of size at least $\frac{n}{5}$ can be computed, whose bottleneck length is at most $\btopt$.
\end{theorem}
\begin{proof}

{\em Proof of edge length}: Let $\bt_{\MC}$ be the length of the longest edge in $\MC$ and consider a component $G_i$ of $DG(\bt_{\MC}, P)$. All the selected edges in Steps 1 and 2 belong to $T_i$ except $(u_j,u_{j+1})$ and $(p,p')$. $T_i$ is a subgraph of $G_i$, and the edge $(u_j,u_{j+1})$ belongs to $G_i$, and the edge $(p,p')$ belongs to $\MC$ (which belongs to $G_i$ as well). So all the selected edges belong to $G_i$, and $\bt_{M_i}\le \bt_{\MC}$. Since $\bt_{\MC}\le \btopt$, we have $\bt_{M_i}\le \bt_{\MC}\le \btopt$ for all $i$, $1\le i \le k$.

 {\em Proof of planarity}: The edges of $M_i$ belong to the minimum spanning forest of $DG(\bt_{\MC}, P)$ which is plane, except $(u_j,u_{j+1})$ and $(p,p')$. According to Corollary \ref{bm:non-crossing-edges} and Lemma \ref{bm:empty-triangle-lemma} the edge $(u_j,u_{j+1})$ does not cross the edges of the minimum spanning forest. In Step 2 we select edges of $T_i$ in such a way that avoid $(p,p')$. Note that $(p,p')$ belongs to the component $G_i$ and by Corollary \ref{bm:non-crossing-edges} it does not cross any edge of the other components of $DG(\bt_{\MC}, P)$. So $M$ is plane.

 {\em Proof of matching size}: Since $M=M_1\cup \dots \cup M_k$, and for each $1\le i\le k$, $|M_i|\ge \frac{n_i}{5}$, hence
$$
 |M|\ge \sum_{i=1}^{k} |M_i| \ge \sum_{i=1}^{k} \frac{n_i}{5} = \frac{n}{5}.
$$

 {\em Proof of complexity}: The initial matching $\MC$ can be computed in time $O(n^{1.5}\sqrt{\log n})$ by using the algorithm of Chang et al. \cite{Chang1992}. By Lemma \ref{bm:n-minus-one-lemma} algorithm {\scshape PlaneMatching} runs in $O(n \log n)$ time. In Step 1 we spend constant time for checking the angles and the number of leaves connected to $v$ during the while-loop. In Step 2, the matched leaves $p$ and $p'$ can be computed in $O(n)$ time for all $T_i$'s by checking all the edges of $\MC$ before running the algorithm again. So the modified {\scshape PlaneMatching} still runs in $O(n \log n)$ time, and the total running time of our method is $O(n^{1.5}\sqrt{\log n})$.
\end{proof}

Since the running time of the algorithm is bounded by the time of computing the initial bottleneck matching $\MC$, any improvement in computing $\MC$ leads to a faster algorithm for computing a plane matching  $M$. In the next section we improve the running time.

\subsubsection{Improving the Running Time}
In this section we present an algorithm that improves the running time to $O(n\log^2 n)$. We first compute a forest $F$, such that each edge in $F$ is of length at most $\btopt$ and, for each tree $T \in F$, we have a leaf $p \in T$ and a point $p' \in T$ such that $|pp'| \le \btopt$. Once we have this forest $F$, we apply Step 1 and Step 2 on $F$ to obtain the matching $M$ as in the previous section. 

Let $MST(P)$ be a (five-degree) minimum spanning tree of $P$. 
Let $F_{\lambda}=\{ T_1, T_2, \dots, T_k\}$ be the forest obtained from $MST(P)$ by removing all the edges whose length is greater than $\lambda$, i.e., $F_{\lambda} = \{e \in MST(P): |e| \le \lambda\}$. For a point $p \in P$, let $cl(p, P)$ be the point in $P$ that is closest to $p$.

\begin{lemma}
\label{bm:opt-forest}
For all $T \in F_{\btopt}$, it holds that
\begin{enumerate}
	\item[(i)]  the number of points in $T$ is even, and
	\item[(ii)]  for each two leaves $p,q \in T$ that are incident to the same node $v$ in $T$, let $p' = cl(p, P\setminus \{v\})$, let $q' = cl(q, P\setminus \{v\})$, and assume that $|pp'| \le |qq'|$. Then, $\btopt \geq |pp'|$ and $p'$ belongs to $T$.
\end{enumerate}
\end{lemma}
\begin{proof}
(i) Suppose that $T$ has odd number of points. Thus in $M^*$ one of the points in $T$ should be matched to a point in a tree $T' \neq T$ by an edge $e$. Since $e \notin F_{\btopt}$, we have $|e|>\btopt$, which contradicts that $\btopt$ is the minimum bottleneck. (ii) Note that $v$ is the closest point to both $p$ and $q$. In $M^*$, at most one of $p$ and $q$ can be matched to $v$, and the other one must be matched to a point which is at least as far as its second closest point. Thus, $\btopt$ is at least $|pp'|$. The distance between any two trees in $F_{\btopt}$ is greater than $\btopt$. Now if $p'$ is not in $T$, then in any bottleneck perfect matching, either $p$ or $q$ is matched to a point of distance greater than $\btopt$, which contradicts that $\btopt$ is the minimum bottleneck.
\end{proof}

Let $E=(e_1, e_2, \dots ,e_{n-1})$ be the edges of $MST(P)$ in sorted order of their lengths. Our algorithm performs a binary search on $E$, and for each considered edge $e_i$, it uses Algorithm~\ref{bm:proc2} to decide whether $\lambda < \btopt$, where $\lambda=|e_i|$. 
The algorithm constructs the forest $F_{\lambda}$, and for each tree $T$ in $F_{\lambda}$, it picks two leaves $p$ and $q$ from $T$ and finds their second closest points $p'$ and $q'$. Assume w.l.o.g. that $|pp'| \le |qq'|$. Then, the algorithm returns FALSE if $p'$ does not belong to $T$. By Lemma~\ref{bm:opt-forest}, if the algorithm returns FALSE, then we know that $\lambda <\btopt$. 

Let $e_j$ be the shortest edge in $MST(P)$, for which Algorithm~\ref{bm:proc2} does not return FALSE. This means that Algorithm~\ref{bm:proc2} returns FALSE for $e_{j-1}$, and there is a tree $T$ in $F_{|e_{j-1}|}$ and a leaf $p$ in $T$, such that $|pp'|\ge |e_j|$. Thus $|e_j|\leq\btopt$ and for each tree $T$ in the forest $F_{|e_j|}$, we stored a leaf $p$ of $T$ and a point $p' \in T$, such that $|pp'| \le \btopt$. Since each tree in $F_{|e_j|}$ is a subtree of $MST(P)$, $F_{|e_j|}$ is planar and each tree in $F_{|e_j|}$ is of degree at most five. 

\floatname{algorithm}{Algorithm}

\begin{algorithm}                      
\caption{{\scshape CompareToOpt}$(\lambda)$}          
\label{bm:proc2} 
\begin{algorithmic}[1]

  \State compute $F_{\lambda}$ 
  \State {$L\gets$ empty list}
  \For {each $T \in F_{\lambda}$}
	\If {$T$ has an odd number of points} 
		\State return {FALSE} \EndIf
	\If {there exist two leaves $p$ and $q$ incident to a node $v\in T$}
  \State $p' \gets cl(p, P\setminus \{v\} )$
	\State $q' \gets cl(q, P\setminus \{v\} )$
	
  \If {$|pp'| \le |qq'|$}
			\If {$p'$ does not belong to $T$} 
				\State \Return {FALSE}
			\Else
				\State {add the triple $(p,p',T)$ to $L$}  \EndIf
  \Else 
			\If {$q'$ does not belong to $T$} 
				\State \Return {FALSE}
			\Else
				\State {add the triple $(q,q',T)$ to $L$}  \EndIf
  \EndIf \EndIf
	\EndFor
\State \Return $L$
	
\end{algorithmic}
\end{algorithm}

Now we can apply Step 1 and Step 2 on $F_{|e_j|}$ as in the previous section. Note that in Step 2, for each tree $T_i$ we have a pair $(p,p')$ (or $(q,q')$) in the list $L$ which can be matched. In Step 2, in each iteration $v$ has degree five, thus, $p'$ should be a vertex
of degree two or a leaf in $T_i$. If $p'$ is a leaf we run the modified version of {\scshape PlaneMatching} as in the previous section. If $p'$ has degree two, we remove all the edges of $M_i$ and initialize $M_i = {(p, p')}$. Then remove $p'$ from $T_i$ and run {\scshape PlaneMatching} on the resulted subtrees. Finally, set $M=\bigcup_{i=1}^{k}M_i$.

\begin{lemma}
\label{bm:planarity}
 The matching $M$ is planar.
\end{lemma}
\begin{proof}
Consider two edges $e=(p,p')$ and $e'=(q,q')$ in $M$. We distinguish between four cases:
\begin{enumerate}
   \item $e\in F_{|e_j|}$ and $e' \in F_{|e_j|}$. In this case, both $e$ and $e'$ belong to $MST(P)$ and hence they do not cross each other.
  \item $e\notin F_{|e_j|}$ and $e'\in F_{|e_j|}$. If $e$ and $e'$ cross each other, then this contradicts the selection of $(q,q')$ in Step 2 (which prevents $(p,p')$).
  \item $e\in F_{|e_j|}$ and $e'\notin F_{|e_j|}$. It leads to a contradiction as in the previous case.
  \item $e\notin F_{|e_j|}$ and $e'\notin F_{|e_j|}$. If $e$ and $e'$ cross each other, then either $\min\{|pq|,|pq'|\} < |pp'|$ or $\min\{|qp|,|qp'|\} < |qq'|$, which contradicts the selection of $p'$ or $q'$. Note that $p$ cannot be the second closest point to $q$, because $p$ and $q$ are in different trees. 
\end{enumerate} 
\end{proof}

\begin{lemma}
\label{bm:running-time}
 The matching $M$ can be computed in $O(n\log^2 n)$ time.
\end{lemma}
\begin{proof}
Computing $MST(P)$ and sorting its edges take $O(n\log{n})$~\cite{deBerg08}. Since we performed a binary search on the edges of $MST(P)$, we need $\log{n}$ iterations. In each iteration, for an edge $e_i$, we compute the forest $F_{|e_i|}$ in $O(n)$ and the number of the trees in the forest can be $O(n)$ in the worst case. We compute in advance the second order Voronoi diagram of the points together with a corresponding point location data structure, in $O(n\log{n})$~\cite{deBerg08}. For each tree in the forest, we perform a point location query to find the closest points $p'$ and $q'$, which takes $O(\log{n})$ for each query. Therefore the total running time is $O(n\log^2{n})$.
\end{proof}

\begin{theorem}
\label{bm:two-over-five-theorem-2}
Let $P$ be a set of $n$ points in the plane, where $n$ is even, and let $\btopt$ be the minimum bottleneck length of any plane perfect matching of $P$. In $O(n\log^2 n)$ time, a plane matching of $P$ of size at least $\frac{n}{5}$ can be computed, whose bottleneck length is at most $\btopt$.
\end{theorem}

\subsection{Second Approximation Algorithm}
\label{bm:bottleneck-five-over-four}

In this section we present another approximation algorithm which gives a plane matching $M$ of size $|M|\ge\frac{2}{5}n$ with bottleneck length $\bt_M\le(\sqrt{2}+\sqrt{3})\btopt$. Let $DT(P)$ denote the Delaunay triangulation of $P$. Let the edges of $DT(P)$ be, in sorted order of their lengths, $e_1, e_2, \dots$. Initialize a forest $F$ consisting of $n$ tress, each one being a single node for one point of $P$. Run Kruskal's algorithm on the edges of $DT(P)$ and terminate as soon as every tree in $F$ has an even number of nodes. Let $e_l$ be the last edge that is added by Kruskal's algorithm. Observe that $e_l$ is the longest edge in $F$. Denote the trees in $F$ by $T_1 ,\dots, T_k$ and for $1\le i \le k$, let $P_i$ be the vertex set of $T_i$ and let $n_i=|P_i|$. 

\begin{lemma}
\label{bm:longest-edge}
 $\btopt \ge |e_l|.$
\end{lemma}
\begin{proof}
Let $i$ be such that $e_l$ is an edge in $T_i$. Let $T'_i$ and $T''_i$ be the trees obtained by removing $e_l$ from $T_i$. Let $P'_i$ be the vertex set of $T'_i$. Then $|e_l|=\min \{|pq|: p \in P'_i, q\in P\setminus P'_i\}$. Consider the optimal matching $M^*$ with bottleneck length $\btopt$. Since $e_l$ is the last edge added, $P'_i$ has odd size. The matching $M^*$ contains an edge joining a point in $P'_i$ with a point in $P\setminus P'_i$. This edge has length at least $|e_l|$.
\end{proof}

By Lemma \ref{bm:longest-edge} the length of the longest edge in $F$ is at most $\btopt$. For each $T_i\in F$, where $1\le i\le k$, our algorithm will compute a plane matching $M_i$ of $P_i$ of size at least $\frac{2n_i}{5}$ with edges of length at most $(\sqrt{2}+\sqrt{3})\btopt$ and returns $\bigcup_{i=1}^{k} M_i$. To describe the algorithm for tree $T_i$ on vertex set $P_i$, we will write $P$, $T$, $n$, $M$ instead of $P_i$, $T_i$, $n_i$, $M_i$, respectively. Thus, $P$ is a set of $n$ points, where $n$ is even, and $T$ is a minimum spanning tree of $P_i$.

Consider the minimum spanning tree $T$ of $P$ having degree at most five, and let ${T'}$ be the skeleton tree of $T$. Suppose that $T'$ has at least two vertices. We will use the following notation. Let $v'$ be a leaf in $T'$, and let $w'$ be the neighbor of $v'$. Recall that $v'$ and $w'$ are copies of vertices $v$ and $w$ in $T$. In $T$, we consider the clockwise ordering of the neighbors of $v$. Let this ordering be $w, u_1, u_2, \dots, u_k$ for some $1\le k\le 4$. Clearly $u_1,\dots,u_k$ are leaves in $T$. Consider two leaves $u_i$ and $u_j$ where $i<j$. We define $\cw(u_ivu_j)$ as the clockwise angle from $\overline{vu_i}$ to $\overline{vu_j}$. We say that the leaf $v'$ (or its copy $v$) is an {\em anchor} if $k=2$ and $\cw(u_1vu_2)\ge\pi$. See Figure \ref{bm:empty-skeleton-angle-fig}.

\begin{figure}[ht]
  \centering
    \includegraphics[width=0.7\textwidth]{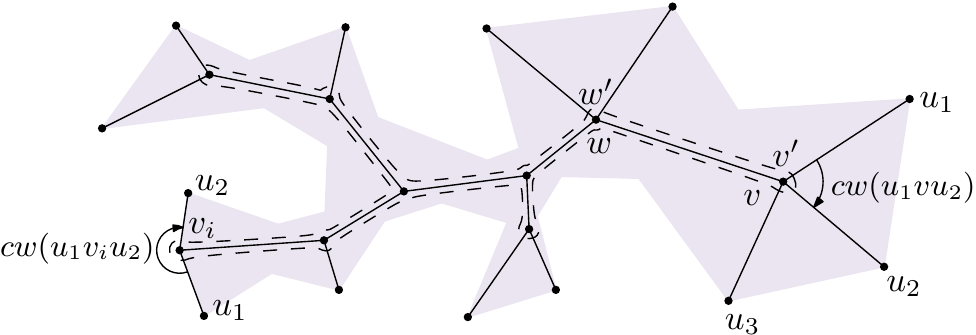}
  \caption{The vertices around $v$ are sorted clockwise, and $v_i$ is an anchor vertex.}
\label{bm:empty-skeleton-angle-fig}
\end{figure}

Now we describe how one can iteratively compute a plane matching of proper size with bounded-length edges from $T$. We start with an empty matching $M$. Each iteration consists of two steps, during which we add edges to $M$. As we prove later, the output is a plane matching of $P$ of size at least $\frac{2}{5}n$ with bottleneck at most $(\sqrt{2}+\sqrt{3})\btopt$.

\subsubsection{Step 1}

We keep applying the following process as long as $T$ has more than six vertices and $T'$ has some non-anchor leaf. Note that $T'$ has at least two vertices.
 
Take a non-anchor leaf $v'$ in $T'$ and according to the number $k$ of leaves connected to $v$ in $T$ do the following:
 \begin{description}
  \item[{\em k} = 1] add $(v, u_1)$ to $M$, and set $T=T\setminus\{v, u_1\}$. 
  \item[{\em k} = 2] since $v'$ is not an anchor, $\cw(u_1vu_2)< \pi$. By Lemma \ref{bm:empty-triangle-lemma} the triangle $\bigtriangleup u_1vu_2$ is empty. We add $(u_1, u_2)$ to $M$, and set $T=T\setminus\{u_1, u_2\}$.
  \item[{\em k} = 3] in this case $v$ has degree four and at least one of $\cw(u_1vu_2)$ and $\cw(u_2vu_3)$ is less than $\pi$. W.l.o.g. suppose that $\cw(u_1vu_2)< \pi$. By Lemma \ref{bm:empty-triangle-lemma} the triangle $\bigtriangleup u_1vu_2$ is empty. Add $(u_1, u_2)$ to $M$ and set $T=T\setminus\{u_1, u_2\}$. 
  \item[{\em k} = 4] this case is handled similarly as the case $k=3$.
 \end{description}

At the end of Step 1, $T$ has at most six vertices or all the leaves of $T'$ are anchors. In the former case, we add edges to $M$ as will be described in Section \ref{bm:base-cases} and after which the algorithm terminates. In the latter case we go to Step 2.

\subsubsection{Step 2}

In this step we deal with the case that $T$ has more than six vertices and all the leaves of $T'$ are anchors. We define the {\em second level skeleton tree} ${T''}$ to be the skeleton tree of $T'$. In other words, $T''$ is the tree which is obtained from $T'$ by removing all the leaves. For clarity we use $w$ to refer to a leaf of $T''$, and we use $w$, $w'$, and $w''$, respectively, to refer to the copies of vertex $w$ in $T$, $T'$, and $T''$. For now suppose that $T''$ has at least two vertices. Consider a leaf $w''$ and its neighbor $y''$ in $T''$. Note that in $T$, $w$ is connected to $y$, to at least one anchor, and possibly to some leaves of $T$. After Step 1, the copy of $w''$ in $T'$, i.e. $w'$, is connected to anchors $v'_1,\dots,v'_k$ in $T'$ (or $v_1,\dots,v_k$ in $T$) for some $1 \le k \le 4$, and connected to at most $4-k$ leaves of $T$. In $T$, we consider the clockwise ordering of the non-leaf neighbors of $w$. Let this ordering be $y, v_1, v_2, \dots, v_k$. We denote the pair of leaves connected to anchor $v_i$ by $a_i$ and $b_i$ in clockwise order around $v_i$; see Figure \ref{bm:anchor-orientation}. 

\begin{figure}[ht]
  \centering
    \includegraphics[width=0.52\textwidth]{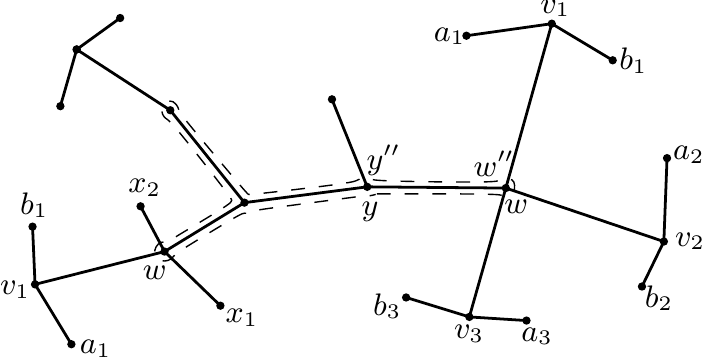}
  \caption{Second level skeleton tree $T''$ is surrounded by dashed line. $v_i$'s are ordered clockwise around leaf vertex $w''$, as well as $x_i$'s. $a_i$ and $b_i$ are ordered clockwise around $v_i$.}
\label{bm:anchor-orientation}
\end{figure}

In this step we pick an arbitrary leaf $w''\in T''$ and according to the number of anchors incident to $w''$, i.e. $k$, we 
add edges to $M$. Since $1 \le k \le 4$, four cases occur and we discuss each case separately. Before that, we state some lemmas.

\begin{figure}[ht]
  \centering
    \includegraphics[width=0.32\textwidth]{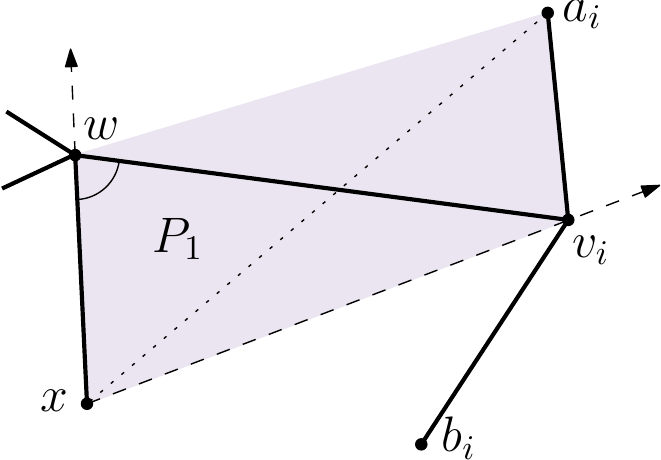}
  \caption{Illustrating Lemma \ref{bm:empty-quadrilateral-lemma}.}
\label{bm:empty-quadrilateral-fig}
\end{figure}

\begin{lemma}
\label{bm:empty-quadrilateral-lemma}
Let $w''$ be a leaf in $T''$. Consider an anchor $v_i$ which is adjacent to $w$ in $T$. For any neighbor $x$ of $w$ for which $x\neq v_i$, if $\cw(v_iwx)\le\pi/2$ (resp. $\cw(xwv_i)\le\pi/2$), the polygon $P_1=\{v_i,x,w,a_i\}$ (resp. $P_2 = \{v_i,b_i,w,x\}$) is convex and empty.
\end{lemma}
\begin{proof}
We prove the case when $\cw(v_iwx)\le\pi/2$; see Figure \ref{bm:empty-quadrilateral-fig}. The proof for the second case is symmetric. To prove the convexity of $P_1$ we show that the diagonals $\overline{v_iw}$ and $\overline{a_ix}$ of $P_1$ intersect each other. To show the intersection we argue that $a_i$ lies to the left of $\overrightarrow{xv_i}$ and to the right of $\overrightarrow{xw}$.

Consider $\overrightarrow{xv_i}$. According to Lemma \ref{bm:empty-triangle-lemma}, triangle $\bigtriangleup v_ixw$ is empty so $b_i$ lies to the right  of $\overrightarrow{xv_i}$. On the other hand, $v_i$ is an anchor, so $\cw(a_iv_ib_i) \ge \pi$, and hence $a_i$ lies to the left of $\overrightarrow{xv_i}$.
Now consider $\overrightarrow{xw}$. For the sake of contradiction, suppose that $a_i$ is to the left of $\overrightarrow{xw}$. Since $\cw(v_iwx)\le\pi/2$, the angle $\cw(a_iwv_i)$ is greater than $\pi/2$. This means that $\overline{v_ia_i}$ is the largest side of $\bigtriangleup a_iwv_i$, which contradicts that $\overline{v_ia_i}$ is an edge of $MST(P)$. So $a_i$ lies to the right of $\overrightarrow{xw}$. Therefore, $\overline{v_iw}$ intersects $\overline{a_ix}$ and $P_1$ is convex. $P_1$ is empty because by Lemma \ref{bm:empty-triangle-lemma}, the triangles $\bigtriangleup v_iwx$ and $\bigtriangleup v_iwa_i$ are empty. 
\end{proof}

\begin{figure}[ht]
  \centering
    \includegraphics[width=0.4\textwidth]{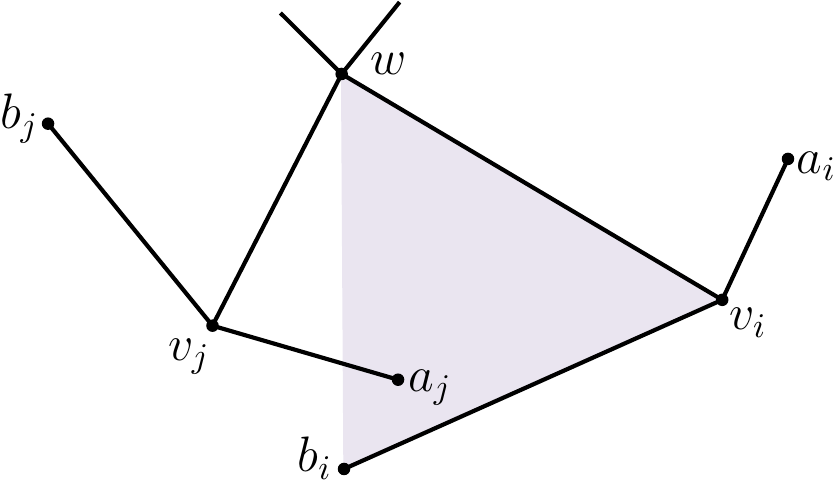}
  \caption{Proof of Lemma \ref{bm:clockwise-order-lemma}.}
\label{bm:clockwise-order-fig}
\end{figure}

\begin{lemma}
\label{bm:clockwise-order-lemma}
Let $w''$ be a leaf in $T''$ and consider the clockwise sequence $v_1,\dots,v_k$ of anchors that are incident on $w$. The sequence of vertices $a_1, v_1, \allowbreak b_1,\allowbreak  \dots,a_k,v_k,b_k$ are angularly sorted in clockwise order around $w$. 
\end{lemma}
\begin{proof}
 Using contradiction, consider two vertices $v_i$ and $v_j$, and assume that $v_i$ comes before $v_j$ and $a_j$ comes before $b_i$ in the clockwise order; see Figure \ref{bm:clockwise-order-fig}. Either $a_j$ is in $\bigtriangleup wv_ib_i$ or $b_i$ is in $\bigtriangleup wv_ja_j$. However, by Lemma \ref{bm:empty-triangle-lemma}, neither of these two cases can happen.
\end{proof}

\begin{figure}[ht]
  \centering
\setlength{\tabcolsep}{0in}
  $\begin{tabular}{cc}
  \multicolumn{1}{m{.5\columnwidth}}{\centering\includegraphics[width=.35\columnwidth]{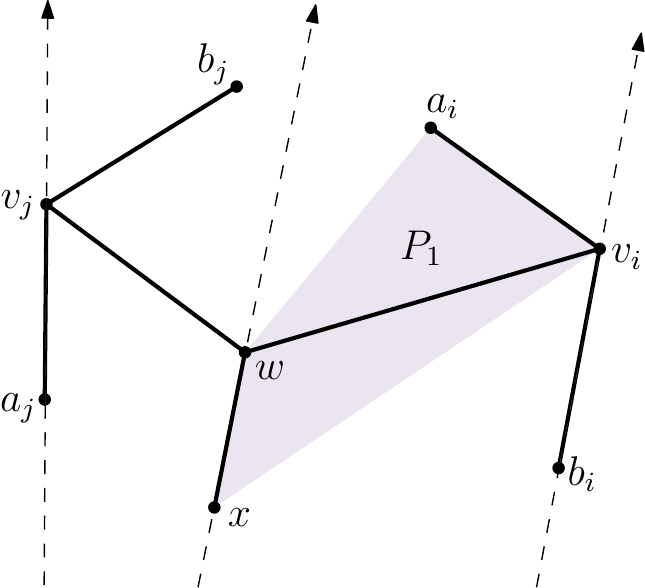}}
  &\multicolumn{1}{m{.5\columnwidth}}{\centering\includegraphics[width=.35\columnwidth]{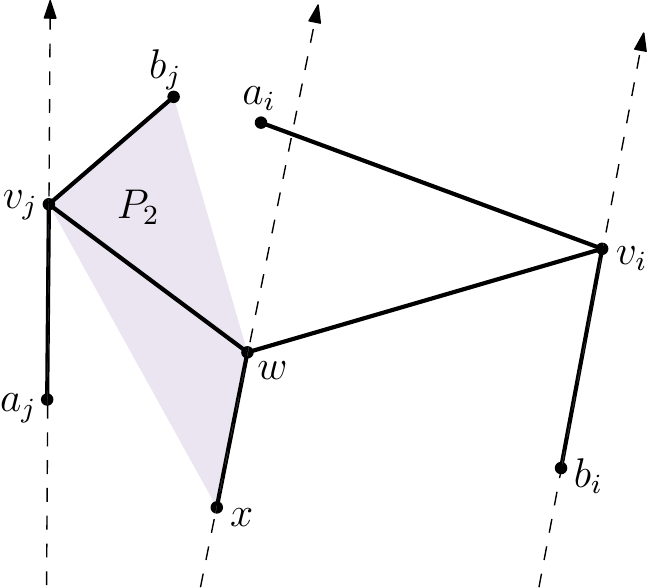}}
  \\
  (a) & (b)
  \end{tabular}$
  \caption{Illustrating Lemma \ref{bm:anchor2-vertex1-lemma}.}
\label{bm:anchor2-vertex1-1-fig}
\end{figure}

\begin{lemma}
\label{bm:anchor2-vertex1-lemma}
Let $w''$ be a leaf in $T''$ and consider the clockwise sequence $v_1,\dots,v_k$ of anchors that are adjacent to $w$. Let $1\le i < j \le k$ and let $x$ be a neighbor of $w$ for which $x\neq v_i$ and $x\neq v_j$. 
\begin{enumerate}
 \item If $x$ is between $v_i$ and $v_j$ in the clockwise order:
    \begin{enumerate}
      \item if $a_i$ is to the right of $\overrightarrow{xw}$, then $P_1=\{x, w, a_i, v_i\}$ is convex and empty.
      \item if $a_i$ is not to the right of $\overrightarrow{xw}$, then $P_2=\{w, x,v_j, b_j\}$ is convex and empty.
    \end{enumerate}
  \item If $v_j$ is between $v_i$ and $x$, or $v_i$ is between $x$ and $v_j$ in the clockwise order:
    \begin{enumerate}
      \item if $b_i$ is to the left of $\overrightarrow{xw}$, then $P_3=\{w, x, v_i, b_i\}$ is convex and empty.
      \item if $b_i$ is not to the left of $\overrightarrow{xw}$, then $P_4=\{x, w, a_j, v_j\}$ is convex and empty.
    \end{enumerate}
\end{enumerate}
\end{lemma}

\begin{proof}
We only prove the first case, the proof for the second case is symmetric. Thus, we assume that $x$ is between $v_i$ and $v_j$ in the clockwise order. First assume that $a_i$ is to the right of $\overrightarrow{xw}$. See Figure \ref{bm:anchor2-vertex1-1-fig}(a). Consider $\overrightarrow{b_iv_i}$. Since $v_i$ is an anchor, $a_i$ cannot be to the right of $\overrightarrow{b_iv_i}$, and according to Lemma \ref{bm:empty-triangle-lemma}, $x$ cannot be to the right of $\overrightarrow{b_iv_i}$. For the same reasons, both the vertices $b_j$ and $x$ cannot be to the left of $\overrightarrow{a_jv_j}$. Now consider $\overrightarrow{xw}$. By assumption, $a_i$ is to the right of $\overrightarrow{xw}$. Therefore $\overline{xa_i}$ intersects $\overline{wv_i}$ and hence $P_1$ is convex.

Now assume that $a_i$ is not to the right of $\overrightarrow{xw}$; see Figure \ref{bm:anchor2-vertex1-1-fig}(b). By Lemma \ref{bm:clockwise-order-lemma}, $b_j$ is to the left of $\overrightarrow{xw}$. Therefore, $\overline{xb_j}$ intersects $\overline{wv_j}$ and hence $P_2$ is convex. 
The emptiness of the polygons follows directly from Lemma \ref{bm:empty-triangle-lemma}.
\end{proof}

\begin{figure}[ht]
  \centering
    \includegraphics[width=0.37\textwidth]{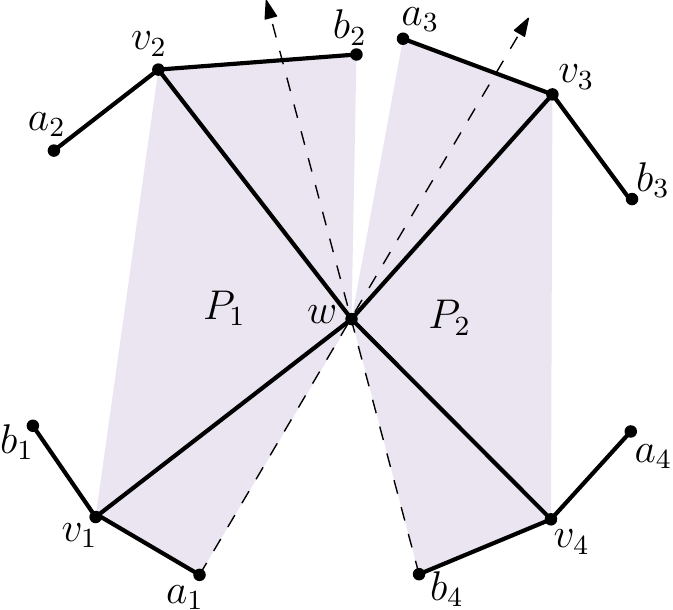}
  \caption{Illustrating Lemma \ref{bm:separation-lemma}.}
\label{bm:anchor4-5-fig}
\end{figure}

\begin{lemma}
\label{bm:separation-lemma}
Let $w''$ be a leaf in $T''$ and consider the clockwise sequence $v_1,\dots,v_k$ of anchors that are incident on $w$. If $4\le k \le 5$, then at least one of the pentagons $P_1=\{a_1,v_1,v_2,b_2,w\}$ and $P_2=\{a_{k-1},v_{k-1},v_k,b_k,w\}$ is convex and empty.
\end{lemma}
\begin{proof}
We prove the case when $k=4$; the proof for $k=5$ is analogous. By Lemma \ref{bm:clockwise-order-lemma}, $a_1$ comes before $b_4$ in the clockwise order. Consider $\overrightarrow{a_1w}$; see Figure \ref{bm:anchor4-5-fig}. Let $l(a_1, b_2)$ denote the line through $a_1$ and $b_2$. If $b_2$ is to the left of $\overrightarrow{a_1w}$, then $l(a_1, b_2)$ separates $w$ from $v_1$ and $v_2$. Otherwise, $b_2$ is to the right of $\overrightarrow{a_1w}$. Now consider $\overrightarrow{b_4w}$. If $a_3$ is to the right of $\overrightarrow{b_4w}$, then $l(a_3, b_4)$ separates $w$ from $v_3$ and $v_4$. The remaining case, i.e., when $a_3$ is to the left of $\overrightarrow{b_4w}$, cannot happen by Lemma \ref{bm:clockwise-order-lemma} (because, otherwise, $a_3$ would come before $b_2$ in the clockwise order).

Thus, we have shown that (i) $l(a_1, b_2)$ separates $w$ from $v_1$ and $v_2$ or (ii) $l(a_3, b_4)$ separates $w$ from $v_3$ and $v_4$. Assume w.l.o.g. that (i) holds. Now to prove the convexity of $P_1$ we show that all internal angles of $P_1$ are less than $\pi$. Since $v_1$ is an anchor, $\cw(b_1v_1a_1)\le \pi$. By Lemma~\ref{bm:empty-triangle-lemma}, $b_1$ is to the left of of $\overrightarrow{v_1v_2}$. Therefore, $\cw(v_2v_1a_1) < \cw(b_1v_1a_1) \le \pi$. On the other hand $cw(wv_1a_1)\ge \pi/3$, so in $\bigtriangleup a_1v_1w$, $\cw(v_1a_1w) \le 2\pi/3$. By a similar analysis $\cw(b_2v_2v_1)$ and $\cw(wb_2v_2)$ are less than $\pi$. In addition, $\cw(a_1wb_2)$ in $\bigtriangleup a_1wb_2$ is less than $\pi$. Thus, $P_1$ is convex. Its emptiness is assured from the emptiness of the triangles $\bigtriangleup v_2wb_2$, $\bigtriangleup v_1wa_1$ and $\bigtriangleup v_1wv_2$.
\end{proof}

Now we are ready to present the details of Step 2. Recall that $T$ has more than six vertices and all leaves of $T'$ are anchors. In this case $T''$ has at least one vertex. If $T''$ has exactly one vertex, we add edges to $M$ as will be described in Section \ref{bm:base-cases} and after which the algorithm terminates. Assume that $T''$ has at least two vertices. Pick a leaf $w''$ in $T''$. As before, let $v_1,\dots,v_k$ where $1\le k \le 4$ be the clockwise order of the anchors connected to $w$. Let $x_1,\dots,x_{\ell}$ where $\ell \le 4-k\le3$ be the clockwise order of the leaves of $T$ connected to $w$; see Figure \ref{bm:anchor-orientation}. This means that $deg(w)=1+k+\ell$, where $k+\ell \le 4$. Now we describe different configurations that may appear at $w$, according to $k$ and $\ell$.
 \begin{description}
  \item[Case 1:] assume $k=1$. If $\ell=0$ or $1$, add $(a_1,v_1)$ and $(b_1,w)$ to $M$ and set $T=T\setminus\{a_1,\allowbreak b_1,\allowbreak v_1,\allowbreak w\}$. If $\ell=1$ remove $x_1$ from $T$ as well. If $\ell=2$, consider $\alpha_i$ as the acute angle between segments $\overline{wv_1}$ and $\overline{wx_i}$. W.l.o.g. assume $\alpha_1 = \min(\alpha_1,\alpha_2)$. Two cases may arise: (i) $\alpha_1\le \pi/2$, (ii) $\alpha_1> \pi/2$. If (i) holds, w.l.o.g. assume that $x_1$ comes before $v_1$ in clockwise order around $w$. According to Lemma \ref{bm:empty-quadrilateral-lemma} polygon $P=\{v_1, b_1, w, x_1\}$ is convex and empty. So we add $(v_1,a_1)$, $(b_1,x_1)$ and $(x_2,w)$ to $M$. Other cases are handled in a similar way. If (ii) holds, according to Lemma \ref{bm:empty-triangle-lemma} triangle $\bigtriangleup x_1wx_2$ is empty. So we add $(v_1, a_1)$, $(b_1,w)$ and $(x_1,x_2)$ to $M$. In both cases set $T=T\setminus\{a_1,b_1,v_1, x_1, x_2, w\}$. If $\ell=3$, remove $x_3$ from $T$ and handle the rest as $\ell=2$.
  \item[Case 2:] assume $k=2$. If $\ell = 0$, we add $(v_1, a_1)$, $(b_1, w)$, and $(v_2, a_2)$ to $M$. If $\ell = 1$ suppose that $x_1$ comes before $v_1$ in clockwise ordering. According to Lemma \ref{bm:anchor2-vertex1-lemma} one of polygons $P_3$ and $P_4$ is empty; suppose it be $P_3=\{w,x_1,v_1,b_1\}$ (where $i=1$ and $j=2$ in Lemma \ref{bm:anchor2-vertex1-lemma}). Thus we add $(v_1, a_1)$, $(b_1,x_1)$, $(v_2,a_2)$ and $(b_2,w)$ to $M$. In both cases set $T=T\setminus\{a_1,\allowbreak b_1,\allowbreak v_1, \allowbreak a_2,b_2,v_2,\allowbreak w\}$ and if $\ell = 1$ remove $x_1$ from $T$ as well. If $\ell = 2$, remove $x_2$ from $T$ and handle the rest as $\ell=1$.
  \item[Case 3:] assume $k=3$. If $\ell = 0$ then set $M=M\cup\{(v_1,a_1), (b_1,w), \allowbreak (v_2,a_2), (v_3,a_3)\}$. If $\ell=1$, consider $\beta_i$ as the acute angle between segments $\overline{wx_1}$ and $\overline{wv_i}$. W.l.o.g. assume $\beta_2$ has minimum value among all $\beta_i$'s and $x_1$ comes after $v_2$ in clockwise order. According to Lemma \ref{bm:anchor2-vertex1-lemma} one of polygons $P_1$ and $P_2$ is empty; suppose it be $P_1=\{w,x_1,v_2,a_2\}$ (where $i=2$ and $j=3$ in Lemma \ref{bm:anchor2-vertex1-lemma}). Thus we set $M=M\cup\{(v_2, b_2), (a_2,x_1),\allowbreak (v_1,a_1), (b_1,w),\allowbreak  (v_3,a_3)\}$. In both cases set $T = T\setminus\{a_1,b_1,\allowbreak v_1, a_2,\allowbreak b_2,v_2, \allowbreak a_3,b_3,v_3, w\}$ and if $\ell = 1$ remove $x_1$ from $T$ as well. Other cases can be handled similarly.

\begin{figure}[ht]
  \centering
    \includegraphics[width=0.37\textwidth]{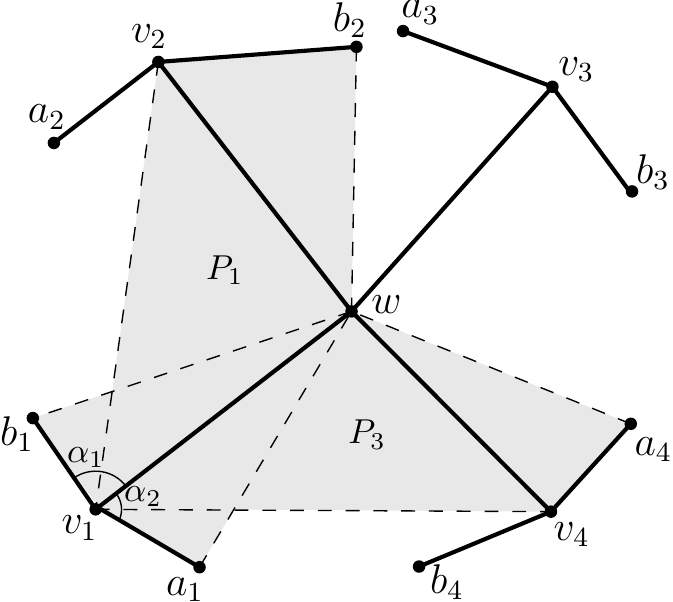}
  \caption{The vertex $v_1$ is shared between $P_1$ and $P_3$.}
\label{bm:case4-fig}
\end{figure}

  \item[Case 4:] assume $k=4$. According to Lemma \ref{bm:separation-lemma} one of $P_1=\{a_1,v_1,\allowbreak v_2,\allowbreak b_2,w\}$ and $P_2=\{a_3,v_3,v_4,b_4,w\}$ is convex and empty. Again by Lemma \ref{bm:separation-lemma} one of $P_3=\{a_4,v_4,v_1,b_1,w\}$ and $P_4=\{a_2,v_2,\allowbreak v_3,\allowbreak b_3,w\}$ is also convex and empty. Without loss of generality assume that $P_1$ and $P_3$ are empty. See Figure \ref{bm:case4-fig}. Clearly, these two polygons share a vertex ($v_1$ in Figure \ref{bm:case4-fig}). Let $\alpha_1=\cw(b_1,v_1,w)$ which is contained in $P_3$ and $\alpha_2=\cw(w,v_1,a_1)$ which is contained in $P_1$. We pick one of the polygons $P_1$ and $P_3$ which minimizes $\alpha_i, i=1,2$. Let $P_1$ be that polygon. So we set $M = M \cup  \allowbreak \{(v_1,b_1) , \allowbreak (v_2,a_2) ,\allowbreak  (a_1,b_2),\allowbreak (v_3, a_3),\allowbreak  (b_3,w),\allowbreak  (v_4,a_4)\}$ and set $T = T\setminus\{a_1,b_1,\allowbreak v_1,\allowbreak  a_2,\allowbreak b_2,v_2, \allowbreak  a_3, b_3,v_3,\allowbreak a_4,\allowbreak b_4\allowbreak,v_4,w \}$. 
 \end{description}

This concludes Step 2. Go back to Step 1.

\subsubsection{Base Cases}
\label{bm:base-cases}

In this section we describe the base cases of our algorithm. As mentioned in Steps 1 and 2, we may have two base cases: (a) $T$ has at most $t\le 6$ vertices, (b) $T''$ has only one vertex.

\begin{paragraph}{(a) {\em t}${\le 6}$} Suppose that $T$ has at most six vertices.
\begin{description}
  \item [{\em t} = 2] it can happen only if $t=n=2$, and we add the only edge to $M$.
  \item [{\em t} = 4, 5] in this case we match four vertices. If $t=4$, $T$ could be a star or a path of length three, and in both cases we match all the vertices. If $t=5$, remove one of the leaves and match other four vertices.
  \item [{\em t} = 6] in this case we match all the vertices. If one of the leaves connected to a vertex of degree two, we match those two vertices and handle the rest as the case when $t=4$, otherwise, each leaf of $T$ is connected to a vertex of degree more than two, and hence $T'$ has at most two vertices. Figure \ref{bm:n-small}(a) shows the solution for the case when $T'$ has only one vertex and $T$ is a star; note that at least two angles are less than $\pi$. 
  Now consider the case when $T'$ has two vertices, $v_1$ and $v_2$, which have degree three in $T$. Figure \ref{bm:n-small}(b) shows the solution for the case when neither $v_1$ nor $v_2$ is an anchor. 
  Figure \ref{bm:n-small}(c) shows the solution for the case when $v_1$ is an anchor but $v_2$ is not.
  Figure \ref{bm:n-small}(d) shows the solution for the case when both $v_1$ and $v_2$ are anchors. Since $v_2$ is an anchor in Figure \ref{bm:n-small}(d), at least one of $\cw(b_2v_2v_1)$ and $\cw(v_1v_2a_2)$ is less than or equal to $\pi/2$. W.l.o.g. assume $\cw(v_1v_2a_2)\le\pi/2$. By Lemma \ref{bm:empty-quadrilateral-lemma} polygon $P=\{v_1,a_2,v_2,a_1\}$ is convex and empty. We add $(v_1,b_1)$, $(v_2,b_2)$, and $(a_1,a_2)$ to $M$.

  \item [{\em t} = 1, 3] this case could not happen. Initially $t=n$ is even. Consider Step 1; before each iteration $t$ is bigger than six and during the iteration two vertices are removed from $T$. So, at the end of Step 1, $t$ is at least five. Now consider Step 2; before each iteration $T''$ has at least two vertices and during the iteration at most one vertex is removed from $T''$. So, at the end of Step 2, $T''$ has at least one vertex that is connected to at least one anchor. This means that $t$ is at least four. Thus, $t$ could never be one or three before and during the execution of the algorithm.

\begin{figure}[ht]
  \centering
\setlength{\tabcolsep}{0in}
  $\begin{tabular}{cccc}
  \multicolumn{1}{m{.25\columnwidth}}{\centering\includegraphics[width=.21\columnwidth]{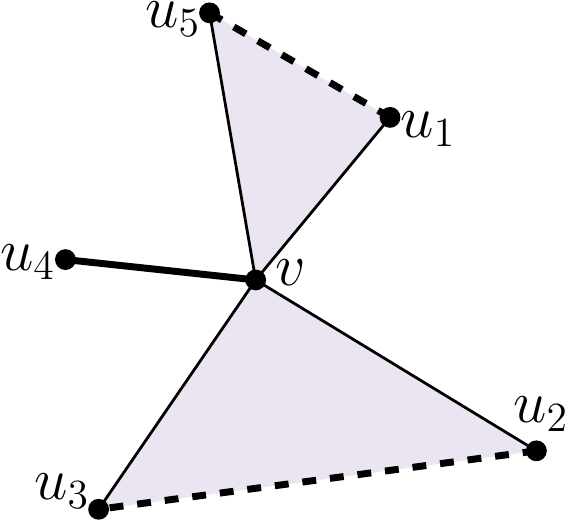}}
  &\multicolumn{1}{m{.25\columnwidth}}{\centering\includegraphics[width=.23\columnwidth]{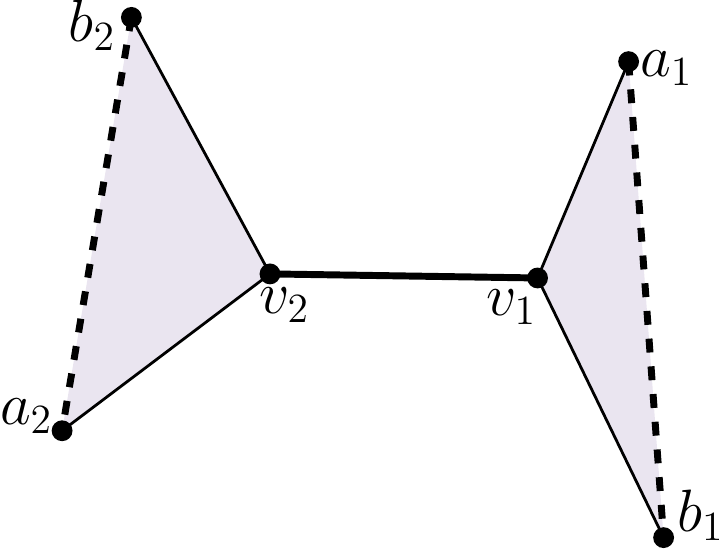}}
  &\multicolumn{1}{m{.25\columnwidth}}{\centering\includegraphics[width=.23\columnwidth]{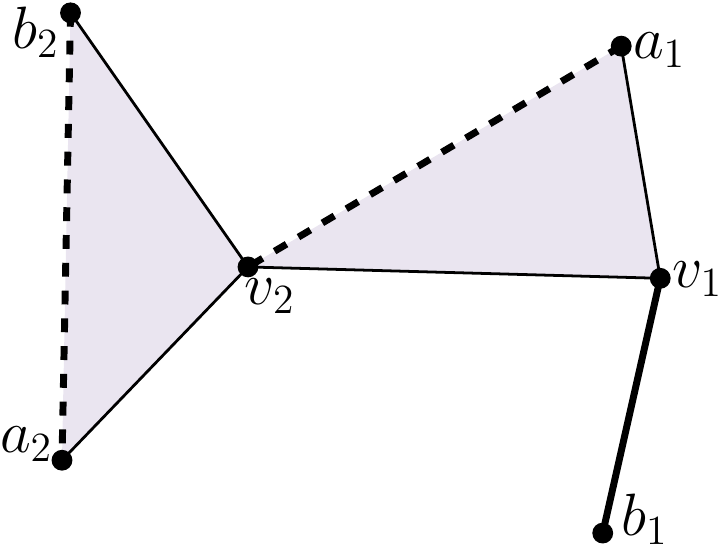}}
  &\multicolumn{1}{m{.25\columnwidth}}{\centering\includegraphics[width=.23\columnwidth]{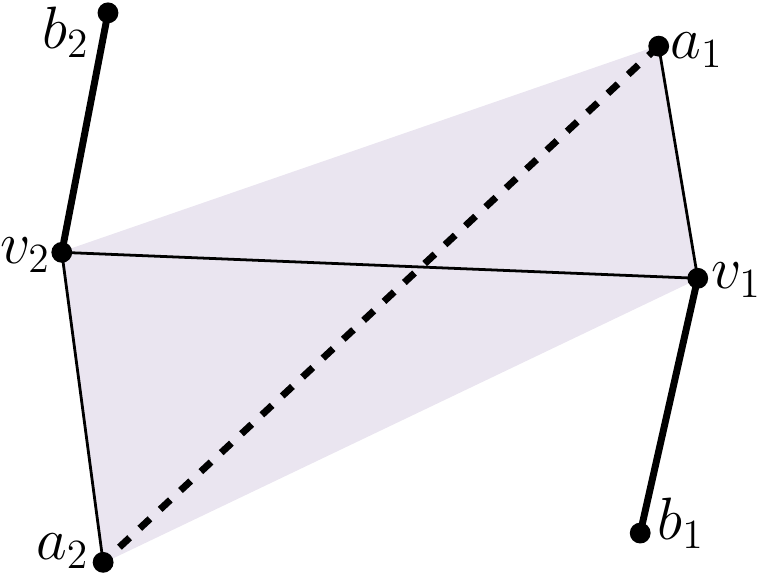}}\\
  (a) & (b)&(c)&(d)
  \end{tabular}$
  \caption{The bold (solid and dashed) edges are added to $M$ and all vertices are matched. (a) a star, (b) no anchor, (c) one anchor, and (d) two anchors.}
\label{bm:n-small}
\end{figure} 

\end{description}
\end{paragraph}

\begin{paragraph}{(b) $T''$ has one vertex}
In this case, the only vertex $w''\in T''$ is connected to at least two anchors, otherwise $w''$ would have been matched in Step 1. So we consider different cases when $w$ is connected to $k$, $2\le k \le 5$ anchors and $\ell \le 5-k$ leaves of $T$:
\begin{description}
 \item[{\em k} = 2] if $\ell=0,1,2$ we handle it as Case 2 in Step 2. If $\ell=3$, at least two leaves are consecutive, say $x_1$ and $x_2$. Since $cw(x_1wx_2) <\pi$ we add $(x_1,x_2)$ to $M$ and handle the rest like the case when $\ell=1$. 
\item[{\em k} = 3] if $\ell=2$ remove $x_2$ from $T$. Handle the rest as Case 3 in Step 2.
\item[{\em k} = 4] if $\ell=1$ remove $x_1$ from $T$. Handle the rest as Case 4 in Step 2. 
\item[{\em k} = 5] add $(v_5, a_5)$ to $M$, remove $a_5$, $b_5$, $v_5$ from $T$, and handle the rest as Case 4 in Step 2.
\end{description}
\end{paragraph}

This concludes the algorithm.

\begin{lemma}
\label{bm:convex-disjoint-lemma}
The convex empty regions that are considered in different iterations of the algorithm, do not overlap.
\end{lemma}
\begin{proof}
In Step 1, Step 2 and the base cases, we used three types of convex empty regions; see Figure \ref{bm:convex-disjoint-fig}.
Using contradiction, suppose that two convex regions $P_1$ and $P_2$ overlap. Since the regions are empty, no vertex of $P_1$ is in the interior of $P_2$ and vice versa. Then, one of the edges in $MST(P)$ that is shared with $P_1$ intersects some edge in $MST(P)$ that is shared with $P_2$, which is a contradiction.   
\end{proof}

\begin{figure}[ht]
  \centering
\setlength{\tabcolsep}{0in}
  $\begin{tabular}{ccc}
  \multicolumn{1}{m{.33\columnwidth}}{\centering\includegraphics[width=.2\columnwidth]{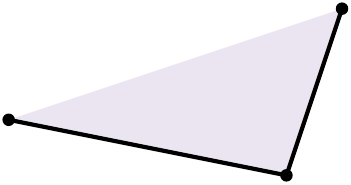}}
  &\multicolumn{1}{m{.33\columnwidth}}{\centering\includegraphics[width=.25\columnwidth]{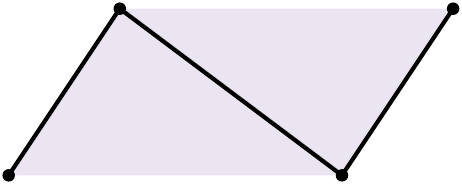}}
  &\multicolumn{1}{m{.33\columnwidth}}{\centering\includegraphics[width=.3\columnwidth]{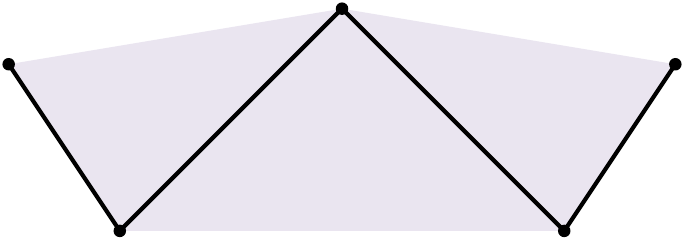}}
  \end{tabular}$
  \caption{Empty convex regions. Bold edges belong to $MST(P)$.}
\label{bm:convex-disjoint-fig}
\end{figure}

\begin{theorem}
Let $P$ be a set of $n$ points in the plane, where $n$ is even, and let $\btopt$ be the minimum bottleneck length of any plane perfect matching of $P$. In $O(n\log n)$ time, a plane matching of $P$ of size at least $\frac{2n}{5}$ can be computed, whose bottleneck length is at most $(\sqrt{2}+\sqrt{3})\btopt$.
\end{theorem}
\begin{proof}

 {\em Proof of planarity}: In each iteration, in Step 1, Step 2, and in the base cases, the edges added to $M$ are edges of $MST(P)$ or edges inside convex empty regions. By Lemma \ref{bm:convex-disjoint-lemma} the convex empty regions in each iteration will not be intersected by the convex empty regions in the next iterations. Therefore, the edges of $M$ do not intersect each other and $M$ is plane.

{\em Proof of matching size}: In Step 1, in each iteration, all the vertices which are excluded from $T$ are matched. In Step 2, when $k=\ell=1$ we match four vertices out of five, and when $k=3, \ell =0$ we match eight vertices out of ten. In base case (a) when $t=5$ we match four vertices out of five. In base case (b) when $k=3, \ell =0$ we match eight vertices out of ten. In all other cases of Step 2 and the base cases, the ratio of matched vertices is more than $4/5$. Thus, in each iteration at least $4/5$ of the vertices removed from $T$ are matched and hence $|M_i|\ge\frac{2n_i}{5}$. Therefore, 
$$
 |M|= \sum_{i=1}^{k} |M_i| \ge \sum_{i=1}^{k} \frac{2n_i}{5} = \frac{2n}{5}.
$$

{\em Proof of edge length}: By Lemma \ref{bm:longest-edge} the length of edges of $T$ is at most $\btopt$. Consider an edge $e\in M$ and the path $\delta$ between its end points in $T$. If $e$ is added in Step 1, then $|e|\le 2\btopt$ because $\delta$ has at most two edges. If $e$ is added in Step 2, $\delta$ has at most three edges ($|e|\le 3\btopt$) except in Case 4. In this case we look at $\delta$ in more detail. We consider the worst case when all the edges of $\delta$ have maximum possible length $\btopt$ and the angles between the edges are as big as possible; see Figure \ref{bm:path-fig}. Consider the edge $e=(a_1,b_2)$ added to $M$ in Case 4. Since $v_2$ is an anchor and $\cw(wv_2a_2)\ge \pi/3$, the angle $\cw(b_2v_2w)\le 2\pi/3$. As our choice between $P_1$ and $P_3$ in Case 4, $\cw(w,v_1,a_1)\le\pi/2$. Recall that $e$ avoids $w$, and hence $|e|\le(\sqrt{2}+\sqrt{3})\btopt$. The analysis for the base cases is similar.

\begin{figure}[ht]
  \centering
    \includegraphics[width=0.52\textwidth]{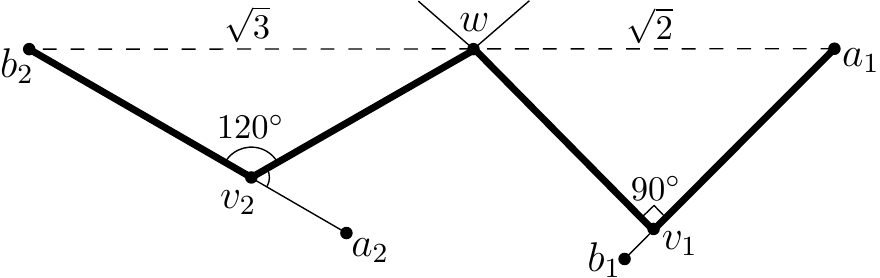}
  \caption{Path $\delta$ (in bold) with four edges of length $\btopt$ between end points of edge $(a_1,b_2)$.}
\label{bm:path-fig}
\end{figure}

{\em Proof of complexity}: The Delaunay triangulation of $P$ can be computed in $O(n \log n)$ time. Using Kruskal's algorithm, the forest $F$ of even components can be computed in $O(n\log n)$ time. 
In Step 1 (resp. Step 2) in each iteration, we pick a leaf of $T'$ (resp. $T''$) and according to the number of leaves (resp. anchors) connected to it we add some edges to $M$. Note that in each iteration we can update $T$, $T'$ and $T''$ by only checking the two hop neighborhood of selected leaves. Since the two hop neighborhood is of constant size, we can update the trees in $O(1)$ time in each iteration. 
Thus, the total running time of Step 1, Step 2, and the base cases is $O(n)$ and the total running time of the algorithm is $O(n\log n)$.
\end{proof}

\section{Conclusions}
\label{bm:conclusion}
We considered the NP-hard problem of computing a bottleneck plane perfect matching of a point set. Abu-Affash et al. \cite{Abu-Affash2014} presented a $2\sqrt{10}$-approximation for this problem. We used the maximum plane matching problem in unit disk graphs (UDG) as a tool for approximating a bottleneck plane perfect matching. In Section \ref{bm:three-approximation} we presented a polynomial time algorithm which computes a plane matching of size $\frac{n}{6}$ in UDG. We also presented a $\frac{2}{5}$-approximation algorithm for computing a maximum matching in UDG. By extending this algorithm we showed how one can compute a bottleneck plane matching of size $\frac{n}{5}$ with edges of optimum-length. A modification of this algorithm gives us a plane matching of size at least $\frac{2n}{5}$ with  edges of length at most $\sqrt{2}+\sqrt{3}$ times the optimum. 

\bibliographystyle{abbrv}
\bibliography{../thesis}

%% file: chapters/ch4-bottleneckmatching-bipartite.tex
\chapter{Bottleneck Plane Matchings in Bipartite Graphs}
\label{ch:bmb}

Given a set of $n$ red points and $n$ blue points in the plane, we are interested to match the red points with the blue points by straight line segments in such a way that the segments do not cross each other and the length of the longest segment is minimized. In general, this problem in NP-hard. We give exact solutions for some special cases of the input point set. 

\vspace{10pt}
This chapter is published is published in the proceedings of the 26th Canadian Conference on Computational Geometry (CCCG 2014) \cite{Biniaz2014-bichromatic}. 

\section{Introduction}
We study the problem of computing a bottleneck non-crossing matching of red and blue points in the plane.
Let $R=\{r_1, \dots, r_n\}$ be a set of $n$ red points and $B=\{b_1,\dots,b_n\}$ be a set of $n$ blue points in the plane. A {\em RB-matching} is a non-crossing perfect matching of the points by straight line segments in such a way that each segment has one endpoint in $B$ and one in $R$.  
The length of the longest edge in an RB-matching $M$ is known as {\em bottleneck} which we denote by $\lambda_M$. The {\em bottleneck bichromatic matching} (BBM) problem is to find a non-crossing matching $M^*$ with minimum bottleneck $\lambda^*$. Carlsson et al. \cite{Carlsson2010} showed that the bottleneck bichromatic matching problem is NP-hard. Moreover, when all the points have the same color, the bottleneck non-crossing perfect matching problem is NP-hard \cite{Abu-Affash2014}.

Notice that, the bottleneck (possibly crossing) perfect matching of red and blue points can be computed exactly in $O(n^{1.5}\log n)$ time \cite{Efrat2001}.
In addition, a non-crossing perfect matching of red and blue points always exists and can be computed in $O(n\log n)$ time by applying the ham sandwich cut recursively. In \cite{Aloupis2013} the authors considered the problem of non-crossing matching of points with different geometric objects. 

In this chapter we present exact solutions for some special cases of the BBM problem when the points are arranged in convex position, boundary of a circle, and on a line. For simplicity, in the rest of the chapter we refer to a RB-matching as a ``matching''.

\section{Points in Convex Position}
\label{bmb:convex}
In this section we deal with the case when $R\cup B$ form the vertices of a convex polygon. Carlsson et al. \cite{Carlsson2010} presented an $O(n^4\log n)$-time algorithm for points on convex position. We improve their result to $O(n^3)$ time. Let $P$ denote the union of $R$ and $B$, that is $P=\{r_1,\dots,r_n, b_1, \dots, b_n\}$. We have the following observation:

\begin{observation}
\label{bmb:obs1}
 Let $(r_i, b_j)$ be an edge in any RB-matching of $P$, then there are the same number of red and blue points on each side of the line passing through $r_i$ and $b_j$.
\end{observation}

Using Observation \ref{bmb:obs1}, we present a dynamic programming algorithm which solves the BBM problem for $P$. For simplicity of notation, let $P=\{p_1 ,\dots, p_{2n}\}$ denote the sequence of the vertices of the convex polygon in counter clockwise order, starting at an arbitrary vertex $p_1$; see Figure \ref{bmb:convex-fig}. 
By Observation \ref{bmb:obs1}, we denote $(p_i,p_j)$ as a {\em feasible edge} if $p_i$ and $p_j$ have different colors and the sequence $p_{i+1},\dots, p_{j-1}$ contains the same number of red and blue points. In other words we say that $p_j$ is a {\em feasible match} for $p_i$, and vice versa. Let $F_i$ denote the set of feasible matches for $p_i$. Figure \ref{bmb:convex-fig} shows that $F_1=\{p_4, p_8, p_{10}\}$. Therefore, we define a weight function $w$ which assigns a weight $w_{i,j}$ to each pair $(p_i,p_j)$, where
$$w_{i,j} =\left\{
  \begin{array}{l l}
    |p_ip_j| &  \text{: if $(p_i,p_j)$ is a feasible edge}\\
    +\infty &  \text{: otherwise}
  \end{array} \right.$$

\begin{figure}[ht]
  \centering
    \includegraphics[width=0.4\textwidth]{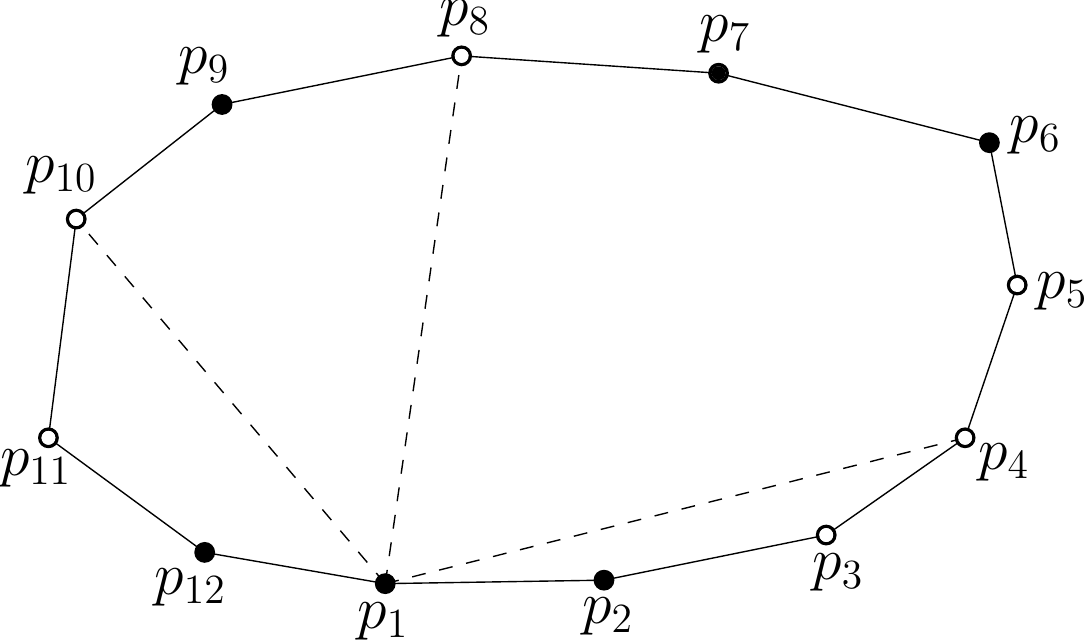}
  \caption{Points arranged on convex position.}
\label{bmb:convex-fig}
\end{figure}

Consider any subsequence $P_{i,j}=\{p_i, \dots,p_j\}$ of $P$, where $1\le i<j\le2n$. Let
$A[i, j]$ denote the bottleneck of the optimal matching in $P_{i,j}$ if $P_{i,j}$ has an RB-matching; otherwise, $A[i,j]=+\infty$. So $A[1, 2n]$
denotes the optimal solution for $P$. We use dynamic programming to compute $A[1, 2n]$.
We derive a recurrence for $A[i,j]$. For a feasible edge $(p_i,p_k)$ where $i+1\le k\le j$ and $p_k\in F_i$, the values of the sub-problems to the left and right of $(p_i,p_k)$ are $A[i+1,k-1]$ and $A[k+1,j]$. We match $p_i$ to a feasible point $p_k$ which minimizes the bottleneck. Thus,
$$
\begin{tabular}{l}
$A[i,j]=\min\limits_{\shortstack{$\scriptstyle i+1\le k\le j $\\$\scriptstyle p_k\in F_i$}}\{\max\{w_{i,k},A[i+1,k-1],A[k+1,j]\}\}.$
\end{tabular}
$$

The size of $A$ (which is the total number of sub-problems) is $O(n^2)$. For each sub-problem $A[i,j]$ we have at most $k=j-i$ lookups in $A$. Therefore, the total running time is $O(n^3)$.

\begin{theorem}
Given a set $B$ of $n$ blue points and a set $R$ of $n$ red points in convex position, one can compute a bottleneck non-crossing RB-matching in time $O(n^3)$ and in space $O(n^2)$.
\end{theorem}

Note that in \cite{Abu-Affash2014} the authors showed that for points in convex position and when all the points have the same color, a bottleneck plane matching can be computed in $O(n^3)$ time and $O(n^2)$ space via dynamic programming. They obtained the same time and space complexities for the bichromatic set of points. 

\subsection{Points on Circle}
\label{bmb:circle}
In this section we consider the BBM problem when the points in $R$ and $B$ are arranged on the boundary of a circle. Clearly, we can use the same algorithm as for points in convex position to solve this problem in $O(n^3)$ time. But for points on a circle we can do better; we present an algorithm running in $O(n^2)$ time. Consider $P=\{p_1 ,\dots, p_{2n}\}$ as the sequence of the points in counter clockwise order on a circle. We prove that there is an optimal matching $M^*$, such that each point $p_i\in P$ is connected to its first feasible match in the clockwise or counter clockwise order from $p_i$.

\begin{figure}[ht]
  \centering
\setlength{\tabcolsep}{0in}
  $\begin{tabular}{cc}
  \multicolumn{1}{m{.5\columnwidth}}{\centering\includegraphics[width=.3\columnwidth]{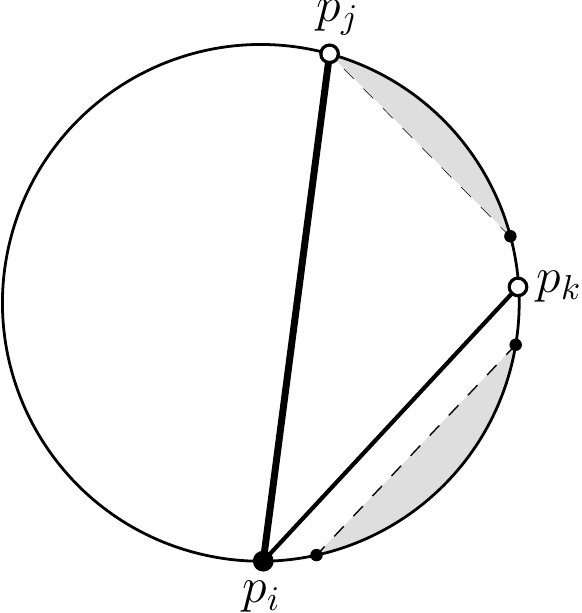}}
  &\multicolumn{1}{m{.5\columnwidth}}{\centering\includegraphics[width=.3\columnwidth]{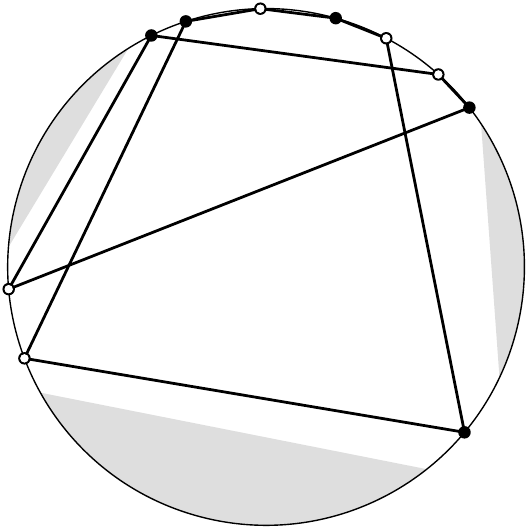}}\\
  (a) & (b)
  \end{tabular}$
  \caption{(a) illustrating the proof of Lemma \ref{bmb:circle-lemma}, (b) the resulting graph of procedure CompareToOpt.}
\label{bmb:circle-fig}
\end{figure}

\begin{lemma}
\label{bmb:circle-lemma}
There is an optimal RB-matching for a point set $P$ on a circle, such that each 
$p_i\in P$ is connected to its first feasible match in the clockwise or counter clockwise order from $p_i$. 
\end{lemma}
\begin{proof}
Consider an optimal matching $M^*$ with an edge $(p_i,p_j)$. Consider two arcs $\widehat{p_ip_j}$ and $\widehat{p_jp_i}$. W.l.o.g. let $\widehat{p_ip_j}$ be the smaller one. Clearly, the distance between any two points on $\widehat{p_ip_j}$ is at most $|p_ip_j|$. If $p_{i+1},\dots,p_{j-1}$ contains no feasible match for $p_i$, then $p_j$ is the first feasible match to the right of $p_i$. Otherwise, let $p_k$ be the first feasible match for $p_i$ in $\widehat{p_ip_j}$; see Figure \ref{bmb:circle-fig}(a). By connecting $p_i$ to $p_k$ we have two smaller arcs $\widehat{p_{i+1}p_{k-1}}$ and $\widehat{p_{k+1}p_j}$. Obviously,  $|p_ip_k|<|p_ip_j|$, and any matching of the vertices on $\widehat{p_{i+1}p_{k-1}}$ and $\widehat{p_{k+1}p_j}$ have the bottleneck smaller than $|p_ip_j|$. By repeating this process for all edges of $M^*$ and all new edges, we obtain a matching $M$ which satisfies the statement of the lemma and $\lambda_M\le \lambda^*$. 
\end{proof}

As a result of Lemma \ref{bmb:circle-lemma}, for each point $p_i \in P$, we consider at most two feasible matches in $F_i$. Thus, using the dynamic programming idea of the previous section, for each sub-problem $A[i,j]$ we have at most two lookups in $A$. Thus, it takes $O(n^2)$ time to fill the table $A$. By preprocessing $P$, for each point $p_i\in P$ we can find its first matched points in $O(n^2)$ time. Thus, the total running time of the algorithm is $O(n^2)$.

\subsubsection{A Faster Algorithm}
Let $R=\{r_1, \dots, r_n\}$ be a set of $n$ red points and $B=\{b_1,\dots,b_n\}$ be a set of $n$ blue points on the boundary of a circle $C$. Without loss of generality let $P=\{p_1,\dots,p_{2n}\}$ be the clockwise ordered set of all the points. In this section we present an $O(n\log n)$ time algorithm which solves the BBM problem for $P$. 

Let $F_i$ denote the first feasible matches of $p_i$ in clockwise and counter-clockwise order. Note that $|F_i|\le 2$. We describe how one can compute $F_i$ for all points in $P$ in linear time. First, consider the case that we are looking for the first clockwise-feasible match for each red point. We make a copy $P'$ of $P$. Consider an empty stack, and start from an arbitrary red point $r_{start}$ and walk on $P'$ clockwise. If we see a red point, push it onto the stack. If we see a blue point $p_j$ and the stack is not empty, we pop a red point $p_i$ from the stack and add $p_j$ to $F_i$, and delete $p_i$ and $p_j$ from $P'$. If we see a blue point $p_j$ and stack is empty, we do nothing. The process stops as soon as we find the proper match for each red vertex. As we visit each point in $P$ at most twice, this step takes linear time. We can do the same process for the counter-clockwise order. Therefore, $F_i$ for all $1\le i\le 2n$ can be computed in $O(n)$ time. 

Let $F$ denote the set of all feasible edges, in sorted order of their lengths. Let $G$ be the graph with vertex set $P$ and edge set $F$. Note that the degree of each vertex in $G$ is at most two and hence the total number of edges is $2n$. Let $G_\lambda$ be the subgraph of $G$ containing all the edges of length at most $\lambda$.
Our algorithm performs a binary search on the edges in $G$ and for each considered edge $e$, we use the following procedure to decide whether $G_\lambda$, where $\lambda=|e|$, has a non-crossing perfect matching. The running time of the algorithm is $O(n\log n)$. 

\begin{figure*}[ht]
  \centering
\setlength{\tabcolsep}{0in}
  $\begin{tabular}{ccc}
  \multicolumn{1}{m{.33\textwidth}}{\centering\includegraphics[width=.31\textwidth]{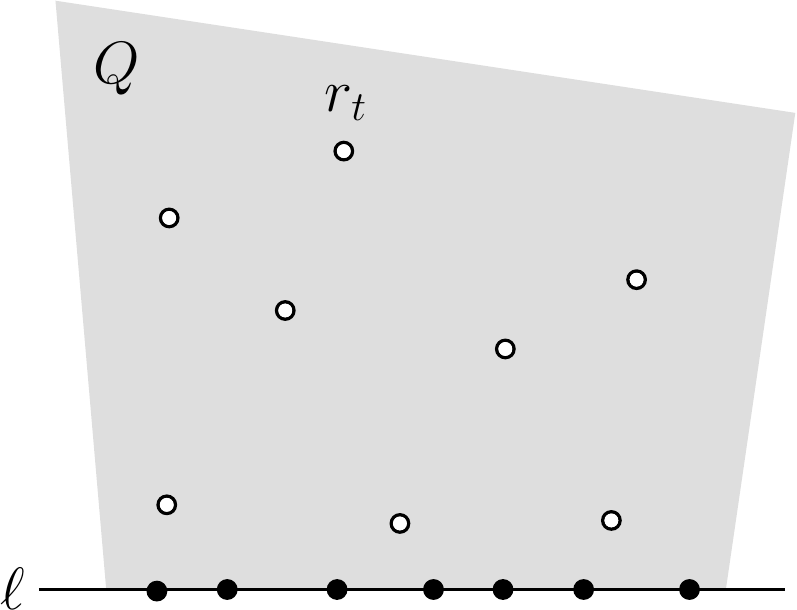}}
  &\multicolumn{1}{m{.33\textwidth}}{\centering\includegraphics[width=.31\textwidth]{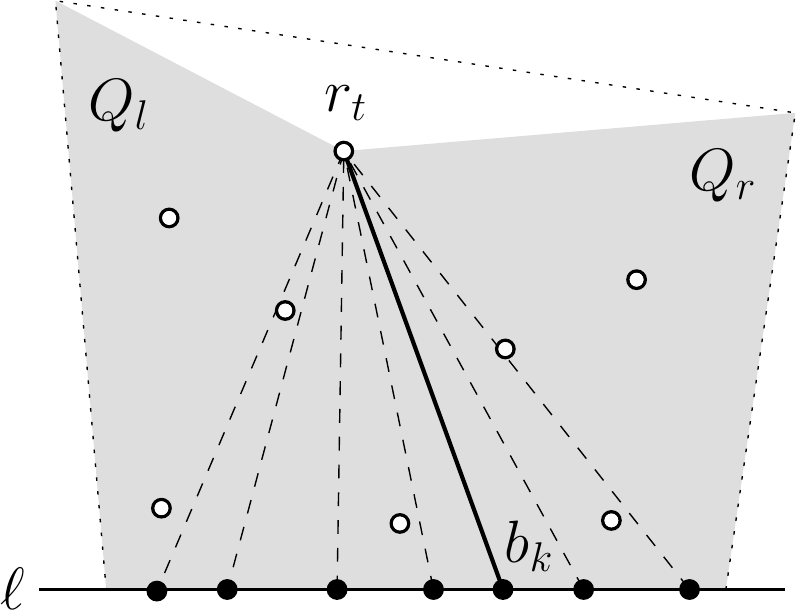}}
  &\multicolumn{1}{m{.33\textwidth}}{\centering\includegraphics[width=.31\textwidth]{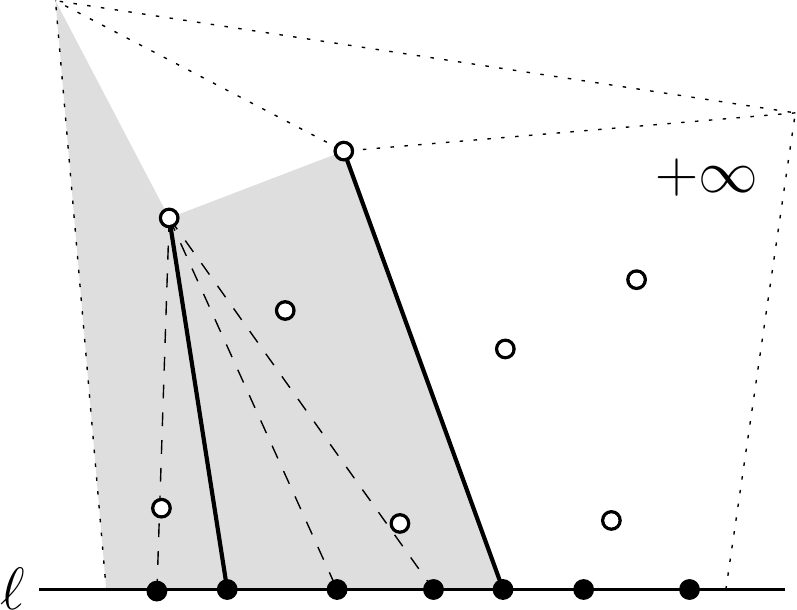}}
  \\
  (a) & (b)&(c)
  \end{tabular}$
  \caption{(a) definition of a sub-problem, (b) possible matching edges for $r_t$, and (c) $Q_r$ returns $+\infty$ as it does not contain a matching; recurse on $Q_l$.}
\label{bmb:quadrilateral}
\end{figure*}

For each edge $e=(p_i,p_j)$ in $G_\lambda$ let $I_e$ be the set of all vertices of $P$ in the smaller arc between $p_i$ and $p_j$, including $p_i$ and $p_j$. Let $P_0$ and $P_1$ be the lists of vertices of degree zero and one in $G_\lambda$, respectively. If $P_0$ is non-empty, then it is obvious that a perfect matching does not exist. If $P_0$ is empty and $P_1$ is non-empty, then for each point $p\in P_1$, do the following. Let $e=(p, q)$ be the only edge incident to $p$. It is obvious that any perfect matching in $G_\lambda$ should contain $e$. In addition, $(p,q)$ is a feasible edge, and then all the points in $I_e$ can be matched properly. Thus, we can remove the points of $I_e$ from $G_\lambda$. Note that this changes the lists $P_0$ and $P_1$. The algorithm CompareToOpt receives $G_\lambda$ as input and decides whether it has a perfect non-crossing matching.

\begin{algorithm}                      
\caption{CompareToOpt$(G_\lambda)$}          
\label{bmb:alg1} 
\require{a graph $G_\lambda$}\\
\ensure{TRUE, if $G_\lambda$ has a non-crossing perfect matching, FALSE, otherwise}
\begin{algorithmic}[1]
      \State $P_0 \gets$ vertices of degree zero in $G_\lambda$
      \State $P_1 \gets$ vertices of degree one in $G_\lambda$
      \While{$P_0\neq \emptyset$ or $P_1\neq \emptyset$}
	  \If {$P_0\neq \emptyset$} \Return {FALSE}\EndIf
	  \State $p \gets$ a vertex in $P_1$
	  \State $q \gets$ the vertex adjacent to $p$ in $G_\lambda$
	  \For {each $r$ in $I_{(p,q)}$}
	      \State remove $r$ and its adjacent edges from $G_\lambda$
	      \State update $P_0$ and $P_1$
	  \EndFor
      \EndWhile
      \State \Return {TRUE}
\end{algorithmic}
\end{algorithm} 

The algorithm CompareToOpt consider each vertex and each edge once, so it executes in linear in the size of $G_\lambda$.
At the end of the while loop, we have $P_0=P_1=\emptyset$. All the vertices of the remaining part of $G_\lambda$ have degree two and this case is the same as the problem that we started with (BBM problem) and by Lemma~\ref{bmb:circle-lemma}, it has a perfect non-crossing matching, thus we return TRUE. See Figure~\ref{bmb:circle-fig}(b). 

Notice that, if the procedure returns FALSE for some $\lambda$, then we know that $\lambda<\lambda^*$. Let $e$ be the shortest edge for which the procedure returns TRUE. Thus $|e| \ge \lambda^*$, and a bottleneck RB-matching is contained in $G_\lambda$, where $\lambda=|e|$. 

\begin{theorem}
Given a set $B$ of $n$ blue points and a set $R$ of $n$ red points on a circle, one can compute 
a bottleneck non-crossing RB-matching in time $O(n\log n)$ and in space $O(n)$.
\end{theorem}

\section{Blue Points on Straight Line}
\label{bmb:line}
In this section we deal with the case where the blue points are on a horizontal line and the red points are on one side of the line. 
Formally, given a sequence $B_{1,n}=b_1,\dots,b_n$ of $n$ blue points on a horizontal line $\ell$ and $n$ red points above $\ell$, we are interested to find a non-crossing matching $M$ between the points in $R$ and $B$, such that the length of the longest edge in $M$ is minimized. We show how to build dynamic programming algorithms that solve this problem. In Section \ref{bmb:line-alg1} we present a bottom-up dynamic programming algorithm that solves this problem in $O(n^5)$ time. In Section \ref{bmb:line-alg2} we present a top-down dynamic programming algorithm for this problem running in $O(n^4)$ time.

\subsection{First algorithm}
\label{bmb:line-alg1}
In this section we present a dynamic programming algorithm for the problem. We define a subproblem $(R',B')$ in the following way: given a quadrilateral $Q$ with one face on $\ell$, we are looking for a bottleneck RB-matching in $Q$, where $R'=R\cap Q$ and $B'=B\cap Q$. For simplicity, we may refer to the sub-problem $(R',B')$ as its bounding box $Q$. In the top level we imagine a bounding quadrilateral which contains all the points of $R$ and $B$. See Figure \ref{bmb:quadrilateral}(a). Let $b(Q)$ denote the bottleneck of the sub-problem $Q$. If $Q$ is empty, we set $b(Q)=0$. If $Q$ is not empty but $|R'|\neq|B'|$, we set $b(Q)=+\infty$, as it is not possible to have a RB-matching for $(R',B')$. Otherwise, we have $|R'|=|B'|>0$; let $r_t$ be the topmost red point in $R'$ in $Q$. It has at most $|B'|$ possible matching edges. Each of the matching edges defines two new independent sub-problems $Q_l$ and $Q_r$ to its left and right sides, respectively. See Figures \ref{bmb:quadrilateral}(b) and \ref{bmb:quadrilateral}(c). Thus, we can compute the bottleneck of a sub-problem $Q$, using the following recursion:
$$b(Q)=\min\limits_{b_k\in B'}\{\max\{|r_tb_k|, b(Q_l),b(Q_r)\}\}.$$ 
Note that the $y$-coordinate of all the red points in $Q_l$ and $Q_r$ are smaller than $y$-coordinate of $r_t$. If we recurse this process on $Q_l$ and $Q_r$, it is obvious that each sub-problem $(R',B')$ is bounded by the left and right sides of its corresponding quadrilateral. Thus, each sub-problem is defined by a pair of edges (or possibly the edges of the outer bounding box). 

Note that the total number of edges is $n^2+2$ (including the edges of the outer box). 
The dynamic programming table contains $n^2+2$ rows and $n^2+2$ columns, each corresponds to an edge. The cells correspond to sub-problems.
The dynamic programming table contains $O(n^4)$ cells, and for each we have at most $n$ pairs of possible sub-problems, which implies at most $2n$ lookups in the table. Therefore, the algorithm runs in time $O(n^5)$ and space $O(n^4)$.

\subsection{Second algorithm}
\label{bmb:line-alg2}
In this section we present a top-down dynamic programming algorithm that improves the result of Section \ref{bmb:line-alg1}. Consider the problem $(R,B)$, where $B=B_{1,n}=\{b_1,\dots,b_n\}$. Let $r_t$ be the topmost red point. In any solution $M$ to the problem, consider the edge $(r_t,b_k)\in M$ which matches $r_t$ to a point $b_k$ in $B$, then there is no edge in $M$ that intersects $(r_t,b_k)$. Thus, $(r_t,b_k)$ is a feasible edge if on each side of $(r_t,b_k)$ the number of red points equals the number of blue points. In this case, $b_k$ is a feasible match for $r_t$. Recall that $F_t$ denotes the set of all feasible matches for $r_t$. See Figure \ref{bmb:trapezoidal}(a). In other words,
$$F_t=\{k: (r_t,b_k) \text{ is a feasible edge}\}.$$

For a feasible edge $(r_t,b_k)$, let $R_l$ (resp. $R_r$) and $B_l$ (resp. $B_r$) be the red and blue points to the left (resp. right) of $(r_t,b_k)$, respectively. That is, the edge $(r_t,b_k)$ divides the $(R,B)$ problem into two sub-problems $(R_l,B_l)$ and $(R_r,B_r)$, where $|R_l|=|B_l|$ and $|R_r|=|B_r|$. Clearly, $B_l=B_{1,k-1}=\{b_1,\dots, b_{k-1}\}$ and $B_r=B_{k+1,n}=\{b_{k+1},\dots,\allowbreak  b_n\}$. We develop the following recurrence to solve the problem:
$$b(R,B)=\min\limits_{k\in F_t}\{\max\{|r_tb_k|, b(B_l,R_l),b(B_r,R_r)\}\}.$$

Let $l_t$ denote the horizontal line passing through $r_t$. Note that the $y$-coordinate of all red points in $R_l$ and $R_r$ is smaller than the $y$-coordinate of $r_t$, and hence they lie below $l_t$. This implies that the left (resp. right) sub-problem is contained in a trapezoidal region $T_l$ (resp. $T_r$) with bounding edges $\ell$, $l_t$, and $(r_t,b_k)$. See Figure \ref{bmb:trapezoidal}(a). Since, in each step we have two sub-problems, in the rest of this section we describe the process for the right sub-problem; the process for the left sub-problem is symmetric. Note that $r_t$ is the top-left corner of the right sub-problem. Thus, given $B_r$ and $r_t$, we know that $r_t$ is connected to a blue point immediately to the left of $B_r$. In addition, we can find the red points assigned to the right sub-problem in the following way. Stand at a blue point immediately to the right of $B_r$ and scan the plane clockwise, starting from $\ell$. Count the red points in $T_r$ while scanning, and stop as soon as the number of red points seen equals the number of blue points in $B_r$. These red points form the set $R_r$. See Figures \ref{bmb:trapezoidal}(b) and \ref{bmb:trapezoidal}(c).

Since, $r_t$ defines the right (resp. left) and top boundaries of $T_l$ (resp. $T_r$) which contains the left (resp. right) sub-problem, we call $r_t$ a ``boundary vertex''. We define a sub-problem as a sequence $B_{i,j}=\{b_i,\dots,b_j\}$ of blue points, a boundary vertex, $r_t$, connected to $b_{j+1}$ (resp. $b_{i-1}$) for the left (resp. right) sub-problem. More precisely, a sub-problem $(B_{i,j},r_t,d)$ consists of an interval $B_{i,j}$, a boundary vertex $r_t$, and a direction $d=\{left,right\}$ which indicates that $r_t$ is connected to a point immediately to the left or to the right of $B_{i,j}$. 
For a sub-problem $(B_{i,j},t,d)$, where $d=left$ we find the vertex set $R_{i,j}$ in the following way. Scan the plane by a clockwise rotating line $s$ anchored at $b_{j+1}$. Count the red points in trapezoidal region formed by $\ell$, $l_t$, and $(r_t,b_{i-1})$, and stop as soon as $j-i+1$ red points have been encountered. These red points form the set $R_{i,j}$. See Figures \ref{bmb:trapezoidal}(b) and \ref{bmb:trapezoidal}(c). 

\begin{figure*}[ht]
  \centering
\setlength{\tabcolsep}{0in}
  $\begin{tabular}{ccc}
  \multicolumn{1}{m{.32\textwidth}}{\centering\includegraphics[width=.3\textwidth]{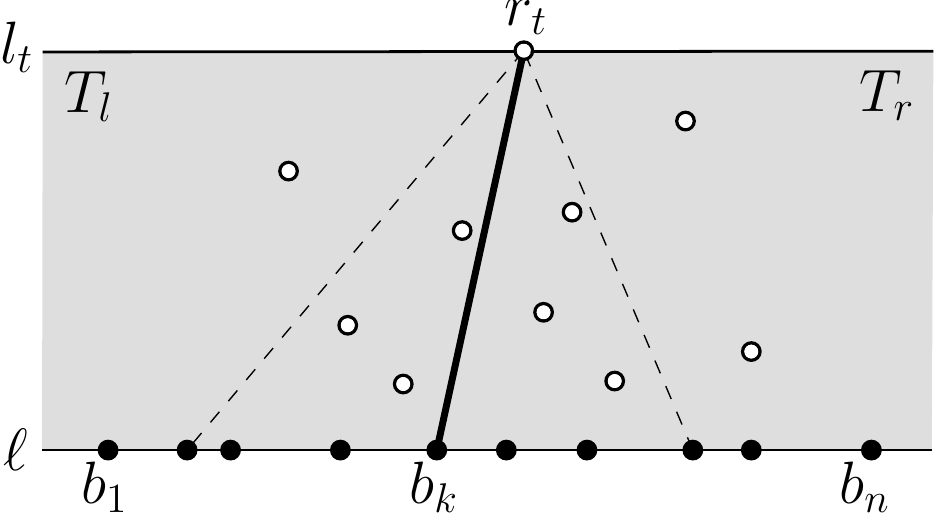}}
  &\multicolumn{1}{m{.34\textwidth}}{\centering\includegraphics[width=.32\textwidth]{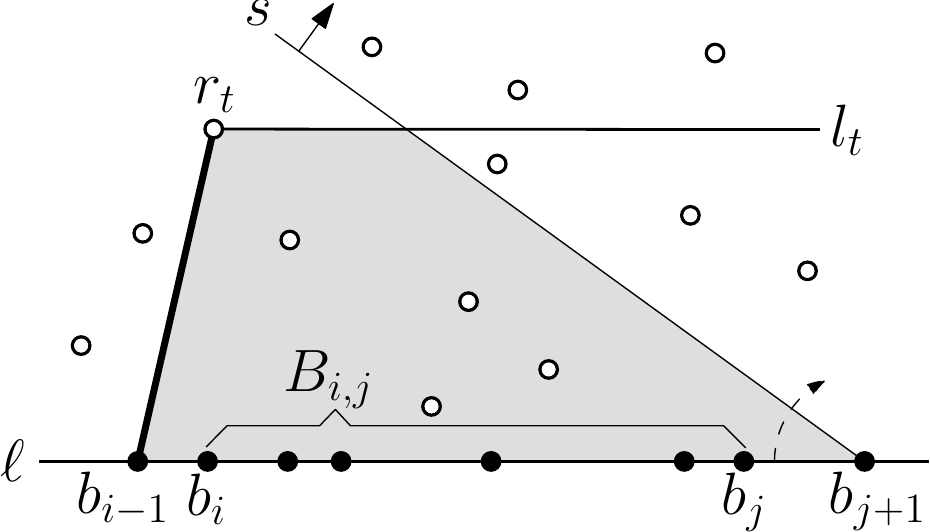}}
  &\multicolumn{1}{m{.34\textwidth}}{\centering\includegraphics[width=.32\textwidth]{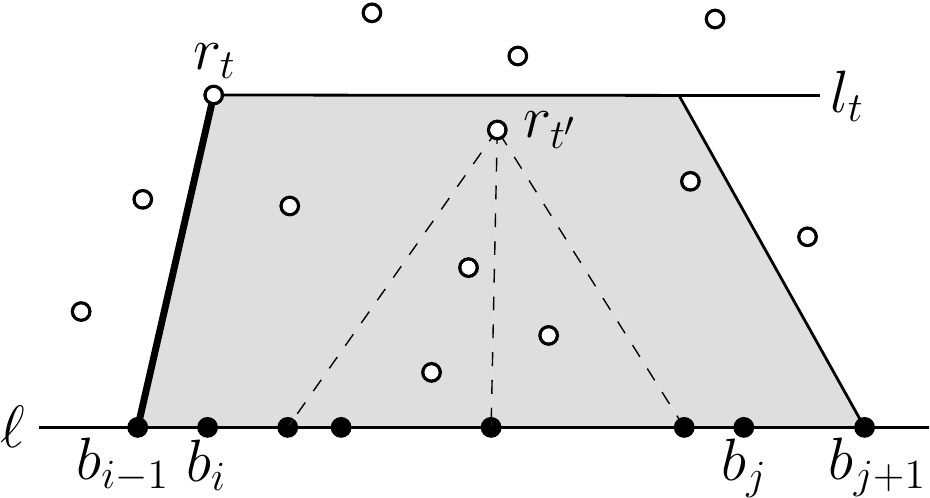}}
  \\
  (a) & (b)&(c)
  \end{tabular}$
  \caption{(a) feasible matches for $r_t$, (b) scanning the red points in the trapezoidal region (shaded area), and (c) the trapezoidal region which contains the same number of red and blue points.}
\label{bmb:trapezoidal}
\end{figure*}

In the top level, we add points $b_0$ and $b_{n+1}$ on $\ell$ to the left and right of $B$, respectively. We add a point $r_0$ as the boundary vertex of the $(R,B)$ problem in such a way that $R$ and $B$ are contained in the trapezoid formed by $\ell$, $l_0$, and the line segment $r_0b_0$. Thus in the top level we have the sub-problem $(B_{1,n},r_0,left)$.

The dynamic programming table is a four-dimensional table $A[1..n,\allowbreak 1..n,\allowbreak 0..n,\allowbreak 1..2]$, where the first and second dimensions correspond to an interval of blue points, the third dimension corresponds a boundary vertex, and the fourth dimension corresponds to the directions. For simplicity we use $l$ and $r$ for $left$ and $right$ directions, respectively. Each cell $A[i,j,t, d]$ stores the bottleneck of the sub-problem $(B_{i,j}, r_t, d)$, and we are looking for $A[1, n, 0, l]$ which corresponds to the bottleneck of $M^*$. We fill $A$ in the following way: 
$$
\begin{tabular}{l}
$A[i,j,t,d]= \min\limits_{k\in F_{t'}}\{\max\{|r_{t'}b_k|,A[i,k-1,t',r],A[k+1,j,t',l]\}\},$
\end{tabular}
$$
where $r_{t'}$ is the topmost red point in the point set $R_{i,j}$ assigned to $(B_{i,j},t,d)$.

Algorithm~\ref{bmb:alg2} computes the bottleneck of each subproblem using top-down dynamic programming.
In the top level, we execute LineM\-atching($1,n,0,l$). Before running algorithm LineMatching, for each point, we pre-sort the red points in the following way. For each red point $r$, we keep a sorted list of all the red points below $l_r$ in clockwise order. For each blue point, we keep two sorted lists of red points in clockwise and counter-clockwise orders. This step takes $O(n^2\log n)$ time.

\begin{algorithm}                      
\caption{LineMatching$(i,j,t,d)$}          
\label{bmb:alg2} 
\require{sequence $B_{i,j}$, top point $r_t$, and direction $d$.}\\
\ensure{bottleneck of $M^*$.}
\begin{algorithmic}[1]
      \If {$A[i,j,t,d]>0$}
	  \State \Return $A[i,j, t,d]$
      \EndIf
      \If {$i>j$}
	  \State \Return $A[i,j,t,d]\gets 0$
      \EndIf
      \State $R_{i,j}\gets$ $j-i+1$ red points assigned to $B_{i,j}$
      \State $t'\gets top$-$index(R_{i,j})$
      \If {$i=j$}
	  \State \Return $A[i,j,t,d]\gets |r_{t'}b_i|$
      \EndIf
      \State $b\gets +\infty$
      \State $F_{t'}\gets$ indices of feasible blue points for $r_{t'}$ 
      \For {each $k \in F_{t'}$}
	  \State $A[i,k-1,t',r]\gets$ LineMatching$(i,k-1,t',r)$
	  \State $A[k+1,j,t',l]\gets$ LineMatching$(k+1,j,t',l)$
	  \State $m\gets\max\{|r_{t'}b_k|, A[i,k-1,t',r], A[k+1,j,t',l]\}$
	  \If {$m<b$}
	      \State $b\gets m$
	  \EndIf
      \EndFor
      \State \Return $A[i,j,t,d]\gets b$
\end{algorithmic}
\end{algorithm}

\begin{lemma}
 Algorithm LineMatching computes the bottleneck of $M^*$ in $O(n^4)$ time. 
\end{lemma}
\begin{proof}
Each cell $A[i,j,t,d]$ corresponds to a sub-problem formed by an interval $B_{i,j}$, a boundary vertex $r_t$, and a direction $d$. The total number of possible $B_{i,j}$ intervals is ${n\choose 2}+n$ ($i$ can be equal to $j$). For each interval, any of the $n$ red points can be the corresponding boundary vertex, which can be connected to the left or right side of the interval. Thus, the total number of subproblems is $2n{n\choose 2}+2n^2=O(n^3)$. In order to compute $R_{i,j}$ for each sub-problem, we use the sorted lists assigned to $b_{i-1}$ (or $b_{j+1}$) and scan for the red points in the trapezoidal region. To compute the feasible blue vertices for $r_{t'}\in R_{i,j}$, we use the sorted list assigned to $r_{t'}$ and keep track of feasible matches for $r_{t'}$ in $B_{i,j}$. Thus, for each sub-problem, we can compute $R_{i,j}$, $r_{t'}$, and $F_{t'}$ in linear time. Therefore, the total running time of the algorithm is $O(n^4)$.
\end{proof}

Finally, we reconstruct $M^*$ from $A$ in linear time.
\begin{theorem}
Given a set $B$ of $n$ blue points on a horizontal line $\ell$, a set $R$ of $n$ red points above $\ell$, one can compute a bottleneck non-crossing RB-matching in time $O(n^4)$ and in space $O(n^3)$.
\end{theorem}

\bibliographystyle{abbrv}
\bibliography{../thesis}

%% file: chapters/ch5-planematching.tex
\chapter{Plane Matchings in Complete Multipartite Graphs}
\label{ch:pm}

Let $P$ be a set of $n$ points in general position in the plane which is partitioned into {\em color} classes. $P$ is said to be {\em color-balanced} if the number of points of each color is at most $\lfloor n/2\rfloor$. Given a color-balanced point set $P$, a {\em balanced cut} is a line which partitions $P$ into two color-balanced point sets, each of size at most $2n/3  + 1$. A {\em colored matching} of $P$ is a perfect matching in which every edge connects two points of distinct colors by a straight line segment. A {\em plane colored matching} is a colored matching which is non-crossing. In this chapter, we present an algorithm which computes a balanced cut for $P$ in linear time. Consequently, we present an algorithm which computes a plane colored matching of $P$ optimally in $\Theta(n\log n)$ time.

\vspace{10pt}
This chapter is published in the proceedings of the 14th International Symposium on Algorithms and Data Structures (WADS'15)~\cite{Biniaz2015-RGB}. 

\vspace{10pt}
We also extended the results of this chapter to the case where the points are in a simple polygon; the extension is not included in this thesis. However, the results have been published in the Proceedings of the First international conference on Topics in Theoretical Computer Science (TTCS'15) \cite{Biniaz2015-geodesic},
and also published in the journal of Computational Geometry: Theory and Applications~\cite{Biniaz2016-CGTA-geodesic}.

\section{Introduction}
\label{introduction-section}
Let $P$ be a set of $n$ points in general position (no three points on a line) in the plane. Assume $P$ is partitioned into {\em color} classes, i.e., each point in $P$ is colored by one of the given colors. $P$ is said to be {\em color-balanced} if the number of points of each color is at most $\lfloor n/2\rfloor$. In other words, $P$ is color-balanced if no color is in strict majority. For a color-balanced point set $P$, we define a {\em feasible cut} as a line $\ell$ which partitions $P$ into two point sets $Q_1$ and $Q_2$ such that both $Q_1$ and $Q_2$ are color-balanced. In addition, if the number of points in each of $Q_1$ and $Q_2$ is at most $2n/3+1$, then $\ell$ is said to be a {\em balanced cut}. 
The well-known ham-sandwich cut (see~\cite{Lo1994}) is a balanced cut: given a set of $2m$ red points and $2m$ blue points in general position in the plane, a ham-sandwich cut is a line $\ell$ which partitions the point set into two sets, each of them having $m$ red points and $m$ blue points. Feasible cuts and balanced cuts are useful for convex partitioning of the plane and for computing plane structures, e.g., plane matchings and plane spanning trees. 

Assume $n$ is an even number. Let $\{R,B\}$ be a partition of $P$ such that $|R|=|B|=n/2$. Let $K_n(R,B)$ be the complete bipartite geometric graph on $P$ which connects every point in $R$ to every point in $B$ by a straight-line edge. An $RB${\em -matching} in $P$ is a perfect matching in $K_n(R,B)$. Assume the points in $R$ are colored red and the points in $B$ are colored blue. An $RB$-matching in $P$ is also referred to as a {\em red-blue matching} or a {\em bichromatic matching}. 
A {\em plane} $RB${\em -matching} is an $RB$-matching in which no two edges cross.
Let $\{P_1,\dots,P_k\}$, where $k\ge 2$, be a partition of $P$. Let $K_n(P_1,\dots,P_k)$ be the complete multipartite geometric graph on $P$ which connects every point in $P_i$ to every point in $P_j$ by a straight-line edge, for all $1\le i<j\le k$. Imagine the points in $P$ to be colored, such that all the points in $P_i$ have the same color, and for $i\neq j$, the points in $P_i$ have a different color from the points in $P_j$. We say that $P$ is a $k$-{\em colored} point set. A {\em colored matching} of $P$ is a perfect matching in $K_n(P_1,\dots,P_k)$. A {\em plane colored matching} of $P$ is a perfect matching in $K_n(P_1,\dots,P_k)$ in which no two edges cross. See Figure~\ref{RB-fig}(a). 

In this chapter we consider the problem of computing a balanced cut for a given color-balanced point set in general position in the plane. We show how to use balanced cuts to compute plane matchings in multipartite geometric graphs.

\vspace{-5pt}
\begin{figure}[htb]
  \centering
\setlength{\tabcolsep}{0in}
  $\begin{tabular}{cc}
\multicolumn{1}{m{.5\columnwidth}}{\centering\includegraphics[width=.28\columnwidth]{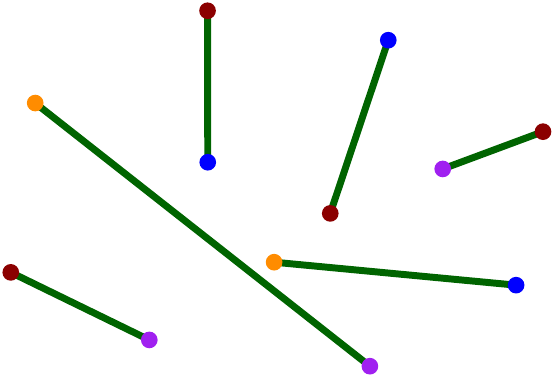}}
&\multicolumn{1}{m{.5\columnwidth}}{\centering\includegraphics[width=.3\columnwidth]{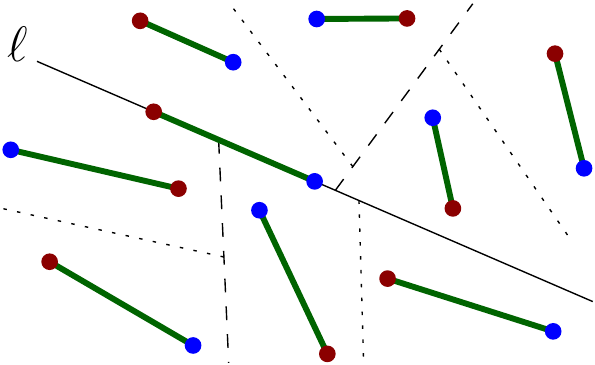}} \\
(a) & (b)
\end{tabular}$
\vspace{-5pt}
  \caption{(a) A plane colored matching. (b) Recursive ham sandwich cuts.}
\label{RB-fig}
\end{figure}
\vspace{-25pt}
\addtocontents{toc}{\protect\setcounter{tocdepth}{1}}
\subsection{Previous Work}
\addtocontents{toc}{\protect\setcounter{tocdepth}{2}}
\label{previous-work}

\subsubsection{2-Colored Point Sets}
Let $P$ be a set of $n=2m$ points in general position in the
plane. Let $\{R, B\}$ be a partition of $P$ such that $|R| = |B|=m$. Assume the points in $R$ are colored red and the points in $B$ are colored blue.
It is well-known that $K_n(R,B)$ has a plane $RB$-matching \cite{Putnam1979}. In fact, a minimum weight $RB$-matching, i.e., a perfect matching that minimizes the sum of Euclidean
length of the edges, is plane. A minimum weight $RB$-matching in $K_n(R,B)$ can be computed in $O(n^{2.5}\log n)$ time~\cite{Vaidya1989}, or even in $O(n^{2+\epsilon})$ time~\cite{Agarwal1999}. Consequently, a plane $RB$-matching can be computed in $O(n^{2+\epsilon})$ time.
As a plane $RB$-matching is not necessarily a minimum weight $RB$-matching, one may compute a plane $RB$-matching faster than computing a minimum weight $RB$-matching.
Hershberger and Suri~\cite{Hershberger1992} presented an $O(n\log n)$ time algorithm for computing a plane $RB$-matching. They also proved a lower bound of $\Omega(n\log n)$ time for computing a plane $RB$-matching, by providing a reduction from sorting.

Alternatively, one can compute a plane $RB$-matching  by recursively applying the ham sandwich theorem; see Figure~\ref{RB-fig}(b). We say that a line $\ell$ {\em bisects} a point set $R$ if both sides of $\ell$ have the same number of points of $R$; if $|R|$ is odd, then $\ell$ contains one point of $R$. 
\begin{theorem}[Ham Sandwich Theorem]
\label{ham-sandwich-thr}
 For a point set $P$ in general position in the plane which is partitioned into sets $R$ and $B$, there exists a line that simultaneously bisects $R$ and $B$.
\end{theorem}
A line $\ell$ that simultaneously bisects $R$ and $B$ can be computed in $O(|R|+|B|)$ time, assuming $R\cup B$ is in general position in the plane~\cite{Lo1994}. By recursively applying Theorem~\ref{ham-sandwich-thr}, we can compute a plane $RB$-matching in $\Theta(n\log n)$ time.

\subsubsection{3-Colored Point Sets}
Let $P$ be a set of $n=3m$ points in general position in the
plane. Let $\{R, G, B\}$ be a partition of $P$ such that $|R| = |G|=|B|=m$. Assume the points in $R$ are colored red, the points in $G$ are colored green, and the points in $B$ are colored blue.
A lot of research has been done to generalize the ham sandwich theorem to 3-colored point sets, see e.g.~\cite{Bereg2015,Bereg2012,Kano2013}. It is easy to see that there exist configurations of $P$ such that there exists no line which bisects $R$, $G$, and $B$, simultaneously. Furthermore, for some configurations of $P$, for any $k\in\{1,\dots, m-1\}$, there does not exist any line $\ell$ such that an open half-plane bounded by $\ell$ contains $k$ red, $k$ green, and $k$ blue points (see~\cite{Bereg2012} for an example).
For the special case, where the points on the convex hull of $P$ are monochromatic, Bereg and Kano~\cite{Bereg2012} proved that  there exists an integer $1\le k\le m-1$ and an open half-plane containing exactly $k$ points from each color.

Bereg et al.~\cite{Bereg2015} proved that if the points of $P$ are on any closed Jordan curve $\gamma$, then for every integer $k$ with $0 \le k \le m$ there exists a pair of disjoint intervals on $\gamma$ whose union contains exactly $k$ points of each color. In addition, they showed that if $m$ is even, then there exists a double wedge that contains exactly $m/2$ points of each color.

Now, let $P$ be a 3-colored point set of size $n$ in general position in the plane, with $n$ even. Assume the points in $P$ are colored red, green, and blue such that $P$ is color-balanced. Let $R$, $G$, and $B$ denote the set of red, green, and blue points, respectively. Note that $|R|$, $|G|$, and $|B|$ are at most $\lfloor n/2\rfloor$, but, they are not necessarily equal. Kano et al.~\cite{Kano2013} proved the existence of a feasible cut, when the points on the convex hull of $P$ are monochromatic. 
\begin{theorem}[Kano et al.~\cite{Kano2013}]
\label{Kano-thr}
Let $P$ be a 3-colored point set in general position in the plane, such that $P$ is color-balanced and $|P|$ is even. If the points on the convex hull of $P$ are monochromatic, then there exists a line $\ell$ which partitions $P$ into $Q_1$ and $Q_2$ such that both $Q_1$ and $Q_2$ are color-balanced and have an even number of points and $2\le |Q_i|\le |P|-2$, for $i=1,2$.
\end{theorem}

They also proved the existence of a plane perfect matching in $K_n(R,\allowbreak G,B)$ by recursively applying Theorem~\ref{Kano-thr}. Their proof is constructive. Although they did not analyze the running time, it can be shown that their algorithm runs in $O(n^2\log n)$ time as follows. If the size of the largest color class is exactly $n/2$, then consider the points in the largest color class as $R$ and the other points as $B$, then compute a plane $RB$-matching; and we are done. If there are two adjacent points of distinct colors on the convex hull, then match these two points and recurse on the remaining points. Otherwise, if the convex hull is monochromatic, pick a point $p\in P$ on the convex hull and sort the points in $P\setminus\{p\}$ around $p$. A line $\ell$\textemdash partitioning the point set into two color-balanced point sets\textemdash is found by scanning the sorted list. Then recurse on each of the partitions. To find $\ell$ they spend $O(n\log n)$ time. The total running time of their algorithm is $O(n^2\log n)$. 

Based on the algorithm of Kano et al.~\cite{Kano2013}, we can show that a plane perfect matching in $K_n(R,G,B)$ can be computed in $O(n\log^3 n)$ time. We can prove the existence of a feasible cut for $P$, even if the points on the convex hull of $P$ are not monochromatic. To find feasible cuts recursively, we use the dynamic convex hull structure of Overmars and Leeuwen~\cite{Overmars1981}, which uses $O(\log^2 n)$ time for each insertion and deletion. Pick a point $p\in P$ on the convex hull of $P$ and look for a point $q\in P\setminus\{p\}$, such that the line passing through $p$ and $q$ is a feasible cut. Search for $q$, alternatively, in clockwise and counterclockwise directions around $p$. To do this, we repeatedly check if the line passing through $p$ and its (clockwise and counterclockwise in turn) neighbor on the convex hull, say $r$, is a feasible cut. If the line through $p$ and $r$ is not a feasible cut, then we delete $r$. At some point we find a feasible cut $\ell$ which divides $P$ into $Q_1$ and $Q_2$. Add the two points on $\ell$ to either $Q_1$ or $Q_2$ such that they remain color-balanced. Let $|Q_1|=k$ and $|Q_2|\ge k$. In order to compute the data structure for $Q_2$, we use the current data structure and undo the deletions on the side of $\ell$ which contains $Q_2$. We rebuild the data structure for $Q_1$. Then, we recurse on $Q_1$ and $Q_2$. The running time can be expressed by $T(n)= T(n-k)+T(k)+O(k \log^2 n)$, where $k\le n-k$. This recurrence solves to $O(n \log^3 n)$. Notice that, because we undo the deletions on one side of $\ell$ and rebuild the data structure for the points on the other side of $\ell$, any dynamic data structure that performs insertions and deletions in faster amortized time may not be feasible.
 
\subsubsection{Multicolored Point Sets}
Let $\{P_1,\dots, P_k\}$, where $k \ge 2$, be a partition of $P$ and $K_n(P_1,\dots,\allowbreak P_k)$ be the complete multipartite geometric graph on $P$. 
A necessary and sufficient condition for the existence of a perfect matching in $K_n(P_1,\dots, P_k)$ follows from the following result of Sitton~\cite{Sitton1996}.

\begin{theorem}[Sitton~\cite{Sitton1996}]
\label{Sitton-thr}
The size of a maximum matching in any complete multipartite graph $K_{n_1,\dots,n_k}$, with $n=n_1+\dots+n_k$ vertices, where $n_1\ge \dots \ge n_k$, is
$$|M_{max}|=\min\left\{\sum_{i=2}^{k}{n_i},\left\lfloor\frac{1}{2}\sum_{i=1}^{k}{n_i}\right\rfloor\right\}.$$
\end{theorem}

Theorem~\ref{Sitton-thr} implies that if $n$ is even and $n_1\le \frac{n}{2}$, then $K_{n_1,\dots,n_k}$ has a perfect matching. It is obvious that if $n_1>\frac{n}{2}$, then $K_{n_1,\dots,n_k}$ does not have any perfect matching. Therefore,

\begin{corollary}
\label{colored-matching-cor}
Let $k\ge 2$ and consider a partition $\{P_1,\dots,P_k\}$ of a point set $P$, where $|P|$ is even. Then, $K_n(P_1,\dots,P_k)$ has a colored matching if and only if $P$ is color-balanced. 
\end{corollary}

Aichholzer et al.~\cite{Aichholzer2010}, and Kano et al.~\cite{Kano2013} show that the same condition as in Corollary~\ref{colored-matching-cor} is necessary and sufficient for the existence of a plane colored matching in $K_n(P_1,\dots,P_k)$:

\begin{theorem}[Aichholzer et al.~\cite{Aichholzer2010}, and Kano et al.~\cite{Kano2013}]
\label{Aichholzer-thr}
Let $k\ge 2$ and consider a partition $\{P_1,\dots,P_k\}$ of a point set $P$, where $|P|$ is even. Then, $K_n(P_1,\dots,P_k)$ has a plane colored matching if and only if $P$ is color-balanced. 
\end{theorem}

In fact, they show something stronger. Aichholzer et al.~\cite{Aichholzer2010} show that a minimum weight colored matching in $K_n(P_1,\dots,P_k)$, which minimizes the total Euclidean length of the edges, is plane. Gabow \cite{Gabow1990} gave an implementation of Edmonds' algorithm which computes a minimum weight matching in general graphs in $O(n(m+n\log n))$ time, where $m$ is the number of edges in $G$. Since $P$ is color-balanced, $K_n(P_1,\dots,P_k)$ has $\Theta(n^2)$ edges. Thus, a minimum weight colored matching in $K_n(P_1,\dots,P_k)$, and hence a plane colored matching in $K_n(P_1,\dots,P_k)$, can be computed in $O(n^3)$ time. Kano et al.~\cite{Kano2013} extended their $O(n^2\log n)$-time algorithm for the 3-colored point sets to the multicolored case. 

Since the problem of computing a plane $RB$-matching in $K_n(R,B)$ is a special case of the problem of computing a plane colored matching in $K_n(P_1,\dots,P_k)$, the $\Omega(n\log n)$ time lower bound for computing a plane $RB$-matching holds for computing a plane colored matching.

\addtocontents{toc}{\protect\setcounter{tocdepth}{1}}
\subsection{Our Contribution}
\addtocontents{toc}{\protect\setcounter{tocdepth}{2}}
Our main contribution, which is presented in Section~\ref{balanced-cut-section}, is the following: given any color-balanced point set $P$ in general position in the plane, there exists a balanced cut for $P$. Further, we show that if $n$ is even, then there exists a balanced cut which partitions $P$ into two point sets each of even size, and such a balanced cut can be computed in linear time. In Section~\ref{algorithm-section}, we present a divide-and-conquer algorithm which computes a plane colored matching in $K_n(P_1,\dots,P_k)$ in $\Theta(n\log n)$ time, by recursively finding balanced cuts in color-balanced subsets of $P$. In case $P$ is not color-balanced, then $K_n(P_1,\dots,P_k)$ does not admit a perfect matching; we describe how to find a plane colored matching with the maximum number of edges in Section~\ref{maximum-section}. In addition, we show how to compute a maximum matching in any complete multipartite graph in linear time. 

\section{Balanced Cut Theorem}
\label{balanced-cut-section}

Given a color-balanced point set $P$ with $n\ge 4$ points in general position in the plane, recall that a {\em balanced cut} is a line which partitions $P$ into two point sets $Q_1$ and $Q_2$, such that both $Q_1$ and $Q_2$ are color-balanced and $\max\{|Q_1|,|Q_2|\}\le \frac{2n}{3}+1$. Let $\{P_1,\dots,P_k\}$ be a partition of $P$, where the points in $P_i$ are colored $C_i$. In this section we prove the existence of a balanced cut for $P$. Moreover, we show how to find such a balance cut in $O(n)$ time.
 
If $k=2$, the existence of a balanced cut follows from the ham sandwich cut theorem. If $k\ge 4$, we reduce the $k$-colored point set $P$ to a three colored point set. Afterwards, we prove the statement for $k=3$.

\begin{lemma}
\label{k-to-3}
Let $P$ be a color-balanced point set of size $n$ in the plane with $k\ge4$ colors. In $O(n)$ time $P$ can be reduced to a color-balanced point set $P'$ with 3 colors such that any balanced cut for $P'$ is also a balanced cut for $P$.
\end{lemma}
\begin{proof}
We repeatedly merge the color families in $P$ until we get a color-balanced point set $P'$ with three colors. Afterwards, we show that any balanced cut for $P'$ is also a balanced cut for $P$. 

Without loss of generality assume that $C_1,\dots, C_k$ is a non-increasing order of the color classes according to the number of points in each color class. That is, $\lfloor|P|/2\rfloor\ge |P_1|\ge \dots\ge|P_k|\ge 1$ (note that $P$ is color-balanced). In order to reduce the $k$-colored problem to a 3-colored problem, we repeatedly merge the two color families of the smallest cardinality. In each iteration we merge the two smallest color families, $C_{k-1}$ and $C_k$, to get a new color class, $C'_{k-1}$, where $P'_{k-1}=P_{k-1}\cup P_k$. In order to prove that $P'=P_1\cup\dots\cup P_{k-2}\cup P'_{k-1}$ is color-balanced with respect to the coloring $C_1,\dots, C_{k-2},C'_{k-1}$ we have to show that $|P'_{k-1}|\le \lfloor|P'|/2\rfloor$. Note that before the merge we have $|P|=|P_1|+\dots+|P_{k-2}|+|P_{k-1}|+|P_k|$, while after the merge we have $|P'|=|P_1|+\dots+|P_{k-2}|+|P'_{k-1}|$, where $|P'_{k-1}|=|P_{k-1}|+|P_k|$. Since $P_{k-1}$ and $P_k$ are the two smallest and $k\ge4$, $|P'_{k-1}|\le |P_1|+\dots+|P_{k-2}|$. This implies that after the merge we have $|P'_{k-1}|\le \lfloor|P'|/2\rfloor$. Thus $P'$ is color-balanced. By repeatedly merging the points of the two smallest color families, at some point we get a 3-colored point set $P'$ which is color-balanced. Without loss of generality assume that $P'$ is colored by $R$, $G$, and $B$. Consider any balanced cut $\ell$ for $P'$; $\ell$ partitions $P'$ into two sets $Q_1$ and $Q_2$, each of size at most $\frac{2}{3}n+1$, such that the number points of each color in $Q_i$ is at most $\lfloor |Q_i|/2\rfloor$, where $i=1,2$. Note that the set of points in $P$ colored $C_j$, for $1\le j\le k$, is a subset of points in $P'$ colored either $R$, $G$, or $B$. Thus, the number of points colored $C_j$ in $Q_i$ is at most $\lfloor |Q_i|/2\rfloor$, where $j=1,\dots,k$ and $i=1,2$. Therefore, $\ell$ is a balanced cut for $P$.

In order to merge the color families, a monotone priority queue (see~\cite{Cherkassky1999}) can be used, where the priority of each color $C_j$ is the number of points colored $C_j$. The monotone priority queue offers {\sf insert} and {\sf extract-min} operations where the priority of an inserted element is greater than the priority of the last element extracted from the queue. We store the color families in a monotone priority queue of size $\frac{n}{2}$ (because all elements are in the range of 1 up to $\frac{n}{2}$). Afterwards, we perform a sequence of $O(k)$ {\sf extract-min} and {\sf insert} operations. Since $k\le n$, the total time to merge $k$ color families is $O(n)$.
\qed\end{proof}

According to Lemma~\ref{k-to-3}, from now on we assume that $P$ is a color-balanced point set consisting of $n$ points colored by three colors.

\begin{lemma}
\label{balanced-cut-lemma}
Let $P$ be a color-balanced point set of $n\ge 4$ points in general position in the plane with three colors. In $O(n)$ time we can compute a line $\ell$ such that
\begin{enumerate}
  \item $\ell$ does not contain any point of $P$.
  \item $\ell$ partitions $P$ into two point sets $Q_1$ and $Q_2$, where
      \begin{enumerate}
	\item both $Q_1$ and $Q_2$ are color-balanced,
	\item both $Q_1$ and $Q_2$ contains at most $\frac{2}{3}n+1$ points.
      \end{enumerate}
\end{enumerate}
\end{lemma}

\begin{figure}[htb]
  \centering
\setlength{\tabcolsep}{0in}
  $\begin{tabular}{cc}
\multicolumn{1}{m{.5\columnwidth}}{\centering\includegraphics[width=.35\columnwidth]{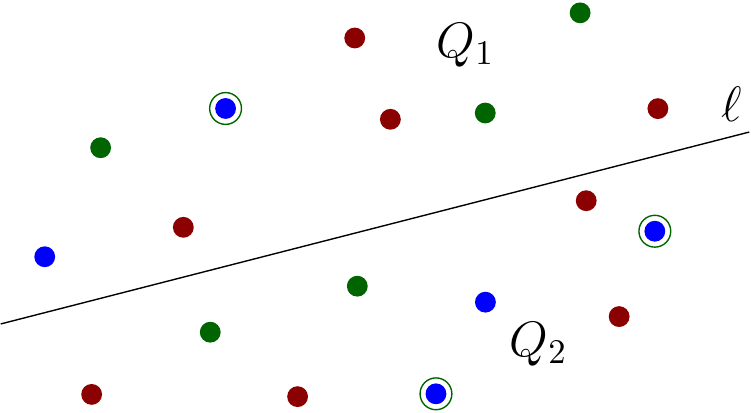}}
&\multicolumn{1}{m{.5\columnwidth}}{\centering\includegraphics[width=.35\columnwidth]{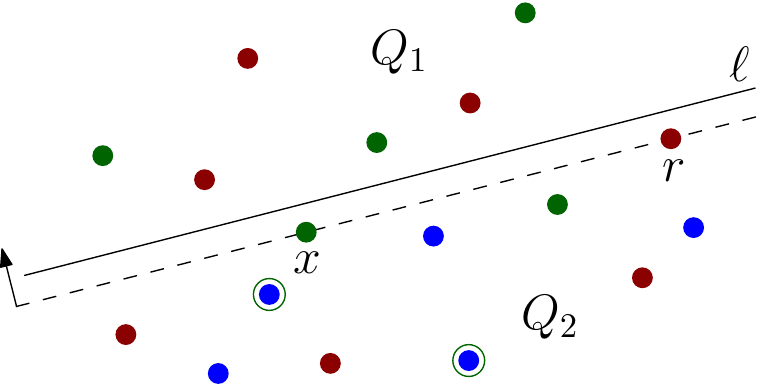}} \\
(a) & (b)
\end{tabular}$

  \caption{Illustrating the balanced cut theorem. The blue points in $X$ are surrounded by circles. The line $\ell$ is a balanced cut where: (a) $|R|$ is even, and (b) $|R|$ is odd.}
\label{balanced-cut-fig}
\end{figure}

\begin{proof}
Assume that the points in $P$ are colored red, green, and blue.
Let $R$, $G$, and $B$ denote the set of red, green, and blue points, respectively. Without loss of generality assume that $1\le|B|\le|G|\le|R|$. Since $P$ is color-balanced, $|R|\le \lfloor\frac{n}{2}\rfloor$. Let $X$ be an arbitrary subset of $B$ such that $|X|=|R|-|G|$; note that $X=\emptyset$ when $|R|=|G|$, and $X=B$ when $|R|=\frac{n}{2}$ (where $n$ is even). Let $Y=B-X$. Let $\ell$ be a ham sandwich cut for $R$ and $G\cup X$ (pretending that the points in $G\cup X$ have the same color). Let $Q_1$ and $Q_2$ denote the set of points on each side of $\ell$; see Figure~\ref{balanced-cut-fig}(a). If $|R|$ is odd, then $|G\cup X|$ is also odd, and thus $\ell$ contains a point $r\in R$ and a point $x\in G\cup X$; see Figure~\ref{balanced-cut-fig}(b). In this case without loss of generality assume that the number of blue points in $Q_2$ is at least the number of blue points in $Q_1$; slide $\ell$ slightly such that $r$ and $x$ lie in the same side as $Q_2$, i.e. $Q_2$ is changed to $Q_2\cup\{r,x\}$. We prove that $\ell$ satisfies the statement of the theorem. The line $\ell$ does not contain any point of $P$ and by the ham sandwich cut theorem it can be computed in $O(n)$ time. 

Now we prove that both $Q_1$ and $Q_2$ are color-balanced. Let $R_1$, $G_1$, and $B_1$ be the set of red, green, and blue points in $Q_1$. Let $X_1=X\cap Q_1$ and $Y_1=Y\cap Q_1$. Note that $B_1=X_1\cup Y_1$. Similarly, define $R_2$, $G_2$, $B_2$, $X_2$, and $Y_2$ as subsets of $Q_2$.
Since $|R|=|G\cup X|$ and $\ell$ bisects both $R$ and $G\cup X$, we have $|R_1|=\lfloor|R|/2\rfloor$ and $|G_1|+|X_1|=|R_1|$. In the case that $|R|$ is odd, we add the points on $\ell$ to $Q_2$ (assuming that $|B_2|\ge |B_1|$). Thus, in either case ($|R|$ is even or odd) we have $|R_2|=\lceil|R|/2\rceil$ and $|G_2|+|X_2|=|R_2|$. Therefore,
\begin{align}
\label{PRGX}
|Q_1|& \ge |R_1|+|G_1|+|X_1|=2\lfloor|R|/2\rfloor, \nonumber\\
|Q_2|& \ge |R_2|+|G_2|+|X_2|=2\lceil|R|/2\rceil.
\end{align}
Let $t_1$ and $t_2$ be the total number of red and green points in $Q_1$ and $Q_2$, respectively. Then, we have the following inequalities: 

\begin{equation}
\label{tt}
\begin{tabular}{lll}
{$\begin{aligned}
t_1 & = |R_1|+|G_1| \\
       & = 2|R_1|-|X_1| \\
       & \ge 2|R_1|-|X| \\
       &= 2\lfloor|R|/2\rfloor -(|R|-|G|) \\
      &= \left\{
	\begin{array}{l l}
	  |G| & \quad \text{if $|R|$ is even}\\
	  |G|-1 & \quad \text{if $|R|$ is odd,}
      \end{array} \right.
\end{aligned}$}
&&
{$\begin{aligned}
t_2 & = |R_2|+|G_2|\\
       & = 2|R_2|-|X_2|\\
       & \ge 2|R_2|-|X|\\
       &= 2\lceil|R|/2\rceil -(|R|-|G|)\\
      &= \left\{
	\begin{array}{l l}
	  |G| & \quad \text{if $|R|$ is even}\\
	  |G|+1 & \quad \text{if $|R|$ is odd.}
      \end{array} \right.
\end{aligned}$}
\end{tabular}
\end{equation}
In addition, we have the following equations:
\begin{align}
\label{PtB}
|Q_1| = t_1+|B_1|\quad\text{ and }\quad
|Q_2| = t_2+|B_2|.
\end{align}

Note that $|R_1|=\lfloor|R|/2\rfloor$ and $|G_1|\le |Q_1\cap (G\cup X)|=|R_1|$, thus, by Inequality~(\ref{PRGX}) we have $|R_1|\le \lfloor|Q_1|/2\rfloor$ and $|G_1|\le \lfloor|Q_1|/2\rfloor$. Similarly, $|R_2|\le \lfloor|Q_2|/2\rfloor$ and $|G_2|\le \lfloor|Q_2|/2\rfloor$. Therefore, in order to argue that $Q_1$ and $Q_2$ are color-balanced, it only remains to show that $|B_1|\le \lfloor|Q_1|/2\rfloor$ and $|B_2|\le \lfloor|Q_2|/2\rfloor$. Note that $|B_1|,|B_2|\le|B|$ and by initial assumption $|B|\le|G|$. We differentiate between two cases where $|R|$ is even and $|R|$ is odd. If $|R|$ is even, by Inequalities~(\ref{tt}) we have $t_1,t_2\ge |G|$. Therefore, by the fact that $\max\{|B_1|,|B_2|\}\le |B|\le |G|$ and Equation~(\ref{PtB}), we have $|B_1|\le \lfloor|Q_1|/2\rfloor$ and $|B_2|\le \lfloor|Q_2|/2\rfloor$. If $|R|$ is odd, we slide $\ell$ towards $Q_1$; assuming that $|B_2|\ge |B_1|$. In addition, since $|B_1|+|B_2|=|B|$ and $|B|\ge 1$, $|B_2|\ge 1$. Thus, $|B_1|\le |B|-1\le |G|-1$, while by Inequality~(\ref{tt}), $t_1\ge |G|-1$. Therefore, Equality~(\ref{PtB}) implies that $|B_1|\le \lfloor|Q_1|/2\rfloor$. Similarly, by Inequality~(\ref{tt}) we have $t_2\ge |G|+1$ while $|B_2|\le |G|$. Thus, Equality~(\ref{PtB}) implies that $|B_2|\le \lfloor|Q_2|/2\rfloor$. Therefore, both $Q_1$ and $Q_2$ are color-balanced.

We complete the proof by providing the following upper bound on the size of $Q_1$ and $Q_2$. Since we assume that $R$ is the largest color class, $|R|\ge \lceil\frac{n}{3}\rceil$. By Inequality~(\ref{PRGX}), $\min\{|Q_1|,\allowbreak |Q_2|\}\allowbreak\ge 2\lfloor|R|/2\rfloor$, which implies that
$$\max\{|Q_1|, |Q_2|\}\le n-2\left\lfloor\frac{|R|}{2}\right\rfloor\le n-2\left(\frac{|R|-1}{2}\right)\le n-\frac{n}{3}+1=\frac{2n}{3}+1.$$
\qed\end{proof}

Therefore, by Lemma~\ref{k-to-3} and Lemma~\ref{balanced-cut-lemma}, we have proved the following theorem:
\begin{theorem}[Balanced Cut Theorem]
\label{balanced-cut-thr}
Let $P$ be a color-balanced point set of $n\ge 4$ points in general position in the plane. In $O(n)$ time we can compute a line $\ell$ such that
\begin{enumerate}
  \item $\ell$ does not contain any point of $P$.
  \item $\ell$ partitions $P$ into two point sets $Q_1$ and $Q_2$, where
      \begin{enumerate}
	\item both $Q_1$ and $Q_2$ are color-balanced,
	\item both $Q_1$ and $Q_2$ contains at most $\frac{2}{3}n+1$ points.
      \end{enumerate}
\end{enumerate}
\end{theorem}

By Theorem~\ref{Aichholzer-thr}, if $P$ has even number of points and no color is in strict majority, then $P$ admits a plane perfect matching. By Theorem~\ref{balanced-cut-thr}, we partition $P$ into two sets $Q_1$ and $Q_2$ such that in each of them no point is in strict majority. But, in order to apply the balanced cut theorem, recursively, to obtain a perfect matching on each side of the cut, we need both $Q_1$ and $Q_2$ to have an even number of points. Thus, we extend the result of Theorem~\ref{balanced-cut-thr} to a restricted version of the problem where $|P|$ is even and we are looking for a balanced cut which partitions $P$ into $Q_1$ and $Q_2$ such that both $|Q_1|$ and $|Q_2|$ are even. The following theorem describes how to find such a balanced cut.

\begin{figure}[htb]
  \centering
\setlength{\tabcolsep}{0in}
  $\begin{tabular}{cc}
\multicolumn{1}{m{.5\columnwidth}}{\centering\includegraphics[width=.35\columnwidth]{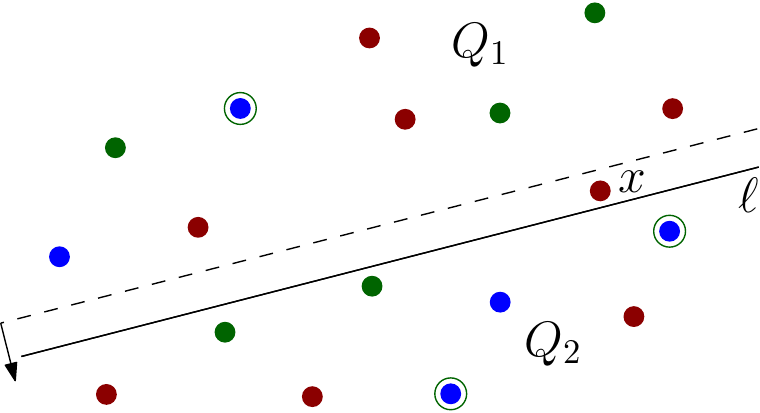}}
&\multicolumn{1}{m{.5\columnwidth}}{\centering\includegraphics[width=.35\columnwidth]{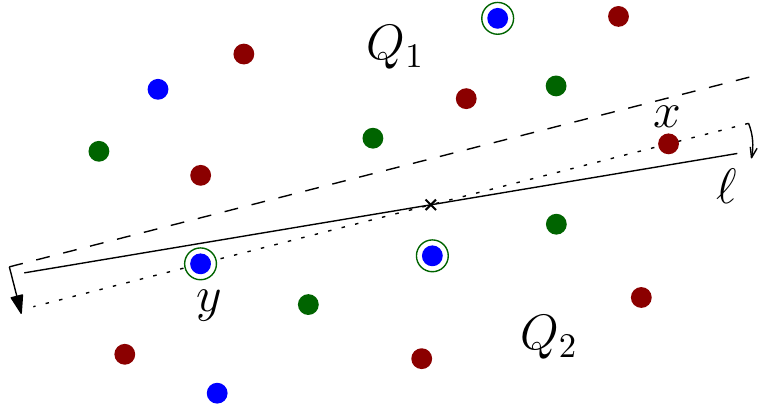}} \\
(a) & (b)
\end{tabular}$
  \caption{Updating $\ell$ to make $|Q_1|$ and $|Q_2|$ even numbers, where: (a) $\ell$ passes over one point, and (b) $\ell$ passes over two points.}
\label{even-balanced-cut-fig}
\end{figure}

\begin{theorem}
\label{even-cut-cor}
Let $P$ be a color-balanced point set of $n\ge 4$ points in general position in the plane with $n$ even and three colors. In $O(n)$ time we can compute a line $\ell$ such that
\begin{enumerate}
  \item $\ell$ does not contain any point of $P$.
  \item $\ell$ partitions $P$ into two point sets $Q_1$ and $Q_2$, where
      \begin{enumerate}
	\item both $Q_1$ and $Q_2$ are color-balanced,
	\item both $Q_1$ and $Q_2$ have even number of points,
	\item both $Q_1$ and $Q_2$ contains at most $\frac{2}{3}n+1$ points.
      \end{enumerate}
\end{enumerate}
\end{theorem}
\begin{proof}
 Let $\ell$ be the balanced cut obtained in the proof of Lemma~\ref{balanced-cut-lemma}, which divides $P$ into $Q_1$ and $Q_2$. Note that $\ell$ does not contain any point of $P$. If $|Q_1|$ is even, subsequently $|Q_2|$ is even, thus $\ell$ satisfies the statement of the theorem and we are done. Assume that $|Q_1|$ and $|Q_2|$ are odd. Let $R_1$, $G_1$, and $B_1$ be the set of red, green, and blue points in $Q_1$. Let $X_1=X\cap Q_1$ and $Y_1=Y\cap Q_1$. Note that $B_1=X_1\cup Y_1$. Similarly, define $R_2$, $G_2$, $B_2$, $X_2$, and $Y_2$ as subsets of $Q_2$. Note that $|Q_1|=|R_1|+|G_1|+|X_1|+|Y_1|$ and $|Q_2|=|R_2|+|G_2|+|X_2|+|Y_2|$. Recall that $|R_1|=|G_1|+|X_1|=\lfloor|R|/2\rfloor$ and $|R_2|=|G_2|+|X_2|=\lceil|R|/2\rceil$. Thus, $|R_1|+|G_1|+|X_1|$ and $|R_2|+|G_2|+|X_2|$ are even. In order to make $|Q_1|$ and $|Q_2|$ to be odd numbers, both $|Y_1|$ and $|Y_2|$ have to be odd numbers. Thus, $|Y_1|\ge 1$ and $|Y_2|\ge 1$, which implies that 
\begin{align}
\label{QRGXY}
|Q_1|& = |R_1|+|G_1|+|X_1|+|Y_1|\ge 2\lfloor|R|/2\rfloor+1, \nonumber\\
|Q_2|& = |R_2|+|G_2|+|X_2|+|Y_2|\ge 2\lceil|R|/2\rceil+1.
\end{align}
In addition,
\begin{align}
\label{BXY}
|B_1|& = |B|-(|X_2|+|Y_2|)\le |B|-1, \nonumber\\
|B_2|& = |B|-(|X_1|+|Y_1|)\le |B|-1.
\end{align}

Note that $Q_1$ is color-balanced. That is, $|R_1|,|G_1|,|B_1|\le \lfloor|Q_1|/2\rfloor$, where $|Q_1|$ is odd. Thus, by addition of one point (of any color) to $Q_1$, it still remain color-balanced. Therefore, we slide $\ell$ slightly towards $Q_2$ and stop as soon as it passes over a point $x\in Q_2$; see Figure~\ref{even-balanced-cut-fig}(a). If $\ell$ passes over two points $x$ and $y$, rotate $\ell$ slightly, such that $x$ lies on the same side as $Q_1$ and $y$ remains on the other side; see Figure~\ref{even-balanced-cut-fig}(b). We prove that $\ell$ satisfies the statement of the theorem. It is obvious that updating the position of $\ell$ takes $O(n)$ time. Let $Q'_1=Q_1\cup \{x\}$ and $Q'_2=Q_2- \{x\}$. By the previous argument $Q'_1$ is color-balanced. Now we show that $Q'_2$ is color-balanced as well. Note that $|Q'_2|=|Q_2|-1$, thus, by Inequality~(\ref{QRGXY}) we have $$|Q'_2|\ge 2\lceil|R|/2\rceil.$$ 

Let $R'_2$, $G'_2$, and $B'_2$ be the set of red, green, and blue points in $Q'_2$, and let $t'_2$ be the total number of red and green points in $Q'_2$. Then, 
\begin{align}
\label{PtB2}
|Q'_2|& = t'_2+|B'_2|.
\end{align}
To prove that $Q'_2$ is color-balanced we differentiate between three cases, where $x\in R_2$, $x\in G_2$, or $x\in B_2$:
\begin{itemize}
  \item $x\in R_2$. In this case: (i) $|R'_2|=|R_2|-1=\allowbreak \lceil|R|/2\rceil-1 \allowbreak\le \lfloor|Q'_2|/2\rfloor$. (ii) $|G'_2|\allowbreak =|G_2|\allowbreak \le |R_2|\allowbreak =\lceil|R|/2\rceil\le \lfloor|Q'_2|/2\rfloor$. (iii) $t'_2=t_2-1\ge |G|-1$, while $|B'_2|\allowbreak =|B_2|\allowbreak \le |B|-1\allowbreak\le |G|-1$; Inequality~(\ref{PtB2}) implies that  $|B'_2|\allowbreak\le\lfloor|Q'_2|/2\rfloor$.

  \item $x\in G_2$. In this case: (i) $|R'_2|=|R_2|=\lceil|R|/2\rceil\allowbreak \le \lfloor|Q'_2|/2\rfloor$. (ii) $|G'_2|\allowbreak =|G_2|-1\allowbreak \le |R_2|-1=\lceil|R|/2\rceil-1\allowbreak \le \lfloor|Q'_2|/2\rfloor$. (iii) $t'_2=t_2-1\ge |G|-1$, while $|B'_2|\allowbreak=|B_2|\allowbreak\le |B|-1\allowbreak \le |G|-1$; Inequality~(\ref{PtB2}) implies that  $|B'_2|\le\lfloor|Q'_2|/2\rfloor$.

  \item $x\in B_2$. In this case: (i) $|R'_2|=|R_2|=\lceil|R|/2\rceil\allowbreak \le \lfloor|Q'_2|/2\rfloor$. (ii) $|G'_2|\allowbreak =|G_2|\allowbreak \le |R_2|=\lceil|R|/2\rceil\le \lfloor|Q'_2|/2\rfloor$. (iii) $t'_2=t_2\allowbreak \ge |G|$, while $|B'_2|\allowbreak=|B_2|-1\allowbreak\le |B|-2\allowbreak \le |G|-2$; Inequality~(\ref{PtB2}) implies that $|B'_2|\le\lfloor|Q'_2|/2\rfloor$.
\end{itemize}
In all cases $|R'_2|, |G'_2|,|B'_2|\le \lfloor|Q'_2|/2\rfloor$, which imply that $Q'_2$ is color-balanced. 

As for the size condition, $$\min\{|Q'_1|,\allowbreak |Q'_2|\}\allowbreak= \min\{|Q_1|+1,\allowbreak |Q_2|-1\}\allowbreak\ge 2\lfloor|R|/2\rfloor,$$ where the last inequality resulted from Inequality~(\ref{QRGXY}). This implies that
$\max\{|Q'_1|,\allowbreak |Q'_2|\}\allowbreak\le \frac{2n}{3}+1$. 
Thus, $\ell$ satisfies the statement of the theorem, with $Q_1=Q'_1$ and $Q_2=Q'_2$.
\qed\end{proof}

We note that the constant $\frac{2}{3}$ in Theorem~\ref{even-cut-cor} is tight. Consider an equilateral triangle with vertices $a$, $b$, and $c$. Consider a color-balanced point set $P$ of $\frac{n}{3}$ red points, $\frac{n}{3}$ green points, and $\frac{n}{3}$ blue points such that the red points are located around $a$, the green points are located around $b$, and the red points are located around $c$. Then, any feasible cut for $P$ partitions it into two point sets such that one of them has size at least $\frac{2}{3}n$.

Note that both Theorem~\ref{even-cut-cor} and Theorem~\ref{Kano-thr} prove the existence of a line $\ell$ which partitions a color-balanced point set $P$ into two color-balanced point sets $Q_1$ and $Q_2$. But, there are two main differences: (i) Theorem~\ref{even-cut-cor} can be applied on any color-balanced point set $P$ in general position. Theorem~\ref{Kano-thr} is only applicable on color-balanced point sets in general position, where the points on the convex hull are monochromatic. (ii) Theorem~\ref{even-cut-cor} proves the existence of a balanced cut such that $\frac{n}{3}-1\le |Q_i|\le \frac{2n}{3}+1$, while the cut computed by Theorem~\ref{Kano-thr} is not necessarily balanced, as $2\le |Q_i|\le n-2$, where $i=1,2$. In addition, the balanced cut in Theorem~\ref{even-cut-cor} can be computed in $O(n)$ time, while the cut in Theorem~\ref{Kano-thr} is computed in $O(n\log n)$ time.

\section{Plane Colored Matching Algorithm}
\label{algorithm-section}
Let $P$ be a color-balanced point set of $n$ points in general position in the plane with respect to a partition $\{P_1,\dots,P_k\}$, where $n$ is even and $k\ge 2$. In this section we present an algorithm which computes a plane colored matching in $K_n(P_1,\dots,P_k)$ in $\Theta(n\log n)$ time. 

Let $\{C_1,\dots,C_k\}$ be a set of $k$ colors. Imagine all the points in $P_i$ are colored $C_i$ for all $1\le i\le k$. Without loss of generality, assume that $|P_1|\ge|P_2|\ge\dots\ge|P_k|$. If $k=2$, then we can compute an $RB$-matching in $O(n\log n)$ time by recursively applying the ham sandwich theorem. If $k\ge 4$, as in Lemma~\ref{k-to-3}, in $O(n)$ time, we compute a color-balanced point set $P$ with three colors. Any plane colored matching for $P$ with respect to the three colors, say $(R,G,B)$, is also a plane colored matching with respect to the coloring $C_1,\dots, C_k$. Hereafter, assume that $P$ is a color-balanced point set which is colored by three colors.

By Theorem~\ref{even-cut-cor}, in linear time we can find a line $\ell$ that partitions $P$ into two sets $Q_1$ and $Q_2$, where both $Q_1$ and $Q_2$ are color-balanced with an even number of points, such that $\max\{|Q_1|,\allowbreak |Q_2|\}\le \frac{2n}{3}+1$. Since $Q_1$ and $Q_2$ are color-balanced, by Corollary~\ref{colored-matching-cor}, both $Q_1$ and $Q_2$ admit plane colored matchings. Let $M(Q_1)$ and $M(Q_2)$ be plane colored matchings in $Q_1$ and $Q_2$, respectively. Since $Q_1$ and $Q_2$ are separated by $\ell$, $M(Q_1)\cup M(Q_2)$ is a plane colored matching for $P$. Thus, in order to compute a plane colored matching in $P$, one can compute plane colored matchings in $Q_1$ and $Q_2$ recursively, as described in Algorithm~\ref{alg1}.
The $\RGB$ function receives a colored point set $P$ of $n$ points, where $n$ is even and the points of $P$ are colored by three colors, and computes a plane colored matching in $P$. The $\BC$ function partitions $P$ into $Q_1$ and $Q_2$ where both are color-balanced and have even number of points.
 
\begin{algorithm}                      
\caption{\RGB$(P)$}          
\label{alg1} 
\require{a color-balanced point set $P$ with respect to $(R,G,B)$, where $|P|$ is even.}\\
\ensure{a plane colored matching in $P$.}
\begin{algorithmic}[1]
    \If {$P$ is 2-colored}
	  \State \Return $\RB(P)$
    \Else
    \State $\ell \gets \BC(P)$
    \State $Q_1\gets$ points of $P$ to the left of $\ell$
    \State $Q_2\gets$ points of $P$ to the right of $\ell$
    \State \Return $\RGB(Q_1)\cup\RGB(Q_2)$
    \EndIf
\end{algorithmic}
\end{algorithm}

Now we analyze the running time of the algorithm. If $k=2$, then in $O(n\log n)$ time we can find a plane $RB$-matching for $P$. If $k\ge 4$, then by Lemma~\ref{k-to-3}, in $O(n)$ time we reduce the $k$-colored problem to a 3-colored problem. Then, the function $\RGB$ computes a plane colored matching in $P$. Let $T(n)$ denote the running time of $\RGB$ on the 3-colored point set $P$, where $|P|=n$. 
As described in Theorem~\ref{balanced-cut-thr} and Theorem~\ref{even-cut-cor}, in linear time we can find a balanced cut $\ell$ in line 4 in Algorithm~\ref{alg1}. The recursive calls to $\RGB$ function in line 7 takes $T(|Q_1|)$ and $T(|Q_2|)$ time. Thus, the running time of $\RGB$ can be expressed by the following recurrence:
$$T(n)= T(|Q_1|)+T(|Q_2|)+O(n).$$

Since $|Q_1|,|Q_2| \le\frac{2n}{3}+1$ and $|Q_1|+|Q_2|=n$, this recurrence solves to $T(n)=O(n\log n)$.
\begin{theorem}
 \label{running-time-thr}
Given a color-balanced point set $P$ of size $n$ in general position in the plane with $n$ even, a plane colored matching in $P$ can be computed in $\Theta(n\log n)$ time.
\end{theorem}

\subsection{Maximum Matching}
\label{maximum-section}
If $P$ is not color-balanced, then $K_n(P_1,\dots,P_k)$ does not admit a perfect matching. In this case we compute a maximum matching. 
\begin{theorem}
\label{max-plane-matching}
Given a colored point set $P$ of size $n$ in general position in the plane, a maximum plane colored matching $M$ in $P$ can be computed optimally in $\Theta(n+|M|\log |M|)$ time.
\end{theorem}

\begin{proof}
Let $\{P_1,\dots,P_k\}$, where $k\ge 2$, be a partition of the points in $P$ such that the points in $P_i$ colored $C_i$ for $1\le i\le k$. Without loss of generality assume that $|P_1|\ge \dots\ge|P_k|$. If $|P_1|\le\lfloor|P|/2\rfloor$, then $P$ is color-balanced, and hence, by Theorem~\ref{running-time-thr} we can compute a plane colored matching in $\Theta(n\log n)$ time. Assume $|P_1|>\lfloor|P|/2\rfloor$. Then $P$ is not color-balanced, and hence, $P$ does not admit a perfect matching. In this case, by Theorem~\ref{Sitton-thr}, the size of any maximum matching, say $M$, is
$$|M|=\sum_{i=2}^{k}{|P_i|}.$$ 
Let $P'_1$ be any arbitrary subset of $P_1$ such that $|P'_1|=|P_2|+\dots+|P_{k}|$. Imagine the points in $P_2\cup\dots\cup P_{k}$ are colored red and the points in $P'_{1}$ are colored blue. Let $P'=P'_1\cup P_2\cup\dots\cup P_{k}$. Any plane $RB$-matching in $P'$ is a maximum plane colored matching in $P$, and has $|P_1|+\dots+|P_k|=|M|$ edges. An $RB$-matching of size $|M|$ can be computed in $\Theta(|M|\log |M|)$ time.  
\qed\end{proof}

\begin{theorem}
\label{max-matching}
Given any complete multipartite graph $K_n(V_1,\dots, V_k)$ on $n$ vertices and $k\ge 2$, a maximum matching in $K_n(V_1,\dots, V_k)$ can be computed optimally in $\Theta(n)$ time.
\end{theorem}

\begin{proof}
 If $n$ is odd, then by Theorem~\ref{Sitton-thr}, we can remove a vertex from the largest vertex set, without changing the size of a maximum matching. Thus, assume that $n$ is even. Without loss of generality assume that $|V_1|\ge |V_2|\ge \dots\ge|V_k|$. If $|V_1|\ge n/2$, then let $R$ be an arbitrary subset of $V_1$ such that $|R|=|V_2|+\dots+|V_{k}|$, and let $B=V_2\cup \dots\cup V_k$. Then, any maximal matching in $K_n(R,B)$\textemdash which is also a perfect matching\textemdash is a maximum matching in $K_n(V_1,\dots, V_k)$. 

If $|V_1|< n/2$, then by a similar argument as in Lemma~\ref{k-to-3}, in $O(n)$ time we merge $V_1, \dots, V_k$ to obtain a partition $\{R,G,B\}$ of vertices, such that $\max\{|R|,\allowbreak|G|,\allowbreak|B|\}\le n/2$ and $K_n(R,G,B)$ is a subgraph of $K_n(V_1,\dots, V_k)$. Now we describe how to compute a perfect matching in $K_n(R,G,B)$. Without loss of generality assume that $|R|\ge |G|\ge |B|$. Let $m=n-2\cdot|R|$; observe that $m$ is an even number. Since $m=n-2\cdot|R|\le n-(|R|+|G|)$, we have $m\le|B|$, and subsequently $m\le|G|$. Let $G'$ (resp. $B'$) be an arbitrary subset of $G$ (resp. $B$) of size $m/2$. Thus, $|B'|=|G'|=m/2$. Let $G''=G\setminus G'$ and $B''=B\setminus B'$. Thus, 
\begin{align}
|G''\cup B''|&=\notag n-|R|-|G'\cup B'|=\notag n-|R|-m=\notag n-|R|-(n-2|R|)=\notag |R|.
\end{align}

Thus, both $K_m(G',B')$ and $K_{n\text{-}m}(R, G''\cup B'')$ have perfect matchings. Therefore, the union of any maximal matching in $K_m(G',B')$ and any maximal matching in $K_{n\text{-}m}(R, G''\cup B'')$ is a perfect matching in $K_n(R,G,B)$, and subsequently in $K_n(V_1,\dots, V_k)$. 

Since a maximal matching in a complete bipartite graph is also a maximum matching and can be computed in linear time, the presented algorithm takes $O(n)$ time.   
\qed\end{proof}

\bibliographystyle{abbrv}
\bibliography{../thesis}

%% file: chapters/ch6-matchingpacking.tex
\chapter{Packing Matchings into a Point Set}
\label{ch:mp}

Given a set $P$ of $n$ points in the plane, where $n$ is even, we consider the following question: How many plane perfect matchings can be packed into $P$? For points in general position we prove the lower bound of $\lfloor\log_2{n}\rfloor-1$. For some special configurations of point sets, we give the exact answer. We also consider some restricted variants of this problem.

\vspace{10pt}
This chapter is published in the journal of Discrete Mathematics {\&} Theoretical Computer Science~\cite{Biniaz2015-packing}. 

\section{Introduction}
\label{mp:introduction-section}
Let $P$ be a set of $n$ points in general position in the plane (no three points on a line). A {\em geometric graph} $G=(P,E)$ is a graph whose vertex set is $P$ and whose edge set $E$ is a set of straight-line segments with endpoints in $P$. We say that two edges of $G$ {\em cross} each other if they have a point in common that is interior to both edges. Two edges are {\em disjoint} if they have no point in common. A subgraph $S$ of $G$ is said to be {\em plane} ({\em non-crossing} or {\em crossing-free}) if its edges do not cross. A {\em plane matching} is a plane graph consisting of pairwise disjoint edges. Two subgraphs $S_1$ and $S_2$ are {\em edge-disjoint} if they do not share any edge. A {\em complete geometric graph} $\Kn{(P)}$ is a geometric graph on $P$ which contains a straight-line edge between every pair of points in $P$. 

We say that a set of subgraphs of $\Kn{(P)}$ is {\em packed into} $\Kn{(P)}$, if the subgraphs in the set are pairwise edge-disjoint. In a packing problem, we ask for the largest number of subgraphs of a given type that can be packed into $\Kn{(P)}$. Among all subgraphs of $\Kn{(P)}$, {\em plane perfect matchings}, {\em plane spanning trees}, and {\em plane spanning paths} are of interest~\cite{Aichholzer2015-disjoint-compatibility, Aichholzer2009-compatible-matchings, Aichholzer2010-edge-removal, Aichholzer2014-packing, Aloupis2013, Bose2006, Ishaque2013, Nash-Williams1961}. That is, one may look for the maximum number of plane spanning trees, plane Hamiltonian paths, or plane perfect matchings that can be packed into $\Kn{(P)}$. Since $\Kn{(P)}$ has $\frac{n(n-1)}{2}$ edges, at most $\frac{n}{2}$ spanning trees, at most $\frac{n}{2}$ spanning paths, and at most $n-1$ perfect matchings can be packed into it. In this chapter we consider perfect matchings. A {\em perfect matching} in $\Kn{(P)}$ is a set of edges that do not share any endpoint and cover all the points in $P$. 

A long-standing open question is to determine if the edges of $\Kn{(P)}$ (where $n$ is even) can be partitioned into $\frac{n}{2}$ plane spanning trees. In other words, is it possible to pack $\frac{n}{2}$ plane spanning trees into $\Kn{(P)}$? If $P$ is in convex position, the answer in the affirmative follows from the result of Bernhart and Kanien~\cite{Bernhart1979}. For $P$ in general position, Aichholzer et al.~\cite{Aichholzer2014-packing} prove that $\Omega(\sqrt{n})$ plane spanning trees can be packed into $\Kn{(P)}$. They also show the existence of at least 2 edge-disjoint plane spanning paths. 

In this chapter we consider a closely related question: How many plane perfect matchings can be packed into $\Kn{(P)}$, where $P$ is a set of $n$ points in general position in the plane, with $n$ even? 

\addtocontents{toc}{\protect\setcounter{tocdepth}{1}}
\subsection{Previous Work}
\addtocontents{toc}{\protect\setcounter{tocdepth}{2}}
\label{mp:previous-work-section}
\subsubsection{Existence of Plane Subgraphs}
The existence of certain plane subgraphs in a geometric graph on a set $P$ of $n$ points is one of the classic problems in combinatorial and computational geometry. 

One of the extremal problems in geometric graphs which was first studied by Avital and Hanani~\cite{Avital1966}, Kuptiz~\cite{Kupitz1979}, Erd\H{o}s~\cite{Erdos1946}, and Perles (see reference \cite{Toth1999}) is the following. What is the smallest number $\e{k}{n}$ such that any geometric graph with $n$ vertices and more than $\e{k}{n}$ edges contains $k+1$ pairwise disjoint edges, i.e., a plane matching of size at least $k+1$. Note that $k\le\lfloor n/2\rfloor-1$. By a result of Hopf and Pannwitz~\cite{Hopf1934}, Sutherland~\cite{Sutherland1935}, and Erd\H{o}s~\cite{Erdos1946}, $\e{1}{n}=n$, i.e., any geometric graph with $n+1$ edges contains a pair of disjoint edges, and there are some geometric graphs with $n$ edges which do not contain any pair of disjoint edges. 

Alon and Erd\H{o}s~\cite{Alon1989} proved that $\e{2}{n}<6n-5$, i.e., any geometric graph with $n$ vertices and at least $6n-5$ edges contains a plane matching of size three. This bound was improved to $\e{2}{n}\le 3n$ by Goddard et al.~\cite{Goddard1996}. Recently {\v{C}}ern{\'{y}}~\cite{Cerny2005} proved that $\e{2}{n}\le \lfloor 2.5 n\rfloor$; while the lower bound of $\e{2}{n}\ge \lceil2.5n\rceil-3$ is due to Perles (see \cite{Cerny2005}). For $\e{3}{n}$, Goddard et al.~\cite{Goddard1996} showed that $3.5n-6\le\e{3}{n}\le 10n$, which was improved by T{\'{o}}th and Valtr~\cite{Toth1999} to $4n-9\le\e{3}{n}\le 8.5n$.

For general values of $k$, Akiyama and Alon~\cite{Akiyama1989} gave the upper bound of $\e{k}{n}=O(n^{2-1/(k+1)})$. Goddard et al.~\cite{Goddard1996} improved the bound to $\e{k}{n}=O(n(\log n)^{k-3})$. Pach and T{\"{o}}r{\H{o}}csik~\cite{Pach1994} obtained the upper bound of $\e{k}{n} \le k^4n$; which is the first upper bound that is linear in $n$. The upper bound was improved to $k^3(n+1)$ by T{\'{o}}th and Valtr~\cite{Toth1999}; they also gave the lower bound of $\e{k}{n}\ge \frac{3}{2}(k-1)n-2k^2$. T{\'{o}}th~\cite{Toth2000} improved the upper bound to $\e{k}{n}\le 2^9k^2n$, where the constant has been improved to $2^8$ by Felsner~\cite{Felsner2004}. It is conjectured that $\e{k}{n}\le ckn$ for some constant $c$.

For the maximum value of $k$, i.e., $k=\frac{n}{2}-1$, with $n$ even, Aichholzer et al.~\cite{Aichholzer2010-edge-removal} showed that $\e{n/2-1}{n}={n \choose 2}-\frac{n}{2}=\frac{n(n-2)}{2}$. That is, by removing $\frac{n}{2}-1$ edges from any complete geometric graph, the resulting graph has $k+1=\frac{n}{2}$ disjoint edges, i.e., a plane perfect matching. This bound is tight; there exist complete geometric graphs, such that by removing $\frac{n}{2}$ edges, the resulting graph does not have any plane perfect matching. Similar bounds were obtained by Kupitz and Perles for complete convex graphs, i.e., complete graphs of point sets in convex position. Kupitz and Perles showed that any convex geometric graph with $n$ vertices and more than $kn$ edges contains $k+1$ pairwise disjoint edges; see \cite{Goddard1996} (see also \cite{Akiyama1989} and \cite{Alon1989}). In particular, in the convex case, $2n + 1$ edges guarantee a plane matching of size three. In addition, Keller and Perles~\cite{Keller2012} gave a characterization of all sets of $\frac{n}{2}$ edges whose removal prevents the resulting graph from having a plane perfect matching.

\v{C}ern{\'{y}} et al.~\cite{Cerny2007} considered the existence of Hamiltonian paths in geometric graphs. They showed that after removing at most $\sqrt{n}/(2\sqrt{2})$ edges from any complete geometric graph of $n$ vertices, the resulting graph still contains a plane Hamiltonian path. Aichholzer et al.~\cite{Aichholzer2010-edge-removal} obtained tight bounds on the maximum number of edges that can be removed from a complete geometric graph, such that the resulting graph contains a certain plane subgraph; they considered plane perfect matchings, plane subtrees of a given size, and triangulations. 
\subsubsection{Counting Plane Graphs}
The number of plane graphs of a given type in a set of $n$ points is also of interest. In 1980, Newborn and Moser~\cite{Newborn1980} asked for the maximal number of plane Hamiltonian cycles; they give an upper bound of $2\cdot6^{n-2}\lfloor\frac{n}{2}\rfloor!$, but conjecture that it should be of the form $c^n$, for some constant $c$. In 1982, Ajtai et al.~\cite{Ajtai1982} proved that the number of plane graphs is at most $10^{13n}$. Every plane graph is a subgraph of some triangulation (with at most $3n-6$ edges). Since a triangulation has at most $2^{3n-6}$ plane subgraphs, as noted in~\cite{Garcia2000}, any bound of $\alpha^n$ on the number of triangulations implies a bound of $2^{3n-6}\alpha^n<(8\alpha)^n$ on the number of plane graphs. The best known upper bound of $30^n$, for the number of triangulations is due to Sharir and Sheffer~\cite{Sharir2011}. This implies the bound $240^n$ for plane graphs. As for plane perfect matchings, since a perfect matching has $\frac{n}{2}$ edges, Dumitrescu~\cite{Dumitrescu1999} obtained an upper bound of ${{3n-6} \choose {n/2}}\alpha^n\le (3.87\alpha)^n$, where $\alpha=30$. Sharir and Welzl~\cite{Sharir2006} improved this bound to $O(10.05^n)$. They also showed that the number of all (not necessarily perfect) plane matchings is at most $O(10.43^n)$. 

Garc{\'{\i}}a et al.~\cite{Garcia2000} showed that the number of plane perfect matchings of a fixed size set of points in the plane is minimum when the points are in convex position. Motzkin~\cite{Motzkin1948} showed that points in convex position have $C_{n/2}$ many perfect matchings (classically referred to as non-crossing configurations of chords on a circle), where $C_{n/2}$ is the $(n/2)^{th}$ {\em Catalan number}; $C_{n/2}=\Theta(n^{-3/2}2^n)$. Thus, the number of plane perfect matchings of $n$ points in the plane is at least $C_{n/2}$. Garc{\'{\i}}a et al.~\cite{Garcia2000} presented a configuration of $n$ points in the plane which has $\Omega(n^{-4}3^n)$ many plane perfect matchings. See Table~\ref{mp:table1}.

\subsubsection{Counting Edge-Disjoint Plane Graphs}

The number of edge-disjoint plane graphs of a given type in a point set $P$ of $n$ points is also of interest. Nash-Williams~\cite{Nash-Williams1961} and Tutte~\cite{Tutte1961} independently considered the number of (not necessarily plane) spanning trees. They obtained necessary and sufficient conditions for a graph to have $k$ edge-disjoint spanning trees. Kundu~\cite{Kundu1974} showed that any $k$-edge-connected graph contains at least $\lceil\frac{k-1}{2}\rceil$ edge-disjoint spanning trees. 

As for the plane spanning trees a long-standing open question is to determine if the edges of $\Kn{(P)}$ (where $n$ is even) can be partitioned into $\frac{n}{2}$ plane spanning trees. In other words, is it possible to pack $\frac{n}{2}$ plane spanning trees into $\Kn{(P)}$? If $P$ is in convex position, the answer in the affirmative follows from the result of Bernhart and Kanien~\cite{Bernhart1979}. In \cite{Bose2006}, the authors characterize the partitions of the complete convex graph into plane spanning trees. They also describe a sufficient condition, which generalizes the convex case, for points in general position. Aichholzer et al.~\cite{Aichholzer2014-packing} showed that if the convex hull of $P$ contains $h$ vertices, then $\Kn{(P)}$ contains at least $\lfloor\frac{h}{2}\rfloor$ edge-disjoint plane spanning trees, and if $P$ is in a ``regular wheel configuration'', $\Kn{(P)}$ can be partitioned into $\frac{n}{2}$ spanning trees. For $P$ in general position they showed that $\Kn{(P)}$ contains $\Omega(\sqrt{n})$ edge-disjoint plane spanning trees. They obtained the following trade-off between the number of edge-disjoint plane spanning trees and the maximum vertex degree in each tree: For any $k\le \sqrt{n/12}$, $\Kn{(P)}$ has $k$ edge-disjoint plane spanning trees with maximum vertex degree $O(k^2)$ and diameter $O(\log(n/k^2))$. They also showed the existence of at least 2 edge-disjoint plane Hamiltonian paths. 
\addtocontents{toc}{\protect\setcounter{tocdepth}{1}}
\subsection{Our Results}
\addtocontents{toc}{\protect\setcounter{tocdepth}{2}}
\label{mp:our-results-section}

Given a set $P$ of $n$ points in the plane, with $n$ even, we consider the problem of packing plane perfect matchings into $\Kn{(P)}$. 
From now on, a {\em matching} will be a {\em perfect matching}. 

In Section~\ref{mp:edge-disjoint-plane-section} we prove bounds on the number of plane matchings that can be packed into $\Kn{(P)}$. 
In Section~\ref{mp:convex-position-section} we show that if $P$ is in convex position, then $\frac{n}{2}$ plane matchings can be packed into $\Kn{(P)}$; this bound is tight. 

The points in wheel configurations are considered in Section~\ref{mp:wheel-section}. We show that if $P$ is in regular wheel configuration, then $\frac{n}{2}-1$ edge-disjoint plane matchings can be packed into $\Kn{(P)}$; this bound is tight as well. In addition, for a fixed size set of points, we give a wheel configuration of the points which contains at most $\lceil\frac{n}{3}\rceil$ edge-disjoint plane matchings. 

Point sets in general position are considered in Section~\ref{mp:general-position-section}. We show how to find three edge-disjoint plane matchings in any set of at least 8 points. If $n$ is a power of two, we prove that $\Kn{(P)}$ contains at least $\log_2n$ many edge-disjoint plane matchings. For the general case, where $n$ is an even number, we prove that $\Kn{(P)}$ contains at least $\lceil\log_2n\rceil-2$ edge-disjoint plane matchings. 

In Section~\ref{mp:non-crossing-matching-section} we count the number of pairwise non-crossing plane matchings. Two plane matchings $M_1$ and $M_2$ are called {\em non-crossing} (or {\em compatible}) if the edges of $M_1$ and $M_2$ do not cross each other. We show that $\Kn{(P)}$ contains at least two and at most five non-crossing plane matchings; these bounds are tight. Table~\ref{mp:table1} summarizes the results. 

In Section~\ref{mp:persistency-section} we study the concept of {\em matching persistency} in a graph. A graph $G$ is called {\em matching-persistent}, if by removing any perfect matching $M$ from $G$, the resulting graph, $G-M$, still contains a perfect matching. We define the {\em plane matching persistency} of a point set $P$, denoted by $\pmp{P}$, to be the smallest number of edge-disjoint plane matchings such that, if we remove them from $\Kn{(P)}$ the resulting graph does not have any plane perfect matching. In other words, $\pmp{P}=|\mathcal{M}|$, where $\mathcal{M}$ is the smallest set of edge-disjoint plane matchings such that $\Kn{(P)}-\bigcup_{M\in\mathcal{M}}{M}$ does not have any plane perfect matching. Here, the challenge is to find point sets with high plane matching persistency. We show that $\pmp{P}\ge 2$ for all point sets $P$. We give a configuration of $P$ with $\pmp{P}\ge 3$. 
Concluding remarks and open problems are presented in Section~\ref{mp:conclusion}.

\begin{table}
\caption{Number of plane perfect matchings in a point set $P$ of $n$ points ($n$ is even).}
\label{mp:table1}
\centering
    \begin{tabular}{|l|c|c|c|c|}
         \hline
             Matching 	& $\forall P: \ge$ &$\exists P:\le$&$\exists P: \ge$&$\forall P:\le$  \\ \hline\hline
             total& 	$2^n$\cite{Garcia2000, Motzkin1948}&$2^n$\cite{Motzkin1948}&$3^n$\cite{Garcia2000}& $O(10.05^n)$\cite{Sharir2006}\\\hline\hline
		 edge-disjoint&$\lfloor\log_2{n}\rfloor-1$ &$\lceil \frac{n}{3}\rceil$& $\frac{n}{2}$&$n-1$\\
             plane edge-disj.& 2 & 2&5&5 \\
         \hline
    \end{tabular}
\end{table}

\section{Preliminaries}
\label{mp:preliminaries}
\addtocontents{toc}{\protect\setcounter{tocdepth}{1}}
\subsection{Graph-Theoretical Background}
\addtocontents{toc}{\protect\setcounter{tocdepth}{2}}
\label{mp:graph-background-section}
Consider a graph $G=(V,E)$ with vertex set $V$ and edge set $E$. If $G$ is a complete graph on a vertex set $V$ of size $n$, then $G$ is denoted by $K_n$.
A $k${\em -factor} is a regular graph of degree $k$. If $G$ is the union of pairwise edge-disjoint $k$-factors, their union is called a $k${\em-factorization} and $G$ itself is $k${\em-factorable}~\cite{Harary1991}. A {\em matching} in a graph $G$ is a set of edges that do not share vertices. A perfect matching of $G$ is a 1-factor of $G$. 
In this chapter only perfect matchings are considered and they are simply called matchings. Since a perfect matching is a regular graph of degree one, it is a $1$-factor. It is well-known that for $n$ even, the complete graph $K_n$ is 1-factorable (See~\cite{Harary1991}).
Note that $K_n$ has $\frac{n(n-1)}{2}$ edges and every $1$-factor has $\frac{n}{2}$ edges. Thus, $K_n$ can be partitioned into at most $n-1$ edge-disjoint perfect matchings.

On the other hand it is well-known that the edges of a complete graph $K_n$, where $n$ is even, can be colored by $n-1$ colors such that any two adjacent edges have a different color. Each color is assigned to $\frac{n}{2}$ edges, so that each color defines a $1$-factor. The following geometric construction of a coloring, which uses a ``regular wheel configuration'', is provided in \cite{Soifer2009}. In a regular wheel configuration, $n-1$ equally spaced points are placed on a circle and one point is placed at the center of the circle. For each color class, include an edge $e$ from the center to one of the boundary vertices, and all of the edges perpendicular to the line through $e$, connecting pairs of boundary vertices.

The number of perfect matchings in a complete graph $K_n$ (with $n$ even), denoted by $M(n)$, is given by the double factorial; $M(n)=(n-1)!!$ \cite{Callan2009}, where $(n-1)!!=1\cdot 3\cdot 5\cdots (n-3)\cdot(n-1)$. 
\addtocontents{toc}{\protect\setcounter{tocdepth}{1}}
\subsection{Plane Matchings in Colored Point Sets}
\addtocontents{toc}{\protect\setcounter{tocdepth}{2}}
\label{mp:colored-matching-section}
Let $P$ be a set of $n$ colored points in general position in the plane with $n$ even. A {\em colored matching} of $P$, is a perfect matching such that every edge connects two points of distinct colors. A {\em plane colored matching} is a colored matching which is non-crossing.  
A special case of a plane colored matching, where $P$ is partitioned into a set $R$ of $\frac{n}{2}$ red points and a set $B$ of $\frac{n}{2}$ blue points, is called {\em plane bichromatic matching}, also known as {\em red-blue matching} ($RB${\em-matching}). In other words, an $RB$-matching of $P$ is a non-crossing perfect matching such that every edge connects a red point to a blue point. It is well-known that if no three points of $P$ are collinear, then $P$ has an $RB$-matching~\cite{Putnam1979}. As shown in Figure~\ref{mp:RB-fig}(a), some point sets have a unique $RB$-matching. Hershberger and Suri~\cite{Hershberger1992} construct an $RB$-matching in $O(n\log n)$ time, which is optimal.

\begin{figure}[htb]
  \centering
\setlength{\tabcolsep}{0in}
  $\begin{tabular}{cc}
\multicolumn{1}{m{.5\columnwidth}}{\centering\includegraphics[width=.28\columnwidth]{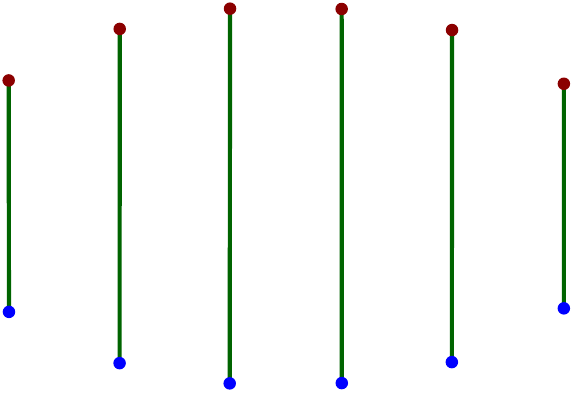}}
&\multicolumn{1}{m{.5\columnwidth}}{\centering\includegraphics[width=.3\columnwidth]{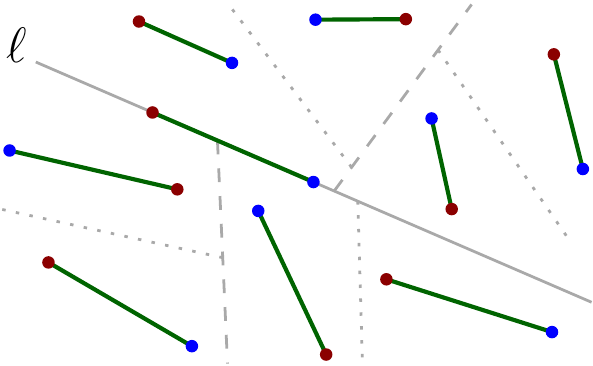}} \\
(a) & (b)
\end{tabular}$
  \caption{(a) A point set with a unique $RB$-matching, (b) Recursive ham sandwich cuts: first cut is in solid, second-level cuts are in dashed, and third-level cuts are in dotted lines.}
\label{mp:RB-fig}
\end{figure}

We review some proofs for the existence of a plane perfect matching between $R$ and $B$:
\begin{itemize}
\item {\Min{R}{B}:} Consider a matching $M$ between $R$ and $B$ which minimizes the total Euclidean length of the edges. The matching $M$ is plane. To prove this, suppose that two line segments $r_1b_1$ and $r_2b_2$ in $M$ intersect. By the triangle inequality, $|r_1b_2|+|r_2b_1|<|r_1b_1|+|r_2b_2|$. This implies that by replacing $r_1b_1$ and $r_2b_2$ in $M$ by $r_1b_2$ and $r_2b_1$, the total length of the matching is decreased; which is a contradiction.

\item {\Cut{R}{B}:} The {\em ham sandwich theorem} implies that there is a line $\ell$, known as a {\em ham sandwich cut}, that splits both $R$ and $B$ exactly in half; if the size of $R$ and $B$ is odd, the line passes through one of each. Match the two points on $\ell$ (if there are any) and recursively solve the problem on each side of $\ell$; the recursion stops when each subset has one red point and one blue point. By matching these two points in all subsets, a plane perfect matching for $P$ is obtained. See Figure~\ref{mp:RB-fig}(b). A ham sandwich cut can be computed in $O(n)$ time \cite{Lo1994}, and hence the running time can be expressed as the recurrence $T(n)=O(n)+2\cdot T(\lfloor\frac{n}{2}\rfloor)$. Therefore, an $RB$-matching can be computed in $O(n\log n)$ time. 

\item {\Tangent{R}{B}:} If $R$ and $B$ are separated by a line, we can compute an $RB$-matching in the following way. W.l.o.g. assume that $R$ and $B$ are separated by a vertical line $\ell$. Let \CH{R} and \CH{B} denote the convex hulls of $R$ and $B$. Compute the upper tangent $rb$ of \CH{R} and \CH{B} where $r\in R$ and $b\in B$. Match $r$ and $b$, and recursively solve the problem for $R-\{r\}$ and $B-\{b\}$; the recursion stops when the two subsets are empty. In each iteration, all the remaining points are below the line passing through $r$ and $b$, thus, the line segments representing a matched pair in the successor iterations do not cross $rb$. Therefore, the resulting matching is plane.
\end{itemize}

Consider a set $P$ of $n$ points where $n$ is even, and a partition $\{P_1,\dots,P_k\}$ of $P$ into $k$ color classes. Sufficient and necessary conditions for the existence of a colored matching in $P$ follows from the following theorem by Sitton~\cite{Sitton1996}:

\begin{theorem}[Sitton~\cite{Sitton1996}]
\label{mp:Sitton}
Let $K_{n_1,\dots,n_k}$ be a complete multipartite graph with $n$ vertices, where $n_1\le\dots\le n_k$. If $n_k\allowbreak \le \allowbreak n_1+\dots+\allowbreak n_{k-1}$, then $K_{n_1,\dots,n_k}$ has a matching of size $\lfloor\frac{n}{2}\rfloor$. 
\end{theorem}

Aichholzer et al.~\cite{Aichholzer2010-edge-removal} showed that if $K_{n_1,\dots,n_k}$ is a geometric graph corresponding to a colored point set $P$, then the minimum-weight colored matching of $P$ is non-crossing. Specifically, they extend the proof of 2-colored point sets to multi-colored point sets:

\begin{theorem}[Aichholzer et al.~\cite{Aichholzer2010-edge-removal}]
\label{mp:Aichholzer}
Let $P$ be a set of colored points in general position in the plane with $|P|$ even. Then $P$
admits a non-crossing perfect matching such that every edge connects two points of distinct colors if and only if at most half the points in $P$ have the same color.
\end{theorem}

\section{Packing Plane Matchings into Point Sets}
\label{mp:edge-disjoint-plane-section}
Let $P$ be a set of $n$ points in the plane with $n$ even. In this section we prove lower bounds on the number of plane matchings that can be packed into $\Kn{(P)}$. It is obvious that every point set has at least one plane matching, because a minimum weight perfect matching in $\Kn{(P)}$, denoted by ${\sf Min}(P)$, is plane. A trivial lower bound of 2 (for $n\ge4$) is obtained from a minimum weight Hamiltonian cycle in $\Kn{(P)}$, because this cycle is plane and consists of two edge-disjoint matchings. We consider points in convex position (Section~\ref{mp:convex-position-section}), wheel configuration (Section~\ref{mp:wheel-section}), and general position (Section~\ref{mp:general-position-section}).

\subsection{Points in Convex Position}
\label{mp:convex-position-section}
In this section we consider points in convex position. We show that if $P$ is in convex position, $\frac{n}{2}$ plane matchings can be packed into $\Kn{(P)}$; this bound is tight.
\begin{lemma}
\label{mp:two-convex-edges}
 If $P$ is in convex position, where $|P|$ is even and $|P|\ge4$, then every plane matching in $P$ contains at least two edges of \CH{P}.
\end{lemma}
\begin{proof}
Let $M$ be a plane matching in $P$. We prove this lemma by induction on the size of $P$. If $|P|=4$, then $|M|=2$. None of the diagonals of $P$ can be in $M$, thus, the two edges in $M$ belong to \CH{P}. If $|P|>4$ then $|M|\ge 3$. If all edges of $M$ are edges of \CH{P}, then the claim in the lemma holds. Assume that M contains a diagonal edge $pq$, where $pq$ is not an edge
of \CH{P}. Let $P_1$ and $P_2$ be the sets of points of $P$ on each side of $\ell(p,q)$ (both including $p$ and $q$). Let $M_1$ and $M_2$ be the edges of $M$ in $P_1$ and $P_2$, respectively. It is obvious that $P_1$ (resp. $P_2$) is in convex position and $M_1$ (resp. $M_2$) is a plane matching in $P_1$ (resp. $P_2$). By the induction hypothesis $M_1$ (resp. $M_2$) contains two edges of \CH{P_1} (resp. \CH{P_2}). Since $\CH{P}=\CH{P_1}\cup \CH{P_2}$ and $|M_1\cap M_2|=1$, $M$ contains at least two edges of \CH{P}.
\end{proof}

\begin{theorem}
\label{mp:convex}
For any set $P$ of $n$ points in convex position in the plane, with $n$ even, the maximum number of plane matchings that can be packed into $\Kn{(P)}$ is $\frac{n}{2}$.
\end{theorem}
\begin{proof}
By Lemma~\ref{mp:two-convex-edges}, every plane matching in $P$ contains at least two edges of \CH{P}. On the other hand, \CH{P} has $n$ edges. Therefore, the number of plane matchings that can be packed into $\Kn{(P)}$ is at most $\frac{n}{2}$. Bernhart and Kainen~\cite[Theorem 3.4]{Bernhart1979} showed that the book thickness of the complete graph is $\frac{n}{2}$.
Their construction of $\frac{n}{2}$ edge-disjoint paths directly
carries over to packing the same amount of plane matchings into $\Kn{(P)}$.

In order to be self-contained, we show how to pack $\frac{n}{2}$ plane matchings into $\Kn{(P)}$.
Let $P=\{p_0,\dots,p_{n-1}\}$, and w.l.o.g. assume that $p_0,p_1,\dots,p_{n-1}$ is the radial ordering of the points in $P$ with respect to a fixed point in the interior of \CH{P}. For each $p_i$ in the radial ordering, where $0\le i<\frac{n}{2}$, let $M_i=\{p_{i+j-1}p_{n+i-j}:j=1,\dots,\frac{n}{2}\}$ (all indices are modulo $n$). Informally speaking, $M_i$ is a plane perfect matching obtained from edge $p_ip_{i-1}$ and all edges parallel to $p_ip_{i-1}$; see Figure~\ref{mp:convex-fig}. Let $\mathcal{M}=\{M_i:i=0,\dots,\frac{n}{2}-1\}$. The matchings in $\mathcal{M}$ are plane and pairwise edge-disjoint. Thus, $\mathcal{M}$ is a set of $\frac{n}{2}$ plane matchings that can be packed into $\Kn{(P)}$.
\end{proof}

\begin{figure}[htb]
  \centering
\setlength{\tabcolsep}{0in}
\includegraphics[width=.27\columnwidth]{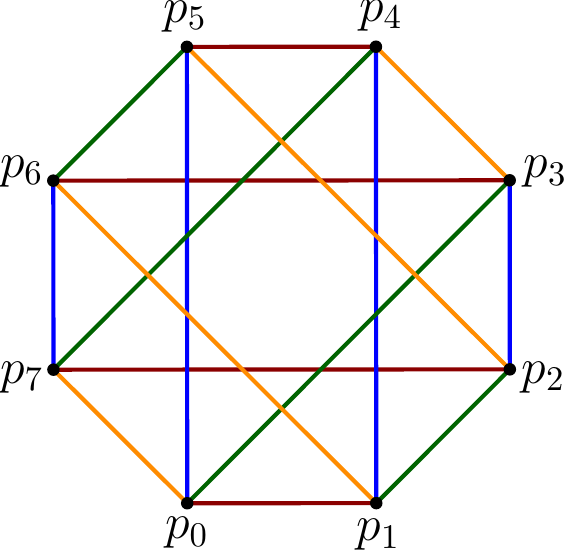}
  \caption{Points in convex position.}
\label{mp:convex-fig}
\end{figure}

\subsection{Points in Wheel Configurations}
\label{mp:wheel-section}
A point set $P$ of $n$ points is said to be in ``regular wheel configuration'' in the plane, if $n-1$ points of $P$ are equally spaced on a circle $C$ and one point of $P$ is at the center of $C$.
We introduce a variation of the regular wheel configuration as follows. 
Let the point set $P$ be partitioned into $X$ and $Y$ such that $|X|\ge 3$ and $|X|$ is an odd number. The points in $X$ are equally spaced on a circle $C$. For any two distinct points $p, q\in X$ let $\ell(p,q)$ be the line passing through $p$ and $q$. Since $X$ is equally spaced on $C$ and $|X|$ is an odd number, $\ell(p,q)$ does not contain the center of $C$. Let $H(p,q)$ and $H'(p,q)$ be the two half planes defined by $\ell(p,q)$ such that $H'(p,q)$ contains the center of $C$. Let $C'=\bigcap_{p,q\in X}H'(p,q)$. The points in $Y$ are in the interior of $C'$; see Figure~\ref{mp:n-over-3-fig}(a). For any two points $p,q\in X$, the line segment $pq$ does not intersect the interior of $C'$. The special case when $|Y|=1$ is the regular wheel configuration.

\begin{figure}[htb]
  \centering
\setlength{\tabcolsep}{0in}
  $\begin{tabular}{ccc}
\multicolumn{1}{m{.5\columnwidth}}{\centering\includegraphics[width=.43\columnwidth]{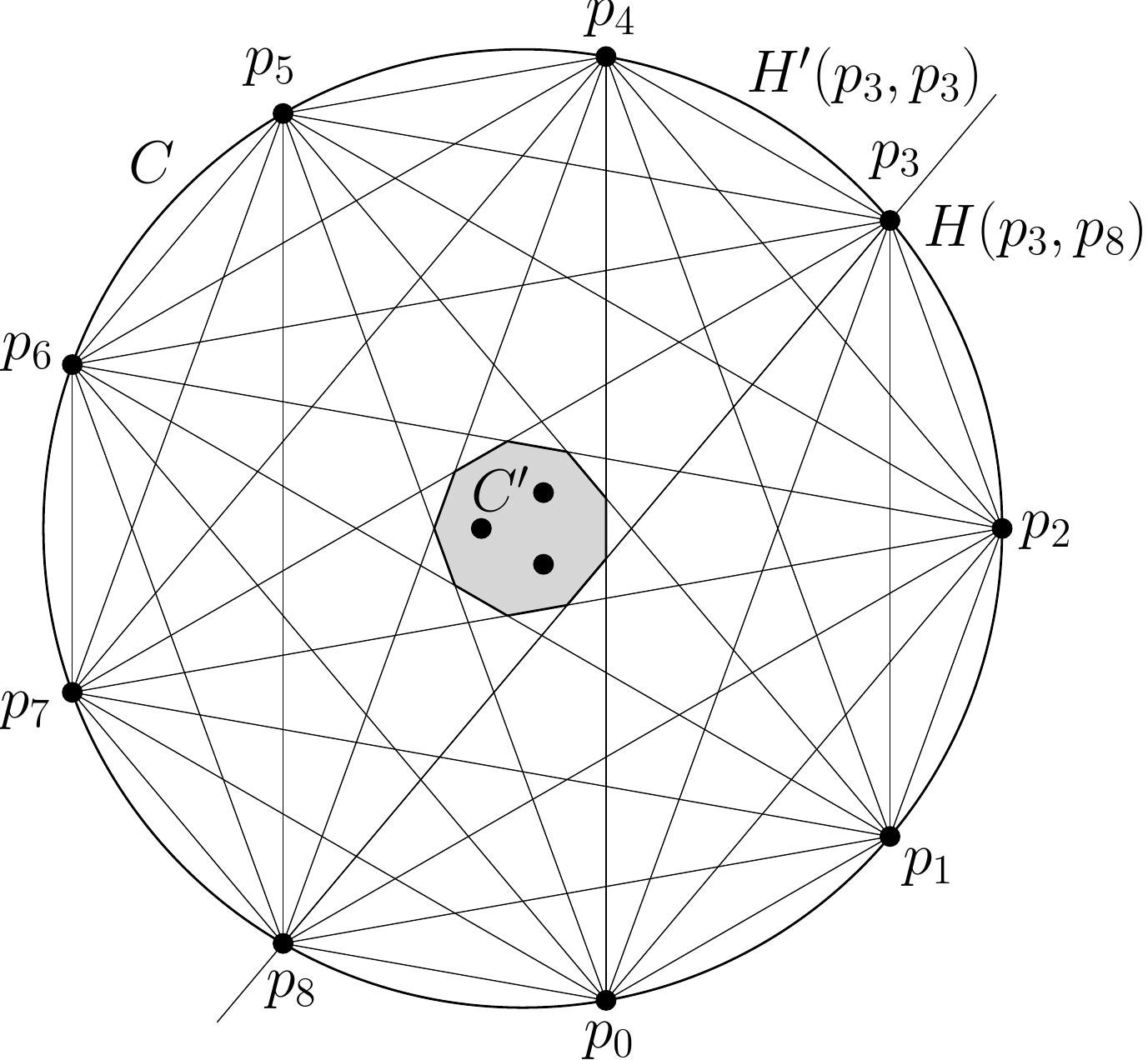}}&
\multicolumn{1}{m{.5\columnwidth}}{\centering\includegraphics[width=.43\columnwidth]{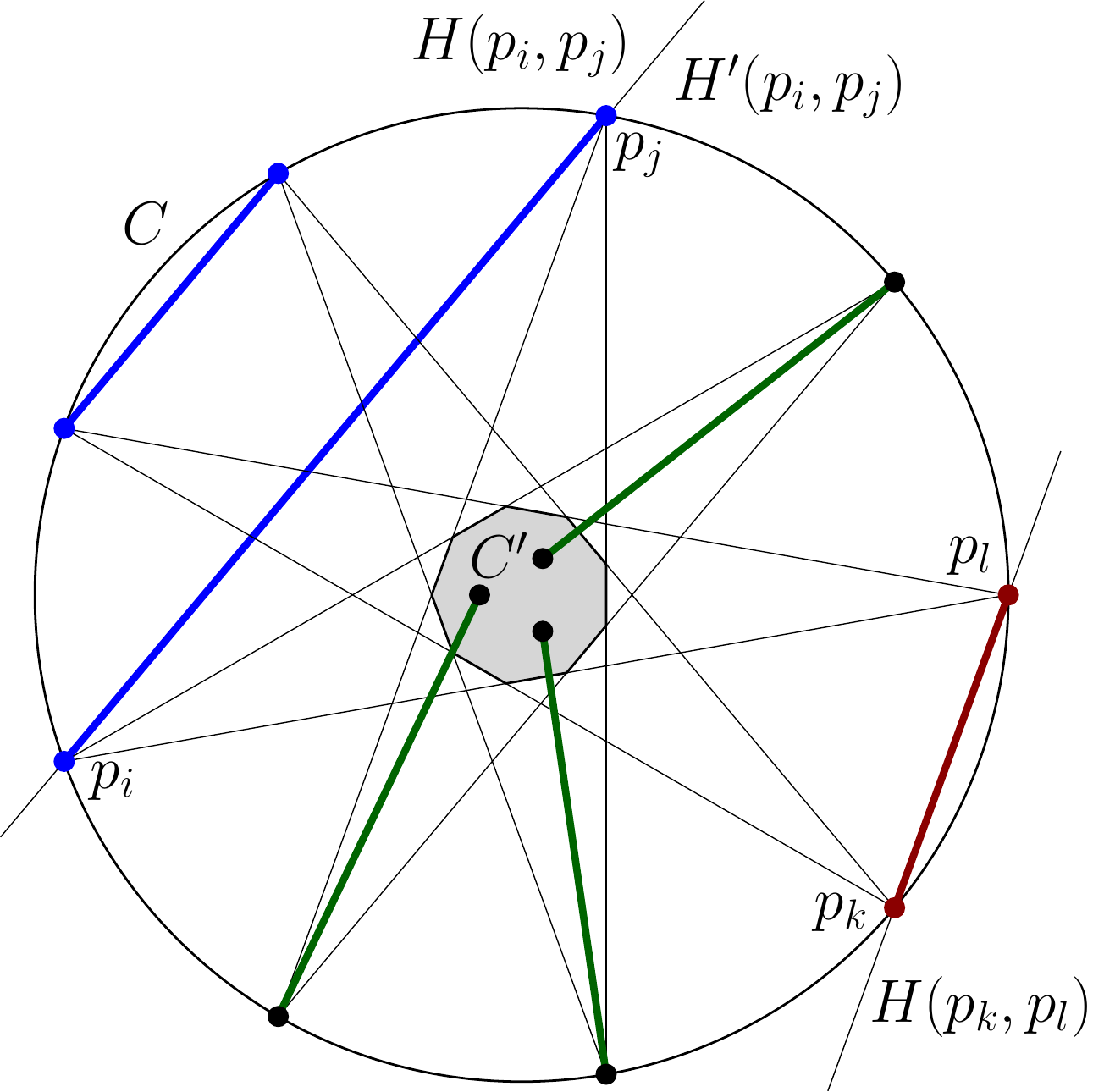}}
\\
(a) & (b)
\end{tabular}$
  \caption{(a) A variation of the regular wheel configuration. (b) Illustration of Lemma~\ref{mp:two-wheel-edges}. The points of $A$ and the edges of $M(A)$ are in blue, and the points of $B$ and the edges of $M(B)$ are in red.}
\label{mp:n-over-3-fig}
\end{figure}

\begin{lemma}
\label{mp:two-wheel-edges}
 Let $P$ be a set of points in the plane where $|P|$ is an even number and $|P|\ge 6$. Let $\{X,Y\}$ be a partition of the points in $P$  such that $|X|$ is an odd number and $|Y|\le 2\lfloor\frac{|P|}{6}\rfloor-1$. If $P$ is in the wheel configuration described above, then any plane matching in $P$ contains at least two edges of \CH{P}.
\end{lemma}
\begin{proof}
 Consider a plane matching $M$ of $P$. It is obvious that $\CH{P}=\CH{X}$; we show that $M$ contains at least two edges of $\CH{X}$. Note that $|X|=|P|-|Y|$, and both $|X|$ and $|Y|$ are odd numbers. Observe that $|X|\ge 4\lfloor\frac{|P|}{6}\rfloor+1=2|Y|+3\ge 5$; which implies that $|Y|\le \frac{|X|-1}{2}-1$. Thus, $|X|>|Y|$, and hence there is at least one edge in $M$ with both endpoints in $X$. Let $p_ip_j$ be the longest such edge. Recall that $C'\subset H'(p_i,p_j)$. Let $A$ be the set of points of $P$ in $H(p_i,p_j)$ (including $p_i$ and $p_j$), and let $A'$ be the set of points of $P$ in $H'(p_i,p_j)$ (excluding $p_i$ and $p_j$). By definition, $H(p_i,p_j)\cup \ell(p_i,p_j)$ does not contain any point of $Y$. Thus, $A\subset X$ and $A$ is in convex position with $|A|\le \frac{|X|-1}{2}$ (note that $|X|$ is an odd number). Let $M(A)$ and $M(A')$ be the edges of $M$ induced by the points in $A$ and $A'$, respectively. Clearly, $\{M(A),M(A')\}$ is a partition of the edges of $M$, and hence $M(A)$ (resp. $M(A')$) is a plane perfect matching for $A$ (resp. $A'$). We show that each of $M(A)$ and $M(A')$ contains at least one edge of \CH{X}. First we consider $M(A)$. If $|A|=2$, then $p_ip_j$ is the only edge in $M(A)$ and it is an edge of \CH{X}. Assume that $|A|\ge 4$. By Lemma~\ref{mp:two-convex-edges}, $M(A)$ contains at least two edges of \CH{A}. On the other hand each edge of \CH{A}, except for $p_ip_j$, is also an edge of \CH{X}; see Figure~\ref{mp:n-over-3-fig}(b). This implies that $M(A)-\{p_ip_j\}$ contains at least one edge of \CH{X}. Now we consider $M(A')$. Let $X'=A'\cap X$, that is, $\{A,X'\}$ is a partition of the points in $X$. Since $|A|\le \frac{|X|-1}{2}$, we have $|X'|\ge\frac{|X|+1}{2}$. Recall that $|Y|\le \frac{|X|-1}{2}-1$. Thus, $|Y|<|X'|$, and hence there is an edge $p_kp_l\in M(A')$ with both $p_k$ and $p_l$ in $X'$. Let $B$ be the set of points of $P$ in $H(p_k,p_l)$ (including $p_k$ and $p_l$). By definition, $H(p_k,p_l)\cup \ell(p_k,p_l)$ does not contain any point of $Y$. Thus, $B\subset X$ and $B$ is in convex position. On the other hand, by the choice of $p_ip_j$ as the longest edge, $A$ cannot be a subset of $B$ and hence $B\subset X'$. Let $M(B)$  be the edges of $M(A')$ induced by the points in $B$. We show that $M(B)$ contains at least one edge of \CH{X}. If $|B|=2$, then $p_kp_l$ is the only edge in $M(B)$ and it is an edge of \CH{X}. Assume that $|B|\ge 4$. By Lemma~\ref{mp:two-convex-edges}, $M(B)$ contains at least two edges of \CH{B}. On the other hand, each edge of \CH{B}, except for $p_kp_l$, is also an edge of \CH{X}; see Figure~\ref{mp:n-over-3-fig}(b). This implies that $M(B)-\{p_kp_l\}$ contains at least one edge of \CH{X}. This completes the proof. 
\end{proof}

\begin{figure}[htb]
  \centering
\setlength{\tabcolsep}{0in}
  $\begin{tabular}{cc}

\multicolumn{1}{m{.5\columnwidth}}{\centering\includegraphics[width=.33\columnwidth]{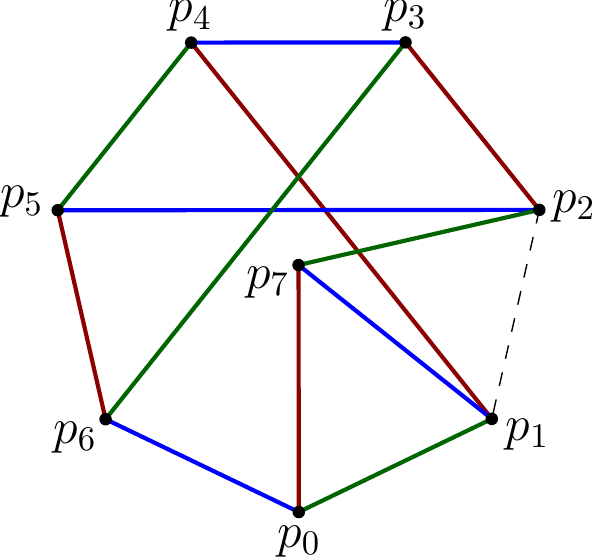}}
&\multicolumn{1}{m{.5\columnwidth}}{\centering\includegraphics[width=.33\columnwidth]{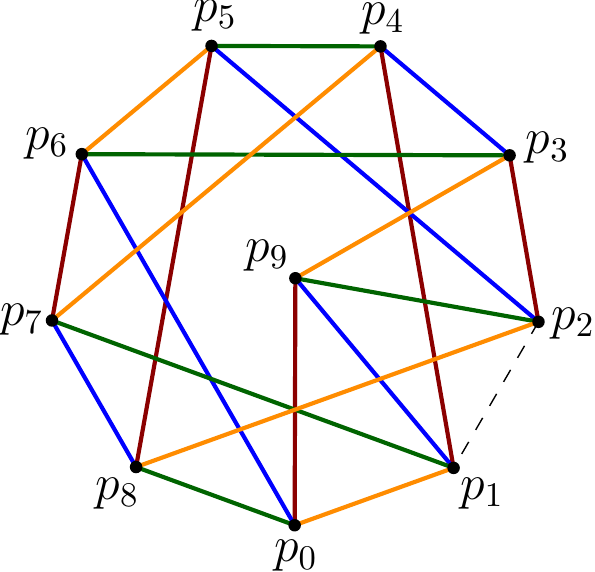}} \\
(a) & (b)
\end{tabular}$
  \caption{Points in the regular configuration with (a) $n=4k$ and (b) $n=4k+2$; one of the edges in \CH{P} cannot be matched.}
\label{mp:wheel-fig}
\end{figure}

\begin{theorem}
\label{mp:wheel}
For a set $P$ of $n\ge 6$ points in the regular wheel configuration in the plane with $n$ even, the maximum number of plane matchings that can be packed into $\Kn{(P)}$ is $\frac{n}{2}-1$.
\end{theorem}
\begin{proof}
In the regular wheel configuration, $P$ is partitioned into a point set $X$ of size $n-1$ and a point set $Y$ of size 1. The points of $X$ are equally spaced on a circle $C$ and the (only) point of $Y$ is the center of $C$. By Lemma~\ref{mp:two-wheel-edges}, every plane matching in $P$ contains at least two edges of \CH{P}. On the other hand, \CH{P} has $n-1$ edges. Therefore, the number of plane matchings that can be packed into $\Kn{(P)}$ is at most $\frac{n-1}{2}$. Since $n$ is an even number and the number of plane matchings is an integer, we can pack at most $\frac{n}{2}-1$ plane matchings into $\Kn{(P)}$.

Now we show how to pack $\frac{n}{2}-1$ plane matchings into $\Kn{(P)}$.
Let $P=\{p_0,\dots,p_{n-1}\}$, and w.l.o.g. assume that $p_{n-1}$ is the center of $C$. Let $P'=P-\{p_{n-1}\}$, and let $p_0,p_1,\dots,p_{n-2}$ be the radial ordering of the points in $P'$ with respect to $p_{n-1}$. For each $p_i$ in the radial ordering, where $0\le i\le \frac{n}{2}-2$, let $$R_i=\{p_{i+j}p_{i+2\lceil(n-2)/4\rceil-j+1}:j=1,\dots,\lceil(n-2)/4\rceil\},$$ and $$L_i=\{p_{i-j}p_{i-2\lfloor(n-2)/4\rfloor+j-1}:j=1,\dots,\lfloor(n-2)/4\rfloor\}$$ (all indices are modulo $n-1$). Let $M_i=R_i\cup L_i\cup \{p_ip_{n-1}\}$; informally speaking, $M_i$ is a plane perfect matching obtained from edge $p_ip_{n-1}$ and edges parallel to $p_ip_{n-1}$. See Figure~\ref{mp:wheel-fig}(a) for the case where $n=4k$ and Figure~\ref{mp:wheel-fig}(b) for the case where $n=4k+2$. Let $\mathcal{M}=\{M_i:i=0,\dots,\frac{n}{2}-2\}$. The matchings in $\mathcal{M}$ are plane and pairwise edge-disjoint. Thus, $\mathcal{M}$ is a set of $\frac{n}{2}-1$ plane matchings that can be packed into $\Kn{(P)}$. 
\end{proof}

In the following theorem we use the wheel configuration to show that for any even integer $n\ge 6$, there exists a set $P$ of $n$ points in the plane, such that no more than $\lceil\frac{n}{3}\rceil$ plane matchings can be packed into $\Kn{(P)}$. 

\begin{theorem}
\label{mp:n-over-3-thr}
For any even number $n\ge 6$, there exists a set $P$ of $n$ points in the plane such that no more than $\lceil\frac{n}{3}\rceil$ plane matchings can be packed into $\Kn{(P)}$.
\end{theorem}
\begin{proof}
The set $P$ of $n$ points is partitioned into $X$ and $Y$, where $|Y|=2\lfloor\frac{n}{6}\rfloor-1$ and $|X|=n-|Y|$. The points in $X$ are equally spaced on a circle $C$ and the points in $Y$ are in the interior $\bigcap_{p,q\in X}H'(p,q)$. By Lemma~\ref{mp:two-wheel-edges}, any plane matching in $P$ contains at least two edges of \CH{P}.
Since $\CH{P}=\CH{X}$, any plane matching of $P$ contains at least two edges of \CH{X}. Thus, if $\mathcal{M}$ denotes any set of plane matchings which can be packed into $\Kn{(P)}$, we have (note that $|X|$ is odd)
$$|\mathcal{M}|\le \frac{|X|-1}{2}=\frac{n-2\lfloor n/6\rfloor}{2}=\frac{n}{2}-\lfloor\frac{n}{6}\rfloor\le \frac{n}{2}-\frac{n-5}{6}\le \lceil\frac{n}{3}\rceil.$$
\end{proof}

\subsection{Points in General Position}
\label{mp:general-position-section}
In this section we consider the problem of packing plane matchings for point sets in general position (no three points on a line) in the plane. 
Let $P$ be a set of $n$ points in general position in the plane, with $n$ even. Let $M(P)$ denote the maximum number of plane matchings that can be packed into $\Kn{(P)}$. As mentioned earlier, a trivial lower bound of $2$ (when $n \ge 4$) is obtained from a minimum weight Hamiltonian cycle, which is plane and consists of two edge-disjoint perfect matchings. 

In this section we show that at least $\lfloor\log_2{n}\rfloor-1$ plane matchings can be packed into $\Kn{(P)}$. As a warm-up, we first show that if $n$ is a power of two, then $\log_2{n}$ plane matchings can be packed into $\Kn{(P)}$. Then we extend this result to get a lower bound of $\lfloor\log_2{n}\rfloor-1$ for every point set with an even number of points. We also show that if $n\ge 8$, then at least three plane matchings can be packed into $\Kn{(P)}$, which improves the result for $n=10$, $12$, and $14$. Note that, as a result of Theorem~\ref{mp:n-over-3-thr}, there exists a set of $n=6$ points such that no more than $\lceil\frac{n}{3}\rceil=2$ plane matchings can be packed into $\Kn{(P)}$. First consider the following observation.

\begin{observation}
\label{mp:partition-obs}
 Let $\mathcal{P}=\{P_1,\dots, P_k\}$ be a partition of the point set $P$, such that $|P_i|$ is even and $\CH{P_i}\cap\CH{P_j}=\emptyset$ for all $1\le i,j\le k$ where $i\neq j$. Let $i$ be an index such that, $M(P_i)=\min\{M(P_j):1\le j\le k\}$. Then, $M(P)\ge M(P_i)$.
\end{observation}

\begin{lemma}
\label{mp:n-power2}
For a set $P$ of $n$ points in general position in the plane, where $n$ is a power of 2, at least $\log_2{n}$ plane matchings can be packed into $\Kn{(P)}$.
\end{lemma}
\begin{proof}
We prove this lemma by induction. The statement of the lemma holds for the base case, where $n=2$. Assume that $n\ge 4$. Recall that $M(P)$ denotes the maximum number of plane matchings that can be packed into $\Kn{(P)}$. W.l.o.g. assume that a vertical line $\ell$ partitions $P$ into sets $R$ and $B$, each of size $\frac{n}{2}$. By the induction hypothesis, $M(R),M(B)\ge\log_2{(\frac{n}{2})}$. By Observation~\ref{mp:partition-obs}, $M(P)\ge \min\{M(R),\allowbreak M(B)\}\allowbreak \ge\allowbreak \log_2{(\frac{n}{2})}$. That is, by pairing a matching $M_R$ in $R$ with a matching $M_B$ in $B$ we get a plane matching $M_P$ in $\Kn{(P)}$, such that each edge in $M_P$ has both endpoints in $R$ or in $B$. If we consider the points in $R$ as red and the points in $B$ as blue, \Cut{R}{B} (see Section~\ref{mp:colored-matching-section}) gives us a plane perfect matching $M'_P$ in $\Kn{(P)}$, such that each edge in $M'_P$ has one endpoint in $R$ and one endpoint in $B$. That is $M'_P\cap M_P=\emptyset$. Therefore, we obtain one more plane matching in $\Kn{(P)}$, which implies that $M(P)\ge \log_2{(\frac{n}{2})}+1=\log_2n$.
\end{proof}

Let $R$ and $B$ be two point sets which are separated by a line. A {\em crossing tangent} between $R$ and $B$ is a line $l$ touching \CH{R} and \CH{B} such that $R$ and $B$ lie on different sides of $l$. Note that $l$ contains a point $r\in R$, a point $b\in B$, and consequently the line segment $rb$; we say that $l$ is subtended from $rb$. It is obvious that there are two (intersecting) crossing tangents between $R$ and $B$; see Figure~\ref{mp:three-matching-fig}. 

\begin{figure}[htb]
  \centering
\setlength{\tabcolsep}{0in}
  $\begin{tabular}{cc}
\multicolumn{1}{m{.5\columnwidth}}{\centering\includegraphics[width=.38\columnwidth]{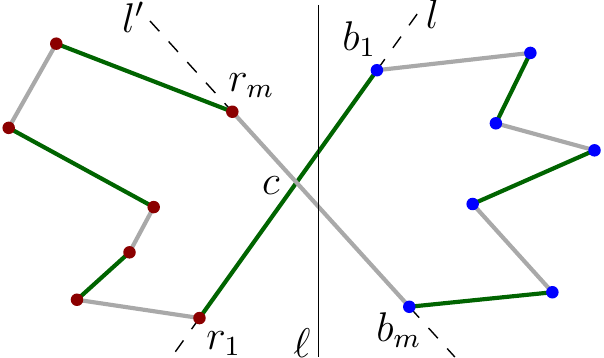}}
&\multicolumn{1}{m{.5\columnwidth}}{\centering\includegraphics[width=.35\columnwidth]{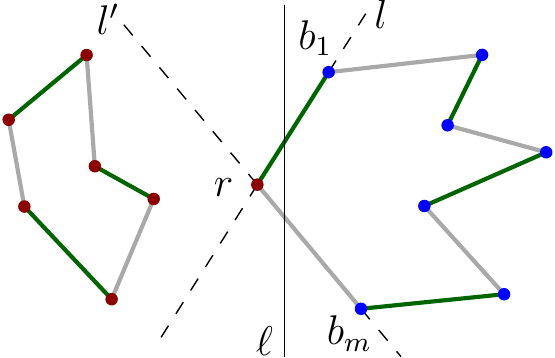}} \\
(a) & (b)
\end{tabular}$
  \caption{(a) The crossing tangents intersect at a point $c\notin P$: $R$ and $B$ are sorted clockwise around $c$, (b) The crossing tangents intersect at a point $r\in R$: $B$ is sorted clockwise around $r$. $M_1$ and $M_2$ are shown by green and gray line segments.}
\label{mp:three-matching-fig}
\end{figure}

\begin{lemma}
\label{mp:3-matching-theorem}
For a set $P$ of $n\ge 8$ points in general position in the plane with $n$ even, at least three plane matchings can be packed into $\Kn{(P)}$.
\end{lemma}
\begin{proof}
We describe how to extract three edge-disjoint plane matchings, $M_1,M_2,M_3$, from $\Kn{(P)}$. Let $\ell$ be a vertical line which splits $P$ into sets $R$ and $B$, each of size $\frac{n}{2}$. Consider the points in $R$ as red and the points in $B$ as blue. We differentiate between two cases: (a) $n=4k$ and (b) $n= 4k+2$, for some integer $k>1$.

In case (a), both $R$ and $B$ have an even number of points. Let \M{1}{R} and \M{2}{R} (resp. \M{1}{B} and \M{2}{B}) be two edge-disjoint plane matchings in $R$ (resp. $B$) obtained by a minimum length Hamiltonian cycle in $R$ (resp. $B$). Let $M_1=\M{1}{R}\cup \M{1}{B}$ and $M_2=\M{2}{R}\cup \M{2}{B}$. Clearly $M_1$ and $M_2$ are edge-disjoint plane matchings for $P$. Let $M_3=\Cut{R}{B}$. It is obvious that $M_3$ is edge-disjoint from $M_1$ and $M_2$, which completes the proof in the first case.

In case (b), both $R$ and $B$ have an odd number of points and we cannot get a perfect matching in each of them. Let $l$ and $l'$ be the two crossing tangents between $R$ and $B$, subtended from $rb$ and $r'b'$, respectively. We differentiate between two cases: (i) $l$ and $l'$ intersect in the interior of $rb$ and $r'b'$, (ii) $l$ and $l'$ intersect at an endpoint of both $rb$ and $r'b'$; see Figure~\ref{mp:three-matching-fig}.
\begin{itemize}
\item In case (i), let $c$ be the intersection point; see Figure~\ref{mp:three-matching-fig}(a). Let $r_1, r_2, \dots, r_m$ and $b_1, b_2, \allowbreak\dots, \allowbreak b_m$ be the points of $R$ and $B$, respectively, sorted clockwise around $c$, where $m=\frac{n}{2}$, $r_1=r,r_m=r', b_1=b, b_m=b'$. Consider the Hamiltonian cycle $H=\{r_ir_{i+1}:1\le i< m\}\cup \{b_ib_{i+1}:1\le i< m\}\cup \{r_1b_1,r_mb_m\}$. Let $M_1$ and $M_2$ be the two edge-disjoint matchings obtained from $H$. Note that $r_1b_1$ and $r_mb_m$ cannot be in the same matching, thus, $M_1$ and $M_2$ are plane. Let $M_3=\Tangent{R}{B}$. As described in Section~\ref{mp:colored-matching-section}, $M_3$ is a plane matching for $P$. In order to prove that $M_3\cap (M_1\cup M_2)=\emptyset$, we show that $rb$ and $r'b'$\textemdash which are the only edges in $M_1\cup M_2$ that connect a point in $R$ to a point in $B$\textemdash do not belong to $M_3$. Note that \Tangent{R}{B} iteratively selects an edge which has the same number of red and blue points below its supporting line, whereas the supporting lines of $rb$ and $r'b'$ have different numbers of red and blue points below them. Thus $rb$ and $r'b'$ are not considered by \Tangent{R}{B}. Therefore $M_3$ is edge-disjoint from $M_1$ and $M_2$. 

\item In case (ii), w.l.o.g. assume that $l$ and $l'$ intersect at the red endpoint of $rb$ and $r'b'$, i.e., $r=r'$; See Figure~\ref{mp:three-matching-fig}(b). Let $R'=R\setminus \{r\}$ and $B'=B\cup \{r\}$. Note that both $R'$ and $B'$ have an even number of points and $|R'|,|B'|\ge 4$. Let \M{1}{R'} and \M{2}{R'} be two edge-disjoint plane matchings in $R'$ obtained by a minimum length Hamiltonian cycle in $R'$. Let $b_1, b_2, \dots, b_m$ be the points of $B$ sorted clockwise around $r$, where $m=\frac{n}{2}$, $b_1=b, b_m=b'$. Consider the Hamiltonian cycle $\HC{R'}=\{b_ib_{i+1}:1\le i< m\}\cup \{rb_1,rb_m\}$. Let $M_1(B')$ and $M_2(B')$ be the two edge-disjoint plane matchings in $B'$ obtained from $\HC{B'}$. Let $M_1=\M{1}{R'}\cup \M{1}{B'}$ and $M_2=\M{2}{R'}\cup \M{2}{B'}$. Clearly $M_1$ and $M_2$ are edge-disjoint plane matchings in $P$. Let $M_3=\Tangent{R}{B}$. As described in case (i), $M_3$ is a plane matching in $P$ and $M_3\cap (M_1\cup M_2)=\emptyset$. Therefore, $M_3$ is edge-disjoint from $M_1$ and $M_2$.
\end{itemize}
\end{proof}

As a direct consequence of Lemma~\ref{mp:n-power2} and Lemma~\ref{mp:3-matching-theorem} we have the following theorem.
\begin{theorem}
For a set $P$ of $n=2^i\cdot m$ points in general position in the plane with $n$ even, $m\ge 4$, $i\ge 0$ , at least $i+2$ plane matchings can be packed into $\Kn{(P)}$.
\end{theorem}
\begin{proof}
If $i=0$, then a minimum weight Hamiltonian cycle in $\Kn{(P)}$ consists of two plane matchings. Assume $i\ge 1$. Partition $P$ by vertical lines, into $2^{i-1}$ point sets, each of size $2m$. By Lemma~\ref{mp:3-matching-theorem}, at least three plane matchings can be packed into each set. Considering these sets as the base cases in Lemma~\ref{mp:n-power2}, we obtain $i-1$ plane matchings between these sets. Thus, in total, $i+2$ plane matchings can be packed into $\Kn{(P)}$.
\end{proof}

\begin{theorem}
For a set $P$ of $n$ points in general position in the plane, with $n$ even, at least $\lfloor\log_2{n}\rfloor-1$ plane matchings can be packed into $\Kn{(P)}$.
\end{theorem}
\begin{proof}
If $n$ is a power of two, then by Lemma~\ref{mp:n-power2} at least $\log_2 n\ge \lfloor \log_2 n\rfloor -1$ matchings can be packed into 
$\Kn{(P)}$. Assume $n$ is not a power of two. We describe how to pack a set $\mathcal{M}$ of $\lfloor\log_2{n}\rfloor-1$ plane perfect matchings into $\Kn{(P)}$. The construction consists of the following three main steps which we will describe in detail.

\begin{enumerate}
  \item Building a binary tree $T$.
  \item Assigning the points of $P$ to the leaves of $T$.
  \item Extracting $\mathcal{M}$ from $P$ using internal nodes of $T$.
\end{enumerate}

\begin{paragraph}{\em \small 1. Building the tree T.}
In this step we build a binary tree $T$ such that each node of $T$ stores an even number, and each internal node of $T$ has a left and a right child. For an internal node $u$, let \LC{u} and \RC{u} denote the left child and the right child of $u$, respectively. Given an even number $n$, we build $T$ in the following way:
\begin{itemize}
  \item The root of $T$ stores $n$.
  \item If a node of $T$ stores $2$, then that node is a leaf.
  \item For a node $u$ storing $m$, with $m$ even and $m\ge 4$, we store the following even numbers into \LC{u} and \RC{u}:
    \begin{itemize}
	\item If $m$ is divisible by $4$, we store $\frac{m}{2}$ in both \LC{u} and \RC{u}; see Figure~\ref{mp:tree-construction-fig}(a).
	\item If $m$ is not divisible by $4$ and $u$ is the root or the left child of its parent then we store $2\lfloor\frac{m}{4}\rfloor$ in \LC{u} and $m-2\lfloor\frac{m}{4}\rfloor$ in \RC{u}; see Figure~\ref{mp:tree-construction-fig}(b).
	\item If $m$ is not divisible by $4$ and $u$ is the right child of its parent then we store $m-2\lfloor\frac{m}{4}\rfloor$ in \LC{u} and $2\lfloor\frac{m}{4}\rfloor$ in \RC{u}; see Figure~\ref{mp:tree-construction-fig}(c).
    \end{itemize}
Note that in the last two cases\textemdash where $m$ is not divisible by four\textemdash the absolute difference between the values stored in \LC{u} and \RC{u} is exactly 2. See Figure~\ref{mp:matching-example-fig}.
\end{itemize}
\end{paragraph}

\begin{figure}[htb]
  \centering
\setlength{\tabcolsep}{0in}
  $\begin{tabular}{ccc}
\multicolumn{1}{m{.33\columnwidth}}{\centering\includegraphics[width=.2\columnwidth]{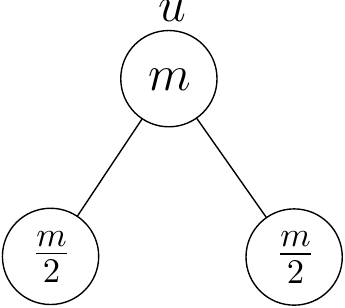}}&
\multicolumn{1}{m{.33\columnwidth}}{\centering\includegraphics[width=.24\columnwidth]{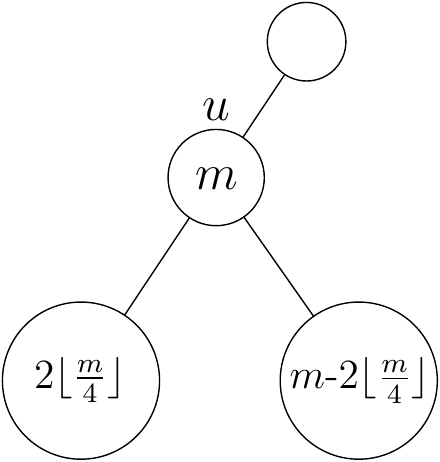}}
&\multicolumn{1}{m{.33\columnwidth}}{\centering\includegraphics[width=.24\columnwidth]{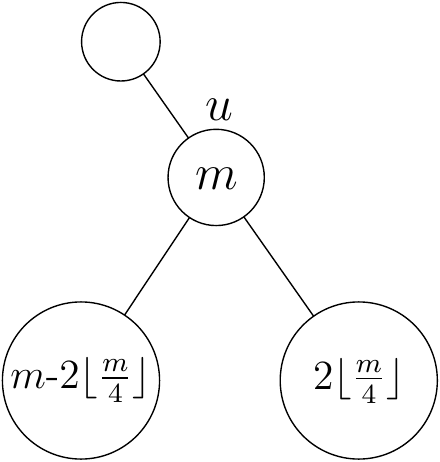}} \\
(a) & (b)&(c)
\end{tabular}$
  \caption{(a) $m$ is divisible by four, (b) $m$ is not divisible by four and $u$ is a left child, and (c) $m$ is not divisible by four and $u$ is a right child.}
\label{mp:tree-construction-fig}
\end{figure}

\begin{paragraph}{\em \small 2. Assigning the points to the leaves of the tree.}
In this step we describe how to assign the points of $P$, in pairs, to the leaves of $T$. We may assume without loss of generality that no two points of $P$ have the same $x$-coordinate. Sort the points of $P$ in a increasing order of their $x$-coordinate.
Assign the first two points to the leftmost leaf, the next two points to the second leftmost leaf, and so on. Note that $T$ has $\frac{n}{2}$ leaves, and hence all the points of $P$ are assigned to the leaves of $T$. See Figure~\ref{mp:matching-example-fig}.  
\end{paragraph}

\begin{paragraph}{\em \small 3. Extracting the matchings.}
Let $L$ be the number of edges in a shortest path from the root to any leaf in $T$; in Figure~\ref{mp:matching-example-fig}, $L=3$. For an internal node $u\in T$, let $T_u$ be the subtree rooted at $u$. Let $L_u$ and $R_u$ be the set of points assigned to the left and right subtrees of $T_u$, respectively, and let $P_u=L_u\cup R_u$. Consider the points in $L_u$ as red and the points in $R_u$ as blue. Since the points in $L_u$ have smaller $x$-coordinates than the points in $R_u$, we say that $L_u$ and $R_u$ are separated by a vertical line $\ell(u)$. For each internal node $u$ where $u$ is in level $0\le i< L$ in $T$\textemdash assuming the root is in level $0$\textemdash we construct a plane perfect matching $M_u$ in $P_u$ in the following way. Let $m$ be the even number stored at $u$.
\begin{itemize}
  \item If $m$ is divisible by $4$ (Figure~\ref{mp:tree-construction-fig}(a)), then let $M_u=\Min{L_u}{R_u}$; see Section~\ref{mp:colored-matching-section}. Since $|L_u|=|R_u|$, $M_u$ is a plane perfect matching for $P_u$. See vertices $u_2, u_3$ in Figure~\ref{mp:matching-example-fig}.

  \item If $m$ is not divisible by $4$ and $u$ is the root or a left child (Figure~\ref{mp:tree-construction-fig}(b)), then $|R_u|-|L_u|=2$. Let $a,b$ be the two points assigned to the rightmost leaf in $T_u$, and let $M_u=\{ab\}\cup\\\allowbreak \Min{L_u}{R_u-\{a,b\}}$. Since $|L_u|=|R_u-\{a,b\}|$, $M_u$ is a perfect matching in $P_u$. In addition, $a$ and $b$ are the two rightmost points in $P_u$, thus, $ab$ does not intersect any edge in $\\\Min{L_u}{R_u-\{a,b\}}$, and hence $M_u$ is plane. See vertices $u_0, u_1,\allowbreak u_5$ in Figure~\ref{mp:matching-example-fig}.

  \item If $m$ is not divisible by $4$ and $u$ is a right child (Figure~\ref{mp:tree-construction-fig}(c)), then $|L_u|-|R_u|=2$. Let $a,b$ be the two points assigned to the leftmost leaf in $T_u$ and let $M_u=\{ab\}\cup\Min{L_u-\{a,b\}}{R_u}$. Since $|L_u-\{a,b\}|=|R_u|$, $M_u$ is a perfect matching in $P_u$. In addition, $a$ and $b$ are the two leftmost points in $P_u$, thus, $ab$ does not intersect any edge in $\Min{L_u-\{a,b\}}{R_u}$, and hence $M_u$ is plane. See vertices $u_4, u_6$ in Figure~\ref{mp:matching-example-fig}.
\end{itemize}

For each $i$, where $0\le i< L$, let $S(i)$ be the set of vertices of $T$ in level $i$; see Figure~\ref{mp:matching-example-fig}. For each level $i$ let $M_i=\bigcup_{u\in S(i)}{M_u}$. Let $\mathcal{M}=\{M_i:0\le i< L\}$. 
In the rest of the proof, we show that $\mathcal{M}$ contains $\lfloor\log_2{n}\rfloor-1$ edge-disjoint plane matchings in $P$.
\end{paragraph}

\begin{figure}[htb]
  \centering
  \includegraphics[width=.75\columnwidth]{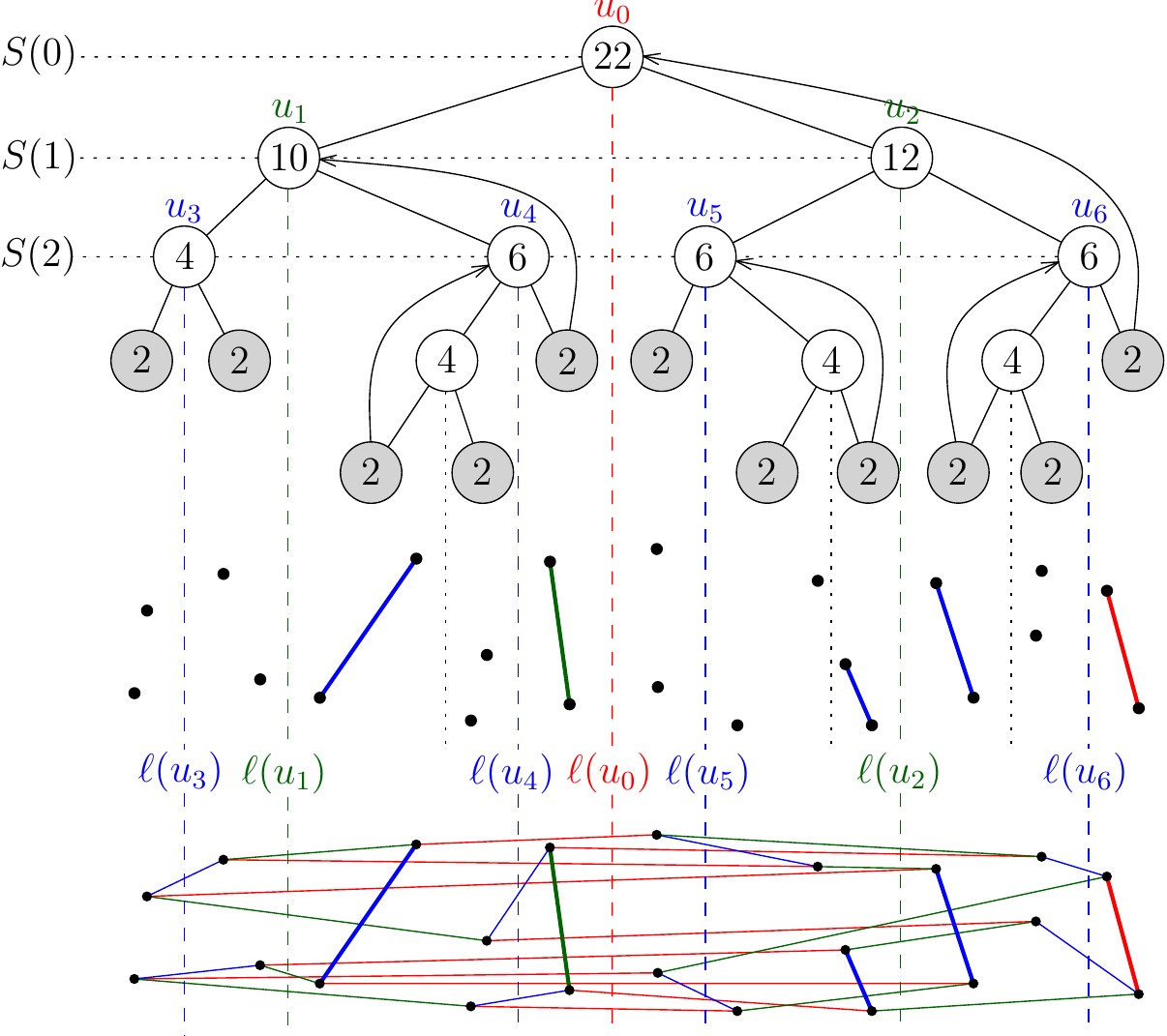}
 \caption{The points in $P$ are assigned, in pairs, to the leaves of $T$, from left to right. The point set $P$ with the edge-disjoint plane matchings is shown as well. $M_0$ contains the bold red edge and the red edges crossing $\ell(u_0)$. $M_1$ contains the bold green edge and the green edges crossing $\ell(u_1), \ell(u_2)$. $M_2$ contains the bold blue edges and the blue edges crossing $\ell(u_3), \allowbreak \ell(u_4), \allowbreak\ell(u_5), \allowbreak\ell(u_6)$.}
  \label{mp:matching-example-fig}
\end{figure}

\begin{paragraph}{Claim 1.}{\em For each $i$, where $0\le i< L$, $M_i$ is a plane perfect matching in $P$.}\end{paragraph} 

Note that if $u$ is the root of the tree, then $P_u=P$. In addition, for each internal node $u$ (including the root), $\{L_u,R_u\}$ is a partition of the point set $P_u$. This implies that in each level $i$ of the tree, where $0\le i <L$, we have $P=\bigcup_{u\in S(i)}{P_u}$. Moreover, the points in $P$ are assigned to the leaves of $T$ in non-decreasing order of their $x$-coordinate. Thus, $\mathcal{P}_i=\{P_u: u\in S(i)\}$ is a partition of the point set $P$; the sets $P_u$ with $u\in S(i)$ are separated by vertical lines; see Figure~\ref{mp:matching-example-fig}. Therefore, $M_i$ is a perfect plane matching in $P$; which proves the claim.

\begin{paragraph}{Claim 2.}{\em For all $M_i, M_j\in \mathcal{M}$, where $0\le i,j<L$ and $i\neq j$, $M_i\cap M_j=\emptyset$.}\end{paragraph}

In order to prove that $M_i$ and $M_j$ are edge-disjoint, we show that for each pair of distinct internal nodes $u$ and $v$, $M_u\cap M_v=\emptyset$. If $u$ and $v$ are in the same level, then $P_u$ and $P_v$ are separated by $\ell(u)$, thus, $M_u$ and $M_v$ do not share any edge. Thus, assume that $u\in S(i)$ and $v\in S(j)$ such that $0\le i, j< L$, $i\neq j$, and w.l.o.g. assume that $i< j$. If $v\notin T_u$, then $P_u$ and $P_v$ are separated by line $\ell(w)$, where $w$ is the lowest common ancestor of $u$ and $v$; this implies that $M_u$ and $M_v$ do not share any edge. Therefore, assume that $v\in T_u$, and w.l.o.g. assume that $v$ is in the left subtree of $T_u$. Thus, $P_v$\textemdash and consequently $M_v$\textemdash is to the left of $\ell(u)$. The case where $v$ is in the right subtree of $T_u$ is symmetric. Let $m\ge 4$ be the number stored at $u$. We differentiate between three cases:

\begin{itemize}
  \item If $m$ is divisible by $4$, then all the edges in $M_u$ cross $\ell(u)$, while the edges in $M_v$ are to the left of $\ell(u)$. This implies that $M_u$ and $M_v$ are disjoint. 

  \item If $m$ is not divisible by $4$ and $u$ is the root or a left child, then all the edges of $M_u$ cross $\ell(u)$, except the rightmost edge $ab$ which is to the right of $\ell(u)$. Since $M_v$ is to the left of $\ell(u)$, it follows that $M_u$ and $M_v$ are disjoint. 

  \item If $m$ is not divisible by $4$ and $u$ is a right child, then all the edges of $M_u$ cross $\ell(u)$, except the leftmost edge $ab$. If $a,b\notin P_v$, then $ab\notin M_v$, and hence $M_u$ and $M_v$ are disjoint. If $a,b\in P_v$ then $v$ is the left child of its parent and all the edges in $M_v$ cross $\ell(v)$ (possibly except one edge which is to the right of $\ell(v)$), while $ab$ is to the left of $\ell(v)$. Therefore $M_u$ and $M_v$ do not share any edge. This completes the proof of the claim.
\end{itemize}

\begin{paragraph}{Claim 3.}{\em For every two nodes $u$ and $v$ in the same level of $T$ which store $m$ and $m'$, respectively, $|m-m'| \le 2$.}\end{paragraph}

We prove the claim inductively for each level $i$ of $T$. For the base case, where $i=1$: (a) if $n$ is divisible by four, then both $u$ and $v$ store $\frac{n}{2}$ and the claim holds, (b) if $n$ is not divisible by four then $u$ stores $2\lfloor\frac{n}{4}\rfloor$ and $v$ stores $n-2\lfloor\frac{n}{4}\rfloor$; as $0\le n-2\lfloor\frac{n}{4}\rfloor-2\lfloor\frac{n}{4}\rfloor\le 2$, the claim holds for $i=1$. 
Now we show that if the claim is true for the $i$th level of $T$, then the claim is true for the $(i+1)$th level of $T$.
Let $u$ and $v$, storing $m$ and $m'$, respectively, be in the $i$th level of $T$. By the induction hypothesis, the claim holds for the $i$th level, i.e., $|m-m'|\le 2$. We prove that the claim holds for the $(i+1)$th level of $T$, i.e., for the children of $u$ and $v$. Since $m$ and $m'$ are even numbers, $|m-m'|\in\{0,2\}$. If $|m-m'|=0$, then $m=m'$, and by a similar argument as in the base case, the claim holds for the children of $u$ and $v$. If $|m-m'|=2$, then w.l.o.g. assume that $m'=m+2$. Let $\alpha$ be the smallest number and $\beta$ be the largest number stored at the children of $u$ and $v$ (which are at the $(i+1)$th level). We show that $\beta-\alpha \le 2$. It is obvious that $\alpha= 2\lfloor\frac{m}{4}\rfloor$ and $\beta= m'-2\lfloor\frac{m'}{4}\rfloor$. Thus,
\begin{align}
\label{mp:alpha-beta}
\beta-\alpha & = m'-2\left\lfloor\frac{m'}{4}\right\rfloor-2\left\lfloor\frac{m}{4}\right\rfloor= m+2-2\left\lfloor\frac{m+2}{4}\right\rfloor-2\left\lfloor\frac{m}{4}\right\rfloor 
\end{align}
Now, we differentiate between two cases, where $m=4k$ or $m=4k+2$. If $m=4k$, then by Equation~\ref{mp:alpha-beta},
\begin{align}
\beta-\alpha & = 4k+2-2\left\lfloor\frac{4k+2}{4}\right\rfloor-2\left\lfloor\frac{4k}{4}\right\rfloor= 4k+2-2k-2k= 2.\nonumber
\end{align}
If $m=4k+2$, then by Equation~\ref{mp:alpha-beta},
\begin{align}
\beta-\alpha & = 4k+4-2\left\lfloor\frac{4k+4}{4}\right\rfloor-2\left\lfloor\frac{4k+2}{4}\right\rfloor = 4k+4-2(k+1)-2k= 2\nonumber
\end{align}
which completes the proof of the claim.

\begin{paragraph}{Claim 4.}{$L\ge\lfloor\log_2{n}\rfloor-1$.}\end{paragraph}

It follows from Claim 3 that all the leaves of $T$ are in the last two levels. Since $T$ has $\frac{n}{2}$ leaves, $T$ has $n-1$ nodes. Recall that $L$ is the number of edges in a shortest path from the root to any leaf in $T$. Thus, $L\ge h-1$, where $h$ is the height of $T$. To give a lower bound on $h$, one may assume that the last level of $T$ is also full, thus,

$$n-1\le 2^0+2^1+2^2+\dots+2^h\le 2^{h+1}-1$$
and hence, $h\ge \log_2{n} -1$. Therefore, $L\ge h-1\ge \log_2{n} -2$.  Since $L$ is an integer and $n$ is not a power of two, $L\ge\lceil\log_2{n}\rceil-2=\lfloor\log_2{n}\rfloor-1$; which proves the claim.

\vspace{20pt}

Claim 1 and Claim 2 imply that $\mathcal{M}$ contains $L$ edge-disjoint plane perfect matchings. Claim 4 implies that $L\ge\lfloor\log_2{n}\rfloor-1$, which proves the statement of the theorem.
\end{proof}

\subsection{Non-crossing Plane Matchings}
\label{mp:non-crossing-matching-section}
In this section we consider the problem of packing plane matchings into $\Kn{(P)}$ such that any two different matchings in the packing are non-crossing.
Two edge-disjoint plane matchings $M_1$ and $M_2$ are {\em non-crossing} (or {\em compatible}), if no edge in $M_1$ crosses any edge in $M_2$. For a set $P$ of $n$ points in general position in the plane, with $n$ even, we give tight lower and upper bounds on the number of pairwise non-crossing plane perfect matchings that can be packed into $\Kn{(P)}$. The union of $k$ disjoint perfect matchings in $\Kn{(P)}$ has $\frac{kn}{2}$ edges. By Euler's formula for planar graphs we know that the number of edges of a planar graph is at most $3n-6$. This implies the following lemma.

\begin{lemma}
\label{mp:5-non-crossing}
For a set $P$ of $n$ points in general position in the plane, with $n$ even, at most five pairwise non-crossing plane matchings can be packed into $\Kn{(P)}$.
\end{lemma}

\begin{figure}[htb]
  \centering
  \includegraphics[width=.35\columnwidth]{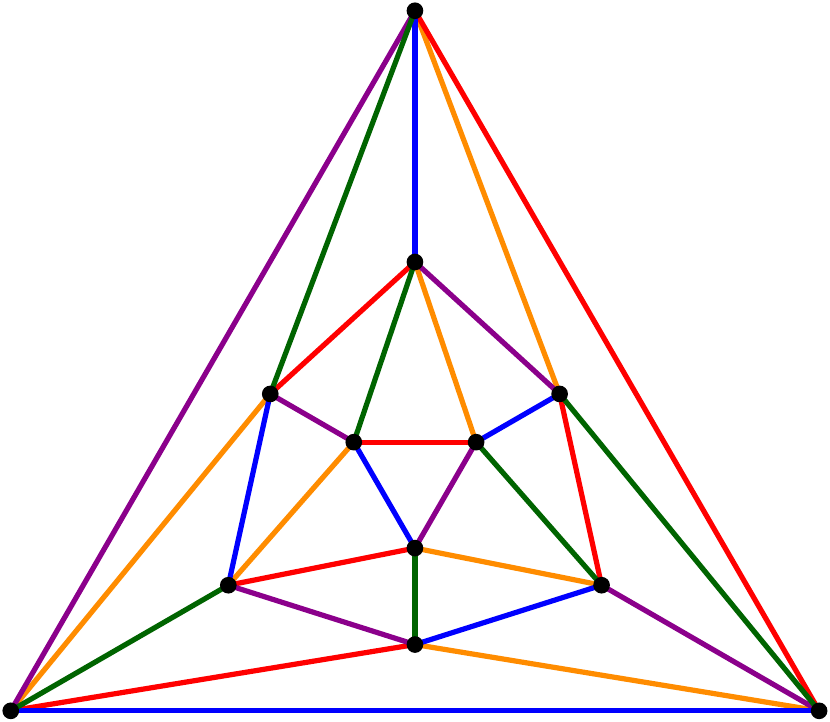}
 \caption{A point set with five non-crossing edge-disjoint plane perfect matchings.}
  \label{mp:5-regular-fig}
\end{figure}

Figure~\ref{mp:5-regular-fig} shows a $5$-regular geometric graph on a set of 12 points in the plane which contains five non-crossing edge-disjoint plane matchings. In \cite{Hasheminezhad2011}, the authors showed how to generate an infinite family of 5-regular planar graphs using the graph in Figure~\ref{mp:5-regular-fig}. By an extension of the five matchings shown in Figure~\ref{mp:5-regular-fig}, five non-crossing matchings for this family of graphs is obtained. Thus, the bound provided by Lemma~\ref{mp:5-non-crossing} is tight.  

It is obvious that if $P$ contains at least four points, the minimum length Hamiltonian cycle in $\Kn{(P)}$ contains two non-crossing edge-disjoint plane matchings. Alternatively, in \cite{Ishaque2013} the following result is proved. Given a plane perfect matching $M_1$ in $\Kn{(P)}$, we can always find another plane perfect matching $M_2$ that is disjoint from $M_1$ and for which $M_1\cup M_2$ is plane. In the following lemma we show that there exist point sets which contain at most two non-crossing edge-disjoint plane matchings.
\begin{lemma}
 \label{mp:2-non-crossing}
For a set $P$ of $n\ge 4$ points in convex position in the plane, with $n$ even, at most two pairwise non-crossing plane matchings can be packed into $\Kn{(P)}$.
\end{lemma}
\begin{proof}
 The proof is by contradiction. Consider three pairwise non-crossing plane matchings $M_1$, $M_2$, $M_3$ in $\Kn{(P)}$. Let $G$ be the induced subgraph of $\Kn{(P)}$ by $M_1\cup M_2\cup M_3$. Note that $G$ is a $3$-regular plane graph. Moreover, $G$ is an outerplanar graph. It is well known that every outerplanar graph has a vertex of degree at most 2. This contradicts that every vertex in $G$ has degree 3.
\end{proof}

We conclude this section with the following theorem.
\begin{theorem}
For a set $P$ of $n\ge 4$ points in general position in the plane, with $n$ even, at least two and at most five pairwise non-crossing plane matchings can be packed into $\Kn{(P)}$. These bounds are tight.
\end{theorem}

\section{Matching Removal Persistency}
\label{mp:persistency-section}
In this section we define the matching persistency of a graph. A graph $G$ is {\em matching persistent} if by removing any perfect matching $M$ from $G$, the resulting graph, $G-M$, has a perfect matching. We define the {\em matching persistency} of $G$, denoted by $mp(G)$, as the size of the smallest set $\mathcal{M}$ of edge-disjoint perfect matchings that can be removed from $G$ such that $G-\mathcal{M}$ does not have any perfect matching. 
In other words, if $mp(G)=k$, then

\begin{enumerate}
 \item by removing an arbitrary set of $k-1$ edge-disjoint perfect matchings from $G$, the resulting graph still contains a perfect matching, and 
\item there exists a set of $k$ edge-disjoint perfect matchings such that by removing these matchings from $G$, the resulting graph does not have any perfect matching.
\end{enumerate}

In particular, $G$ is matching persistent iff $mp(G)\ge 2$.

\begin{figure}[htb]
  \centering
\setlength{\tabcolsep}{0in}
  $\begin{tabular}{cc}
\multicolumn{1}{m{.5\columnwidth}}{\centering\includegraphics[width=.33\columnwidth]{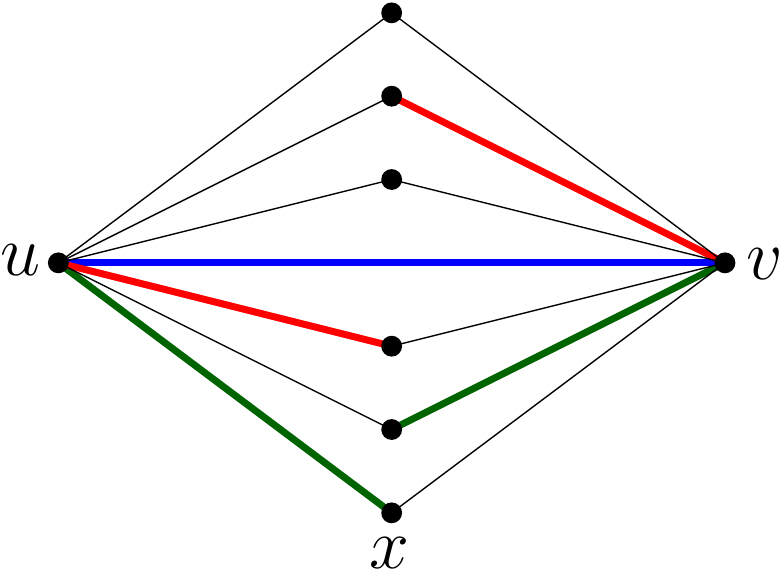}}&
\multicolumn{1}{m{.5\columnwidth}}{\centering\includegraphics[width=.15\columnwidth]{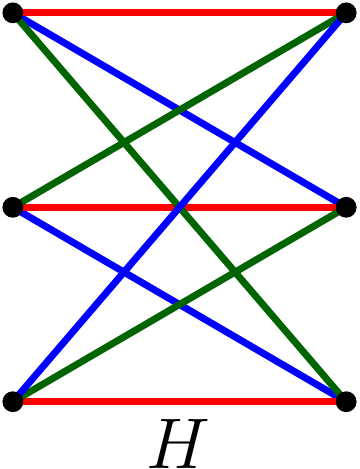}}\\
(a) & (b)
\end{tabular}$
  \caption{(a) By removing any matching (red, green, or blue) from $K_n$, at most two paths between $u$ and $v$ disappear. (b) The edges of $H$ are partitioned to $\frac{n}{2}$ perfect matchings, thus, $K_n-H$ does not have any perfect matching.}
\label{mp:path-matching}
\end{figure}

\begin{lemma}
\label{mp:regularity-connectivity}
Let $K_n$ be a complete graph with $n$ vertices, where $n$ is even, and let $\mathcal{M}$ be a set of $k$ edge-disjoint perfect matchings in $K_n$. Then, $K_n-\mathcal{M}$ is an $(n-1-k)$-regular graph which is $(n-1-2k)$-connected.
\end{lemma}
\begin{proof}
The regularity is trivial, because $K_n$ is $(n-1)$-regular and every vertex has degree $k$ in $\mathcal{M}$, thus, $K_n-\mathcal{M}$ is an $(n-1-k)$-regular graph. Now we prove the connectivity. Consider two vertices $u$ and $v$ in $V(K_n)$. There are $n-1$ many edge-disjoint paths between $u$ and $v$ in $K_n$: $n-2$ many paths of length two of the form $(u,x,v)$, where $x\in V(K_n)-\{u,v\}$ and a path $(u,v)$ of length one; see Figure~\ref{mp:path-matching}(a). By removing any matching in $\mathcal{M}$ from $K_n$, at most two paths disappear, because $u$ and $v$ have degree one in each matching. Thus in $G-\mathcal{M}$, there are $(n-1-2k)$ many edge-disjoint paths between $u$ and $v$, which implies that $K_n-\mathcal{M}$ is $(n-1-2k)$-connected.
\end{proof}

\begin{lemma}
\label{mp:mp-thr}
Let $n$ to be an even number. Then, $mp(K_n)\ge\frac{n}{2}$.
\end{lemma}
\begin{proof}
We show that by removing any set $\mathcal{M}$ of $k$ edge-disjoint perfect matchings from $K_n$, where $0\le k<\frac{n}{2}$, the resulting graph still has a perfect matching. By Lemma~\ref{mp:regularity-connectivity}, the graph $K_n-\mathcal{M}$ is an $(n-1-k)$-regular graph which is $(n-1-2k)$-connected. Since $k<\frac{n}{2}$, $K_n-\mathcal{M}$ is a connected graph and the degree of each vertex is at least $\frac{n}{2}$. Thus, by a result of 
Dirac~\cite{Dirac1952}, $K_n-\mathcal{M}$ has a Hamiltonian cycle and consequently a perfect matching. Therefore, by removing $k$ arbitrary perfect matchings from $K_n$, where $k<\frac{n}{2}$, the resulting graph still has a perfect matching, which proves the claim.

\end{proof}

\begin{lemma}
\label{mp:mp-thr2}
If $n\equiv 0 \mod 4$, then $mp(K_n)\ge\frac{n}{2} + 1$.
\end{lemma}
\begin{proof}
By Lemma~\ref{mp:mp-thr}, $mp(K_n)\ge\frac{n}{2}$. Let $\mathcal{M}$ be any set of $\frac{n}{2}$ edge-disjoint perfect matchings in $K_n$. We will show that $K_n-\mathcal{M}$ contains a perfect matching. If $K_n-\mathcal{M}$ contains a Hamiltonian cycle then it has a perfect matching and we are done. Assume $K_n-\mathcal{M}$ does not contain any Hamiltonian cycle, while it is a $(\frac{n}{2}-1)$-regular graph. A result of Cranston and O~\cite{Cranston2013} implies that $K_n-\mathcal{M}$ is disconnected. In order for $K_n-\mathcal{M}$ to be $(\frac{n}{2}-1)$-regular each component has to
have at least $\frac{n}{2}$ vertices. Thus, $K_n-\mathcal{M}$ consists of two disjoint copies of $K_{\frac{n}{2}}$. Each of these components has a Hamiltonian cycle, and hence a perfect matching. Therefore, the union of these two components has a perfect matching.
\end{proof}

\begin{theorem}
\label{mp:bipartite-matchings-lemma}
 If $n\equiv 2 \mod 4$, then $mp(K_n)=\frac{n}{2}$.
\end{theorem}
\begin{proof}
By Lemma~\ref{mp:mp-thr}, $mp(K_n)\ge\frac{n}{2}$. In order to complete the proof, we show that $mp(K_n)\le\frac{n}{2}$.
Let $H=K_{\frac{n}{2},\frac{n}{2}}$ be a complete bipartite subgraph of $K_n$. Note that $\frac{n}{2}$ is an odd number and $H$ is an $\frac{n}{2}$-regular graph.
According to Hall's marriage theorem \cite{Hall1935}, for $k\ge 1$, every $k$-regular bipartite graph contains a perfect matching \cite{Harary1991}. Since by the iterative removal of perfect matchings  from $H$ the resulting graph is still regular, the edges of $H$ can be partitioned into $\frac{n}{2}$ perfect matchings; see Figure~\ref{mp:path-matching}(a). It is obvious that $K_n-H$ consists of two connected components of odd size. Thus, by removing the $\frac{n}{2}$ matchings in $H$, the resulting graph, $K_n-H$, does not have any perfect matching. This proves the claim.
\end{proof}

\begin{theorem}
\label{mp:bipartite-matchings-lemma2}
 If $n\equiv 0 \mod 4$, then $mp(K_n)=\frac{n}{2}+1$.
\end{theorem}
\begin{proof}
By Lemma~\ref{mp:mp-thr2}, $mp(K_n)\ge\frac{n}{2}+1$. In order to complete the proof, we show that $mp(K_n)\le\frac{n}{2}+1$. 
Assume $n=4k$. Let $A = \{a_1,\dots,a_{2k-1}\}$ and $B = \{b_1,\dots,b_{2k+1}\}$ be a partition of vertices of $K_n$. Let $M_i$ be a matching consisting of the edges $b_ib_{i+1}$ and $a_jb_{j+i+1}$, where $+$ is modulo $2k +1$ and $j$ runs from $1$ to $2k-1$. It is easy to see that $M_1, \dots,M_{2k+1}$ are disjoint perfect matchings, and after removing them we have a complete graph on $A$ and a graph on $B$, which are disjoint. Since each of $A$ and $B$ has an odd number of points, there is no more perfect matching.
This proves the claim.
\end{proof}

In the rest of this section we consider plane matching removal from geometric graphs.

Let $P$ be a set of $n$ points in general position in the plane, with $n$ even. Given a geometric graph $G$ on $P$, we say that $G$ is {\em plane matching persistent} if by removing any plane perfect matching $M$ from $G$, the resulting graph, $G-M$, has a plane perfect matching. We define the {\em plane matching persistency} of $G$, denoted by $pmp(G)$ as the size of the smallest set $\mathcal{M}$ of edge-disjoint plane perfect matchings that can be removed from $G$ such that $G-\mathcal{M}$ does not have any plane perfect matching. In particular, $G$ is plane matching persistent iff $pmp(G)\ge 2$.

Aichholzer et al.~\cite{Aichholzer2010-edge-removal} and Perles (see~\cite{Keller2012}) showed that by removing any set of at most $\frac{n}{2}-1$ edges from $\Kn{(P)}$, the resulting graph has a plane perfect matching. This bound is tight~\cite{Aichholzer2010-edge-removal}; that is, there exists a point set $P$ such that by removing a set $H$ of $\frac{n}{2}$ edges from $\Kn{(P)}$ the resulting graph does not have any plane perfect matching. In the examples provided by~\cite{Aichholzer2010-edge-removal}, the $\frac{n}{2}$ edges in $H$ form a connected component which has $\frac{n}{2}+1$ vertices. 

Thus, one may think if the removed edges are disjoint, it may be possible to remove more than $\frac{n}{2}-1$ edges while the resulting graph has a plane perfect matching.
In the following lemma we show that by removing any plane perfect matching, i.e., a set of $\frac{n}{2}$ disjoint edges, from $\Kn{(P)}$, the resulting graph still has a perfect matching.

\begin{lemma}
\label{mp:pmp2-lemma}
Let $P$ be a set of $n$ points in general position in the plane with $n$ even, then $pmp(\Kn{(P)})\ge 2$.
\end{lemma}
\begin{proof}
 Let $M$ be any plane perfect matching in $\Kn{(P)}$. Assign $\frac{n}{2}$ distinct colors to the points in $P$ such that both endpoints of every edge in $M$ have the same color. By Theorem~\ref{mp:Aichholzer}, $P$ has a plane colored matching, say $M'$. Since both endpoints of every edge in $M$ have the same color while the endpoints of every edge in $M'$ have distinct colors, $M$ and $M'$ are edge-disjoint. Therefore, by removing any plane perfect matching from $\Kn{(P)}$, the resulting graph still has a plane perfect matching, which implies that $pmp(\Kn{(P)})\ge 2$.
\end{proof}

\begin{theorem}
For a set $P$ of $n\ge 4$ points in convex position in the plane with $n$ even, $pmp(\Kn{(P)})=2$.
\end{theorem}
\begin{proof}
By Lemma~\ref{mp:pmp2-lemma}, $pmp(\Kn{(P)})\ge2$. In order to prove the theorem, we need to show that $pmp(\Kn{(P)})\le2$. Let $M_1$ and $M_2$ be two edge-disjoint plane matchings obtained from \CH{P}. By Lemma~\ref{mp:two-convex-edges}, any plane perfect matching in $\Kn{(P)}$ contains at least two edges of \CH{P}, while $\Kn{(P)}-\{M_1\cup M_2\}$ does not have convex hull edges, and hence does not have any plane perfect matching. Therefore, $pmp(\Kn{(P)})\le2$. 
\end{proof}

\begin{observation}
\label{mp:even-cycle-obs}
The union of two edge-disjoint perfect matchings in any graph is a set of even cycles.
\end{observation}

\begin{lemma}
There exists a point set $P$ in general position such that $pmp(K(P))\ge 3$.
\end{lemma}
\begin{proof}
We prove this lemma by providing an example. Figure~\ref{mp:pmp3-fig}(a) shows a set $P=\{a_1,\dots,a_n,\allowbreak b_1,\dots,b_n,\allowbreak c_1,\dots,c_n\}$ of $3n$ points in general position, where $n$ is an even number. In order to prove that $pmp(K(P))\ge 3$, we show that by removing any two edge-disjoint plane matchings from $K(P)$, the resulting graph still has a plane perfect matching. Let $M_1$ and $M_2$ be any two plane perfect matchings in $K(P)$. Let $G$ be the subgraph of $K(P)$ induced by the edges in $M_1\cup M_2$. Note that $G$ is a 2-regular graph and by Observation~\ref{mp:even-cycle-obs} does not contain any odd cycle. For each $1\le i\le n$, let $t_i$ be the triangle which is defined by the three points $a_i$, $b_i$, and $c_i$. Let $\mathcal{T}$ be the set of these $n$ (nested) triangles. Since $G$ does not have any odd cycle, for each $t_i\in\mathcal{T}$, at least one edge of $t_i$ is not in $G$. Let $M_3$ be the matching containing an edge $e_i$ from each $t_i\in \mathcal{T}$ such that $e_i\notin G$. See Figure~\ref{mp:pmp3-fig}(b). Now we describe how to complete $M_3$, i.e., complete it to a perfect matching. Partition the triangles in $\mathcal{T}$ into $\frac{n}{2}$ pairs of consecutive triangles. For each pair $(t_i,t_{i+1})$ of consecutive triangles we complete $M_3$ locally\textemdash on $a_i,b_i,c_i,a_{i+1},b_{i+1},c_{i+1}$\textemdash in the following way. Let $t_i=(a_i,b_i,c_i)$ and $t_{i+1}=(a_{i+1},b_{i+1},c_{i+1})$. See Figure~\ref{mp:pmp3-fig}(c). W.l.o.g. assume that $M_3$ contains $a_ib_i$ and $a_{i+1}c_{i+1}$, that is $a_ib_i\notin G$ and $a_{i+1}c_{i+1}\notin G$. If $c_ib_{i+1}\notin G$, then we complete $M_3$ by adding $c_ib_{i+1}$. If $c_ib_{i+1}\in G$, then $a_{i+1}b_{i+1}\notin G$ or $c_{i+1}b_{i+1}\notin G$ because $b_{i+1}$ has degree two in $G$. W.l.o.g. assume that $a_{i+1}b_{i+1}\notin G$. Then we modify $M_3$ by removing $a_{i+1}c_{i+1}$ and adding $a_{i+1}b_{i+1}$. Now, if $c_ic_{i+1}\notin G$, then we complete $M_3$ by adding $c_ic_{i+1}$. If $c_ic_{i+1}\in G$, then by Observation~\ref{mp:even-cycle-obs}, $b_{i+1}c_{i+1}\notin G$. We modify $M_3$ by removing $a_{i+1}b_{i+1}$ and adding $b_{i+1}c_{i+1}$. At this point, since $c_ib_{i+1}$ and $c_ic_{i+1}$ are in $G$, $c_ia_{i+1}\notin G$ and we complete $M_3$ by adding $c_ia_{i+1}$.
\end{proof}
\begin{figure}[htb]
  \centering
\setlength{\tabcolsep}{0in}
  $\begin{tabular}{ccc}
\multicolumn{1}{m{.33\columnwidth}}{\centering\includegraphics[width=.31\columnwidth]{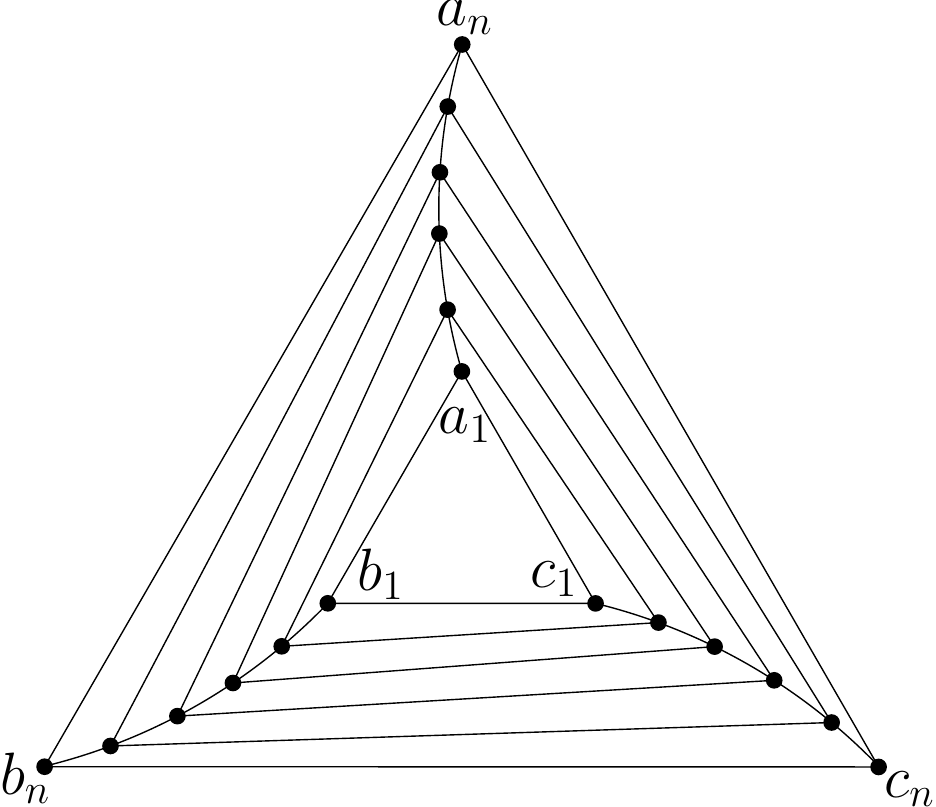}}&
\multicolumn{1}{m{.33\columnwidth}}{\centering\includegraphics[width=.31\columnwidth]{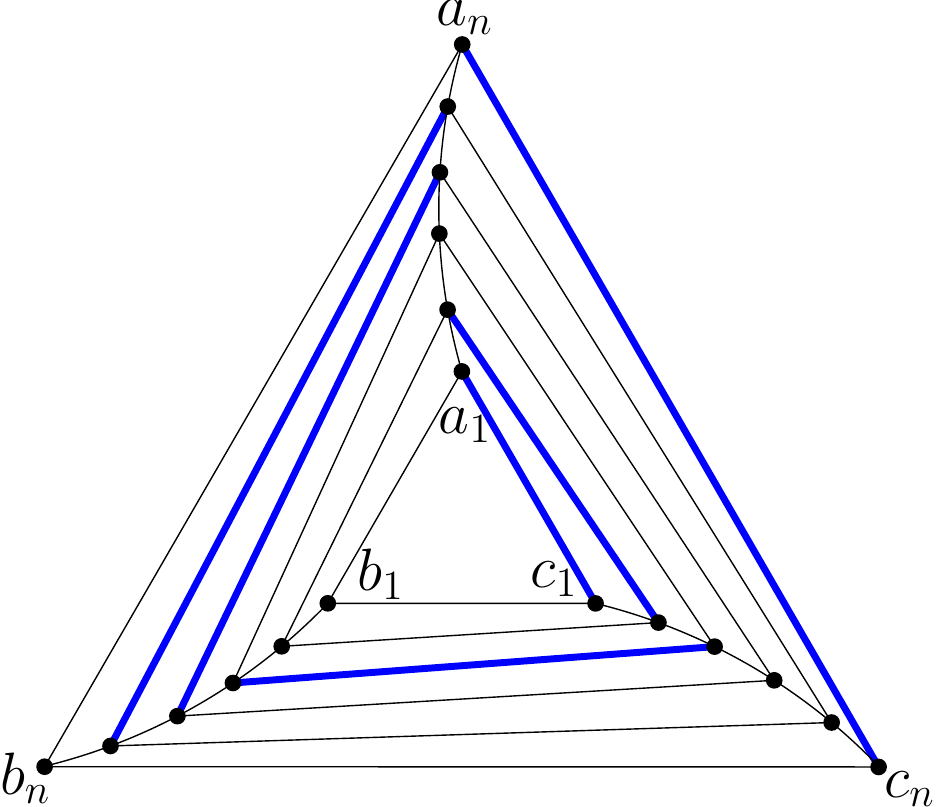}}
&\multicolumn{1}{m{.33\columnwidth}}{\centering\includegraphics[width=.31\columnwidth]{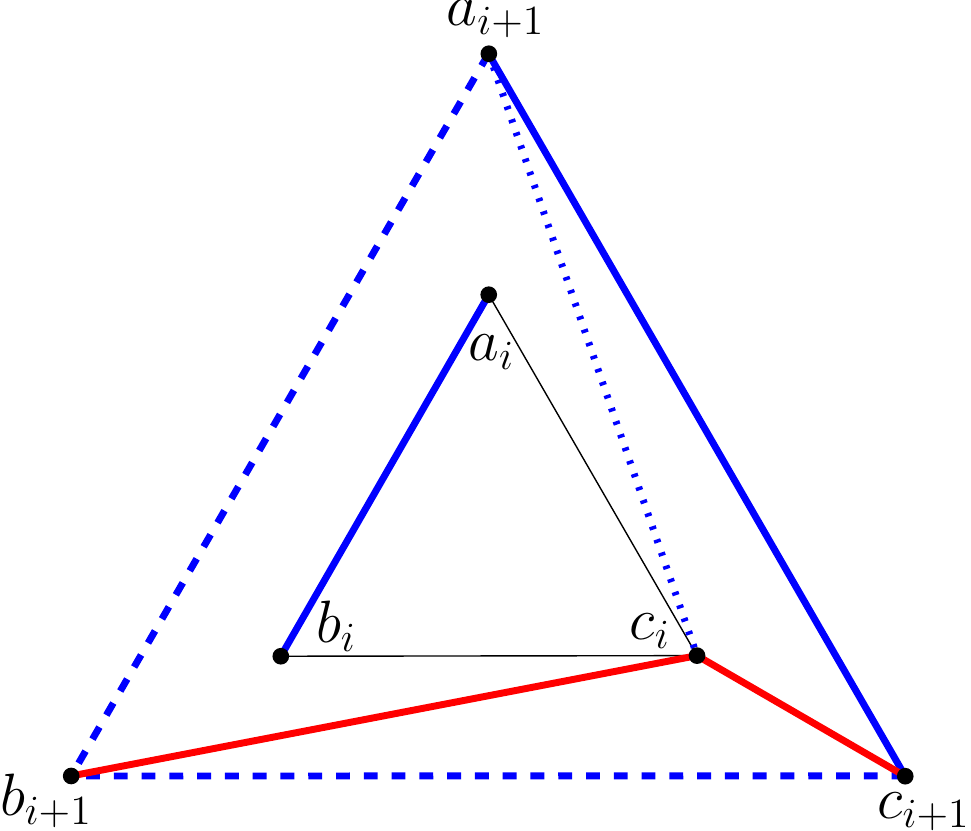}} \\
(a) & (b)&(c)
\end{tabular}$
\caption{(a) Set $P$ of $3n$ points in general position. (b) $M_3$ contains one edge from each triangle. (c) Locally converting $M_3$ to a perfect matching, for $t_i$ and $t_{i+1}$.}
\label{mp:pmp3-fig}
\end{figure}

\section{Conclusions}
\label{mp:conclusion}
We considered the problem of packing edge-disjoint plane perfect matchings into a set $P$ of $n$ points in the plane. If $P$ is in general position, we showed how to pack $\lfloor\log_2{n}\rfloor-1$ matchings. We also looked at some special cases and variants of this problem.
We believe that the number of such matchings is linear. A natural open problem is to improve either the provided lower bound or the trivial upper bound of $n-1$, where $n>6$. Another problem is to provide point sets with large plane matching persistency.

\bibliographystyle{abbrv}
\bibliography{../thesis}

%% file: frontmatter/Colophon.tex
\pagestyle{empty}

\hfill

\vfill

\pdfbookmark[0]{Colophon}{colophon}
\section*{Colophon}
This document was typeset using the typographical look-and-feel \texttt{classicthesis} developed by Andr\'e Miede. 
The style was inspired by Robert Bringhurst's seminal book on typography ``\emph{The Elements of Typographic Style}''. 
\texttt{classicthesis} is available for both \LaTeX\ and \mLyX: 
\begin{center}
\url{http://code.google.com/p/classicthesis/}
\end{center}
Happy users of \texttt{classicthesis} usually send a real postcard to the author, a collection of postcards received so far is featured here: 
\begin{center}
\url{http://postcards.miede.de/}
\end{center}
 
\bigskip

\noindent\finalVersionString